\documentclass[11pt]{article}

\usepackage{amsmath}
\usepackage{epsfig, amssymb, amsfonts}
\usepackage{amsthm}
\usepackage{amstext}
\usepackage{amsopn}
\usepackage{enumerate}
\usepackage{mathrsfs}
\usepackage{subfigure}
\usepackage{graphicx}
\usepackage[left=2.3cm,top=2cm,right=2.3cm,bottom=2cm, nohead,foot=1cm]{geometry}
\numberwithin{equation}{section}

\newtheorem{theorem}{Theorem}[section]
\newtheorem{lemma}[theorem]{Lemma}
\newtheorem{assumptions}[theorem]{List of rate assumptions}
\newtheorem{conjecture}[theorem]{Conjecture}
\newtheorem{proposition}[theorem]{Proposition}

\newcommand{\R}{{\mathbb R}}

\newcommand{\calA}{{\mathcal{A}}}
\newcommand{\calE}{{\mathcal{E}}}
\newcommand{\calF}{{\mathcal{F}}}
\newcommand{\calH}{{\mathcal{H}}}

\newcommand{\calS}{{\mathcal{S}}}

\newcommand{\Z}{{\mathbb Z}}
\newcommand{\N}{{\mathbb N}}
\newcommand{\Tr}{{\textup{Tr}}}

\newcommand{\C}{{\mathbb C}}

\title{Suppressed dispersion for a randomly kicked quantum particle in a Dirac comb}

\author{\textbf{Jeremy Thane Clark}\vspace{.2cm}\\ jclark@mappi.helsinki.fi \\ University of Helsinki, Department of Mathematics\\  Helsinki 00014, Finland }

\begin{document}
\maketitle

\begin{abstract}

 I study a model for a massive one-dimensional particle in a singular  periodic potential that is receiving kicks from a gas.  The model is described by a Lindblad equation in which the Hamiltonian is a Schr\"odinger operator with a periodic $\delta$-potential and the noise has a frictionless form arising in a Brownian limit.  I prove that  an emergent Markov process in a semi-classical limit governs the momentum distribution in the extended-zone scheme.  The main result is a central limit theorem for a time integral of the momentum process, which is closely related to the particle's position.  When normalized by $t^{\frac{5}{4}}$, the integral process converges to a time-changed Brownian motion whose rate depends on the momentum process.  The scaling $t^{\frac{5}{4}}$ contrasts with  $t^{\frac{3}{2}}$, which would be expected for the case of a smooth periodic potential or for a comparable classical process.  The difference is a wave effect driven by Bragg reflections that occur when the particle's momentum is kicked near the half-spaced reciprocal lattice.

\end{abstract}

\section{Introduction}

Mathematical models of a quantum particle in a periodic environment have been used to describe an electron in a metal and, more recently, an atom in an optical lattice.  One topic of mathematical and physical interest is the transport behavior for the particle in the periodic environment.  The periodic situation stands in contrast with the random or quasi-periodic situation,  which may exhibit Anderson localization~\cite{Anderson}.  Another important topic is the study of the motion of the particle when acted upon by a static force. Zener predicted an electron in a metal would exhibit some periodic motion when a constant force was applied~\cite{Zener}.  This behavior, called Bloch oscillations, is related to Bragg scattering and has been  observed experimentally in conductor superlattices~\cite{Feldmann}.  More recently, atoms in optical lattices have provided an analogous setting in which it is possible to measure Bloch oscillations with fewer noise effects~\cite{Ben,AndersonJr}.         

The current article studies a suppressed dispersion effect that, like Bloch oscillations, is generated by a combination of outside forcing (in this case from a noise) and Bragg scattering in a periodic potential.    My model concerns a one-dimensional massive quantum particle in a  periodic singular potential that receives random momentum kicks (e.g. from a gas of light particles).  The massive particle effectively does not ``feel" the potential except for infrequent instances when its momentum is kicked near an element of the half-spaced reciprocal lattice of the potential.  Near the lattice values, the particle's momentum has a chance of being  reflected, and these reflections in momentum occur often enough to inhibit the motion of the particle.  I imagine the model to describe an atom in a very singular one-dimensional optical potential.  In the physics literature, the article~\cite{Friedman} discusses  Bragg reflections of atoms from a weak  a one-dimensional optical potentials. The articles~\cite{Birkl,Kunze} report the experimental observation of Bragg scattering in atoms with lower kinetic energy through optical potentials.

My mathematical starting point for modeling the particle is a quantum Markovian dynamics generated by a Lindblad equation.  In the following section, I introduce the dynamics, state the main theorems, discuss some background for the model, and  make  conjectures for a similar model.  Section~\ref{SecProofOutline} contains an outline for the proof of the central limit theorem that is the main mathematical result of this article. Section~\ref{SecQuantum} contains a proof that the probability density of  the extended-zone scheme momentum behaves approximately as  an autonomous Markov process when the Hamiltonian dynamics operates on a faster scale than the noise.  Section~\ref{SecTransition} connects basic facts from the original quantum model to the limiting Markovian dynamics for the momentum process.  Section~\ref{SecClassical} contains the details for the proof sketched in Sect.~\ref{SecProofOutline}.  I show  that a time integral of the momentum process, when properly rescaled, converges in distribution to a variable diffusion process whose diffusion rate depends on the absolute value of the momentum.

\section{Results and discussion}

\subsection{The model and statement of the main results}

Let $\mathcal{B}_{1}\big(L^{2}(\R)\big)$ be the space of trace class operators over the Hilbert space $L^{2}(\R)$.  I begin with a quantum Markovian dynamics in which the state of the particle, as expressed by a density matrix  $\rho_{\lambda,t}\in \mathcal{B}_{1}\big(L^{2}(\R)\big)$, evolves according to a Lindblad equation  
\begin{align}\label{TheModel}
\frac{d}{dt}\rho_{\lambda,t}= -\frac{\textup{i}}{\lambda} \big[P^{2}+V,\rho_{\lambda,t}\big]+\Psi(\rho_{\lambda,t})-\frac{1}{2}\{\Psi^{*}(I),\rho_{\lambda,t} \}
 \end{align}
from the initial state $\rho_{\lambda,0}=\rho$.   In this equation,  $P=-\textup{i}\frac{d}{dx}$ is the momentum operator, $V$ is a periodic $\delta$-potential (i.e. Dirac comb potential) with strength $\alpha>0$ and period $2\pi$, and $\Psi:\mathcal{B}_{1}\big(L^{2}(\R)\big)$ is a completely positive map  describing the noise acting on the system and having the form
\begin{align}\label{TheNoise}
 \Psi(\rho)= \int_{\R} dv\,j(v)\,e^{\textup{i}vX}\, \rho\, e^{-\textup{i}vX},
 \end{align}
where  $ \rho\in \mathcal{B}_{1}\big(L^{2}(\R)\big)$, $X$ is the position operator, and $j(v)\in L^{1}(\R)$ is the rate-density for momentum kicks of size $v$.  I will assume the rates satisfy $j(v)=j(-v)$ and $\sigma=\int_{\R} dv\, j(v)\,v^{2}<\infty$.  In~(\ref{TheModel}), $\Psi^{*}(I)$ is the adjoint map $\Psi^{*}$ evaluated for the identity operator $I$ on $L^{2}(\R)$, and it happens that $\Psi^{*}(I)= \mathcal{R}\, I  $  for $\mathcal{R}= \int_{\R} dv\,j(v)$ in my case.    

Equation~(\ref{TheModel}) describes a quantum particle in dimension one evolving in a potential $V$ and receiving random momentum kicks $v$ with rate-density $j(v)$.  The noise is effectively frictionless, since intuitively, the rate of momentum kicks does not depend on the current momentum of the particle.  This excludes the possibility of energy relaxation in the model, and there is a linear rate of growth for the mean energy of the particle:   
\begin{align}\label{MeanEnergy}
\Tr[\rho_{\lambda,t}(P^{2}+V)]=\Tr[\rho(P^{2}+V)]+ \sigma\,t.
\end{align}

By Bloch theory, the Hamiltonian $H=P^{2}+V$ has continuous spectrum and decomposes through a fiber decomposition of the Hilbert space over the Brillouin zone $\phi \in [-\frac{1}{2},\frac{1}{2})$ as
\begin{align*}
L^{2}(\R)=\int_{[-\frac{1}{2},\frac{1}{2})}^{\oplus}d\phi \mathcal{H}_{\phi},\hspace{2cm} H_{\phi}= H|_{ \mathcal{H}_{\phi}},
 \end{align*}
where the Hilbert spaces $ \mathcal{H}_{\phi}$ are canonically identified with $L^{2}\big( [-\pi,\pi)   \big) $,  and the restriction of  $H$ to the $\phi$-fiber  is a self-adjoint operator $H_{\phi}$.  The operators $H_{\phi}$ have a complete set of eigenvectors $\psi_{n,\phi}$, $n\in \N$ with eigenvalues $E_{n,\phi}$ satisfying
$$ 0\leq  E_{n,\phi}\leq   E_{n+1,\phi},   \hspace{2cm}  E_{n,\phi}\longrightarrow \infty \hspace{.1cm}\text{ as }\hspace{.1cm}n\longrightarrow \infty .     $$ 
Through the extended-zone scheme, the eigenvectors $\psi_{n,\phi}$ can be associated with a collection of eigenkets $|k\rangle_{\scriptscriptstyle{Q} }$ parameterized by $k\in \R$ such that   
$$H=\int_{\R}dkE(k)|k\rangle_{\scriptscriptstyle{Q} }\,{ }_{\scriptscriptstyle{Q} } \langle  k|, $$
where the dispersion relation has the form $E(k)=\mathbf{q}^{2}(k)$ for the anti-symmetric, increasing function $\mathbf{q}:\R\rightarrow \R$ satisfying the Kr\"onig-Penney relation
\begin{align}\label{Energies}
\cos(2\pi k)=\cos\big(2\pi \mathbf{q}(k) \big)+\frac{\alpha}{2\mathbf{q}(k) }\sin\big(2\pi \mathbf{q}(k)\big)  
\end{align}
for $k\in \R-\frac{1}{2}\Z$ and $\mathbf{q}(\frac{n}{2})=\frac{n}{2}$ for $n\in \Z$ (see~(\ref{Bloch}) for the corresponding Bloch functions in the position representation).    The Bloch structure and my conventions are discussed  in Appendix~\ref{AppendixFiber}.  The dispersion relation essentially has the form 
$E(k)\approx |k|^{2}+\frac{\alpha}{2\pi} $ for $|k|\gg 1$ except for values of $k$ in small neighborhoods around the lattice $\frac{1}{2}\Z$, where $E(k)$ makes jumps $E(\frac{n}{2}+)-E(\frac{n}{2}-)\neq 0$.  The kets $| k\rangle_{\scriptscriptstyle{Q} }$ also have discontinuities at values $k\in\frac{1}{2}\Z$:  $\lim_{\epsilon\rightarrow 0}| k-\epsilon\rangle_{\scriptscriptstyle{Q} }\neq | k+\epsilon\rangle_{\scriptscriptstyle{Q} }$.
 \begin{figure}[htb]
        \center{ \hspace{1cm}
        \includegraphics[scale=.7]{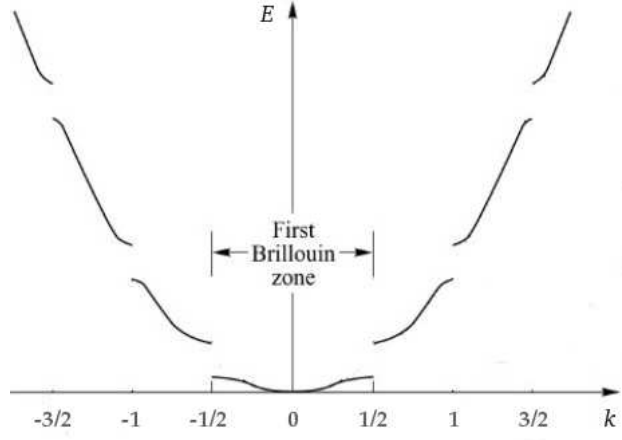} }
        \caption{Qualitative plot of the dispersion relation in the extended-zone scheme.}
 \end{figure}


The first result is concerned with the limiting behavior as $\lambda\rightarrow 0$ of the diagonal distributions  in the extended-zone scheme representation: $ { }_{\scriptscriptstyle{Q}}\langle k|\rho_{\lambda,t} | k \rangle_{\scriptscriptstyle{Q}}=D_{\lambda,t}(k)$.  I show, in a sense defined below, that $D_{\lambda, t}$ converges for small $\lambda$ to the solution $\mathcal{D}_{t}$ of a classical Markov process
\begin{align}\label{Master}
\frac{d}{dt}\mathcal{D}_{t}(k)=L(\mathcal{D}_{t})(k)=\int_{\R}dk^{\prime}\Big(J(k,k^{\prime})\mathcal{D}_{t}(k^{\prime})-J(k^{\prime},k)\mathcal{D}_{t}(k)   \Big),  \quad \quad    \mathcal{D}_{0}(k)=\rho(k,k),
\end{align}
where the rates $J(k,k^{\prime})$ are determined by the rates $j(v)$ and values $\kappa_{v}(k,n)$ (defined below) through the formula
\begin{align}\label{KrissKrossMakesYou}
J(k,k^{\prime}):=\sum_{n\in \Z} j(k-k^{\prime}-n)|\kappa_{k-k^{\prime}-n}(k^{\prime},n)|^{2}.  
\end{align}
The values $\kappa_{v}(k,n)\in \C $ arise as coefficients in the formula  
 \begin{align}\label{DefKappa}
  \sum_{n\in \Z}\kappa_{v}(k,n)  |k+v+n\rangle_{\scriptscriptstyle{Q}}=  e^{\textup{i}vX} |k\rangle_{\scriptscriptstyle{Q}}, \quad \quad k,v\in \R.
  \end{align}
The fact that  $e^{\textup{i}vX} |k\rangle_{\scriptscriptstyle{Q}}$ is a combination of the $|k+v+n\rangle_{\scriptscriptstyle{Q}}$, $n\in Z$ is a consequence of the fiber decomposition.  By the unitarity of $e^{\textup{i}vX}$, the coefficients satisfy  $\sum_{n\in \Z}|\kappa_{v}(k,n)|^{2}=1$, and the process has a constant escape rate: $\mathcal{R}=\int_{\R}dk^{\prime}\,J(k^{\prime},k)$. 
\begin{theorem}[Freidlin-Wentzell/semi-classical limit]   \label{FWLimit}
Let $\rho_{\lambda, t}$ satisfy~(\ref{TheModel}), $\rho_{\lambda, t}(k,k)=D_{\lambda, t}(k)$, and $\mathcal{D}_{t}$ be the solution of~(\ref{Master}).   There exists a $C>0$ such that for all $t>0$     
 $$\|D_{\lambda, t}-\mathcal{D}_{t}\|_{1}\leq C\lambda(1+t) .$$
Also, for $t< \lambda$, the norm of the difference is bounded by a multiple of $t$. 

\end{theorem}

Let $K_{r}$ be the Markov process satisfying the master equation~(\ref{Master}) and define the integral functional $Y_{t}=\int_{0}^{t}dr K_{r} $.  My main result concerns the limiting distributional behavior for the processes $(\sigma t)^{-\frac{1}{2}}|K_{st}|$ and $t^{-\frac{5}{4}} Y_{st}$, $s\in [0,\,1]$ for $t\gg 1$.  I will make the following technical assumptions on $j(v)$:
\begin{assumptions}\label{Assumptions} \text{ }
There is a $\mu>0$ such that
\begin{enumerate} 

\item $  \int_{\R}dv\,j(v)\,e^{a\, |v|}<\mu$ for some $a>0$,

\item $\sup_{-\frac{1}{4}\leq \theta \leq \frac{1}{4}} \sum_{n\in \Z} j(\theta+\frac{n}{2})<\mu $,    

\item $\inf_{v\in [-1,1]   } j(v)\geq \mu^{-1}$.

\end{enumerate}

\end{assumptions}

  The theorem below states that the processes $|t^{-\frac{1}{2}}K_{st}|$ converge in law as $t\rightarrow \infty$ to the absolute value of a Brownian motion and 
  $t^{-\frac{5}{4}}Y_{st}$ converges to a time-changed Brownian motion whose rate of diffusion emerges as the limit law of  $\nu^{-1} |t^{-\frac{1}{2}}K_{st}|^{3}$ for $\nu:=\mathcal{R}\alpha $.  It is clear from the above statement  that   the process $t^{-\frac{1}{2}}K_{st}$ itself does not behave as a Brownian motion; otherwise the appropriate scaling for $Y_{st}$ would be $t^{\frac{3}{2}}$, and the limiting process would be differentiable rather than diffusive.  
 \begin{theorem}[Main result]\label{Main}
Assume  $j(v)$ satisfies List~\ref{Assumptions}.  Let $K_{t}$ be the Markov process whose probability densities $q_{t}$ satisfy~(\ref{Master}), $q_{0}$ have finite second moment, and $Y_{t}$ be the time integral of $K_{r}$ up to time $t$.   As $t\rightarrow \infty$, there is convergence in law with respect to the Skorokhod metric
$$\Big((\sigma t)^{-\frac{1}{2}}  |K_{st}|,\, \sigma^{-\frac{3}{4}}\nu^{\frac{1}{2}} t^{-\frac{5}{4}}Y_{st}     \Big)\stackrel{\frak{L}}{\Longrightarrow}\Big(|\mathbf{B}_{s}|,\, \int_{0}^{s}d \mathbf{B}^{\prime}_{r}\,|\mathbf{B}_{r}|^{\frac{3}{2}}\Big), \quad \quad s\in [0,\,1],
$$ 
where $\mathbf{B}_{s}$ and $\mathbf{B}_{s}^{\prime}$ are independent standard Brownian motions. 
 \end{theorem}

\subsection{Further discussion and background}\label{SubSecDisc}

\subsubsection{The Lindblad dynamics,  the noise,  and the Hamiltonian. }
Introductory material on Lindblad equations can be found in~\cite{Fannes}.   Some basic mathematical questions regarding existence and uniqueness of solutions to Lindblad equations with unbounded generators are not completely understood except for specific classes such as those for which the generator is translation covariant~\cite{Holevo}.  Because the Hamiltonian part of the Lindblad equation~(\ref{TheModel}) is unbounded,  the mathematical definition of a  solution to the Lindblad equation is less direct than the bounded case.  I discuss the rigorous definition for the dynamics and related technical issues in Appendix~\ref{AppendixSemigroup}.  In the case discussed here, these issues are not interesting or challenging, since the noise term $\Psi:\mathcal{B}_{1}\big(L^{2}(\R)  \big)$ is bounded and $\Psi^{*}(I)=\mathcal{R}I$ commutes with the Hamiltonian. 

   Consider a L\'evy process with density $j(v)$, and let $t_{1},\dots, t_{n}$ and $v_{1},\dots,v_{n}$ be the Poisson times and increments of the L\'evy process up to time $t$.  The state for the particle $\rho_{\lambda,t}=\Phi_{\lambda,t}(\rho)$ at time $t$ is equal to  
\begin{align}\label{Couture}
 \Phi_{\lambda,t}(\rho)=  \mathbb{E}\Big[  U_{\lambda,t}(\xi)\,   \rho \, U_{\lambda,t}^{*}(\xi) \Big], 
 \end{align}
where the expectation is with respect to the law of the L\'evy process, $\xi=(v_{1},t_{1};\, v_{2},t_{2};\dots )\in (\mathbb{R}\times \mathbb{R}_{+})^{\infty}$ is the full sequence of random events, and the unitary operator $U_{\lambda,t}(\xi):L^{2}(\R)$ is defined by the product
 \begin{align}\label{LaughableMan}
 U_{\lambda,t}(\xi):= e^{-\frac{\textup{i}(t-t_{n }) }{\lambda}  H }e^{\textup{i}v_{n }X}  \cdots  e^{-\frac{i(t_{2}-t_{1})}{\lambda} H} e^{\textup{i}v_{1}X}  e^{-\frac{\textup{i} t_{1}}{\lambda} H}.
 \end{align}
The construction of the maps $\Phi_{\lambda,t}$ thus only depends on the existence of the unitary groups $e^{-irH}$, $r\in \R$. The equation~(\ref{Couture}) implies the trace for the state is preserved such that $\Tr[ \rho_{\lambda,t}]=\Tr[\rho]=1$, since the expression on the right is a convex combination of unitary conjugations of $\rho$.

A noise of the type appearing in~(\ref{TheModel}) was originally introduced as a phenomenological model for the study of wave collapse in quantum mechanics~\cite{Ghirardi}.  It was later derived in~\cite{ESL} starting with a heuristic scattering analysis that was meant to model an interaction of a test particle with a gas in the limit that the test particle has much greater mass than  the gas particles.  This scattering analysis was clarified in~\cite{Sipe}, which yielded a minor correction by a unitless multiplicative factor in the final expression for the jump rates $j(v)$.   Also, the article~\cite{Hellmich} contains a mathematical derivation  for a noise of the form~(\ref{TheModel}) through a singular coupling limit of a simple system-reservoir Hamiltonian dynamics.  The  noise model has been discussed in relation to experimental frameworks in matter-wave optics~\cite{Alicki,Exper} and appears in other discussions of decoherence~\cite{Vacchini}.  See~\cite[Sec.7.1]{VaccHorn} for the connection of the noise with a quantum linear Boltzmann equation in a large mass limit.    
A similar frictionless noise with some spatial dependence recently appeared in~\cite{Kolovsky} to model the dampening of Bloch oscillations for an atom in an optical lattice.

It is clear from the mean energy growth~(\ref{MeanEnergy}) that  the model for the noise is transient in nature.   The classical analog of the noise map 
  $$\hspace{3.5cm}\rho\longrightarrow \Psi(\rho)-2^{-1}\{\rho,\Psi^{*}(I)\},\hspace{2cm}\rho \in \mathcal{B}_{1}\big(L^{2}(\R)\big),$$
 appearing in the Lindblad equation~(\ref{TheModel})  is given by the map
\begin{align}\label{NoiseAnalog}
\Upsilon(x,\,p)\longrightarrow \int_{\R}dv\,j(v)\big(\Upsilon(x,p-v)-\Upsilon(x,p)\big), 
\end{align}
 for  joint  position-momentum densities $\Upsilon\in L^{1}(\R^{2})$.  The association   of~(\ref{NoiseAnalog}) with the quantum noise can be justified by looking at the Wigner representation $W_{\rho}(x,p)$ of $\rho$.  Of course, most quantum noises are not so readily identifiable with classical analogs.  Equation~(\ref{NoiseAnalog}) describes the momentum as receiving random kicks $v$ with rate-density $j(v)$.    Based on this evidence, the momentum undergoes an unbiased random walk, and the ensuing stochastic acceleration explains the mean energy growth in~(\ref{MeanEnergy}).

 The Hamiltonian $H=P^{2}+\alpha\sum_{n\in\Z}\delta(X-2\pi n)$, $\alpha>0$ is defined as a particular self-adjoint extension of the symmetric operator $-\frac{d^{2}}{dx^{2}}$ with domain consisting of all functions $f\in\mathbf{H}^{2,2}(\R)\cap \{g\,|\,g(2\pi n)=0,\,n\in\mathbb{Z} \}$ having two weak derivatives and taking the value $0$ on the lattice $2\pi \Z$.  The domain of the self-adjoint extension is the  space of functions $f\in \mathbf{H}^{2,1}(\R)\cap \mathbf{H}^{2,2}(\R-2\pi \Z)$ that have one weak derivative in the domain $\R$ and two weak derivatives in $\R-2\pi\Z$, and that satisfy  
$$ \alpha f(2\pi n)= \frac{df}{dx}( 2\pi n+ )-\frac{df}{dx}(2\pi n-).   $$
The Dirac comb is a limiting case of the Kr\"onig-Penney model~\cite{Penney}, which is a periodic Schr\"odinger equation in which the potential has the form  $V(x)=\bar{V}\sum_{n\in\Z}1_{[2\pi n-a,\,2\pi n+a]}$ for $a<2\pi$.  The limit connecting them is  $\bar{V}\rightarrow \infty$ with $2\bar{V}a=\alpha$.  The Kr\"onig-Penney model has been used to model the transport of electrons through a crystal.  One computational advantage of these models is that there are closed equations determining the spectral values and the form of the eigenkets.  Both the periodic $\delta$-potential and the Kr\"onig-Penney model are discussed in~\cite{Solve}.  Some general theory regarding the structure of periodic Schr\"odinger equations can be found in~\cite{East,Reed,Asch}.  The articles ~\cite{Firsova,Cai,Scipio} contain recent results on the dispersion of wave packets evolving according a periodic Schr\"odinger equation with a smooth potential.         

\subsubsection{Bragg reflections   }\label{SecBragg}

For a high momentum quantum particle in a periodic potential, the dominant behavior is simply transmission over the potential.  However, there is a lattice $\frac{1}{2}\Z$ of wave frequencies around which the particle is likely to be reflected by the potential.  For my purpose, the difference between a smooth periodic potential and the Dirac comb is  the proximity a high momentum particle must have to a lattice momentum to experience reflections.  These zones are wider for the Dirac comb, and it is possible for the test particle to score accidental reflections in the process of colliding with the gas.  A more immediate indication of a contrast between the Dirac comb and a smooth period potential is in the size of the jumps in the dispersion relation $E(k)$ at momenta $k\in \frac{1}{2}\Z$ (as pictured in Figure~$1$).  Unlike smooth potentials, where the energy gaps vanish as $k\rightarrow \infty$, the gaps for the periodic $\delta$-potential approach the constant value $\frac{\alpha}{\pi}$.   
     
     Due to the spatial translation symmetry of the potential by $2\pi$, the kets $|k\rangle_{\scriptscriptstyle{Q}}$ can be written as discrete combinations of the momentum kets $|k+n\rangle$ for $n\in \Z$:
\begin{align}\label{FromStanToQuasi}
|k\rangle_{\scriptscriptstyle{Q}}=\sum_{n\in \Z}\eta(k,n)|k+n\rangle,  
\end{align}
where $\sum_{n\in \Z}|\eta(k,n)|^{2}=1$.  For $|k|\gg 1$,  the $n=0$ term dominates the sum except when $k$ is close to a lattice point $\frac{n}{2}\in \frac{1}{2}\Z$, in which case the $-n$ term becomes nonnegligible.  When $k\approx \frac{1}{2}n$, the value $k-n\approx -k$ is approximately a reflection of the momentum.  A high-momentum plane wave tuned near a lattice frequency will be driven by the Hamiltonian evolution to oscillate though a cycle of quantum superpositions between the original wave and the reflected wave.  These are called~\textit{Pendell\"osung oscillations}, and the corresponding wave velocities $E'(k)$ are small for  momenta $k$ close enough to an element in  $\frac{1}{2}\Z$ to  exhibit the Pendell\"osung oscillations (see also in Figure-$1$).  A plane wave with a high momentum sufficiently away from any lattice value will transmit nearly freely through the potential with velocity $E'(k)\approx 2k$.  I will refer to the regions near the lattice momenta in which nonnegligible reflection occurs as the~\textit{reflection bands} (see~\cite{Friedman}).  For the Dirac comb potential, the widths of the reflection bands scale as $\propto n^{-1}$ for  $\frac{n}{2}\in \frac{1}{2}\Z $ with $|n|\gg 1$.

\subsubsection{ The small $\lambda$ limit    }\label{SubSecLambda}

 The regime of $\lambda \ll 1$ in~(\ref{TheModel}) should be considered as a semi-classical regime in which phase oscillations generated by the Hamiltonian term occur on a faster time scale than the mean time between collisions with the gas.  As it appears in~(\ref{TheModel}), the parameter $\lambda$ takes the place of $2M\hbar$, where $M$ is the mass of the test particle.  Since the noise arises in a limit in which $M$ is large, the value $\hbar^{-1}$ must be even ``larger" to make $\lambda \ll 1$.   A more honest comparison of relevant scales requires that I include more physical parameters in the model such as the period and strength of the $\delta$-potential, and I  do this for a slightly richer noise model in Sect.~\ref{SecConj}. 

Without the noise, the dynamical evolution generated by the Hamiltonian  in the extended-zone scheme representation is 
\begin{align}\label{LabelItSomething}
  { }_{\scriptscriptstyle{Q}}\langle k'|   e^{-\frac{\textup{i}t}{\lambda}[H,\cdot]}(\rho)    |k\rangle_{\scriptscriptstyle{Q}}=   e^{-\frac{\textup{i}t}{\lambda}E(k')+\frac{\textup{i}t}{\lambda}E(k) }     { }_{\scriptscriptstyle{Q}}\langle k'|   \rho   |k\rangle_{\scriptscriptstyle{Q}}
\end{align}
for $\rho\in \mathcal{B}_{1}\big(L^{2}(\R)\big)$. The noise map $\Psi$ operates in the momentum representation   as
\begin{align}\label{ShyKid}
 \langle k'|   \Psi(\rho)    |k\rangle=  \int_{\R}dv j(v) \langle k'-v|   \rho  |k-v\rangle .
\end{align}
Since the ket $ |k\rangle_{\scriptscriptstyle{Q}}$ is a combination~(\ref{FromStanToQuasi}) of the kets $|k+n\rangle$ for $n\in \Z$,  the equations~(\ref{LabelItSomething}) and~(\ref{ShyKid}) imply that only the values  $  { }_{\scriptscriptstyle{Q}}\langle k'|   \rho   |k\rangle_{\scriptscriptstyle{Q}} $ with  $k'-k\in \Z$ interact dynamically with the diagonal values   $  { }_{\scriptscriptstyle{Q}}\langle k|   \rho   |k\rangle_{\scriptscriptstyle{Q}} $.  The relevant phase velocities are $\lambda^{-1}\big( E(k')-E(k)     \big)$ for $k'-k\in \Z$.   Due to the non-vanishing energy band gaps for the Dirac comb potential, there is a non-zero minimum $\lambda^{-1}\mathbf{s}_{\textup{min}} $ for the phase speeds, where
\begin{align}\label{PhaseSpeeds}
 \mathbf{s}_{\textup{min}} : =\inf_{ \substack{ k\in \R-2^{-1}\Z  \\ k'-k\in \Z  } } \big| E(k')-E(k)\big| >0.
\end{align}   
These phase oscillations thus have period proportional to $\lambda $, which is small compared to the mean time $\mathcal{R}^{-1}$ between momentum kicks.   This provides the mechanism through which the diagonals  $  { }_{\scriptscriptstyle{Q}}\langle k|   \rho_{\lambda, t}   |k\rangle_{\scriptscriptstyle{Q}} $ tend to an autonomous evolution as $\lambda\rightarrow 0$.    


 The $\lambda\rightarrow 0$ limit connecting the diagonal of $\rho_{\lambda,t}$ in the basis of kets diagonalizing $P^{2}+V$  to a Markovian dynamics  is analogous to a  Freidlin-Wentzell limit~\cite[Ch.8]{Wentzell} for a classical Hamiltonian flow perturbed by a weak noise. In the Freidlin-Wentzell limit, a Markovian dynamics emerges on the energy graph in a scaling limit combining a noise of strength $\lambda\ll 1$  and long time intervals $\propto \lambda^{-1}$.   The ``energy graph" is the collection of connected level curves for the Hamiltonian.     
 My dynamics fits the Freidlin-Wentzell limit description with a simple change of time variable $s:=\lambda\,t$:
 \begin{align*}
\frac{d}{dr}\tilde{\rho}_{\lambda,s}= -\textup{i}\big[H,\tilde{\rho}_{\lambda,s}\big]+\lambda\Big(\Psi(\tilde{\rho}_{\lambda,s})-\frac{1}{2}\{\Psi^{*}(I),\tilde{\rho}_{\lambda,s} \}\Big)\quad \text{for}\quad \tilde{\rho}_{\lambda,s}:= \rho_{\lambda,\frac{s}{\lambda} }. 
 \end{align*}
 The connected level curves of the Hamiltonian and the kets $ |k\rangle_{\scriptscriptstyle{Q}}$ have parallel roles in imprinting the Hamiltonian structure on the limiting Markovian dynamics.   A similar study for a time integral of a momentum-related quantity after a Freidlin-Wentzell limit has previously been performed for a classical model in~\cite{Hairer}.  This model  is a Markovian dynamics for a particle in a periodic potential with a white noise, and the authors have shown that the spatial diffusion and  Freidlin-Wentzell limits commute.

  Theorem~\ref{Main} is a statement only about the classical process $K_{t}$ and its time integral $Y_{t}=\int_{0}^{t}dr K_{r}$.  This result has no direct consequence for the original model, and in particular, I have not proven that the position distribution $d_{\lambda,t}(x):=\langle x|\rho_{\lambda,t} |x\rangle$ itself  exhibits  subdiffusion for $\lambda\ll 1$.    However, I believe the suppressed dispersion effect holds in the original model for any finite $\lambda$.   In other words, $d_{\lambda,t}(x)$ has spread $\propto t^{\frac{5}{4}}$ when the Dirac comb is present rather than the scaling $\propto t^{\frac{3}{2}}$ that holds  otherwise (see Appendix~\ref{AppendixNoComb}).   The effect depends on the occasional development of Hamiltonian-generated quantum superpositions between momenta in the reflection bands and their reflected counterparts.  As the particle is kicked out of a reflection band, the quantum superposition collapses into  a  classical superposition.

\subsubsection{Energy submartingales}

I chose the Dirac comb for this study as opposed to some other singular periodic potential because there are convenient closed expressions for the eigenkets $|k\rangle_{\scriptscriptstyle{Q}}$ in the position and momentum representations.  This is helpful for estimating the coefficients $\kappa_{v}(k,n)$ appearing in the definition~(\ref{KrissKrossMakesYou}) of the jump rates for the Markov process $K_{t}$.  Nevertheless, it is easier to perceive certain critical features of the classical process $K_{t}$  by taking a step back into the original quantum framework than through the closed formulas.  For instance, $E^{\frac{1}{2}}(K_{t})$ is a submartingale.  It follows that $E(K_{t})$ is also a submartingale, although it can be additionally learned from the quantum setting that the increasing part of its Doob-Meyer decomposition is $\sigma t$.  That in not surprising  given the mean energy grown~(\ref{MeanEnergy}).  The process   $E^{\frac{1}{2}}(K_{t})$ is useful  for showing Thm.~\ref{Main}, since $|K_{t}|\approx E^{\frac{1}{2}}(K_{t})$ when $|K_{t}|\gg 1$ by the approximately parabolic shape of the dispersion relation, and because the martingale structure allows $E^{\frac{1}{2}}(K_{t})$ to be treated  (see Thm.~\ref{SubMartCLT}).   An equally important application of the submartingale structure is to show that the  process $|K_{r}|$ spends most of the time interval $r\in [0,t]$ with large values on the order $t^{\frac{1}{2}}$ for which the Bragg scattering is dominated by only  transmitted and reflected waves.

The proof that  $E^{\frac{1}{2}}(K_{t})$ is a submartingale is based on  the Heisenberg representation of the equation~(\ref{Couture}).  The stochastic operator-valued process  
$$  H^{\frac{1}{2}}_{\lambda,t}(\xi):=  U_{\lambda,t}^{*}(\xi)\,    H^{\frac{1}{2}} \, U_{\lambda,t}(\xi)  $$   
is a positive submartingale in the sense that for all $f\in \textup{D}(H^{\frac{1}{2}})\subset L^{2}(\R)$,  the process $\langle f| H^{\frac{1}{2}}_{\lambda,t}(\xi) f\rangle $ is a positive submartingale.  Section~\ref{SecTransition} contains the analysis of the various relevant operator martingales.

\subsubsection{Heuristics for the limit theorem }\label{SecHeuristics}

The Kolmogorov equation~(\ref{Master}) determines a pseudo-Poisson process $K_{r}$ whose jump times are determined by a Poisson clock with rate $\mathcal{R}$. The jumps are a sum of two contributions: one coming directly from a particle collision, and another being a lattice-valued Bragg scattering.  More precisely, a jump $w$ from the starting momentum $k\in \R$ is equal to $w=v+n$, where the component $v\in\R$ has density $\mathcal{R}^{-1}j(v)$ and the component $n\in \Z$ has conditional  probabilities $|\kappa_{v}(k,n)|^{2}$ when given $v$ and $k$.  I  refer to these as the \textit{L\'evy} and \textit{lattice} components of the jump.  When $|k|\gg 1 $,  most of  the lattice jump probabilities $|\kappa_{v}(k,n)|^{2}$ are negligible except for $n=0$ and a few values  corresponding to  reflections in momentum.  The high-energy behavior is dominant, since by earlier discussion, $|K_{r}|$ will typically spend most of a time interval $r\in [0,t]$ with $|K_{r}|\propto t^{\frac{1}{2}}\gg 1 $. An idealized picture emerges in which the  process makes a L\'evy jump $v$ from $k$ with an additional optional jump $-\mathbf{n}\in \Z$ for   
  $$  \mathbf{n}=2(k+v-\theta)\quad \text{with} \quad \theta=k+v\,\textup{mod}\, \frac{1}{2},\quad \theta\in \big[-\frac{1}{4},\,\frac{1}{4}\big),            $$
which occurs with probability 
\begin{align}\label{ReflProb}
\textup{R}_{-}\big( 2^{-1}\theta \mathbf{n} \big):=\frac{\alpha^{2}}{8\pi^{2}} \frac{1}{(2^{-1}\theta\mathbf{n})^{2} +\frac{\alpha^{2}}{16\pi^{2}} }. 
\end{align}
If this extra  jump occurs, the resulting value $k+v-\mathbf{n}\approx -k-v$ is approximately a reflection of the value it would have had otherwise.  

The probability~(\ref{ReflProb}) of a momentum reflection decays when $|\theta|\gg |\mathbf{n}|^{-1}$.  Hence, when $|k+v|\gg 1$,  the value $k+v$ must land near a lattice value $\frac{1}{2}\Z$ in order to have a good chance of  reflection.  The integral 
\begin{align}\label{IntRefl}
2\int_{-\frac{1}{4}}^{\frac{1}{4}}d\theta\,\textup{R}_{-}\big( 2^{-1}\theta \mathbf{n} \big)\approx \alpha (\frac{|\mathbf{n}|}{2})^{-1}
\end{align}
 serves as an effective reflection probability  if the particle is dropped randomly in the cell centered around $\frac{\mathbf{n}}{2}$.  The factor of $2$ in~(\ref{IntRefl}) normalizes the integration.  This suggests the number of reflection times $\mathbf{N}_{r}$ is approximately a Poisson process with a rate depending on $K_{r}$ as  $\approx \mathcal{R}\alpha |K_{r}|^{-1}=\nu|K_{r}|^{-1}$ (since the L\'evy jumps occur with rate $\mathcal{R}$).  In other words, the Poisson rate of reflections is inversely proportional to the absolute value of the momentum.  Over a time interval where $|K_{r}|\propto t^{\frac{1}{2}}$,  the average time between reflections will be $\propto t^{\frac{1}{2}}$.  Hence, if $|K_{r}|$ behaves as the absolute value of a random walk, the number of reflections  over a time interval $[0,t]$ will be (at least) on the order of $t^{\frac{1}{2}}$, and  $t^{-\frac{1}{2}}K_{st}$  itself can not have a limiting distribution.   
However, it is reasonable to expect a limit theorem for the  process $Y_{r}$, since the sign-flipping of $K_{r}$ is smoothed by the time integration.  
 Roughly speaking, $Y_{t}$ can be written   
\begin{align}\label{Gargamel}
Y_{t}\approx  \int_{0}^{t}dr\,(-1)^{\mathbf{N}_{r}}|K_{r}|\approx & \sum_{n=0}^{\mathbf{N}_{t}}(-1)^{n}\int_{\tau_{n}}^{\tau_{n+1}}dr|K_{r}|
\nonumber \\ = & t^{\frac{5}{4}}\Big( t^{-\frac{1}{4}} \sum_{n=0}^{\mathbf{N}_{t}}(-1)^{n}t^{-\frac{1}{2}}\int_{\tau_{n}}^{\tau_{n+1}}dr|t^{-\frac{1}{2}} K_{r}|  \Big),
\end{align}
where $\tau_{n}$ are the reflection times. By the considerations above, the integrand $|t^{-\frac{1}{2}} K_{r}|$ and normalized interval $t^{-\frac{1}{2}}(\tau_{n+1}-\tau_{n})$ will both be $\mathit{O}(1)$.   
  Since the number $\mathbf{N}_{t}$ of terms in the sum  is on the order $ \mathit{O}(t^{\frac{1}{2}})$, a scaling factor of $\mathbf{N}_{t}^{\frac{1}{2}}\propto t^{\frac{1}{4}}$ is appropriate if I expect central limit theorem-type cancellation among the summands in~(\ref{Gargamel}).  Hence, it is reasonable to expect $t^{-\frac{5}{4}}Y_{st}$, $s\in [0,1]$ to have a nontrivial diffusive limit.  The diffusion rate will be proportional to $\big|t^{-\frac{1}{2}}K_{st}  \big|^{3}$ with two powers of $\big|t^{-\frac{1}{2}}K_{st}  \big|$ coming from the integrand~(\ref{Gargamel}), and another factor coming from the less frequent reflections (and thus diminished cancellation) that occur when the momentum is large.

\subsubsection{ Related limit theorems }\label{SecButchered}

The article~\cite{Newton} is a classical analog of the current work in which the periodic potential is continuous and the noise is frictionless.    There, the rescaled momentum process $t^{-\frac{1}{2}}K_{st}$, $s\in[0,1]$ converges in distribution to a Brownian motion, and thus the position process simply converges to the integral of a Brownian motion.  Also for a classical model, the articles~\cite{NewtonII,NewtonIII} work to control and characterize a periodic potential as a perturbative contribution to a dissipative dynamics driven by a linear Boltzmann equation in the limit of large mass for the test particle.  The noise there is analogous to the quantum noise discussed in Sect.~\ref{SecConj}.

The study of the limit law of $t^{-\frac{5}{4}}\int_{0}^{st}drK_{r}$ for large $t\in \R_{+}$ fits under the general category of central limit theory for integral functionals of  Markov processes. It is not covered by previous results that I am aware of, since  the process $K_{t}$ is null-recurrent and systematically makes large jumps in the form of sign-flips from arbitrarily high values in phase space (see~\cite{Hopfner} for martingale limit theory relevant to  a broad class of null-recurrent situations).  A simplified version of the problem is given by the following: let the Markov process $K_{t}'$  make jumps at times determined by a Poisson clock with rate $\mathcal{R}$ and transition density $T(k',k)$ from $k$ to $k'$ given by    
$$ T(k',k)= \left\{  \begin{array}{cc} \mathcal{R}^{-1}j(k'-k) \quad   &  \big|\frac{n}{2}\pm(k'-k)   \big|> \frac{\alpha}{ |n|  }\text{ for all }n\in \Z-\{0\},   \\   \quad & \quad \\   (2\mathcal{R})^{-1}\sum j\big(\pm k'\mp k\big)   \quad   & \big| \frac{n}{2} \pm(k'-k)   \big|\leq  \frac{\alpha}{ |n|  }\text{ for some }n\in \Z-\{0\} .  \end{array} \right. 
  $$
In words, the process first makes a jump from $k$ to $k+v$ with density $\mathcal{R}^{-1}j(v)$, and if  $|k+v- \frac{n}{2}  |\leq  \frac{\alpha}{ |n|  }$ for some $n\in \Z-\{0\}$, then the sign either flips or remains the same with equal probability $1/2$.  Note that the process $|K_{t}'|$ has the same law as the absolute value of a L\'evy process with rates $j(v)$.  The limit statement of Thm.~\ref{Main} holds with $K_{t}$ replaced by $K_{t}'$. If  the reflection bands around the lattice values $\frac{n}{2}\in\frac{1}{2}\Z$ in the simplified model above are replaced by bands with diameter $ \alpha  |\frac{n}{2}|^{-2\vartheta }$ for $0<\vartheta<1$, then  the limit law will be
$$\Big((\sigma t)^{-\frac{1}{2}}  |K_{st}'|,\, \sigma^{-\frac{1}{2}-\frac{\vartheta}{2} }\nu^{\frac{1}{2}} t^{-1-\frac{\vartheta}{2}}Y_{st}'     \Big)\Longrightarrow \Big(|\mathbf{B}_{s}|,\, \int_{0}^{s}d \mathbf{B}^{\prime}_{r}|\mathbf{B}_{r}|^{1+\vartheta}\Big), \quad \quad s\in [0,\,1],
$$ 
where $Y_{t}'=\int_{0}^{t}drK_{r}'$.

The Kr\"onig-Penney model lies at the boundary $\vartheta=1$.  I conjecture that if the periodic $\delta$-potential is replaced by a continuous potential, the limiting behavior will agree with the classical case.  The Kr\"onig-Penney model will have an intermediary behavior in which  $t^{-\frac{1}{2}}K_{st}$  converges in law as $t\rightarrow \infty$, but the limiting process is not a Brownian motion  due to a number of random reflections over the interval $s\in[0,1]$.  The limit of the time  integral process will be differentiable rather than diffusive.

\subsubsection{Further mathematical questions}

The first question is whether the original Lindblad model~(\ref{TheModel}) actually exhibits suppressed spatial dispersion for a fixed value of $\lambda$.  Within the current program of passing through a Freidlin-Wentzell limit, there is the mathematical challenge of beginning with a more sophisticated noise that generates energy relaxation for the test particle (see Sect.~\ref{SecConj}).  Also, it would be interesting to compare the Dirac comb with other singular periodic potentials and to see how the situation changes in higher dimensions.

\subsection{Analogous conjectures for a dissipative model}\label{SecConj}

In this section, I introduce a related quantum Markovian dynamics complex enough to include energy relaxation.  This material is presented with the intent to broaden the reader's perspective on the original model.  In other words, this is a continuation of the discussion in the last section and does not concern the mathematical results of this paper.  The model is a one-dimensional version of the quantum linear Boltzmann dynamics discussed in the review~\cite{VaccHorn}, which  models a test particle interacting with a dilute gas of distinguishable particles.  
The one-dimensional case certainly can not be derived from first principles, since even the classical one-dimensional linear Boltzmann equation does not arise in a low density limit from a  microscopic Hamiltonian model for a test particle interacting with a gas.   Analogous mathematical objects in this section to those previously defined will be denoted with a  tilde, and the meaning of symbols introduced here will reset in future sections.        

 Let the state of the particle at time $t\in \R_{+}$ be given by a density matrix $\hat{\rho}_{ t}\in \mathcal{B}_{1}\big(L^{2}(\R)\big)$ whose evolution is determined by the Lindblad equation
\begin{align}\label{TheAltModel}
\frac{d}{dt}\hat{\rho}_{ t}= -\frac{\textup{i} }{\hbar}\big[\hat{H},\hat{\rho}_{ t}\big]+\hat{\Psi}(\hat{\rho}_{ t})-\frac{1}{2}\big\{\hat{\Psi}^{*}(I),\, \hat{\rho}_{ t}  \big\},  
\end{align}
where  $\hat{\rho}_{ 0}=\rho$, and the Hamiltonian $H$ and the completely positive map $\hat{\Psi}$ are defined below.  The Hamiltonian $H$ is a Schr\"odinger operator with Dirac comb potential    
$$\hat{H}=\frac{1}{2M}P^{2}+\alpha\sum_{N\in \Z}\delta\big(X-a N\big),  $$
where $M$ is the mass of the test particle, $a$ is the spatial period, and $\alpha>0$ is the strength of the Dirac comb.  Let $\eta$ be the spatial density of the gas, $m$ be the mass of a single gas particle, and $\mathfrak{R}(p_{rel})$ be the reflection coefficient  determined by the interaction potential between the test particle and a single gas particle.   For the noise map $\hat{\Psi}$, I take     
\begin{align}\label{AltNoiseMap} 
\hat{\Psi}(\rho)= \frac{(1+\frac{m}{M})\eta  }{ 2 m }  \int_{\R}dq\,|q|\,\big|\mathfrak{R}(\frac{q}{2})\big|^{2}\, e^{\textup{i} \frac{q}{\hbar}X}L_{q}(P)\hat{\rho} L_{q}(P)e^{-\textup{i}\frac{q}{\hbar}X},  
\end{align}
where $L_{q}(P)$ is the multiplication operator in the momentum representation with function  
$$ L_{q}(p)= \Big(\mu\big(2^{-1}( 1+\frac{m}{M} )q+\frac{m}{M} p  \big)\Big)^{\frac{1}{2}} \quad \text{ for }\quad  \mu(p)=\frac{e^{-\beta \frac{p^{2}}{2m}} }{\big(2\pi m\beta^{-1}\big)^{\frac{1}{2} }   } .         $$
The operator $\Psi^{*}(I)$ is also a function of the momentum operator given by
$$\hat{\Psi}^{*}(I)=  \frac{(1+\frac{m}{M})\eta }{ 2 m }  \int_{\R}dq\,|q|\,\big|\mathfrak{R}(\frac{q}{2})\big|^{2}\,\mu\big(2^{-1}( 1+\frac{m}{M} )q+\frac{m}{M} P  \big)=:\mathcal{E}(P).  $$
The function $\mathcal{E}(p)$ serves as an escape rate for getting kicked out of the momentum $p$. 

I will assume  the interaction between the test particle and a reservoir particle is a ``hard-point" interaction (i.e. an infinite strength $\delta$-interaction).   For the hard-point case, the reflection coefficient is $\big|\mathfrak{R}(p)\big|^{2}=1$.  When the Dirac comb is set to zero, the distribution in momentum $\langle p|\rho_{\lambda,t}|p\rangle $ converges exponentially fast in $L^{1}$-norm to a Gaussian of width $(\frac{M}{\beta})^{\frac{1}{2}}$.  This is simple  to prove, since it  can be reduced to showing exponential ergodicity for a classical Kolmogorov equation.  I have discussed exponential dissipation in the article~\cite{Boltzmann} for a three-dimensional case (hard-sphere interaction) for the purpose of studying diffusion.

The discussion of the  Hamiltonian $\hat{H}$ is the same as before except for the inclusion of physical constants. The Hamiltonian has a basis of kets    $|p\rangle_{\scriptscriptstyle{Q}}$ with energies given by
 $\hat{E}(p)=\frac{1}{2M}\mathbf{q}^{2}(p)$ for the anti-symmetric, increasing function $\mathbf{q}:\R\rightarrow \R$ determined by the relation
\begin{align*}
\cos\big(\frac{a}{\hbar} p\big)=\cos\big(\frac{a}{\hbar} \mathbf{q}(p) \big)+\frac{\alpha M}{\hbar\mathbf{q}(p) }\sin\big(\frac{ a}{\hbar} \mathbf{q}(p)\big),  \hspace{1.5cm} p\in \R-\frac{  \pi \hbar }{ a}\Z,
\end{align*}
and $\mathbf{q}(\frac{n\pi  \hbar }{ a})=\frac{ n\pi  \hbar }{ a}$ for $n\in \Z$. 
The dispersion relation $\hat{E}(p)$  has jumps $ g_{n}\neq 0 $  at the values $p=\frac{ n \pi \hbar }{ a}$ for  $n\in \Z$, which approach  $\frac{2\alpha}{a}$ as $n$ goes to infinity.

   I will consider the model~(\ref{TheAltModel}) in a limit in which the constants $M,\hbar,\alpha, a$ and a time variable $s\in[0,t_{0}]$ scale with a single parameter $0<\lambda\ll 1$ as follows
\begin{align}\label{Scales}
M=m\lambda^{-1},\quad t=s\lambda^{-1},\quad \hbar=h\lambda^{1+\varrho},\quad \alpha=\alpha_{0}\lambda^{2+\varrho} ,\quad a=a_{0}\lambda^{1+\varrho}  ,   
\end{align}   
 for fixed constants $m$,  $h$, $\alpha_{0}$,  $a_{0}$, $t_{0}$, $\eta$, and an exponent $\varrho>0$.   I will add  ``$\lambda$" as a subscript to the solution $\hat{\rho}_{\lambda,t}$ of the Lindblad equation~(\ref{TheAltModel}) and the escape rate $\mathcal{E}_{\lambda}(p)$ to indicate the parameter dependence.

  Let $\hat{D}_{\lambda,t}(p)=  { }_{\scriptscriptstyle{Q}}\langle p| \hat{\rho}_{\lambda,t } |p\rangle_{\scriptscriptstyle{Q}} $ be the probability density for the momentum given by the diagonal of the density matrix $\hat{\rho}_{\lambda,t}$ in the extended-zone scheme representation.  Define $\hat{\mathcal{D}}_{\lambda,t}\in L^{1}(\R)$ to be the solution of the master equation 
\begin{align} \label{TheLimit}  
\frac{d}{dt}\hat{\mathcal{D}}_{\lambda,t}(p)= \int_{\R}dp'\Big(\hat{J}_{\lambda}(p,p^{\prime})\hat{\mathcal{D}}_{\lambda,t}(p^{\prime} )- \hat{J}_{\lambda}(p^{\prime},p)\hat{\mathcal{D}}_{\lambda,t}(p )   \Big)
\end{align}
with $\hat{\mathcal{D}}_{\lambda,o}(p )=  { }_{\scriptscriptstyle{Q}}\langle p| \hat{\rho} |p\rangle_{\scriptscriptstyle{Q}} $, where the jump rates $\hat{J}_{\lambda}(p,p^{\prime})$ are defined below.  Let $ \hat{\kappa}_{\lambda, q}(p,m)\in \C$ be the coefficients in the equation
\begin{align}
\sum_{n\in \Z}\hat{\kappa}_{\lambda, q}(p,n)|p+q+n\rangle_{\scriptscriptstyle{Q}}  = e^{\textup{i} \frac{qX}{\hbar} } \mu\big(2^{-1}(1+\lambda)q+\lambda P      \big)^{\frac{1}{2}}   | p \rangle_{\scriptscriptstyle{Q}}.
\end{align}
The rates $\hat{J}_{\lambda}(p,p^{\prime})$ have the form
\begin{align}\label{LimitRates} \hat{J}_{\lambda}(p,p^{\prime})= \frac{(1+\lambda)\eta}{2m} \sum_{n\in \Z }| p-p^{\prime}-n    |\, \big| \hat{\kappa}_{\lambda,\, p-p^{\prime}-n}(p^{\prime},n)\big|^{2}.
\end{align}
I will denote the Markovian process whose densities obey the Kolmogorov equation~(\ref{TheLimit}) as $\hat{K}_{t}$, although the reader should note that the process has units of momentum rather than wave number.

The following conjecture is analogous to Thm.~\ref{FWLimit}.  The exponent $\varrho>0$ from~(\ref{Scales}) will only appear in the error of the following theorem.  
\begin{conjecture}{Freidlin-Wentzell/semi-classical limit.}\label{SemiClassical}
Let $\hat{D}_{\lambda,t}$ and $\hat{\mathcal{D}}_{\lambda,t}$ be defined as above.  For $\lambda\ll 1$, then
$$\sup_{r\in[0,\lambda^{-1}t_{0}]}\|\hat{D}_{\lambda,r}-\hat{\mathcal{D}}_{\lambda,r}\|_{1}=\mathit{O}(\lambda^{\varrho}).   $$
\end{conjecture}

Next, I state a conjecture analogous to Thm.~\ref{Main}.  Let  $\mathbf{P}_{t}$ be the Ornstein-Uhlenbeck process satisfying $\mathbf{P}_{0}=0$ and the Langevin equation 
\begin{align}\label{Langy}
d\mathbf{P}_{t}= - \gamma \mathbf{P}_{t} +\big( \frac{2m\gamma}{\beta}\big)^{\frac{1}{2}}d\mathbf{B}_{t},   
\end{align}
where  $\mathbf{B}_{t}$ is standard Brownian motion
and  $  \gamma =8\eta(\frac{2}{\pi m \beta})^{\frac{1}{2}}$. Define $\nu=\big(\frac{ 32 m   }{\beta \pi}  \big)^{\frac{1}{2}} \frac{ \alpha_{0}\eta  }{ h  }$.   Conjecture~\ref{ConjConv} states   that the processes $\lambda^{\frac{1}{2}} |\hat{K}_{\frac{s}{\lambda} }|$, $s\in [0,t_{0}]$ converge in law for small $\lambda$ to the absolute value of the Ornstein-Uhlenbeck process above, and the normalized integrals $ \lambda^{\frac{3}{8}}\int_{0}^{\frac{s}{\lambda} }dr\hat{K}_{r} $ converge in law to a variable-rate diffusion.

\begin{conjecture}\label{ConjConv}
Let $\hat{K}_{t}$ be a Markov process whose probability densities obey the master equation~(\ref{TheLimit}) for a fixed $\lambda>0$.  Define the time integral process $\hat{Y}_{t}=\frac{1}{m} \int_{0}^{t}dr\hat{K}_{r}$.  In the limit  $\lambda\rightarrow 0$, there is convergence in law with respect to the Skorokhod metric    
$$ \hspace{4cm} \left(  \lambda^{\frac{1}{2}} |\hat{K}_{\frac{s}{\lambda} }|,\, \lambda^{\frac{3}{8}}  \hat{Y}_{\frac{s}{\lambda} }  \right)\stackrel{\frak{L}}{\Longrightarrow} \Big( |\mathbf{P}_{s}|,\,  \frac{1}{m \nu^{\frac{1}{2}}} \int_{0}^{s}d\mathbf{B}_{r}'|\mathbf{P}_{r}|^{\frac{3}{2}}\Big),  \quad \quad \quad s\in[0,t_{0}],     $$
where  $\mathbf{P}_{t}$ is the Ornstein-Uhlenbeck process~(\ref{Langy}), and $\mathbf{B}_{t}'$  is standard Brownian motion independent of $\mathbf{P}_{t}$.
\end{conjecture}

If the Dirac comb is set to zero, then there is the standard limit result  for $\hat{Y}_{t}$ given by the convergence in law    
$$ \hspace{4cm} \left(  \lambda^{\frac{1}{2}}\hat{K}_{\frac{t}{\lambda} }, \,\lambda^{\frac{1}{2} }  \hat{Y}_{\frac{t}{\lambda} }  \right)\Longrightarrow  \Big( \mathbf{P}_{t},\,  \frac{1}{m} \int_{0}^{t}dr\mathbf{P}_{r} \Big),  \quad \quad \quad t\in[0,t_{0}].     $$
The Markov process $\hat{K}_{t}$ without the Dirac comb has the same jump rates as the classical linear Boltzmann equation studied in~\cite{NewtonII}.  The case $\alpha>0$ is subdiffusive, since the spread in position is on the order  $ \lambda^{-\frac{3}{8}}$ rather than  $ \lambda^{-\frac{1}{2}}$ for $\lambda\ll 1$.

\subsubsection{Physical characteristics of the scaling regime}

The following lists the essential characteristics for the regime given by~(\ref{Scales}).  The same mathematical mechanisms outlined in Sect.~\ref{SecHeuristics} should apply for this model, so I focus on a qualitative  comparison of the relevant physical scales.

\begin{enumerate}

\item \textbf{The Brownian limit } 

The scaling of the mass ratio as $\lambda=\frac{m}{M}$ while considering the dynamics over a time interval  $[0, \frac{t_{0}}{\lambda}]$ growing proportionally to $\lambda^{-1}$ is the standard regime for a Brownian limit.   
Since the temperature $\beta^{-1}$ is fixed, the typical speed for a single particle from the reservoir  and the test particle are $(\frac{1}{m\beta})^{\frac{1}{2}}$ and $\lambda^{\frac{1}{2}}(\frac{1}{m\beta})^{\frac{1}{2}}$,  respectively.   Hence, the reservoir particles are moving faster than the test particle by a factor $\lambda^{-\frac{1}{2}}\gg 1$. On the other hand, the  momentum is $(\frac{m}{\beta})^{\frac{1}{2}}$ for a single gas particle    and $\lambda^{-\frac{1}{2}}(\frac{m}{\beta})^{\frac{1}{2}}$ for the test particle.  Individual collisions with gas particles impart momenta that are much smaller than the  typical momentum of the test particle.

\item \textbf{Frequency of phase oscillations versus collisions}

As I described in Section~\ref{SubSecLambda}, the autonomous evolution arising for the densities in the extended-zone scheme representation depend on the noise operating on a comparatively slow scale to the Hamiltonian dynamics.  More precisely, certain phase oscillations driven by the Hamiltonian occur on a smaller time scale than the mean time between collisions.  I characterize the relevant Hamiltonian-driven phase cancellations with the frequency
\begin{align*} 
   \liminf_{p\rightarrow \infty}\inf_{\substack{n\in \Z \\ n\neq 0}   } \frac{1}{\hbar}\big| E(p+n)-E(p)   \big|  = &  \frac{1}{\hbar}\liminf_{n\rightarrow \infty}\hat{g}_{n}\\ = &\frac{2\alpha}{ \hbar  a } =\frac{2\alpha_{0}}{ h a_{0}  }\lambda^{-\varrho} \gg 1  
   \end{align*}
for $\lambda\ll 1$.  I associate the frequency of collisions with the escape rates $\mathcal{E}_{\lambda}(p)$ for $|p|$ on the order  $(\frac{M}{\beta})^{\frac{1}{2}}= (\frac{m}{\beta})^{\frac{1}{2}}\lambda^{-\frac{1}{2}}$:  
\begin{align}\label{Trapeze}
\mathcal{E}_{\lambda}(p)= \frac{( 1+\lambda)\eta }{ 2 m }  \int_{\R}dq|q|\mu\big(2^{-1}( 1+\lambda )q+\lambda p  \big) \approx  \eta\big(  \frac{  8 }{ m\beta \pi  } \big)^{\frac{1}{2}}.  
\end{align}
 Therefore, the phase oscillations occur on a shorter time scale than the collisions for small enough $\lambda$.

\item \textbf{Kinetic energy outweighs potential energy    }

The kinetic energy of the test particle is typically larger than the momentum stored in the potential by a factor of $\lambda^{-1}$.   For this comparison, I associate the typical potential energy with the strength $\alpha$ of the $\delta$-potential divided by the period length: $\frac{\alpha}{a}=\frac{\alpha_{0}}{a_{0}}\lambda$.  This is the potential energy, for instance, in a spatial wave $\phi(x)=(2 L)^{-\frac{1}{2}}1_{[-L,L]}(x)$ in the limit $L\rightarrow \infty$.  The mean kinetic energy is simply $\beta^{-1}$.

\item  \textbf{Reciprocal lattice momenta and the reflection bands  }

 The half-spaced reciprocal lattice of momenta are multiples of $\frac{\pi\hbar}{ a} =\frac{\pi h}{ a_{0}}$.  This is on the same order in $\lambda$ as the typical momentum transfers for collisions $\approx 2(\frac{m}{\beta})^{\frac{1}{2}} $, and so the test particle's momentum has a chance of being kicked out of a given Bloch cell after several collisions.  The reflection band around a lattice momentum $\frac{\pi \hbar }{a }n$ for $n\in \Z$, where nonnegligible probabilities for Bragg reflections may be found, has a width of approximately $\frac{2M\alpha}{\hbar |n|}=\frac{2m\alpha_{0}}{h |n|} $ for high enough $n$ so  that  $\frac{2\pi h}{ a_{0}}\gg \frac{m\alpha_{0}}{2h |n|} $.  By a similar idea as in~(\ref{IntRefl}),  the probability of a reflection when the  particle's momentum is randomly dropped in the interval $[\frac{\pi\hbar (2n-1)}{2a},\frac{\pi\hbar (2n+1)}{2a}] $ around $p= \frac{\pi \hbar n}{a} $ is approximately
 \begin{align} \label{Nip}
 \frac{2M\alpha   }{ \hbar |p|  } = \frac{2m\alpha_{0} }{h|p| }\quad \text{for} \quad \frac{p^{2}}{2M}\gg \frac{\alpha}{a}.    
 \end{align}

\item  \textbf{Bragg reflections are frequent over the time period $[0,\lambda^{-1}t_{0}]$ } 

The frequency of Bragg reflections is equal to the frequency of collisions multiplied by a local averaged probability for reflection after a collision, which depends on the current momentum $p$.  Multiplying~(\ref{Trapeze}) with~(\ref{Nip}), the effective frequency of reflections when the test particle has momentum 
 $|p| \gg ( \frac{M\alpha}{a}  )^{\frac{1}{2}} $ is 
$$ \big(\frac{ 32 m   }{\beta \pi}  \big)^{\frac{1}{2}} \frac{ \alpha_{0}\eta  }{ h |p| }=\frac{\nu}{|p|}  .$$  Since the momentum will typically be found on the scale $|p|\approx \lambda^{-\frac{1}{2}}(\frac{m}{\beta})^{\frac{1}{2}} $, a number of reflections  on the order of $\lambda^{-\frac{1}{2}}$ will occur.

\end{enumerate}

\section{ Overview for proof of Theorem~\ref{Main}    }\label{SecProofOutline}

 In this section, I will state the results that enter directly into the proof of Thm.~\ref{Main}, and then proceed with a presentation of the proof assuming those results.  Thus, this section  concerns only the classical Markovian process $K_{r}$ whose densities obey the Kolmogorov equation~(\ref{Master}).

Recall that   $\mathcal{E}_{r}$ is the square root of the energy at time $r\in \R_{+}$: $E^{\frac{1}{2}}(K_{r})$.  In Sect.~\ref{SecTransition}, I show that $\mathcal{E}_{r}$ is a submartingale.  The following theorem states a central limit theorem for a rescaled version of $\mathcal{E}_{r}$.  Since $E^{\frac{1}{2}}(k)-|k|$ is a bounded function, the convergence of $t^{-\frac{1}{2}}\mathcal{E}_{st}$, $s\in [0,1]$  to the absolute value of a Brownian motion as $t\rightarrow \infty$ is equivalent to the same statement for $t^{-\frac{1}{2}}|K_{st}|$.  Thus, the first component for the convergence stated in Thm.~\ref{Main} follows directly from Thm.~\ref{SubMartCLT}.

\begin{theorem}\label{SubMartCLT}
In the limit $t\rightarrow \infty$, the processes $t^{-\frac{1}{2}}\mathcal{E}_{st}$, $s\in [0,1]$ converge in law to the absolute value of a Brownian motion with diffusion constant $\sigma$. Moreover, the martingale and predictable  components $M_{r}$, $A_{r}$  in the Doob-Meyer decomposition for $\mathcal{E}_{r}$ have  convergence in law for large $t$   given by
$$\hspace{4cm} \big((\sigma t)^{-\frac{1}{2}}M_{st},\,(\sigma t)^{-\frac{1}{2}}A_{st}    \big)   \stackrel{\frak{L}}{\Longrightarrow} \Big(\mathbf{B}_{s},\,\sup_{0\leq r\leq s}-\mathbf{B}_{r}\Big), \hspace{1.5cm}  s\in[0,1],  $$
where $\mathbf{B}$ is a standard Brownian motion.
 The convergences are with respect to the uniform metric.     
\end{theorem} 
 
 Theorem~\ref{SubMartCLT} only characterizes the behavior for quantities that depend on the absolute value of the momentum process $K_{r}$, and thus the sign-flipping does not play a role.  The concept of ``sign-flips" is only meaningful when the particle has a high momentum from which a L\'evy jump to a momentum with the opposite sign would be unlikely.  It is useful to define a series of stopping times that parse the time interval $[0,t]$ into a series of excursions from the low momentum region.  The following is a list of rough definitions for the notations related to sign-flipping and the time-integral process $Y_{r}=\int_{0}^{r}dvK_{v}$.  More precise definitions are given below, although the details for the definitions are not strictly necessary to understand the structure of the argument in the proof of Thm.~\ref{Main} at the end of this section.  The technical definition for the sign-flip times will not be quite adapted to the original filtration $\mathcal{F}_{r}$.  
 \begin{eqnarray*}
& Y_{s}^{(t)}   & \text{Normalized integral functional: $ Y_{s}^{(t)}=t^{-\frac{5}{4}}\int_{0}^{st}drK_{r}$}\\
&  \tau_{m}\in \R_{+}     &   \text{Time of $m$th sign-flip} \\
 &\mathbf{N}_{r}\in \mathbb{N} \text{\,}    & \text{Number of $\tau_{m}$ up to time $r\in \R_{+}$ }\\
& (\varpi_{n},\,\varsigma_{n}) \subset \R_{+}       &  \text{The $n$th incursion  interval into the low momentum region $\approx |K_{r}|\leq t^{\frac{3}{8}}$ }\\
& (\varsigma_{n},\varpi_{n+1}) \subset \R_{+}       &  \text{The $n$th excursion interval from the  low momentum region }\\
&  \Upsilon_{r} \in \mathbb{N}         &    \text{Number of excursions to have begun by time $r$} \\
&\mathcal{F}_{r}    &    \text{Information up to time $r$ } \\ 
&\widetilde{\mathcal{F}}_{r}    &    \text{Information up to the time of the sign-flip following $r$  }\\
& \mathbf{h}_{s}^{(t)}& \text{Martingale with respect to $\widetilde{\mathcal{F}}_{s}^{(t)}:=\widetilde{\mathcal{F}}_{st}$ approximating $t^{-\frac{1}{2}}M_{st}$ for $t\gg 1$ }\\
& \mathbf{m}_{s}^{(t)}& \text{Martingale with respect to $\widetilde{\mathcal{F}}_{s}^{(t)}$ approximating $Y_{s}^{(t)}$ as $t\gg 1$  }
\end{eqnarray*}

Let $S:\R\rightarrow \{\pm 1\}$ be the sign function.  A \textit{sign-flip} is said to occur at a Poisson time $t_{n}$ if  $S(K_{t_{n}})=S(K_{t_{n+1}})$ and there are an odd number $m$ of sign changes leading up to $t_{n}$: $S(K_{t_{n-r}})=-S(K_{t_{n-r+1}})$ for $r\in [1,m]$ and $S(K_{t_{n-m-1}})=S(K_{t_{n-m}})$.  This definition avoids counting double-flips, which occur frequently in the dynamics and would distort the counting.  The sign-flips are not hitting times, since information from the following Poisson time is required to identify them.  Related matters are discussed in the beginning of Sect.~\ref{SubSecRefl}.  To define the time intervals where sign-flips will be counted,  let the stopping times
$\varsigma_{j},\varpi_{j}$ be given by $\varsigma_{0}=\varpi_{1}=0$ and 
\begin{align*}
\varpi_{j}= \min \{r\in (\varsigma_{j-1},\infty)\,\big|\, |K_{r}|\leq  t^{\frac{3}{8}-\iota}         \},\quad \quad
\varsigma_{j}= \min \{ r\in (\varpi_{j},\infty)\,\big| \, |K_{r}|\geq 2 t^{\frac{3}{8}-\iota}            \},
\end{align*} 
for some $0<\iota\ll 1$.  The intervals $[\varpi_{j},\,\varsigma_{j})$ are the incursions into low momentum for the process $K_{r}$.  Let $\Upsilon_{r}\in \mathbb{N}$ be the number of excursions begun by time $r$: $\max\{j\in \mathbb{N}\,|\,r\geq \varsigma_{j}     \}$.   The times $\tau_{m}$ and increments $\Delta \tau_{m}$ are defined inductively for $m\in \mathbb{N}$ such that  
\begin{itemize}
\item $\tau_{1}=\varsigma_{1}$,
\item $\Delta \tau_{m}$  is the waiting time after  $\tau_{m}$ such that  either a sign-flip occurs or $|K_{r}|$ jumps out of $[\frac{1}{2}|K_{\tau_{m}}|,\frac{3}{2}|K_{\tau_{m}}|]$,
\item $\tau_{m+1}=\tau_{m}+\Delta \tau_{m}$  when $\tau_{m}+\Delta \tau_{m}$ occurs during an excursion, and
\item  $\tau_{m+1}=\varsigma_{j}$ when $\tau_{m}< \varpi_{j}$ and  $\tau_{m}+\Delta \tau_{m} \geq  \varpi_{j}$.  
\end{itemize}
 In most cases, $\tau_{m}$ are sign-flips and $\Delta \tau_{m}=\tau_{m+1}-\tau_{m}$.  
  Let $\mathbf{N}_{r}$ be the number of $\tau_{m}$'s to have occurred up to time $r$. 
 
 I denote the standard filtration generated by the process $K_{r}$ with $\mathcal{F}_{r}=\sigma\big(K_{s}:\, 0\leq s\leq r\big)$.   Let  $\widetilde{\mathcal{F}}_{r}$ be the $\sigma$-algebra given by
$$\widetilde{\mathcal{F}}_{r}=\left\{  \begin{array}{cc} \sigma\big(\Delta\tau_{m},\, K_{s}:\, 0\leq s\leq \tau_{m}+\Delta\tau_{m}\big) &   r\in [\tau_{m},\tau_{m}+\Delta\tau_{m}   ), \\   \quad & \quad \\   \mathcal{F}_{r}   \quad   &  \forall(m): r\notin [\tau_{m},\tau_{m}+\Delta\tau_{m}   ) .   \end{array}\right.    $$
When $r \in [\tau_{m}, \tau_{m}+\Delta\tau_{m})$ for some $m$,  the $\sigma$-algebra $\widetilde{\mathcal{F}}_{r}$ includes knowledge of the time $\tau_{m}+\Delta\tau_{m}$ and all information about the process $K_{r}$ up to time $\tau_{m}+\Delta\tau_{m}$.  Since the sign-flip times are not necessarily adapted by the remark above,  $\widetilde{\mathcal{F}}_{r}$ usually contains some  information from the Poisson time following $\tau_{m}+\Delta\tau_{m}$ to verify that the momentum does not change sign again.   For $\widetilde{\mathcal{F}}_{s}^{(t)}:=\widetilde{\mathcal{F}}_{st} $, define the $\widetilde{\mathcal{F}}_{s}^{(t)}$-adapted martingales
\begin{align*}
 \mathbf{m}_{s}^{(t)}:= & t^{-\frac{5}{4}}\sum_{m=1}^{\mathbf{N}_{st}}\,K_{\tau_{m}}\Big(\Delta\tau_{m}- \mathbb{E}\big[\Delta \tau_{m}\,\big|\, \widetilde{\mathcal{F}}_{\tau_{m}^{-}}\big]   \Big), \\
 \mathbf{h}_{s}^{(t)}:=&  t^{-\frac{1}{2}}\sum_{m=1}^{\mathbf{N}_{st}}\Big(\mathcal{E}_{\tau_{m}+\Delta \tau_{m}}-\mathcal{E}_{\tau_{m}}-  \mathbb{E}\big[\mathcal{E}_{\tau_{m}+\Delta \tau_{m}}-\mathcal{E}_{\tau_{m}}\,\big|\, \widetilde{\mathcal{F}}_{\tau_{m}^{-}}\big]  \Big).
 \end{align*}

 The process $\mathcal{E}_{r}$ in the statement of Lem.~\ref{LemMartApproxQuad} can be equivalently replaced by $|K_{r}|$.  The lemmas below  all have the purpose of placing my convergence questions within the context of martingale theory.  In particular, Part (2) of Lem.~\ref{LemMartApproxQuad} is to establish independence for the copies of Brownian motion $\mathbf{B}$ and $\mathbf{B}'$ in the statement of Thm.~\ref{Main}.  I will postulate the limiting law in a  different way in the proof of Thm.~\ref{Main}.

\begin{lemma}[Martingale approximations]\label{LemMartApprox}
In the limit $t\rightarrow \infty$, there are the following convergences in probability:
\begin{enumerate}
\item $\sup_{0\leq s\leq 1}\big|  \mathbf{h}_{s}^{(t)}-t^{-\frac{1}{2}}M_{st}  \big|\Longrightarrow 0,$
\item $\sup_{0\leq s\leq 1}\big| \mathbf{m}_{s}^{(t)}-Y^{(t)}_{s}   \big|\Longrightarrow 0$.  
\end{enumerate}
 \end{lemma}
 \vspace{.2cm}
  
\begin{lemma}[Quadratic variation approximations] \label{LemMartApproxQuad}
In the limit $t\rightarrow \infty$, there are the following convergences in probability:
\begin{enumerate}
\item $\sup_{0\leq s\leq 1}\Big|   \big[\mathbf{m}^{(t)},\mathbf{m}^{(t)}\big]_{s}-\frac{1}{\nu}\int_{0}^{s}dr (t^{-\frac{1}{2} }\mathcal{E}_{rt})^{3}\Big| \Longrightarrow 0$,
\item  $\sup_{0\leq s\leq 1}\Big|   \big[\mathbf{m}^{(t)},\, \mathbf{h}^{(t)}\big]_{s}\Big|  \Longrightarrow 0, $ 
\item $  \sup_{0\leq s\leq 1}\Big|   \big[\mathbf{h}^{(t)},\, \mathbf{h}^{(t)}\big]_{s}-\sigma s\Big|  \Longrightarrow 0   $.
\end{enumerate}
 
 \end{lemma}
 \vspace{.2cm}
 
 \begin{lemma}[Lindberg conditions]\label{LemLindberg}
The martingale $\mathbf{m}^{(t)}$ satisfies the Lindberg condition such that as $t\rightarrow \infty$, then
$\mathbb{E}\big[\sup_{0\leq s\leq 1}\big| \mathbf{m}_{s}^{(t)}-\mathbf{m}_{s^{-} }^{(t)}      \big|\big]\rightarrow 0$.   Also, the family $\mathbf{m}_{s}^{(t)}$ for $0\leq s\leq 1$, $t\in \R_{+}$ is  uniformly integrable.  The same statements hold  for $ \mathbf{h}_{s}^{(t)}$. 
 \end{lemma}

 \vspace{.4cm}
   I will make frequent use of the reference book~\cite{Jacod} in the proof below.
 \vspace{.2cm}
 
 \begin{proof}[Proof of Thm.~\ref{Main}]  All convergence in law will be with respect to the Skorokhod metric.   I will  show that there is convergence in law as $t\rightarrow \infty$ 
 $$\hspace{4cm}\big( |t^{-\frac{1}{2}}K_{st}|,\,Y^{(t)}_{s}    \big) \stackrel{\frak{L}}{\Longrightarrow}  \Big(\mathbf{B}_{s}+\sup_{0\leq r\leq s}-\mathbf{B}_{r}   ,\,\frak{m}_{s}\Big), \hspace{1cm} s\in [0,1],  $$  where $\mathbf{B}$ is Brownian motion with diffusion rate $\sigma$ and $\frak{m}$ is a continuous martingale with quadratic variation processes satisfying
\begin{align}\label{MartProblem}
 [\frak{m},\frak{m} ]_{s}=\frac{1}{\nu}\int_{0}^{s}dr|\mathbf{B}_{r}|^{3} \hspace{1cm}\text{and}\hspace{1cm}   [ \frak{m}, \mathbf{B} ]_{s}=0  
 \end{align}
for all $s\in [0,1]$.  Recall that $\mathbf{B}_{s}+\sup_{0\leq r\leq s}-\mathbf{B}_{r}$ is equal in law to the absolute value of a Brownian motion with rate $\sigma$.   
The above description of the limiting law is equivalent to the construction given in the statement of Thm.~\ref{Main}.  By Part (2) of  Lemma~\ref{LemMartApprox} and the fact that $\big| E^{\frac{1}{2}}(k)-|k| \big|$ is bounded,  I can approximate the tuple  $\big( |t^{-\frac{1}{2}}K_{st}|,\,Y^{(t)}_{s}    \big)$ by $\big( t^{-\frac{1}{2}}\mathcal{E}_{st},\,\mathbf{m}_{s}^{(t)}\big)$.   
Define the process quadruple $Q^{(t)}_{s}=  \big( t^{-\frac{1}{2}}\mathbf{h}_{s}^{(t)},\,t^{-\frac{1}{2}}A_{st},\,\mathbf{m}_{s}^{(t)},\, [\mathbf{m}^{(t)},\mathbf{m}^{(t)}]_{s}   \big)   $, where  $A_{r}$ is the increasing part in the Doob-Meyer decomposition of  $\mathcal{E}_{r}$.   The process   $Q^{(t)}_{s}$ is adapted to the filtration $\widetilde{\mathcal{F}}_{s}^{(t)}$, and the first and third components are martingales with respect to $\widetilde{\mathcal{F}}_{s}^{(t)}$.

By~\cite[Cor.VI.3.33]{Jacod}, the  family of processes   $Q^{(t)}$, $t\in \R_{+}$ is $C$-tight if the components are $C$-tight.  The first component $\mathbf{h}^{(t)}$ of $Q^{(t)}$ can be approximated for large $t$ with $t^{-\frac{1}{2}}M_{st}$ by Part (1) of Lem.~\ref{LemMartApprox}.  The families of processes $t^{-\frac{1}{2}}M_{st}$ and $t^{-\frac{1}{2}}A_{st}$ indexed by $t\in \R_{+}$ are each $C$-tight, since they converge in law to continuous limits by Thm.~\ref{SubMartCLT}.  The fourth component  of $Q^{(t)}$ can be approximated with   $ \frac{1}{\nu}\int_{0}^{s}dr(t^{-\frac{1}{2}}\mathcal{E}_{rt})^{3}$ by Part (1) of Lem~\ref{LemMartApproxQuad}.   The processes $ t^{-\frac{1}{2}}\mathcal{E}_{st}$ converge in law as $t \rightarrow \infty$ to the absolute value of a Brownian motion by Thm.~\ref{SubMartCLT},  and  the functional $F:L^{\infty}([0,1])\rightarrow \R$ defined by $F(f)(s)=\int_{0}^{s}dr f^{3}(r)$ is continuous with respect to the Skorokhod metric;  therefore,  the processes  $ \int_{0}^{s}dr(t^{-\frac{1}{2}}\mathcal{E}_{rt})^{3}$ converge in law  to $ \int_{0}^{s}dr|\mathbf{B}_{r}|^{3}$ for large $t$.  It follows that  $[\mathbf{m}^{(t)}]$, $t\in \R_{+}$ is a $C$-tight family.  The family  of martingales $\mathbf{m}^{(t)}$  must be tight by the $C$-tightness of $[\mathbf{m}^{(t)}]$ and~\cite[Thm.VI.4.13]{Jacod}.  Finally,  $\mathbf{m}^{(t)}$ is $C$-tight by~\cite[Prop.VI.3.26]{Jacod} and the Lindberg condition in  Lemma~\ref{LemLindberg}.

 Consider a sequence $r_{n}\in \R_{+}$  such that $ Q^{(r_{n})}_{s}$ converges in law as $n\rightarrow \infty$ to a limit 
 $$\hspace{3cm}\Big(\mathbf{B}_{s},\, \sup_{0\leq r\leq s}- \mathbf{B}_{r} ,\,\frak{m}_{s},\,\frac{1}{\nu}\int_{0}^{s}dr\big(\mathbf{B}_{r}+ \sup_{0\leq v\leq r}- \mathbf{B}_{v}\big)^{3}  \Big), \hspace{1cm} s\in [0,1].$$ 
The relationships between the first, second, and fourth components are determined by the considerations above and Thm.~\ref{SubMartCLT}.  By definition of $C$-tightness for $ Q^{(t)}$, $t\in \R_{+}$, the third component $\mathbf{m}$ has  continuous trajectories.  The second and fourth components are explicitly determined by the copy of Brownian motion $\mathbf{B}$,  so I will focus on determining the joint law of $(\mathbf{B},\frak{m})$.  Since $\big( \mathbf{h}^{(r_{n})},\,\mathbf{m}^{(r_{n})}   \big)$ is a sequence of martingales and the family of random variables  $ \mathbf{h}^{(r_{n})}_{s}$, $\mathbf{m}^{(r_{n})}_{s}$ for $n\in \mathbb{N}$, $s\in [0,1]$ is uniformly integrable by Lem.~\ref{LemLindberg}, the limit law is a martingale with respect to its own filtration~\cite[Prop.IX.1.12]{Jacod}.   The convergence $ \big( \mathbf{h}^{(r_{n})},\mathbf{m}^{(r_{n})} \big)\stackrel{\frak{L}}{\Longrightarrow} (\mathbf{B},\frak{m})$  with~\cite[Cor.VI.6.7]{Jacod} implies  joint convergence with the quadratic variations
   \begin{align*}
 \left(  \mathbf{h}^{(r_{n})},\,\mathbf{m}^{(r_{n})} ;\, \begin{vmatrix}        [\mathbf{h}^{(r_{n})},\mathbf{h}^{(r_{n})}] ,& [\mathbf{h}^{(r_{n})}, \mathbf{m}^{(r_{n})}] \vspace{.1cm} \\   [\mathbf{m}^{(r_{n})},\mathbf{h}^{(r_{n})}],  &   [\mathbf{m}^{(r_{n})},\mathbf{m}^{(r_{n})}]  \end{vmatrix}  \right) 
  \stackrel{\frak{L}}{\Longrightarrow} \left( \mathbf{B},\,\frak{m};\, \begin{vmatrix}     [\mathbf{B},\mathbf{B}],      &  [\mathbf{B}, \frak{m}]  \vspace{.1cm} \\   [\frak{m}, \mathbf{B}],  &  [\frak{m},\frak{m}]   \end{vmatrix} \right). 
 \end{align*}
I thus have that $ [\frak{m},\frak{m}]_{s}  =\frac{1}{\nu}\int_{0}^{s}dr\big(\mathbf{B}_{r}+ \sup_{0\leq v\leq r}- \mathbf{B}_{v}\big)^{3} $.  Finally, by Lem~\ref{LemMartApproxQuad} the sequence of quadratic variation processes  $[\mathbf{m}^{(r_{n})},\mathbf{h}^{(r_{n})}]_{s}$ converges to zero and hence  $[\frak{m}, \mathbf{B}]=0$.  Therefore, the limiting law for subsequences has been determined uniquely as the one given by~(\ref{MartProblem}), and the proof of  convergence  for  $ \big( |t^{-\frac{1}{2}}K_{st}|,\,Y^{(t)}_{s}    \big) $ is complete.

 \end{proof}

 \section{The Freidlin-Wentzell/semi-classical limit}\label{SecQuantum}

In this section, I prove Thm.~\ref{FWLimit}.  First, I present Lem.~\ref{BadTerms}, which collects some technical estimates  on the probabilities $|\kappa_{v}(k,n)|^{2}$ appearing in the jump rates~(\ref{KrissKrossMakesYou}) for the limiting momentum process.  The estimates pertain to $|k|$ large and include a  technical restriction that the jump $v$ is not too large ($|v|\leq \frac{1}{2}|k|$).  This is an important regime in future sections also, since the particle will stochastically accelerate to high momentum values, and the jumps $v$ will be much smaller by the moment assumption (1) of List~\ref{Assumptions} on the jump rates $j(v)$. There are only a few possible values of $n$ for which $|\kappa_{v}(k,n)|^{2}$  is non-vanishing when $|k|,|k+v|\gg 1$.

  As mentioned in the introduction and is discussed further in Appendix~\ref{AppendixFiber}, the Hilbert space $L^{2}(\R)$ for the test particle has a canonical decomposition into a direct integral $\int_{[-\frac{1}{2},\frac{1}{2})}^{\oplus}$ of copies of $L^{2}\big([-\pi,\pi)\big)$ that are invariant under the Hamiltonian dynamics.  It is useful to relate the kets $|k\rangle_{\scriptscriptstyle{Q}}$ with their associated representations in $L^{2}\big([-\pi,\pi)\big)$.    For $k\in \R-\frac{1}{2}\Z$, the kets $|k\rangle_{\scriptscriptstyle{Q}}$ are identified with Bloch functions  $\widetilde{\psi}_{k}\in L^{2}\big([-\pi,\pi)\big)$ given by
\begin{align}\label{Bloch}
\widetilde{\psi}_{k}(x)= N_{k}^{-\frac{1}{2}}\left\{  \begin{array}{cc} \frac{e^{\textup{i}2\pi (\mathbf{q}(k)-k)  }-1    }{ e^{\textup{i}2\pi(\mathbf{q}(k)+k)}  -1 } e^{-\textup{i} x \mathbf{q}(k) }+e^{\textup{i}2\pi( \mathbf{q}(k)-k) }  e^{\textup{i} x \mathbf{q}(k)  } &  -\pi\leq x\leq 0  ,  \\  \quad & \quad \\ \frac{e^{\textup{i}2\pi(\mathbf{q}(k)-k)  }-1     }{1-e^{-\textup{i}2 \pi(\mathbf{q}(k)+k)} }e^{-\textup{i} x \mathbf{q}(k)  }+   e^{\textup{i}x  \mathbf{q}(k)  }     & 0\leq x< \pi ,   \end{array} \right.  
\end{align}
where $N_{k}>0$ is a normalization constant, and $\mathbf{q}:\R\rightarrow \R$ is  determined by~(\ref{Energies}).    The form of the Bloch functions~(\ref{Bloch}) can be found in~\cite[Sec.III.2.3]{Solve} (under the standard quasimomentum and energy band labeling).  The  Bloch functions formally satisfy $\langle x|k\rangle_{\scriptscriptstyle{Q}} = \widetilde{\psi}_{k}(x)$, and $\widetilde{\psi}_{k}$ is an eigenvector with eigenvalue $E(k)$ for the fiber  Hamiltonian $H_{\phi}$ with $ k=\phi\,\textup{mod}\, 1$.   Analogously, the momentum kets $|k\rangle$ are identified with Bloch functions $\psi_{k}\in L^{2}\big([-\pi,\pi)\big)$ defined by $\psi_{k}(x)=(2\pi)^{-\frac{1}{2}}e^{\textup{i}xk}$.

The following notations are designed for a description of the reflection bands around the lattice momenta $\frac{1}{2}\Z$.  Let $\Theta:\R\rightarrow [-\frac{1}{4},\frac{1}{4})$ and $\mathbf{n}:\R\rightarrow \Z$ be defined through
\begin{align}\label{Barrack}
\Theta(k)=k\,\text{mod}\,\frac{1}{2}, \quad \quad \mathbf{n}(k)=2\big(k-\Theta(k)\big).\quad 
\end{align}
Note that the variables $\Theta(k)$ and $\mathbf{n}(k)$ are \textit{not}  the quasimomentum and energy band, respectively.  Define the set $I(k,v)\subset \Z$ to include the integers $0$, $-\mathbf{n}(k)$, $-\mathbf{n}(k+v)$  and $-\mathbf{n}(k+v)+\mathbf{n}(k)$.  Also define $\beta(k)=\frac{1}{2}\mathbf{n}(k)\Theta(k)$, $ \gamma(\beta)= \big( \beta^{2}+\frac{\alpha^{2}}{16\pi^{2}} \big)^{\frac{1}{2}}-\frac{\alpha}{4\pi}    $, 
$$
\mathbf{r}_{-}(k)=  \frac{1}{1+\big|\frac{|\beta(k)|+\gamma(\beta(k))  }{|\beta(k)|-\gamma(\beta(k)) } \big|^{2}    }, 
$$
and $\mathbf{r}_{+}(k)=1-\mathbf{r}_{-}(k)$. The variable $\beta(k)$ is a dilation of $\Theta(k)$, which is scaled to characterize the limiting local profiles of reflection probabilities for high momenta $k$ near lattice values $\frac{1}{2}\Z$.

 In the following lemma, $S:\R\rightarrow \{\pm 1 \}$ is the sign function, and $\textup{dist}\big(n,\,I(k,v)\big) $ is the smallest distance between $n\in \Z$ and an element in the set $I(k,v)$.

\begin{lemma}\label{BadTerms}
Let $|k|\gg 1$.  There exists a $c>0$ such that for all $|k|$ large enough and $|v|\leq \frac{1}{2}|k|$,  the following inequalities hold: 
\begin{enumerate}
\item 
$$ \sum_{n\notin I(k,v)}|\kappa_{v}(k,n)|^{2}\leq \frac{c}{|k|^{2}}. $$

\item $$\sum_{n\notin I(k,v)} \textup{dist}\big(n,\,I(k,v)\big)  \,|\kappa_{v}(k,n)|^{2}\leq \frac{c}{|k|}.$$

\item
  When $\mathbf{n}(k)\neq \mathbf{n}(k+v)$, 
\begin{align*}
 &
\Big| \big|\kappa_{v}\big(k,-\mathbf{n}(k)\big)\big|^{2}- \mathbf{r}_{-}(k)\mathbf{r}_{+} (k+v )\Big|\leq \frac{c}{|k|\big(1+|\beta(k)|   \big) }, \\
& \Big| \big|\kappa_{v}\big(k,-\mathbf{n}(k+v)\big)\big|^{2}- \mathbf{r}_{+}(k)\mathbf{r}_{-} (k+v )\Big|\leq \frac{c}{|k|\big(1+|\beta(k+v)|   \big)  }, \\
&  \Big| \big|\kappa_{v}\big(k,-\mathbf{n}(k+v)+\mathbf{n}(k)\big)\big|^{2}- \mathbf{r}_{-}(k)\mathbf{r}_{-} (k+v )\Big|\leq \frac{c}{|k|\big(1+|\beta(k)|   \big)\big(1+|\beta(k+v)|   \big)  }, 
\end{align*}
 and when $\mathbf{n}(k)=\mathbf{n}(k+v)$,
\begin{align*}
\Big| \big|\kappa_{v}\big(k,0\big)\big|^{2}- \Big(  \mathbf{r}_{+}^{\frac{1}{2}}(k)\mathbf{r}_{+}^{\frac{1}{2}}(k+v)-S\big(\Theta(k)\Theta(k+v)\big)\mathbf{r}_{-}^{\frac{1}{2}}(k)\mathbf{r}_{-}^{\frac{1}{2}}(k+v)\Big)^{2}\Big|\leq \frac{c}{|k|^{2}}, \\  \Big| \big|\kappa_{v}\big(k,-\mathbf{n}(k)\big)\big|^{2}-\Big(  \mathbf{r}_{+}^{\frac{1}{2}}(k)\mathbf{r}_{-}^{\frac{1}{2}}(k+v)+S\big(\Theta(k)\Theta(k+v)\big)\mathbf{r}_{-}^{\frac{1}{2}}(k)\mathbf{r}_{+}^{\frac{1}{2}}(k+v )   \Big)^{2}\Big| \leq \frac{c}{|k|^{2}}.   
\end{align*}

\item For $m\in \Z$ with $|\frac{m}{2}-k|\leq \frac{1}{2}|k| $ and $m\neq \mathbf{n}(k)$, 
$$ \int_{-\frac{1}{4}  }^{\frac{1}{4} }d\theta\,|\kappa_{\frac{m}{2}+\theta-k}\big(k,-m \big)|^{2}  \leq  \frac{c}{|k|}.  $$

\end{enumerate}

\end{lemma}

\begin{proof} \text{  }\\
\noindent Part (1):\hspace{.15cm} Recall that  $\eta(k,m)$ is defined so that $|k\rangle_{\scriptscriptstyle{Q}}=\sum_{m\in \Z}\eta(k,m)|k+m\rangle$.  By mapping through the fiber decomposition,   this is equivalent to the same equality  with $|k+m\rangle$ and $|k\rangle_{\scriptscriptstyle{Q}}$ replaced by $\psi_{k+m}$ and $\widetilde{\psi}_{k}$, respectively.  Moreover, the equation~(\ref{DefKappa}) defining $\kappa_{v}(k,n)$ is translated to
$$    \sum_{n\in \Z}\kappa_{v}(k,n)\widetilde{\psi}_{k+v+n}=e^{\textup{i}v X} \widetilde{\psi}_{k},      $$
where $X$ is the bounded operator on  $L^{2}\big([-\pi,\pi)  \big)   $ acting as multiplication by the spatial variable $(Xf)(x)=xf(x)$ for $x\in [-\pi,\pi)$.

  Below, I will show there exists a $C>0$ such that
\begin{eqnarray}\label{QuasiToStand} 
\big\| \widetilde{\psi}_{k}-S\big(k \Theta(k)\big)\,\mathbf{r}_{-}^{\frac{1}{2}}(k ) \psi_{k-\mathbf{n}(k)}  -\mathbf{r}_{+}^{\frac{1}{2}}(k)  \psi_{k} \big\|_{2} \leq  C|k|^{-1} . 
\end{eqnarray}
Temporarily assuming~(\ref{QuasiToStand}), I will continue with the proof.   Applying~(\ref{QuasiToStand}) directly and for $k$ replaced by $k':=k-\mathbf{n}(k)$, it follows that 
\begin{align}\label{StandToQuasi}
\big\| \psi_{k'}-S\big(k^{\prime} \Theta(k^{\prime})\big)\,\mathbf{r}_{-}^{\frac{1}{2}}(k^{\prime} ) \widetilde{\psi}_{k^{\prime}-\mathbf{n}(k^{\prime})}-\mathbf{r}_{+}^{\frac{1}{2}}(k^{\prime} ) \widetilde{\psi}_{k^{\prime}}  \big\|_{2}\leq 3C|k^{\prime}|^{-1}, 
\end{align}
where I have used the identity $k=k^{\prime}-\mathbf{n}(k^{\prime})$.

  Multiplying $\widetilde{\psi}_{k}$ by $e^{\textup{i}vX}$,  then with~(\ref{QuasiToStand})
$$ e^{\textup{i}vX}\widetilde{\psi}_{k}= S\big(k \Theta(k)\big)\mathbf{r}_{-}^{\frac{1}{2}}(k)\psi_{k+v-\mathbf{n}(k)} +\mathbf{r}_{+}^{\frac{1}{2}}(k)\psi_{k+v} +\mathit{O}(|k|^{-1}),$$
where the error is with respect to the $L^{2}\big([-\pi,\pi)\big)$ norm. 
By the unitarity of $e^{\textup{i}vX}$, the norm of the error is preserved from~(\ref{QuasiToStand}). Applying~(\ref{StandToQuasi}) for $k^{\prime}=k+v$ and $k'=k+v-\mathbf{n}(k)$, 
\begin{align}\label{Compare}
e^{\textup{i}vX}\widetilde{\psi}_{k}+\mathit{O}(|k|^{-1})= &  \mathbf{r}_{+}^{\frac{1}{2}}(k)\mathbf{r}_{+}^{\frac{1}{2}}(k+v)\widetilde{\psi}_{k+v} \nonumber \\ &  +S\big( k\Theta(k+v) \big)\mathbf{r}_{+}^{\frac{1}{2}}(k)\mathbf{r}_{-}^{\frac{1}{2}}(k+v)\widetilde{\psi}_{k+v-\mathbf{n}(k+v) } \nonumber   \\ &+S\big( k\Theta(k)\big)\mathbf{r}_{-}^{\frac{1}{2}}(k)\mathbf{r}_{+}^{\frac{1}{2}}(k+v-\mathbf{n}(k) )\widetilde{\psi}_{k+v-\mathbf{n}(k)}\nonumber   \\ &-S\big( \Theta(k)\Theta(k+v)  \big)\mathbf{r}_{-}^{\frac{1}{2}}(k)\mathbf{r}_{-}^{\frac{1}{2}}(k+v-\mathbf{n}(k)  )\widetilde{\psi}_{k+v-\mathbf{n}(k+v)+\mathbf{n}(k)}, 
 \end{align}
where I have used $S(k+v-\mathbf{n}(k))=-S(k)$ and $S(k)=S(k+v)$ by the restriction $|v|\leq \frac{1}{2}|k|$.  The minus sign on the fourth term on the right side appears because $k-\mathbf{n}(k)$ has the opposite sign of $k$.  It follows   that 
 $$\sum_{m\notin I(k,k+v)} |\kappa_{v}(k,m)|^{2} = \Big\|e^{\textup{i}vX}\widetilde{\psi}_{k}-\sum_{n\in I(k,k+v)} \kappa_{v}(k,n) \widetilde{\psi}_{k+n}\Big\|_{2}^{2}=\mathit{O}(|k|^{-2}),    $$
which will complete the proof once~(\ref{QuasiToStand}) is established.

Now I  prove~(\ref{QuasiToStand}) by examining $\widetilde{\psi}_{k}$   in the Bloch representation~(\ref{Bloch}).  Note that~(\ref{QuasiToStand}) is a little more precise than is strictly required to prove Part (1), but it will be useful later.  Computing the coefficients $\eta(k,n)$ yields  
\begin{align}\label{Etas}
 \eta(k,n)=\langle \psi_{k+n} |\widetilde{\psi}_{k}\rangle=-\textup{i} N_{k}^{-\frac{1}{2}}\big(e^{2\pi \textup{i}(\mathbf{q}(k)-k)}-1 \big) \Big(\frac{1}{\mathbf{q}(k)+k+n}+\frac{1}{\mathbf{q}(k)-k-n} \Big).  
 \end{align}
It follows from~(\ref{Energies}) that for large enough $k$, the inequality holds: $|\mathbf{q}(k)-k|\leq \frac{\alpha}{\pi|k|}$.  The normalization $N_{k}$ can be written as
\begin{align*}
N_{k}&=\big|e^{2\pi \textup{i}(\mathbf{q}(k)-k)}-1 \big|^{2}\sum_{n}\Big| \frac{1}{\mathbf{q}(k)+k+n}+\frac{1}{\mathbf{q}(k)-k-n}   \Big|^{2}\\ & =4\pi^{2}\Big(1+\Big|\frac{\mathbf{q}(k)-k}{\mathbf{q}(k)-k+2\Theta(k) }\Big|^{2}\Big)+\mathit{O}(|k|^{-2}).   
\end{align*}
All terms from the sum have been absorbed into the error $\mathit{O}(|k|^{-2})$ except for $n=0,-\mathbf{n}(k)$.  The approximation above is due to the equality $\mathbf{q}(k)+k-\mathbf{n}(k)=\mathbf{q}(k)-k+2\Theta(k)$, and since 
\begin{align}\label{TreeFire}
\sum_{n\neq 0,-\mathbf{n}(k)}\Big| \frac{1}{\mathbf{q}(k)-k+2\Theta(k)+(n+\mathbf{n}(k))    }+\frac{1}{\mathbf{q}(k)-k-n}   \Big|^{2}  \leq 2\sup_{ |\theta|\leq \frac{3}{4}    }\sum_{n\neq 0}\frac{1}{|n+\theta|^{2}}  \end{align}
for $|k|$ large enough so that $|\mathbf{q}(k)-k|\leq \frac{1}{4}$.
When multiplied by $\big|e^{2\pi \textup{i}(\mathbf{q}(k)-k)}-1 \big|^{2}=\mathit{O}(|k|^{-2})$, then~(\ref{TreeFire}) is  $\mathit{O}(|k|^{-2})$. 
Thus, the sum $\sum_{n\neq 0,-\mathbf{n}(k)}|\kappa_{v}(k,n)|^{2}$ is $\mathit{O}(|k|^{-2})$, and the special terms $ \eta(k,0)$ and  $\eta\big(k,-\mathbf{n}(k)\big)$ have approximations
\begin{align}\label{BarlyFields}
\eta(k,0)&=\frac{ \Big|\frac{\mathbf{q}(k)-k+2\Theta(k)}{\mathbf{q}(k)-k } \Big|   }{\Big(1+\Big|\frac{\mathbf{q}(k)-k+2\Theta(k)}{ \mathbf{q}(k)-k}\Big|^{2}\Big)^{\frac{1}{2}} }    +\mathit{O}(|k|^{-1}), \nonumber \\  \eta\big(k,-\mathbf{n}(k)\big) & =\frac{ S\big(k\Theta(k)\big)    }{\Big(1+\Big|\frac{\mathbf{q}(k)-k+2\Theta(k)}{\mathbf{q}(k)-k }\Big|^{2}\Big)^{\frac{1}{2}} }+\mathit{O}(|k|^{-1}).     
\end{align}
To complete the proof of~(\ref{QuasiToStand}), I require that the top and bottom expressions on the right side are $\mathbf{r}_{+}^{\frac{1}{2}}(k) +\mathit{O}(|k|^{-1})$ and $S\big(k\Theta(k)\big) \mathbf{r}_{-}^{\frac{1}{2}}(k) +\mathit{O}(|k|^{-1})$, respectively.

By the definitions of $\Theta(k)$, $\mathbf{n}(k)$, and $\beta(k) $, I can write
\begin{align*}
  \Big|\frac{\mathbf{q}(k)-k+2\Theta(k)}{\mathbf{q}(k)-k }\Big|= &  \Big|\frac{\big(\mathbf{q}(k)-\frac{\mathbf{n}(k)}{2}\big)+\Theta(k)}{\big(\mathbf{q}(k)-\frac{\mathbf{n}(k)}{2}\big)-\Theta(k) }\Big|  \\ = & \Big|\frac{\frac{\mathbf{n}(k)}{2}\big(\mathbf{q}\big(\frac{\mathbf{n}(k)}{2}+\frac{2\beta(k)}{\mathbf{n}(k)} \big)-\frac{\mathbf{n}(k)}{2}\big)+\beta(k)}{\frac{\mathbf{n}(k)}{2}(\mathbf{q}\big(\frac{\mathbf{n}(k)}{2}+\frac{2\beta(k)}{\mathbf{n}(k)} \big)-\frac{\mathbf{n}(k)}{2}\big)-\beta(k) }\Big|.
  \end{align*}
 Define the variable $  \gamma_{0}(k)=\frac{n}{2}(\mathbf{q}(\frac{n}{2}+\frac{2\beta(k)}{n} )-\frac{n}{2} ) $ for $n=\mathbf{n}(k)$.  The Kr\"onig-Penney relation~(\ref{Energies}) to second-order for large $|k|$ yields that $\gamma_{0}(k)$ satisfies
\begin{align*}
 \beta^{2}(k) = \gamma_{0}^{2}(k)-\frac{\alpha}{2\pi} \gamma_{0}(k) +\mathit{O}(|k|^{-1}) . 
 \end{align*}
This gives the following asymptotics for $\gamma_{0}(k)$:
\begin{align*}
\gamma_{0}(k)=\left\{  \begin{array}{cc} \frac{\alpha}{2\pi}+\gamma\big(\beta(k)\big)+\mathit{O}(k^{-1}) & \hspace{.5cm}  S(k)\Theta(k)>0 ,  \\    \\ -\gamma\big(\beta(k)\big)+\mathit{O}(|k|^{-1}) & \hspace{.5cm} S(k)\Theta(k)<0 .
    \end{array} \right.  
\end{align*}
 I thus obtain 
\begin{align}\label{JizzJazz}
\Big| \Big(1+\Big|\frac{\mathbf{q}(k)-k+2\Theta(k)}{\mathbf{q}(k)-k }\Big|^{2}\Big)^{-\frac{1}{2}}    -\mathbf{r}_{-}^{\frac{1}{2}}(k) \Big|=\mathit{O}(|k|^{-1}),
\end{align}
where for the case $S(k)\Theta(k)>0$, I have used the identity
   $$\Big| \frac{ \frac{\alpha}{2\pi}+\gamma(\beta) +\beta   }{ \frac{\alpha}{2\pi}  +\gamma(\beta) -\beta   }   \Big|= \Big| \frac{  \gamma(\beta)+\beta  }{ \gamma(\beta)-\beta  }  \Big|. $$
Therefore, combining~(\ref{JizzJazz}) with~(\ref{BarlyFields}), the coefficient $\eta\big(k,-\mathbf{n}(k)\big)$ is equal to  $S\big(k\Theta(k)\big) \mathbf{r}_{-}^{\frac{1}{2}}(k)+\mathit{O}(|k|^{-1})$.

\vspace{.5cm}

\noindent Part (2):\hspace{.15cm} By converting  $\widetilde{\psi}_{k}$ to the momentum basis, operating with $e^{\textup{i}vX}$, and translating back to the Bloch basis $(\widetilde{\psi}_{k+v+n})_{n\in \Z}$, the coefficients $\kappa_{v}(k,n)$ can be written as
\begin{align}\label{Kap} \kappa_{v}(k,n)= \sum_{m\in \Z} \overline{\eta}(k+v+n,m-n)\eta(k,m).     
\end{align}
From~(\ref{Etas}), I have the inequality 
\begin{align}\label{SlyGirl}
|\eta(p,r)|\leq \frac{|\mathbf{q}^{2}(p)-p^{2}| }{\big|\mathbf{q}(p)+p+r\big|\,\big|\mathbf{q}(p)-p-r\big|}\leq \frac{ h}{ \big|\mathbf{q}(p)+p+r\big|\,\big|\mathbf{q}(p)-p-r\big|} , 
\end{align}
where $h:=\sup_{k\in \R}\mathbf{q}^{2}(p)-p^{2}$ and the first inequality follows by using the lower bound $\frac{4\mathbf{q}^{2}(p)|e^{2\pi\textup{i}(\mathbf{q}(p)-p)}-1|^{2}   }{|\mathbf{q}^{2}(p)-p^{2}|^{2}}$ for $N_{p}$ (i.e. from the single term $n=0$ from its sum).  Recall that  for large enough $|p|$, then $|\mathbf{q}(p)-p |\leq \frac{\alpha}{\pi|p|} $.  Roughly speaking, I will find the inequality (\ref{SlyGirl})  useful when $|p|$ is large enough so $|\mathbf{q}(p)-p|\leq \frac{1}{4}$ and $r$ is not $0$ or $-\mathbf{n}(p)$.  In that case, the terms in the denominator have the lower bounds $\big|\mathbf{q}(p)+p+r\big|,\,\big|\mathbf{q}(p)-p-r\big|\geq \frac{1}{4}$.  When~(\ref{SlyGirl}) is not of use, I still have the trivial bound  $|\eta(p,r)|\leq 1$.

For notational ease, I will restrict to the case $\mathbf{n}(k)\neq \mathbf{n}(k+v)$ in this proof.  Let me first focus on the values $n$  such that $ \frac{\alpha }{\pi}|k+v+n|^{-1}\leq \frac{1}{4}$.  Define the constants $a_{j}\in \R$ through
\begin{align*}
a_{j}\equiv a_{j}(k,v,n)=\left\{  \begin{array}{cc} \mathbf{q}(k)-k &  j=1,  \\ -\mathbf{q}(k)-k &  j=2,  \\ \mathbf{q}(k+v+n)-k-v & j=3, \\-\mathbf{q}(k+v+n)-k-v & j=4 .
    \end{array} \right.  
\end{align*}
I will use the inequality~(\ref{SlyGirl}) to bound the sum of the terms $m$ from~(\ref{Kap}) with the special values $m= 0, -\mathbf{n}(k),  n,\,-\mathbf{n}(k+v)-n$ removed:
\begin{align}\label{Giorgi}
\frac{1}{h^{2}} \sum_{\substack{ m\neq 0, -\mathbf{n}(k),   \\  n,\,-\mathbf{n}(k+v)-n  }  }  \big|\overline{\eta}(k+v+n,m-n)\eta(k,m)\big|     \leq \sum_{\substack{ m\neq 0, -\mathbf{n}(k),   \\  n,\,-\mathbf{n}(k+v)-n  }  }\frac{1}{|m-a_{1}||m-a_{2}||m-a_{3}||m-a_{4}|}.
\end{align}
 Let $d>0$  be the largest radius for intervals centered at integer points such that the intervals never contain more than two of the $a_{j}'s$:
$$ d\equiv d(k,v,n)= \sup \big\{ s\in \R_{+}\,\big| \,\forall (m\in \Z):\, |\{ a_{1},a_{2},a_{3},a_{4}  \}\cap [m-s,m+s] |\leq 2    \big\}  .   $$
If no three elements of $\{a_{1},a_{2},a_{3},a_{4}\}$ are within a radius $d>0$ from any single integer, then the  sum~(\ref{Giorgi}) has the following bound 
\begin{align*}
\frac{1}{d^{2}}\sum_{j=1,2,3,4} \sum_{\substack{ m\neq 0, -\mathbf{n}(k),   \\  n,\,-\mathbf{n}(k+v)-n  }  }\frac{1}{|m-a_{j}|^{2}}< \frac{8}{d^{2}}\Big(1+\int_{0}^{\infty}dx\frac{1}{(\frac{1}{4}+x)^{2} }\Big)=\frac{40}{d^{2}}.
\end{align*}
The inequality is a Riemann upper bound using that the distance of the $a_{j}$'s to an integer not equal to $0$, $-\mathbf{n}(k)$,  $n$, or $-\mathbf{n}(k+v)-n $ is $\geq \frac{1}{4}$.  I  claim  that the  
maximal radius $d(k,v,n)$ increases proportionally to $|k|\vee |n|$.  Within the nearest integer, I have $a_{1}\approx 0$, $a_{2}\approx-2k$, $a_{2}\approx n$, and $a_{4}\approx-2(k+v)-n$. For $n$ large enough (i.e. $|n|\geq 6|k|$), then the claim is clearly true, since $|v|\leq \frac{1}{2}|k|$.  For $|n|$ on the order of $|k|$ or smaller, I already have that $0$ and $-2k$ are far apart.  In order to have three of the $a_{j}$'s within a small radius,  then both $n$ and  $-2(k+v)-n$ must be near $0$ or $-2k$.  However, $n$ and  $-2(k+v)-n$ can only reach within a distance $\leq \frac{1}{4}|k|$ of each other if both are at least a distance $\geq\frac{1}{4}|k|$ from  $0$ and $-2k$ (again by the constraint $|v|\leq \frac{1}{2}|k|$).  

For the finitely many values of $n$ such that $ \frac{\alpha}{\pi}|k+v+n|^{-1}> \frac{1}{4}$, I can give almost the same treatment.  For the integers $m$ closest to $-\mathbf{q}(k+v+n)-k-v$ and $\mathbf{q}(k+v+n)-k-v$,  I simply use the bound $|\eta(k+v+n,\,m)|\leq 1$.  The $m$'s closest to those values must be a distance $\geq \frac{1}{4}|k|$ from either $0$ or $-2k$, so $|\eta(k,m)|\leq \frac{8h}{|k|}$.  The remainder of the terms can be treated as in the case above.  In both cases, the sums were bounded by a constant multiple of $(|k|\vee |n| )^{-2}$.      

Now, I deal with the four exceptional terms $m=0,\,-\mathbf{n}(k),\, n,\, -n-\mathbf{n}(k+v)$, where exactly one the terms $\eta(k+v+n,m-n)$  or $\eta(k,\,m)$ is of the form $\eta(p,r)$ for  $r=0,\,-\mathbf{n}(p)$.  For that term,  I  use the inequality $|\eta(p,r)|\leq 1$ rather than~(\ref{SlyGirl}).  For the other term, I apply~(\ref{SlyGirl}) and approximate $2(k+v)\approx \mathbf{n}(k+v) $ and $2k\approx \mathbf{n}(k)$ (and I double the constant $h$ to cover the error resulting from the approximation).  The term $|\overline{\eta}(k+v+n,\,m-n)\eta(k,\,m)|$ is smaller than 
\begin{align*}
\left\{  \begin{array}{cc}  2 h\big(|n||n+\mathbf{n}(k+v) | \big)^{-1} &  m=0 , \\  2 h \big(|n+\mathbf{n}(k)||n+ \mathbf{n}(k+v)-\mathbf{n}(k)  |  \big)^{-1} &  m=-\mathbf{n}(k) , \\  2 h \big(|n||n+\mathbf{n}(k)  | \big)^{-1} & m=n , \\2 h\big(|n+\mathbf{n}(k+v)||n+\mathbf{n}(k+v)-\mathbf{n}(k)  | \big)^{-1} & m=-n-\mathbf{n}(k+v).
    \end{array} \right.  
\end{align*}
All of the four terms are bounded by $2h$ times
\begin{align}\label{Slacker}
   \big(|n|\wedge|n-\mathbf{n}(k)|\big)^{-1}\big(|n+\mathbf{n}(k)|\wedge|n+\mathbf{n}(k+v)|\big)^{-1}.
\end{align}
For values of $n$ near $I(k,v)$, this will be larger than the expression $(|k|\vee |n| )^{-2}$ that bounds the remainder of the terms not in $I(k,v)$, although for larger values of $n$, they have the same order.  Thus,  the values $|\kappa_{v}(k,n)|$ are bounded by some multiple $h^{\prime}$ of~(\ref{Slacker}).

I now have that
\begin{align*}
&\sum_{n\notin I(k,v)}\textup{dist}\big(n,\,I(k,v)\big) \,|\kappa_{v}(k,n)|^{2} \\ &< h^{\prime}  \sum_{n\notin I(k,v)} \textup{dist}\big(n,\,I(k,v)\big) \big(|n|\wedge|n-\mathbf{n}(k)+\mathbf{n}(k+v)|\big)^{-2}\big(|n+\mathbf{n}(k)|\wedge|n+\mathbf{n}(k+v)|\big)^{-2}
\\ & \leq h^{\prime}  \sum_{\substack{ n\notin I(k,v) \\ 4S(k)n\geq -|k| } }\big(|n+\mathbf{n}(k)|\wedge|n+\mathbf{n}(k+v)|\big)^{-2}+ h^{\prime}\sum_{ \substack{n\notin I(k,v) \\  4S(k)n \leq -|k|} } \big(|n|\wedge|n-\mathbf{n}(k)+\mathbf{n}(k+v)|\big)^{-2}.
\end{align*}
Since $|v|\leq \frac{1}{2}|k|$, the above is bounded by the following
\begin{align*}
 2h^{\prime}  \sum_{m=1 }^{\infty}\frac{1}{(\frac{1}{4}|k|+m)^{2}}= \frac{8 h^{\prime} }{|k|}\,\frac{4}{|k|}  \sum_{m=1 }^{\infty}\frac{1}{(1+\frac{4m}{|k|})^{2}} \leq  \frac{8h^{\prime}}{|k|}\int_{0}^{\infty}dx\frac{1}{(1+x)^{2}}=\frac{8 h^{\prime} }{|k|},
\end{align*}
where the inequality is by a Riemann upper bound.  

\vspace{.5cm}

\noindent Part (3): \hspace{.15cm} The analysis from Part (1) gives errors for $|\kappa_{v}\big(k,m\big)|^{2},\,m\in I(k,v)$ of order $\mathit{O}(|k|^{-1})$, but there is additional decay when $\beta(k)$ or $\beta(k+v)$ are large.  All the cases involve similar reasoning, so I focus on the case for $|\kappa_{v}\big(k,-\mathbf{n}(k+v )\big)|^{2}$ with $\mathbf{n}(k+v)\neq \mathbf{n}(k)$.  If there is a $C'$ such that 
\begin{align}\label{Clooney}
 \big|\kappa_{v}\big(k,-\mathbf{n}(k+v )\big)- S\big((k+v)\Theta(k+v)\big)  \mathbf{r}_{-}^{\frac{1}{2}}(k+v)\mathbf{r}_{+}^{\frac{1}{2}}(k) \big|\leq \frac{C'}{|k|}
\end{align}
for some $C'>0$, then 
\begin{align*}
\Big| \big|\kappa_{v}\big(k,-\mathbf{n}(k+v )\big)\big|^{2}-\mathbf{r}_{-}(k+v)\mathbf{r}_{+}(k)\Big|\leq & \frac{2C'}{|k|}\mathbf{r}_{-}^{\frac{1}{2}}(k+v)\mathbf{r}_{+}^{\frac{1}{2}}(k) \\ \leq & \frac{C''}{|k|}  \big(1+|\beta(k+v)|\big)^{-1} .
\end{align*}
The second inequality is for some $C''>0$, since $\mathbf{r}_{-}(k)$ is bounded by a constant multiple of $(1+|\beta(k)|)^{-2}$.  Thus, establishing~(\ref{Clooney})  is sufficient for the proof.

By~(\ref{Kap}) and the triangle inequality,
\begin{align*}
 \big|\kappa_{v}\big(k,-\mathbf{n}(k+v )\big)- \overline{\eta}&\big(k+v-\mathbf{n}(k+v),\mathbf{n}(k+v)\big)\eta(k,0)\big| \\ 
  \leq & \big|\overline{\eta}\big(k+v-\mathbf{n}(k+v),-\mathbf{n}(k)+\mathbf{n}(k+v)\big)\eta\big(k,-\mathbf{n}(k)\big)\big| \\ &+\big|\overline{\eta}(k+v-\mathbf{n}(k+v),0)\eta\big(k,-\mathbf{n}(k+v)\big)\big| \\  &+ \big|\sum_{m\neq I(k,k+v)}\overline{\eta}\big(k+v-\mathbf{n}(k+v),m+\mathbf{n}(k+v)\big)\eta(k,m)       \big|.
 \end{align*}
 By~(\ref{QuasiToStand}), $\sum_{n\neq 0,-\mathbf{n}(k)  }| \eta(k,n)|^{2}=\mathit{O}(k^{-2})$ and with Cauchy-Schwarz the last term above is $\mathit{O}(|k|^{-2})$.  Moreover, by the bounds  in Part (2), the first two terms are each bounded by a multiple $c'>0 $  of $|k|^{-1} $.  For large enough $|k|$, I have
$$\big|\kappa_{v}\big(k,-\mathbf{n}(k+v )\big)- \overline{\eta}\big(k+v-\mathbf{n}(k+v),\mathbf{n}(k+v)\big)\eta(k,0)\big|\leq \frac{3c'}{|k|}.  $$
  By the analysis in Part (1), there is a $C>0$ such that
\begin{align*}
\big|\overline{\eta}\big(k+v-\mathbf{n}(k+v),\mathbf{n}(k+v)\big)\eta(k,0)- S\big((k+v)\Theta(k+v)\big)  \mathbf{r}_{-}^{\frac{1}{2}}(k+v)\mathbf{r}_{+}^{\frac{1}{2}}(k) \big|\leq \frac{C}{|k|}.   
\end{align*}
Putting the above inequalities together, then I obtain~(\ref{Clooney}) for $C'=C+3c'$.

\vspace{.5cm}

\noindent Part (4):\hspace{.15cm} By the bound in Part (3) for $\big|\kappa_{v}\big(k,-\mathbf{n}(k+v)\big)\big|^{2}$ when $v=\frac{m}{2}+\theta-k$,
\begin{align*} \int_{-\frac{1}{4}  }^{\frac{1}{4} }d\theta\,\big|\kappa_{\frac{m}{2}+\theta-k}(k,-m)\big|^{2} & \leq   \int_{-\frac{1}{4}  }^{\frac{1}{4}}d\theta\,\Big(\mathbf{r}_{+}(k)\mathbf{r}_{-}\big(\frac{m}{2}+\theta\big)   +\frac{   c}{|k|\big(1+|\beta(\frac{m}{2}+\theta)| \big)      }\Big)\\ & \leq  \frac{2}{|m|} \int_{-\infty }^{\infty}dw\,\frac{1}{1+\big|\frac{|w|+\gamma(w)   }{|w|-\gamma(w)  }   \big|^{2}  }+\frac{c}{2|k|}\\  &=\mathit{O}(|k|^{-1}),   
\end{align*}
where the order equality uses that $|m|\geq |k|$.

\end{proof}

Let  $T:L^{1}(\R)$ be the trace preserving map with integral kernel $T(k_{1},k_{2})=\mathcal{R}^{-1}J(k_{1},k_{2})$.  The idealized momentum process~(\ref{Master}) has a pseudo Poisson form in which jump times are determined by an outside Poisson clock, and the jump transition densities are given by the operator $T$.  A similar structure holds for the original Lindblad dynamics by (2) of Lem.~\ref{LindbladTech}.

\begin{proof}[Proof of Thm.~\ref{FWLimit}]
Define the map $\rho\rightarrow [\rho]_{\scriptscriptstyle{\textup{D}}}$  from $\mathcal{B}_{1}\big(L^{2}(\R)\big)$ to  $L^{1}(\R)$, which sends a density matrix to its diagonal density in the extended-zone scheme representation $[\rho]_{\scriptscriptstyle{\textup{D}}}(k):={ }_{\scriptscriptstyle{Q}}\langle k|\rho|k\rangle_{\scriptscriptstyle{Q}  }$.  The diagonal map is well-defined by the discussion in Appendix~\ref{SecDiaTech}.   Let $\xi=(t_{1},\cdots ;t_{\mathcal{N} })$ be the sequence of Poisson times less than $t$.   Define the map $\Phi_{t,\xi}^{(\lambda)}:\mathcal{B}_{1}\big(L^{2}(\R)\big)$ as
$$\Phi_{t,\xi}^{(\lambda)}(\rho)=\mathcal{R}^{-\mathcal{N}}    e^{-\frac{\textup{i}(t-t_{\mathcal{N}})}{\lambda}H}  \Psi(   \cdots  e^{-\frac{\textup{i}(t_{2}-t_{1})}{\lambda}H}\Psi(e^{-\frac{\textup{i}}{\lambda}t_{1}H}\rho e^{\frac{\textup{i}t_{1}}{\lambda}H})e^{\frac{\textup{i}(t_{2}-t_{1})}{\lambda}H}\cdots)e^{\frac{\textup{i}(t-t_{\mathcal{N} })}{\lambda}H}.     $$ 
The maps $\Phi_{t,\xi}^{(\lambda)}$ are completely positive and preserve trace for all $\xi$ and $t$.
Similarly to the construction~(\ref{Couture}) of the solution to the Lindblad dynamics, I have that $\rho_{\lambda,t}=\mathbb{E}\big[\Phi_{t,\xi}^{(\lambda)}(\rho)   \big]$ by (2) of Lemma~\ref{LindbladTech}, where $\mathbb{E}$ is the expectation with respect to the Poisson process with rate $\mathcal{R}$ for the sequences $\xi$.  Also $\mathcal{D}_{t}=\mathbb{E}\big[T^{\mathcal{N}(t,\xi)}[\rho]_{\scriptscriptstyle{\textup{D}}}\big]$, where $\mathcal{N}(t,\xi)$ is the value of the Poisson process at time $t$.    

  I can write
\begin{align*}
D_{\lambda,t }-\mathcal{D}_{t}& = \mathbb{E}\Big[  \big[\Phi_{t,\xi}^{(\lambda)}  ( \rho)\big]_{\scriptscriptstyle{\textup{D}}}-   T^{\mathcal{N}(t,\xi)}[\rho]_{\scriptscriptstyle{\textup{D}}} \Big]\\ &=\mathbb{E}\Big[\sum_{n=1}^{\mathcal{N}(t,\xi)}      T^{\mathcal{N}(t,\xi)-n}\big[\Phi_{t_{n},\xi}^{(\lambda)}  ( \rho)\big]_{\scriptscriptstyle{\textup{D}}}-T^{\mathcal{N}(t,\xi)-n+1}\big[\Phi_{t_{n-1},\xi}^{(\lambda)}  ( \rho)\big]_{\scriptscriptstyle{\textup{D}}}          \Big],   
\end{align*}
where I have inserted a telescoping sum.   Since $T$ is contractive in the $1$-norm, $\Phi_{r,\xi}^{(\lambda)}$ is contractive in trace norm, and conjugation by $e^{-\frac{\textup{i}r }{\lambda}H}$  is contractive in trace norm,  I have the second inequality below
\begin{align*}
\| D_{\lambda,t}-\mathcal{D}_{t}\|_{1}  \leq &  e^{-\mathcal{R}t}\sum_{\mathcal{N}=1}^{\infty}\mathcal{R}^{\mathcal{N}}\sum_{n=1}^{\mathcal{N}}\left \|\int_{0\leq t_{1}\cdots t_{\mathcal{N}}\leq t}  T^{\mathcal{N}-n}\big[\Phi_{t_{n},\xi}^{(\lambda)}  ( \rho)\big]_{\scriptscriptstyle{\textup{D}}}-T^{\mathcal{N}-n+1}\big[\Phi_{t_{n-1},\xi}^{(\lambda)}  ( \rho) \big]_{\scriptscriptstyle{\textup{D}}}  \right \|_{1} 
\\ \leq &  e^{-\mathcal{R}t}\sum_{\mathcal{N}=1}^{\infty}\mathcal{R}^{\mathcal{N}}\sum_{n=1}^{\mathcal{N}}\int_{0\leq t_{1}\cdots \leq t_{n-1}\leq t_{n+1}\leq \cdots  t_{\mathcal{N}}\leq t}\\ &\times \sup_{\| \rho\|_{1}=1}     \left \| \int_{t_{n-1}}^{t_{n+1}}dt_{n}\,\Big(    \mathcal{R}^{-1}\big[\Psi(e^{-\frac{\textup{i}(t_{n}-t_{n-1} )}{\lambda}H}\rho e^{\frac{\textup{i}(t_{n}-t_{n-1}) }{\lambda}H})\big]_{\scriptscriptstyle{\textup{D}}} -T[\rho]_{\scriptscriptstyle{\textup{D}}}  \Big) \right  \|_{1}
   \\   = & \Big(e^{-\mathcal{R}t}\sum_{\mathcal{N}=1}^{\infty}\frac{\mathcal{R}^{\mathcal{N}}t^{\mathcal{N}-1}}{(\mathcal{N}-1)!}    \Big)  \sup_{a,b\in \R}\sup_{\| \rho\|_{1}=1}     \left \| \int_{a}^{b}dc\,\Big(    \mathcal{R}^{-1}\big[\Psi(e^{-\frac{\textup{i}c}{\lambda}H}\rho e^{\frac{\textup{i}c}{\lambda}H})\big]_{\scriptscriptstyle{\textup{D}}} -T[\rho]_{\scriptscriptstyle{\textup{D}}}  \Big) \right  \|_{1},
\end{align*}
where  I identify $t_{0}\equiv 0$ and $t_{N+1}\equiv t$ for the boundary terms. 
The first inequality above is the triangle inequality.  With the above and the summation formula $$\mathcal{R}+\mathcal{R}^{2}t =  e^{-\mathcal{R}t}\sum_{\mathcal{N}=1}^{\infty}\mathcal{N} \frac{\mathcal{R}^{\mathcal{N}}t^{\mathcal{N}-1}}{(\mathcal{N}-1)!},   $$
it follows that
\begin{align}\label{Tritip}
 \| D_{\lambda,t}-\mathcal{D}_{t} \|_{1}\leq (\mathcal{R}+\mathcal{R}^{2}t)   \sup_{a,b\in \R_{+} }\sup_{\| \rho\|_{1}=1}  \left \| \int_{a}^{b}dc\,\big(  \mathcal{R}^{-1}\big[\Psi(e^{-\frac{\textup{i}c}{\lambda}H}\rho e^{\frac{\textup{i}c}{\lambda}H})\big]_{\scriptscriptstyle{\textup{D}}}-T[\rho]_{\scriptscriptstyle{\textup{D}}}  \big)  \right \|_{1}.
 \end{align}

The remainder of the proof is concerned with proving that the supremum in~(\ref{Tritip}) is bounded by a multiple of $\lambda$ for $\lambda\ll 1$.   By a direct calculation,
\begin{align}
\int_{a}^{b}dc\,&\big(    \mathcal{R}^{-1}\big[\Psi(e^{-\frac{\textup{i}c}{\lambda}H}\rho e^{\frac{ic}{\lambda}H})\big]_{\scriptscriptstyle{\textup{D}}}(k)-T[\rho]_{\scriptscriptstyle{\textup{D}}}(k)\big)\nonumber  \\   = & -\textup{i}\lambda  \int_{\R}dv\,\frac{j(v)}{\mathcal{R}}\sum_{n\neq m}\kappa_{v}(k-n-v,n)\overline{\kappa}_{v}(k-m-v,m)\rho(k-n-v,k-m-v) \nonumber \\ & \times \frac{e^{-\frac{ib}{\lambda}\big(E(k-n-v)-E(k-m-v)\big)}-e^{-\frac{\textup{i}a}{\lambda}\big(E(k-n-v)-E(k-m-v)\big)}     }{E(k-n-v)-E(k-m-v) }. 
\end{align}
Using the triangle inequality and making a change of variables $k-n-v\rightarrow k$ and $M=n-m$, I have
\begin{align}\label{Bouncer}
\Big\| \int_{a}^{b}dc\, &\big(    \mathcal{R}^{-1}\big[\Psi(e^{-\frac{\textup{i}c}{\lambda}H} \rho e^{\frac{\textup{i}c}{\lambda}H})\big]_{\scriptscriptstyle{\textup{D}}}(k)-T[\rho]_{\scriptscriptstyle{\textup{D}}}(k)\big) \Big\|_{1}\nonumber \\ & \leq \lambda\sum_{M\neq 0} \int_{\R} dk\,\int_{\R}dv\,\frac{j(v)}{\mathcal{R}}\,  \frac{\sum_{n}|\kappa_{v}(k,n)|\,|\kappa_{v}(k+M,n-M)|}{|E(k)-E(k+M) | }|\rho(k,k+M)|\nonumber   \\ & \leq \lambda \big(\frac{2\pi}{\alpha}+\sum_{M\neq 0}C_{M}        \big) ,
\end{align}
where  the values $C_{M}$ are defined as
$$C_{M}=  \sup_{|2k+M|\geq  1 } \int_{\R}dv\,\frac{j(v)}{\mathcal{R}}\,  \frac{\sum_{n}|\kappa_{v}(k,n)|\,|\kappa_{v}(k+M,n-M)|}{|E(k)-E(k+M) | }.$$
The second inequality in~(\ref{Bouncer}) can be found by splitting the integration $\int_{\R}dk$ into the regions $|2k+M|\leq 1$ and $|2k+M|>1$.  This splitting isolates some bad behavior (non-decay for large $M$) occurring in regions of $k$ where $|2k+M|$ is small.  The $C_{M}$'s in~(\ref{Bouncer}) arise by applying Holder's inequality over the integration $|2k+M|>1$ and by the  inequality $|\rho(k_{1},k_{2})|\leq \frac{1}{2}\rho(k_{1},k_{1})+\frac{1}{2}\rho(k_{2},k_{2})$ for the integral kernel of $\rho$.  The kernel inequality follows because $\rho$ is a positive operator.   The $\frac{2\pi}{\alpha}$ term in~(\ref{Bouncer}) comes from the $|2k+M|\leq 1$ integration for which I apply mainly brute force:
\begin{align*}
\sum_{M\neq 0} & \int_{|2k+M|\leq 1 } dk\, \int_{\R}dv\,\frac{j(v)}{\mathcal{R}}\,  \frac{\sum_{n}|\kappa_{v}(k,n)|\,|\kappa_{v}(k+M,n-M)|}{|E(k)-E(k+M) | }|\rho(k,k+M)| \\ & \leq \big(\sup_{k,M}|E(k)-E(k+M) |^{-1}\big)\sum_{M\neq 0}\int_{|2k+M|\leq 1 }dk\,\big(\frac{1}{2}\rho(k,k)+\frac{1}{2}\rho(k+M,k+M)\big)\\ & \leq  (\inf_{n} g_{n}  )^{-1}2\int_{\R}dk\,\rho(k,k)=\frac{2\pi}{\alpha},  
\end{align*}
where $g_{n}>0$ are the gaps between the energy bands occurring at momenta $k\in \frac{1}{2}\Z-\{0\}$.   The infemum of the energy gaps $g_{n}$ is  $\frac{\alpha}{\pi}$.  The sum over $n\in \Z$ of $|\kappa_{v}(k,n)|\,|\kappa_{v}(k+M,n-M)|$ is bounded through the Cauchy-Schwarz inequality and $\sum_{n}|\kappa_{v}(k,n)|^{2}=1$.  The key observation is $k$ and $k+M$ must lie on different energy bands, and it follows that $E(k)$ and $E(k+M)$ differ by at least the length of the smallest energy band  gap.

  Next, I need to show that the sum of the $C_{M}$'s is finite.  A single $C_{M}$ can be bounded by $(\inf_{n} g_{n}  )^{-1}$ using some of the same reasoning as above.   I will show  $C_{M}$ decays on the order of $|M|^{-\frac{3}{2}}$, and thus is a summable series.  The  difference $|E(k)-E(k+M)|$ necessarily becomes large for $|M|\gg 1$ except for cases when $k+M$ is close to $-k$.  However, by my restriction $|2k+M|\geq 1$, the momenta $k$ and $k+M$ will not lie on neighboring energy bands, and thus their energies must differ by at least the length $L_{|M|}$ of the $|M|$th energy band.  By~\cite[Thm.2.3.3]{Solve}, $L_{n}$ grows with linear order for $n\gg 1$.  Also, if $|k|\wedge |k+M|\leq \frac{1}{4}|M|$, then  $|E(k)-E(k+M)|\geq \sum_{ \frac{1}{4}|M|\leq n\leq |M|}L_{n}$ will grow on quadratic order in $\mathit{O}(|M|)$.  
  
  I have  $C_{M}< L_{M}^{-1}C^{\prime}_{M}+C_{M}^{\prime \prime}$, where $C^{\prime}_{N}$ and $C^{\prime \prime}_{N}$ are defined as
\begin{eqnarray*}  
  C_{M}^{\prime}&=& \sup_{\substack{|2k+M|\geq  1 \\  |k|\wedge |k+M|\geq \frac{1}{4}|M| }  } \int_{\R}dv\,\frac{j(v)}{\mathcal{R}}\,  \sum_{n}|\kappa_{v}(k,n)|\,|\kappa_{v}(k+M,n-M)|,  \\
   C_{M}^{\prime \prime}&=& \sup_{ |k|\wedge |k+M|\leq \frac{1}{4}|M|   } \int_{\R}dv\,\frac{j(v)}{\mathcal{R}}\,   \frac{\sum_{n}|\kappa_{v}(k,n)|\,|\kappa_{v}(k+M,n-M)|}{|E(k)-E(k+M) | }.
  \end{eqnarray*}
By the observation above, $C_{M}^{\prime \prime}\leq \sup_{ |k|\wedge |k+M|\leq \frac{1}{4}M   }|E(k)-E(k+M)|^{-1} $ decays quadratically. The $C_{M}^{\prime \prime}$ are therefore summable, and I turn to the only somewhat delicate part of the proof, which requires isolating the problematic terms in the sum of the $|\kappa_{v}(k,n)|\,|\kappa_{v}(k+M,n-M)|$ contributing to $C_{M}^{\prime}$ to which I can apply Lem.~\ref{BadTerms}.  

For fixed $k,v$, the  Cauchy-Schwarz inequality and $\sum_{n}|\kappa_{v}(k,n)|^{2}=1$  yield 
\begin{multline}\label{Rabble}
\sum_{n}|\kappa_{v}(k,n)|\,|\kappa_{v}(k+M,n-M)|\leq  \sum_{ \substack{ n\in I(k,v) \\  n\in I(k+M,v)+M}}|\kappa_{v}(k,n)|\,|\kappa_{v}(k+M,n-M)| \\+\Big(\sum_{n\notin I(k,v)}|\kappa_{v}(k,n)|^{2}\Big)^{\frac{1}{2}}+\Big(\sum_{n\notin I(k+M,v)+M}|\kappa_{v}(k+M,n)|^{2}\Big)^{\frac{1}{2}}.      
\end{multline}
Under the constraint $|k|\wedge |k+M|\geq \frac{1}{4}|M|$ and by Part (1) of Lem.~\ref{BadTerms}, the  terms on the bottom line of~(\ref{Rabble}) decay on the order $|k|^{-\frac{1}{2}}\leq 2|M|^{-\frac{1}{2}}$ and $|k+M|^{-\frac{1}{2}}\leq 2|M|^{-\frac{1}{2}}$, respectively.  The weighted integration $\int_{\R} dv\, \frac{j(v)}{\mathcal{R}}$ is finite so these terms make contributions to $C_{M}^{\prime}$ that vanish with order $\mathit{O}(|M|^{-\frac{1}{2}})$.   

Controlling the integral $\int_{\R}dv\,\frac{ j(v)}{\mathcal{R}}$ of the first term on the right side in~(\ref{Rabble}) will now require invoking the decay of $j(v)$ at infinity and its boundedness through Lem.~\ref{BadTerms}.  The constraints $ n\in I(k,v)$, $ n\in I(k+M,v)+M$, $|2k+M|\geq 1$, and $|k|\wedge |k+M|\geq \frac{1}{4}|M|$ leave the following possibilities:
\begin{itemize}
\item The inequality  $|k+v+\frac{1}{2}M|\leq \frac{1}{4}$ holds and either $n=0$ or $n=M$ holds.
\item The inequality $|v-\frac{1}{2}M| \leq \frac{1}{4}$ holds and either     $|k+\frac{1}{2}n|\leq \frac{1}{4} $ or $|k+v+\frac{1}{2}n|\leq \frac{1}{4}$ holds.
\end{itemize}
The second case vanishes, since $|v|\approx \frac{1}{2}|M|\gg 1$ and by Jensen's inequality: 
$$ \int_{|v|\geq \frac{1}{2} M}dv\,\frac{j(v)}{\mathcal{R}}\,  \sum_{n}|\kappa_{v}(k,n)|\,|\kappa_{v}(k+M,n-M)| < \int_{|v|\geq \frac{1}{2}M}dv\,\frac{j(v)}{\mathcal{R}}\leq  \frac{4\sigma}{\mathcal{R}} |M|^{-2},
  $$
  where $\sigma=\int_{\R}dv j(v)v^{2}$.  
The first case follows by Part (4) of Lem.~\ref{BadTerms}.

\end{proof}

\section{Submartingales related to energy } \label{SecTransition}

In this section, I discuss certain key submartingales appearing in both the quantum and the limiting classical settings.  First, let me define an \textit{operator-valued submartingale}.  Consider a probability space $(\Omega, \mathcal{F}, P)$ with a filtration $\mathcal{F}_{t}$, a Hilbert space $\mathcal{H}$, and an operator-valued process $Y_{t}:\Omega \rightarrow \mathcal{B}(\mathcal{H})$ adapted to $\mathcal{F}_{t}$ and satisfying
\begin{align}\label{DefCond}
\mathbb{E}\big[|\langle f|\, Y_{t} f\rangle |   \big]<\infty, \quad \quad f\in \mathcal{H}.
\end{align}   I call $Y_{t}$ a \textit{submartingale} if $ \mathbb{E}\big[Y_{t}\big|\mathcal{F}_{s}  \big]-Y_{s}   $
 is a positive operator for all $t>s$.
Naturally, $Y_{t}$ is a \textit{martingale} if both $Y_{t}$ and $-Y_{t}$ are submartingales.  I can extend my definition to the case in which $Y_{t}$ may take values as an unbounded operator.  In this case, I require that there is a single dense space  
$ \textup{D}\subset \mathcal{H} $  such that domain of $Y_{t}$ almost surely contains  $\textup{D}$ for all $t$ and~(\ref{DefCond}) holds for all $f\in \textup{D}\subset \mathcal{H} $.   In the unbounded case, I refer to the process by the tuple $(Y_{t}, \textup{D})$.

 For the following discussion, I will consider a Schr\"odinger Hamiltonian $H=P^{2}+V(X)$ with positive potential $V$ and domain $\textup{D}(H)\subset L^{2}(\R)$.  The noise in my model is generated by an underlying L\'evy process $L_{t}$ with jump rate density $j(v)$.   As before, let $v_{n}\in \R, t_{n}\in \R_{+}$ denote the jumps and jump-times for the L\'evy process, and $\mathcal{N}_{t}$ be the Poisson counter for the jump-times.   Also let the unitaries $U_{\lambda,t}(\xi)$ be defined as in~(\ref{LaughableMan}).  Define the operator-valued process $G_{t}$ as
\begin{align}\label{FlyOver}
G_{t}\equiv G_{\lambda, t}(\xi)= U_{\lambda, t}^{*}(\xi)G  U_{\lambda,t}(\xi)  
\end{align}
for an observable $G$ acting on $L^{2}(\R)$.  On bounded observables $G$, the trajectories are right weak*-continuous with weak*-limits existing from the left, since the Hamiltonian evolution in the  Heisenberg representation is weak*-continuous. Analogously to~(\ref{Couture}),  the Heisenberg evolution for   $G\in \mathcal{B}\big(L^{2}(\R)\big)$ can be written  as 
\begin{align}\label{Heisenberg}
\Phi_{\lambda,t}^{*}(G)=\mathbb{E} \big[ G_{\lambda, t}(\xi)\big] .  
\end{align}
In Prop.~\ref{StochLem}, the formula~(\ref{Heisenberg}) may be interpreted  as the definition for the dynamical maps $\Phi_{\lambda,t}^{*}:\mathcal{B}\big(L^{2}(\R)\big)$.   I will suppress the $\lambda$ and $\xi$ dependence for the operator processes in the future.

  In the lemma below, I study~(\ref{FlyOver}) for the special cases $G=H$ and $G=H^{\frac{1}{2}}$, and prove that $\big(H_{t}, \textup{D}(H)\big)$ and $\big(H_{t}^{\frac{1}{2}},\textup{D}(H^{\frac{1}{2}})\big)$ are operator submartingales.  I show that each process is a sum of a martingale and an increasing part, and  I have presented the increasing part of the Doob-Meyer decomposition for $H_{t}^{\frac{1}{2}}$ in a form that is not predictable, but which will be useful later.  The linear spaces $\textup{D}(H)$ and $\textup{D}(H^{\frac{1}{2}} )$ are closed under the operation of the unitaries $U_{\lambda,t}(\xi)$, since the evolution $e^{-\frac{\textup{i}t}{\lambda} H}$ clearly leaves the domains invariant, and 
  $$e^{-\textup{i}vX}H^{\frac{1}{2}}e^{\textup{i}vX}-H^{\frac{1}{2}}\hspace{1cm}\text{and}\hspace{1cm} e^{-\textup{i}vX}He^{\textup{i}vX}-H     $$
are relatively bounded to $H^{\frac{1}{2}}$ and $H$, respectively.  These relative bounds can be shown using the Wigner-Weyl identity $e^{-\textup{i}vX}H e^{\textup{i}vX}= (P+v)^{2}+V(X) $ and, in the case of $H^{\frac{1}{2}}$, a resolvent representation for the square root of an operator (see~(\ref{Tybo})).  

Note that I have used the Wigner-Weyl relation $e^{-\textup{i}vX}H^{\frac{1}{2}}e^{\textup{i}vX}= \big((P +v)^{2}+V(X)    \big)^{\frac{1}{2}}$ in the statement of Part (2) of the proposition below. 

 \begin{proposition}\label{StochLem}
Let the operator-valued processes $H_{t}$ and $H_{t}^{\frac{1}{2}}$ be defined as in~(\ref{FlyOver}) for a Schr\"odinger operator  $H=P^{2}+V(X)$ with nonnegative  $V$ and domain  $\textup{D}(H)$. The processes $\big(H^{\frac{1}{2}}_{t},\textup{D}(H^{\frac{1}{2}} )\big)$ and  $\big(H^{\frac{1}{2}}_{t},\textup{D}(H^{\frac{1}{2}} )\big)$  are   submartingales that can be written as a sum of martingale parts $\big(\mathbf{M}_{t},\textup{D}(H)\big)  $, $\big(\mathbf{M}^{\prime}_{t}, L^{2}(\R)\big)$ and increasing parts $\big(\mathbf{A}_{t},\textup{D}(H)\big)   $, $(\mathbf{A}_{t}^{\prime},\textup{D}(H^{\frac{1}{2}})\big)$, respectively, with the forms found below.
  
\begin{enumerate}
\item The submartingale $H_{t}$ is equal to $\mathbf{M}_{t}+\mathbf{A}_{t}$ for
$$
\mathbf{M}_{t}:=2\sum_{n=1}^{\mathcal{N}_{t} }v_{n}P_{t_{n}^{-}}+\sum_{n=1}^{\mathcal{N}_{t} }v_{n}^{2}-\sigma t\quad \text{ and }\quad
\mathbf{A}_{t}:=H+\sigma t.$$

\item The submartingale $H_{t}^{\frac{1}{2}}$ is equal to $\mathbf{M}_{t}^{\prime}+\mathbf{A}_{t}^{\prime}$ for 
\begin{eqnarray*}
\mathbf{M}^{\prime}_{t}&:=&\frac{1}{2}\sum_{n=1}^{\mathcal{N}_{t} }\big((P_{t_{n}^{-}} +v_{n})^{2}+V(X_{t_{n}^{-}})    \big)^{\frac{1}{2}}- \big((P_{t_{n}^{-}}-v_{n})^{2}+V(X_{t_{n}^{-}})    \big)^{\frac{1}{2}},\\
\mathbf{A}^{\prime}_{t}&:=&H^{\frac{1}{2}}+ \frac{1}{2}\sum_{n=1}^{\mathcal{N}_{t} }\big((P_{t_{n}^{-}} +v_{n})^{2}+V(X_{t_{n}^{-}})    \big)^{\frac{1}{2}}+ \big((P_{t_{n}^{-}}-v_{n})^{2}+V(X_{t_{n}^{-}})\big)^{\frac{1}{2}} \\ & &\hspace{7cm}-2\big(P_{t_{n}^{-}}^{2}+V(X_{t_{n}^{-}})    \big)^{\frac{1}{2}}    .
\end{eqnarray*}
\item  Moreover, I have the operator relations  
$$\Phi_{\lambda, t}^{*}(H)= \mathbb{E}\big[ H_{t} \big]=H+\sigma t \quad \text{and} \quad \Phi_{\lambda, t}^{*}(H^{2})= \mathbb{E}\big[  H_{t}^{2}   \big]\leq H^{2}+ 3 \sigma H t +\mathcal{R}\varsigma t +\frac{3}{2}\sigma^{2}t^{2},  $$
where  $\varsigma:= \int_{\R} dv\,\frac{j(v)}{\mathcal{R}}v^{4}$.

\end{enumerate}

 \end{proposition}

\begin{proof}\text{ }\\

\noindent Part (1):\hspace{.15cm}  The energy process $H_{t}$ can be written as
\begin{align*}  H_{t} &=  e^{\frac{\textup{i}t_{1}}{\lambda}  H}e^{-\textup{i}v_{1}X}\cdots e^{-\textup{i}v_{n}X}e^{\frac{\textup{i}(t-t_{n})}{\lambda}H} H\,e^{-\frac{\textup{i}(t-t_{n})}{\lambda}H}e^{\textup{i}v_{n}X}\cdots e^{\textup{i}v_{1}X}e^{-\frac{\textup{i} t_{1}}{\lambda}H}\\ &= H_{t_{n}^{-}}+2v_{n}P_{t_{n}^{-}}+v_{n}^{2}, 
\end{align*}
for $n=\mathcal{N}_{t}$. Through iteration of the above calculation, I obtain the relation $H_{t}=\mathbf{M}_{t}+\mathbf{A}_{t}$. By the symmetry of the rates $j(v)=j(-v)$, it is clear that $\mathbf{M}_{t}$ is a martingale, since $\sum_{m=1}^{\mathcal{N}_{t}}v_{n}^{2}-t\sigma$ is a martingale.

 For $f\in \textup{D}(H)$, I will now verify condition~(\ref{DefCond}) for $H_{t}$ and $\mathbf{M}_{t}$.  For $H_{t}$, I have
$$\mathbb{E}\big[\langle f|  H_{t} |f\rangle \big]=\langle f| H |f \rangle+t\sigma \|f\|_{2}^{2},  $$
since the martingale part vanishes under the expectation. By two applications of  Jensen's inequality for the first inequality below and using that $ \langle f| P^{2}_{t} |f\rangle \leq \langle f |H_{t} | f\rangle$ for the last inequality, 
\begin{align}
\mathbb{E}\big[\big|\big\langle f\big|  \mathbf{M}_{t} \big|f\big\rangle | \big]^{2} & \leq \|f\|_{2}^{2} \mathbb{E}\big[\big\langle f\big| \mathbf{M}_{t}^{2}\big|f\big\rangle ]\nonumber  \\ &=4\|f\|_{2}^{2}\mathbb{E}\Big[\sum_{n=1}^{\mathcal{N}_{t}}v_{n}^{2}\big\langle f\big|P_{t_{n}^{-}}^{2}\big| f\big\rangle  \Big].  
 \nonumber \\ &= 4\sigma\|f\|_{2}^{2} \int_{0}^{t}dr\mathbb{E}\big[\big\langle f\big| P^{2}_{r}   \big|f\big\rangle \big]\nonumber  \\ & \leq 4\sigma\|f\|_{2}^{2} \Big(t\big\langle f\big| H \big|f \big\rangle +\sigma\frac{t^{2}}{2} \|f\|_{2}^{2}\Big).   
\end{align}

The domain for $\mathbf{A}_{t}$ is clear from its form.

\vspace{.5cm}

\noindent Part (2): \hspace{.15cm}  The equality $H_{t}^{\frac{1}{2}}=\mathbf{M}_{t}^{\prime}+\mathbf{A}_{t}^{\prime}$  can be shown through a telescoping sum and the conservation of energy between momentum kicks in a similar way to the argument in Part (1). Also, that  $\mathbf{M}_{t}^{\prime}$ is formally a martingale is clear through the symmetry of the rates $j(v)=j(-v)$, however unlike for Part (1), it is not immediately clear that  $\mathbf{A}_{t}^{\prime}$ is an increasing process.  
  
   I have the equality
\begin{align*}
e^{\textup{i}vX}H^{\frac{1}{2}} e^{-\textup{i}vX}+e^{-\textup{i}vX}H^{\frac{1}{2}} e^{\textup{i}vX}-2H^{\frac{1}{2}}    = \int_{-|v|}^{|v|}da\int_{0}^{a}db\,e^{\textup{i}bX}\textup{i}\big[X,\textup{i}\big[X, H^{\frac{1}{2}}\big]\big]e^{-\textup{i}bX}   
\end{align*}
by the second-order Taylor expansion
\begin{align*}
e^{\textup{i}vX}H^{\frac{1}{2}} e^{-\textup{i}vX}= H^{\frac{1}{2}}+\textup{i}v\big[X,H^{\frac{1}{2}}\big] +\int_{0}^{v}da\int_{0}^{a}db\,e^{\textup{i}bX}\textup{i}\big[X,\textup{i}\big[X,H^{\frac{1}{2}}\big]\big]e^{-\textup{i}bX}.  
\end{align*}
Showing that $\textup{i}[X,\textup{i}[X, H^{\frac{1}{2}}]]$ is  a positive operator will complete the proof that $\mathbf{A}_{t}^{\prime}$ is increasing. 

 Using the identity $a^{\frac{1}{2}}=\frac{1}{\pi}\int_{0}^{\infty}dw\,w^{-\frac{1}{2}}\frac{ a }{a+w }  $
 and functional calculus, the operator $H^{\frac{1}{2}}$ can be represented through its resolvents (see ~\cite[Ch.VIII, Ex.50]{Simon}) as
\begin{align}\label{Tybo}
 H^{\frac{1}{2}}=\frac{1}{\pi}\int_{0}^{\infty}dw\,w^{-\frac{1}{2}}\frac{ H }{H+w }.    
 \end{align}
 Evaluating both sides by the double commutator with $X$,
\begin{align}\label{DoubleCom}
\textup{i}\big[X,\textup{i}\big[X,H^{\frac{1}{2}}\big]\big] &= \textup{i}\Big[X,\,\textup{i}\Big[X, \frac{1}{\pi}\int_{0}^{\infty}dw\,w^{-\frac{1}{2}}\frac{ H }{H+w } \Big]\Big]\nonumber  \\ &= \frac{2}{\pi} \int_{0}^{\infty}dw\,w^{\frac{1}{2}}\Big(\frac{1}{\big(H+w \big)^{2}}-4\frac{1}{H+w}P \frac{1}{H+w}P \frac{1}{H+w}      \Big), 
\end{align}   
since through canonical commutation relations I obtain
\begin{align*}
\textup{i}\Big[X,\,\textup{i}\Big[X,\frac{ H }{H+w } \Big]\Big]=2w\Big(\frac{1}{(H+w)^{2}}- 4 \frac{1}{H+w}P\frac{1}{H+w}P\frac{1}{H+w}\Big)    .
\end{align*}

The following operator inequalities hold:
\begin{align*}
 \frac{1}{H+w}P\frac{1}{H+w}P\frac{1}{H+w}& \leq \frac{1}{H+w}P\frac{1}{P^{2}+w}P\frac{1}{H+w}\\  & \leq  \frac{1}{H+w}\big(\frac{H }{H+w }\big)\frac{1}{H+w}, 
 \end{align*}
where I have used that  $f(u)=\frac{1}{u+w}$ and $f(u)=\frac{u}{u+w}$ are operator monotonically decreasing and increasing functions, respectively.  Applying this inequality to~(\ref{DoubleCom}), 
$$\textup{i}\big[X,\textup{i}\big[X,H^{\frac{1}{2}}\big]\big]\geq \frac{2}{\pi}\int_{0}^{\infty}dw\,w^{\frac{1}{2}}\Big(\frac{1}{(H+w)^{2}}-4\frac{H}{(H+w)^{3}}     \Big)=0.  $$
To see the equality to zero, I compute the two integrals through a change of variables $w=h\tan^{2}(\theta)$ and find  
\begin{eqnarray*}
\frac{1}{\pi}\int_{0}^{\infty}dw\,\frac{w^{\frac{1}{2}}}{(h+w)^{2}}&=&\frac{2}{h^{\frac{1}{2}}\pi}\int_{0}^{\frac{\pi}{2}}d\theta\,\sin^{2}(\theta)=\frac{1}{2h^{\frac{1}{2}}},\\ \frac{4}{\pi}\int_{0}^{\infty}dw\,\frac{ h w^{\frac{1}{2}}}{(h+w)^{3}}&=&\frac{2}{ h^{\frac{1}{2}}\pi} \int_{0}^{\frac{\pi}{2}}d\theta\,\sin^{2}(2\theta)=\frac{1}{2h^{\frac{1}{2}} }.  
\end{eqnarray*}
Hence, $\textup{i}[X,\textup{i}[X,H^{\frac{1}{2}}]]$ is a positive operator, and $\mathbf{A}^{\prime}_{t}$ is increasing. 

Next, I show that $\mathbf{M}_{t}^{\prime}$ satisfies condition~(\ref{DefCond}).  Again by Taylor's formula, 
\begin{align*}
e^{\textup{i}vX}H^{\frac{1}{2}} e^{-\textup{i}vX}-e^{-\textup{i}vX}H^{\frac{1}{2}} e^{\textup{i}vX}=\int_{-v}^{v}db\,e^{\textup{i}bX}\textup{i}\big[X, H^{\frac{1}{2}}  \big]e^{-\textup{i}bX}.
\end{align*}
By similar reasoning as above,
$$\textup{i}\big[X, H^{\frac{1}{2}}  \big]= \frac{2}{\pi}\int_{0}^{\infty}dw\,w^{\frac{1}{2}}\frac{1}{H+w}P\frac{1}{H+w}.     $$
Using that $ |P|\leq H^{\frac{1}{2}}$,  I have the inequality
$$\big|\big\langle f \big|\textup{i}\big[X, H^{\frac{1}{2}}  \big] f\big\rangle \big| \leq \frac{2 }{\pi}\int_{0}^{\infty}dw\,w^{\frac{1}{2}}\Big\langle f\Big| \frac{H^{\frac{1}{2}} }{(H+w)^{2} } f \Big\rangle= \|f\|_{2}^{2}, $$
where I have used the same change of integration as in the functional calculus above.  Since $e^{\textup{i}bX}$  has operator norm one, the above remarks imply
\begin{align}\label{ThinkaDink}
\Big|\left \langle f \Big|
\Big(e^{\textup{i}vX}H^{\frac{1}{2}} e^{-\textup{i}vX}-e^{-\textup{i}vX}H^{\frac{1}{2}} e^{\textup{i}vX} \Big) f \right \rangle \Big|\leq 2|v| \|f\|_{2}^{2}.  
\end{align}
Finally, for $f\in L^{2}(\R)$,
\begin{align*}
 \mathbb{E}\big[\big|\big\langle f\big| \mathbf{M}_{t}' f\big\rangle\big| \big]= &\mathbb{E}\Big[\int_{0}^{t}dr\int_{\R}dv j(v) \Big|\Big\langle f\Big| e^{\textup{i}vX}H_{r}^{\frac{1}{2}} e^{-\textup{i}vX}-e^{-\textup{i}vX}H_{r}^{\frac{1}{2}} e^{\textup{i}vX}      \Big| f \Big\rangle \Big|\Big] \\
  \leq & 2t\|f\|_{2}^{2}\int_{\R} dv|v|j(v)\leq 2t\|f\|_{2}^{2}\,\mathcal{R}^{\frac{1}{2}}\sigma^{\frac{1}{2}}  .   
 \end{align*}
The same calculation shows  that $H_{t}^{\frac{1}{2}}$ satisfies condition (4.1).  That condition also holds for $\mathbf{A}_{t}'$, since the process $\mathbf{A}_{t}'$ is the difference between $H_{t}^{\frac{1}{2}}$ and $\mathbf{M}_{t}'$.  

\vspace{.5cm}

\noindent Part (3):\hspace{.15cm} Since the martingale part has expectation zero,  the equality
\begin{align}\label{Aggie}
\mathbb{E}\big[ H_{t} \big]=H+\sigma t
\end{align}
 follows trivially from the decomposition in Part (1).  For $\mathbb{E}\big[H_{t}^{2}]$, the classical theory would have   
$$\mathbb{E}\big[H_{t}^{2}]=\mathbb{E}\big[ (\mathbf{M}_{t}+\mathbf{A}_{t})^{2} \big]=\mathbb{E}\Big[ H^{2}+ \langle\mathbf{M},\mathbf{M}\rangle_{t}+2\int_{0}^{t}d\mathbf{A}_{r} H_{r}\Big],    $$
where $\langle\mathbf{M},\mathbf{M}\rangle_{t}$ is the predictable quadratic variation.  In my case, $d\mathbf{A}_{t}=\sigma dt$ is a multiple of the identity operator and therefore commutes with everything.  It follows that the equality above holds by the same argument as for the classical case.   The processes  $\sum_{n=1}^{\mathcal{N}_{t}} P_{t_{n}^{-}}v_{n}$, $\sum_{n=1}^{\mathcal{N}_{t}}\big( v_{n}^{2}-\frac{\sigma}{\mathcal{R} }\big)$,  and $\frac{\sigma}{\mathcal{R} }\mathcal{N}_{t}-\sigma t$ are uncorrelated martingales, and thus
\begin{align}
\big\langle\mathbf{M,M}\big\rangle_{t} =&\Big\langle  \sum_{n=1}^{\mathcal{N}_{r}}P_{t_{n}^{-}}v_{n},\sum_{n=1}^{\mathcal{N}_{r}}P_{t_{n}^{-}}v_{n} \Big\rangle_{t}+   \Big\langle \sum_{n=1}^{\mathcal{N}_{r}}\Big(v_{n}^{2}-\frac{\sigma}{\mathcal{R}}\Big), \sum_{n=1}^{\mathcal{N}_{r}}\Big(v_{n}^{2}-\frac{\sigma}{\mathcal{R}}\Big) \Big\rangle_{t}\nonumber  \\ & +\Big\langle \frac{\sigma}{\mathcal{R}} \mathcal{N}_{r}-\sigma r , \frac{\sigma}{\mathcal{R}} \mathcal{N}_{r}-\sigma r   \Big\rangle_{t}\nonumber \\ =&\sigma \int_{0}^{t}drP^{2}_{r}+\mathcal{R}\Big(\varsigma-\frac{\sigma^{2}}{\mathcal{R}^{2}} \Big)t+\frac{\sigma^{2}}{\mathcal{R}}t.     
\end{align}
Using $P^{2}_{r}\leq H_{r}$ and~(\ref{Aggie}) gives the bound for $\mathbb{E}[\langle\mathbf{M},\mathbf{M}\rangle_{t}  ]$.  Bounding $\mathbb{E}\big[ \int_{0}^{t}drH_{r}\big]$ also follows from~(\ref{Aggie}).

\end{proof}

In the proof of Prop.~\ref{Basics},  I apply Prop.~\ref{StochLem} to gain information about certain martingales  related to the Markov process $K_{r}$.  Define the energy process $E_{r}:=E(K_{r})$, where $E(k)$ is the dispersion relation determined by~(\ref{Energies}).   Define $\mathcal{E}:\R\rightarrow \R_{+}$ as the square root of the energy: $\mathcal{E}(k)=E^{\frac{1}{2}}(k)$.  I also use the symbol ``$\mathcal{E}$"  to refer to the corresponding process $\mathcal{E}_{r}=\mathcal{E}(K_{r})$.   Recall that the kets $|k\rangle_{\scriptscriptstyle{Q}}$ are  associated though a fiber decomposition of $L^{2}(\R)$ with normalized Bloch functions $\widetilde{\psi}_{k}\in L^{2}\big([-\pi,\,\pi)   \big)$ given by~(\ref{Bloch}).  The mathematical connection between the results in Prop.~\ref{StochLem} and the classical process $K_{r}$ is made through formulae such as in the equality   
\begin{align}\label{Notice}
\text{  }_{\scriptscriptstyle{Q}}\big\langle k\big| e^{-\textup{i}vX}H^{\frac{1}{2}}e^{\textup{i}vX}\big|k\big\rangle_{\scriptscriptstyle{Q}}= &\sum_{i\in \Z}\big|\kappa_{v}\big(k,i   \big)\big|^{2}\mathcal{E}(k+v+i)\nonumber \\= &\mathbb{E}\big[ \mathcal{E}_{r}   \,\big|\,K_{r-}=k,\,dL_{r}=v   \big],       
\end{align}
where $L_{r}$ is the underlying L\'evy process.   The first equality above follows from the definition of the coefficients $\big|\kappa_{v}\big(k,i \big)\big|^{2}$.  The value $ \big|\kappa_{v}\big(k,i   \big)\big|^{2}\in [0,1]$ is the probability for a lattice jump $i\in \Z$ conditioned on a L\'evy jump $v$ occurring from the momentum $k$.       The rigorous meaning of expressions involving bra-ket notation can  be traced back to the fiber decomposition such as in 
$$\text{  }_{\scriptscriptstyle{Q}}\big\langle k\big| e^{-\textup{i}vX}H^{\frac{1}{2}}e^{\textup{i}vX}\big|k\big\rangle_{\scriptscriptstyle{Q}}=\big\langle e^{\textup{i}vX}\widetilde{\psi}_{k}\big| H^{\frac{1}{2}}_{\phi} \,e^{\textup{i}vX}\widetilde{\psi}_{k}\big\rangle, $$
where $H_{\phi}$ is the fiber Hamiltonian for $\phi\in[-\frac{1}{2},\frac{1}{2})$ with  $k+v=\phi\,\textup{mod}\,1$.

Define the process  $m_{r}$ as
$$m_{r}:=\frac{1}{2}\sum_{n=1}^{\mathcal{N}_{r}}\Big(\sum_{i\in \Z}\big|\kappa_{v_{n}}\big(K_{t_{n}-},i   \big)\big|^{2}\mathcal{E}\big(K_{t_{n}-}+v_{n}+i\big) - \sum_{i}\big|\kappa_{-v_{n}}\big(K_{t_{n}-},\,i   \big)\big|^{2}\mathcal{E}\big(K_{t_{n}-}-v_{n}+i\big)      \Big).$$
By the rate symmetry $j(v)=j(-v)$ of the L\'evy jumps, $m_{t}$ is a martingale.      

\begin{proposition}\label{Basics} \text{ }\\
\begin{enumerate}
\item The process $E_{r}$ is a submartingale, and the predictable increasing part of its Doob-Meyer decomposition is $\sigma r$.  Moreover, the second moment satisfies 
$$\mathbb{E}\big[E_{r}^{2}\big]\leq \mathbb{E}[E_{0}^{2}]+ 3 \sigma r\mathbb{E}[E_{0}] +\mathcal{R}\varsigma r +\frac{3}{2}\sigma^{2}r^{2},$$
where $\varsigma:=\int_{\R}dv\frac{j(v)}{\mathcal{R}}v^{4}$.

\item  The process $\mathcal{E}_{r}$ is a submartingale, and the martingale part $M_{r}$ of its Doob-Meyer decomposition has predictable quadratic variation $\langle M,M \rangle_{r}$ with $\frac{d}{dr}\langle M,M \rangle_{r}\leq \sigma  $.

\item  The martingales $m_{r}$ and $M_{r}-m_{r}$ are uncorrelated  and therefore 
$\frac{d}{dr}\langle M, M \rangle_{r}\geq \frac{d}{dr}\langle m,m \rangle_{r} $.  Also, the quadratic variation of $m_{r}$ satisfies $[m,m]_{r}-[m,m]_{s}\leq \sum_{n>\mathcal{N}_{s}}^{\mathcal{N}_{r}}v_{n}^{2}$ for $r\geq s$, which means that it is dominated by the quadratic variation of the L\'evy process $L_{t}$.

\item The derivative of the predictable quadratic variation for the martingale $m_{r}$ (i.e. $\frac{d}{dr}\langle m,m \rangle_{r}$) is a function $\mathcal{V}:\R\rightarrow [0,\sigma]$ of $K_{r}$.  There exists an $\mathbf{a}>0$ such that for all $k\in \R$, then
\begin{align}\label{FromLattice}
\sigma-\mathcal{V}(k)\leq \frac{\mathbf{a} }{1+|\beta(k)|}.        
\end{align}

\end{enumerate}

\end{proposition}

\begin{proof}\text{  }\\

\noindent Part (1):\hspace{.15cm} Let $  \mathbf{u}\in L^{1}(\R)$ be a probability density with $\int_{\R}dk\,\mathbf{u}(k)E^{2}(k)<\infty$. Construct the density matrix $\rho=|f\rangle \langle f | \in \mathcal{B}_{1}\big(L^{2}(\R)\big)  $ for $f(k)=\mathbf{u}^{\frac{1}{2}}(k)$.  Let $D_{\lambda,t}$, $\rho_{\lambda,t}$, and $\mathcal{D}_{t}$ be defined as in Thm.~\ref{FWLimit}.  
 For all $\lambda>0$,
\begin{align}\label{Peanuts}
  \int_{\R}dk\,D_{\lambda,t}(k)E(k)= & \Tr[\rho_{\lambda,t}\big(P^{2}+V(x)\big)] \nonumber \\  = &\Tr[\rho\,\Phi^{*}_{\lambda,t}  \big(P^{2}+V(x)\big)]\nonumber \\  =  & \int_{\R}dk\,\mathbf{u}(k)E(k)+\sigma\,t.    
 \end{align}
  The third equality is by Part (3) of Prop.~\ref{StochLem}.

By Thm.~\ref{FWLimit}, $D_{\lambda,t}$ converges to $\mathcal{D}_{t}$  in the $1$-norm as $\lambda\rightarrow 0$ for every fixed $t$.  To guarantee the convergence of $\int_{\R}dk D_{\lambda,t}(k)E(k)$ to $\int_{\R}dk \mathcal{D}_{t}(k)E(k)$, it is sufficient to have a uniform bound in $\lambda$ on the second moments $\int_{\R}dk\,D_{\lambda,t}(k)E^{2}(k)$.  Again using Part (3) of Prop.~\ref{StochLem}, there are constants $c_{1}(t),c_{2}(t)\in \R_{+}$ such that
\begin{align}
\Phi_{\lambda,t}^{*}\big( (P^{2}+V(X))^{2}\big)= \mathbb{E}\big[\big(P_{t}^{2}+V(X_{t})\big)^{2} \big]\leq H^{2}+ c_{1}(t) H +c_{2}(t).
\end{align}
  Hence, a uniform bound for the second moment of the energy is given by
\begin{align}
\int_{\R}dk D_{\lambda,t}(k)E^{2}(k)\leq \int_{\R}dk\,\mathbf{u}(k)E^{2}(k)+c_{1}(t)\int_{\R}dk\,\mathbf{u}(k)E(k)+c_{2}(t).
\end{align}
It follows that $\mathbb{E}_{\mathbf{u}}[E_{t}]=\int_{\R}dk\,\mathcal{D}_{t}(k)E(k)$ is finite and equal to~(\ref{Peanuts}), where $\mathbb{E}_{\mathbf{u}}$ is the expectation beginning from an initial distribution $\mathbf{u}$.  The jump rate densities $\mathbf{j}_{k}(k^{\prime})= J(k^{\prime},k)$ are continuous in $L^{1}(\R)$ as a function of $k$ over intervals between lattice points in $\frac{1}{2}\Z$.  I can approximate a $\delta$-distribution at $k\in \R-\frac{1}{2}\Z$ with densities, and I have the result  $\frac{d}{dt}\mathbb{E}_{k}[E_{t}]|_{t=0}=\sigma$ for $k$ not on the half-spaced lattice.  Since my initial distribution will always be a density and the jump rates are densities, the behavior assigned to the lattice values is irrelevant.  It follows that $E_{t}$ is a submartingale with increasing part $E_{0}+\sigma t$.         

The bound for $\mathbb{E}\big[E_{r}^{2}\big]$ follows by plugging in the explicit values for $c_{1}(r)$ and $c_{2}(r)$.  

\vspace{.5cm}

\noindent Part (2):\hspace{.15cm} Let $\mathbf{u}\in L^{1}(\R)$ be defined as in Part (1).    By Part (2) of Prop.~\ref{StochLem},  for every $t,\lambda >0$ the inequality below holds 
\begin{align*}
\int_{\R}dk D_{\lambda,t}(k)E^{\frac{1}{2}}(k)= &\Tr[\rho_{\lambda,t}(P^{2}+V)^{\frac{1}{2}}]\\ =&\Tr[\rho \Phi^{*}_{\lambda,t}\big((P^{2}+V)^{\frac{1}{2}}\big) ]\\ \geq & \Tr[\rho(P^{2}+V)^{\frac{1}{2}}].
\end{align*}
  A similar argument as in Part (1) shows that $H_{t}^{\frac{1}{2}}$ is a submartingale. 

By Part (1),  the increasing part for the Doob-Meyer decomposition for~$\calE_{t}^{2}=E_{t}$ increases with linear rate $\sigma$.  I thus have the relation
$$\sigma= \frac{d}{dt}\langle M,M\rangle_{t}+2\calE_{t}\frac{d}{dt}\mathcal{A}_{t}.$$
Since both terms on the right are positive, it follows that $\frac{d}{dt}\langle M,M\rangle_{t}\leq \sigma$.

\vspace{.5cm}

\noindent Part (3):\hspace{.15cm} By~(\ref{Notice}), the terms of $m_{r}$ can be rewritten 
$$\sum_{i\in \Z}\Big(\big|\kappa_{v}\big(k,i \big)\big|^{2}\mathcal{E}(k+v+i) - \big|\kappa_{-v}\big(k,i  \big)\big|^{2}\mathcal{E}(k-v+i)\Big) =         { }_{\scriptscriptstyle{Q}}\big\langle k \big| e^{-\textup{i}v X} H^{\frac{1}{2}} e^{\textup{i}v X}-e^{\textup{i}vX} H^{\frac{1}{2}} e^{-\textup{i}vX}\big| k  \big\rangle_{\scriptscriptstyle{Q}}. $$
 By~(\ref{ThinkaDink}),  the absolute values for the jumps of $m_{t}$ are bounded by the absolute values for the jumps $v$ of the L\'evy process.  Consequently, the increments for the quadratic variation $[m,m]_{t}-[m,m]_{r}$ are almost surely smaller than those of the L\'evy process: $[L,L]_{t}-[L,L]_{r}$ for all $t\geq r\geq 0$.

The process $\mathcal{E}_{t}$ can be written as 
\begin{align}\label{Libya}
\mathcal{E}_{t}= \mathcal{E}_{0}+\sum_{n=1}^{\mathcal{N}_{t}}(\mathcal{E}_{t_{n}}-\mathcal{E}_{t_{n}^{-} })= \calS_{t}+\sum_{n=1}^{\mathcal{N}_{t}}\big(\mathcal{E}_{t_{n}}-\mathcal{E}_{t_{n}^{-} }   -\mathbb{E}\big[\mathcal{E}_{t_{n}}-\mathcal{E}_{t_{n}^{-}}  \big|\mathcal{F}_{t_{n}^{-}},\,dL_{t_{n}}\big] \big),  
\end{align}
where $\calS_{t}=\mathbf\mathcal{E}_{0}+ \sum_{n=1}^{\mathcal{N}_{t}}\mathbb{E}\big[\mathcal{E}_{t_{n}}-\mathcal{E}_{t_{n}^{-}}  \big|\mathcal{F}_{t_{n}^{-}},\,dL_{t_{n}}\big] $.  For a single term in the sum and $K_{t_{n}^{-}}=k$, $dL_{t_{n}}=v$, then     
\begin{align*}
\mathbb{E}\big[\mathcal{E}_{t_{n}}-\mathcal{E}_{t_{n}^{-}}  \big|\mathcal{F}_{t_{n}^{-}},\,dL_{t_{n}}\big]  & = \sum_{j\in\Z}|\kappa_{v}(k,j)|^{2}\mathcal{E}(k+j+v)-\calE (k)\\  &={ }_{\scriptscriptstyle{Q}} \big\langle k\big| e^{-\textup{i}vX} H^{\frac{1}{2}} e^{\textup{i}vX} \big|k\big\rangle_{\scriptscriptstyle{Q}}- { }_{\scriptscriptstyle{Q}}\big\langle k\big|\calH^{\frac{1}{2}} \big|k\big\rangle_{\scriptscriptstyle{Q}}. \end{align*}

However, I can reorganize $\calS_{t}$ in a way  reminiscent of  Part (2) in Prop.~\ref{StochLem}:
\begin{align*}
 \calS_{t} =& \calE_{0}+ \frac{1}{2}\sum_{n=1}^{\mathcal{N}_{t}}  { }_{\scriptscriptstyle{Q}}\big\langle K_{t_{n}^{-}}\big| e^{-\textup{i}v_{n}X} H^{\frac{1}{2}} e^{\textup{i}v_{n}X} - e^{\textup{i}v_{n}X} H^{\frac{1}{2}} e^{-\textup{i}v_{n}X} \big|K_{t_{n}^{-}}\big\rangle_{\scriptscriptstyle{Q}} \\ &+\frac{1}{2}\sum_{n=1}^{\mathcal{N}_{t}} { }_{\scriptscriptstyle{Q}}\big\langle K_{t_{n}^{-}}\big| e^{-\textup{i}v_{n}X} H^{\frac{1}{2}} e^{\textup{i}v_{n}X} + e^{\textup{i}v_{n}X} H^{\frac{1}{2}} e^{-\textup{i}v_{n}X} -2 H^{\frac{1}{2}} \big|K_{t_{n}^{-}}\big\rangle_{\scriptscriptstyle{Q}}. 
 \end{align*}
The first sum on the right is $m_{t}$, and I denote the second sum on the right by $a_{t}$. By the analysis in the proof of Part (2) of Prop.~\ref{StochLem}, the terms in the sum of $a_{t}$ are positive.  Also, I note that when $K_{t_{n}^{-}}=k$ and $|dL_{t_{n}}|=|v|$, then  
$$ { }_{\scriptscriptstyle{Q}}\big\langle k\big| e^{-\textup{i}v X} H^{\frac{1}{2}} e^{\textup{i}vX} + e^{\textup{i}vX} H^{\frac{1}{2}} e^{-\textup{i}v X} -2 H^{\frac{1}{2}} \big|k\big\rangle_{\scriptscriptstyle{Q}}=2\mathbb{E}\big[\mathcal{E}_{t_{n}}-\mathcal{E}_{t_{n}^{-}}  \big|\mathcal{F}_{t_{n}^{-}},\,|dL_{t_{n}}|\big],      $$
since $\pm v$ occur with equal probability.

The process $\calS_{t}$ is the conditional projection of $\calE_{t}$ on to the set of processes whose value at time $t$ depends only on $\calF_{t^{-}}$ and the L\'evy process $L_{t}$. Moreover, the process $a_{t}$ is the projection of $\calS_{t}$ that depends only on the jump times and the absolute value of the jumps $|dL_{t}|$.  Finally, $\calA_{t}$ is the predictable projection of $a_{t}$.   It follows that  $\mathcal{E}_{t}-\mathcal{S}_{t}$,  $\mathcal{S}_{t}-a_{t}=m_{t}$, and $a_{t}-\mathcal{A}_{t}$ are uncorrelated martingales with the following inequality for their predictable quadratic variations:   
$$\sigma t \geq \langle M ,M \rangle_{t}=\langle \calE-\mathcal{S},\calE-\mathcal{S}\rangle_{t}+ \langle m ,m\rangle_{t}+   \langle a-\mathcal{A},a-\mathcal{A}\rangle_{t}> \langle m,m \rangle_{t}. $$

\vspace{.5cm}

\noindent Part (4):\hspace{.15cm} The predictable quadratic variation $\langle m,m \rangle_{t}$ has the form
\begin{align}\label{Katrina}
\langle m,m \rangle_{t}= \frac{1}{4}\int_{0}^{t}dr\int_{\R} dv\,j(v) \Big|{ }_{\scriptscriptstyle{Q}}\big\langle K_{r} \big|e^{-\textup{i}v X} H^{\frac{1}{2}} e^{\textup{i}v X}-e^{\textup{i}vX} H^{\frac{1}{2}} e^{-\textup{i}vX}    \big| K_{r}  \big\rangle_{\scriptscriptstyle{Q}}\Big|^{2}.
\end{align}
The above gives an expression through which I can examine the dependence of $\frac{d}{dt}\langle m,m\rangle_{t}$ on $K_{t}$. If the expression in the integrand~(\ref{Katrina}) had  $| p \rangle_{\scriptscriptstyle{Q}}$, ${ }_{\scriptscriptstyle{Q}}\langle p| $,  and $H^{\frac{1}{2}}$ is replaced respectively by $| p \rangle$, $\langle p| $,   and  $|P|$, then I would have the explicit computation
\begin{align}
\big\langle k \big| e^{-\textup{i}v X} |P| e^{\textup{i}v X}-e^{\textup{i}vX} |P| e^{-\textup{i}vX}\big| k  \big\rangle \nonumber  & = \big\langle k+v \big|\, |P|\, \big| k+v  \big\rangle-\big\langle k-v \big|\, |P|\, \big| k-v  \big\rangle \\ & =           |k+v|-|k-v|=2vS(k\cdot v),
\end{align}
where the last inequality is restricted to $|v|\leq |k|$.  In my analysis, I will first work to bound the error of substituting ${ }_{\scriptscriptstyle{Q}}\langle p| $,  $| p \rangle_{\scriptscriptstyle{Q}}$ with $\langle p| $,  $| p \rangle$, and secondly, I  bound the error of  substituting $H^{\frac{1}{2}}$ with $|P|$.

  By the proof of Part (2) of Prop.~\ref{StochLem}, the difference $e^{-\textup{i}v X} H^{\frac{1}{2}} e^{\textup{i}v X}-e^{\textup{i}vX} H^{\frac{1}{2}} e^{-\textup{i}vX}$ has operator norm bounded by $2|v|$.  The difference shares the same fiber decomposition of the Hamiltonian. Consequently for $\phi_{+}, \phi_{-}\in [-\frac{1}{2},\frac{1}{2}) $ with  $\phi_{+}-v=\phi_{-}+v\,\textup{mod}\,1$,  then   the linear map  $e^{-\textup{i}v X} H^{\frac{1}{2}}_{\phi_{+}} e^{\textup{i}v X}-e^{\textup{i}vX} H^{\frac{1}{2}}_{\phi_{-}} e^{-\textup{i}vX}$ on  $L^{2}\big([-\pi,\pi)\big)$ has operator norm $\leq 2|v|$.

 By~(\ref{QuasiToStand}), there is $c>0$ such that the distance between Bloch the vectors $\widetilde{\psi}_{k},\psi_{k}\in  L^{2}\big([-\pi,\,\pi)   \big)  $ for $|k|\gg 1$ is bounded by
 \begin{align*} \big\| \widetilde{\psi}_{k}  -  \psi_{k}  \big\|_{2} \leq & \frac{c}{ 1+|\beta(k)|}+\mathbf{r}_{-}^{\frac{1}{2}}(k ) + 1-\mathbf{r}_{+}^{\frac{1}{2}}(k) \leq  \frac{c'}{ 1+|\beta(k)|}  
 \end{align*}
where the second inequality holds for some $c'>0$, since  $\mathbf{r}_{-}^{\frac{1}{2}}(k ) $ and $1-\mathbf{r}_{+}^{\frac{1}{2}}(k)$ are bounded by a multiple of $ \frac{1}{ 1+|\beta(k)|}$.  Hence, for $\phi_{\pm}\in [-\frac{1}{2},\frac{1}{2})$ with $\phi_{\pm}=k\pm v\,\textup{mod}\,1$, then 
   \begin{align}\label{UsesCond}  \Big| { }_{\scriptscriptstyle{Q}}\big\langle k  \big|   e^{-\textup{i}v X} & H^{\frac{1}{2}}  e^{\textup{i}v X}-e^{\textup{i}vX} H^{\frac{1}{2}} e^{-\textup{i}vX}\big| k  \big\rangle_{\scriptscriptstyle{Q}} - \big\langle k \big| e^{-\textup{i}v X} H^{\frac{1}{2}} e^{\textup{i}v X}-e^{\textup{i}vX} H^{\frac{1}{2}} e^{-\textup{i}vX}\big| k  \big\rangle\Big| \nonumber   \\  = &  \Big| \big\langle  \widetilde{\psi}_{k} \big| e^{-\textup{i}v X}H^{\frac{1}{2}}_{\phi_{+}}\, e^{\textup{i}v X}- e^{\textup{i}vX} H^{\frac{1}{2}}_{\phi_{-}}  e^{-\textup{i}vX} \big|\widetilde{\psi}_{k}  \big\rangle   - \big\langle  \psi_{k} \big| e^{-\textup{i}v X}  H^{\frac{1}{2}}_{\phi_{+}}e^{\textup{i}vX} - e^{\textup{i}vX} H^{\frac{1}{2}}_{\phi_{-}} e^{-\textup{i}vX}\big| \psi_{k}  \big\rangle\Big| \nonumber \\  = &\Big| \big\langle  \widetilde{\psi}_{k}-\psi_{k} \big| e^{-\textup{i}v X}H^{\frac{1}{2}}_{\phi_{+}}\, e^{\textup{i}v X}- e^{\textup{i}vX} H^{\frac{1}{2}}_{\phi_{-}}  e^{-\textup{i}vX} \big|\widetilde{\psi}_{k}  \big\rangle \Big|\nonumber \\ &+\Big|   \big\langle \psi_{k} \big| e^{-\textup{i}v X}  H^{\frac{1}{2}}_{\phi_{+}}e^{\textup{i}vX} - e^{\textup{i}vX} H^{\frac{1}{2}}_{\phi_{-}} e^{-\textup{i}vX}\big| \widetilde{\psi}_{k}-\psi_{k} \big\rangle\Big| \nonumber \\   \leq   & \frac{4c'|v|}{ 1+|\beta(k)| }  . 
   \end{align}

 I now bound the difference $  H^{\frac{1}{2}}-|P|$ when evaluated by kets $ | k^{\prime} \rangle$ with $|k^{\prime}|\gg 1$.  By the formula for the square root of an operator in terms of its resolvents,  
\begin{align*}
 H^{\frac{1}{2}}-|P|=\frac{1}{\pi}\int_{0}^{\infty}\frac{dw}{w^{\frac{1}{2}}}\Big(\frac{H}{w+H}-\frac{P^{2}}{w+P^{2}}\Big)      =\frac{1}{\pi}\int_{0}^{\infty}dw\, w^{\frac{1}{2}} \Big( \frac{1}{w+P^{2}} -\frac{1}{w+H}\Big).       
\end{align*}
However, the difference between the resolvent of a Laplacian and the resolvent of the Laplacian perturbed by a $\delta$-potential has a simple form~\cite{Solve}.  To use this, I will focus on a single fiber from the decomposition~(\ref{FiberInt}). For $w\in \R_{+}$ and $\phi \in [-\frac{1}{2},\frac{1}{2})$, the Green's function $G_{\phi,w}:[-\pi,\pi)\rightarrow \C  $ for the operator  $w-\Delta_{\phi} $ is given by the form  
$$G_{\phi,w}(x)  =\frac{1}{2\pi} \sum_{n\in \Z}\frac{1}{w+(n+\phi)^{2}}e^{i (n+\phi)x}.    $$
The operator on $L^{2}\big([-\pi,\pi)\big)$ determined by the  integral kernel $G_{\phi,w}(x-y) $ is equal to $(w-\Delta_{\phi})^{-1}$.   The difference between the resolvents in the $\phi$-fiber is
$$\Big[\frac{1}{w+P^{2}} -\frac{1}{w+H}\Big]_{\phi}=\frac{\alpha }{1+\alpha G_{\phi,w}(0) } A_{\phi,w},$$  where the expression on the left in square brackets denotes the operator on $L^{2}\big([-\pi,\pi)\big)$ corresponding to the $\phi$-fiber, and the operator $A_{\phi,w}$ has integral kernel 
$A_{\phi,w}(x,y)=G_{\phi,w}(x) G_{\phi,w}(-y)$.

For $k=\phi\,\textup{mod}\,1$, then
\begin{align}\label{NextToLast}
\big|\big\langle k' \big| H^{\frac{1}{2}}-|P|  \big| k' \big\rangle\big| & = \frac{1}{4\pi^{3} }\int_{0}^{\infty}dw\, w^{\frac{1}{2}} \frac{\alpha }{1+\alpha G_{\phi,w}(0) } \Big(\frac{1}{w+(k')^{2} }   \Big)^{2}\nonumber  \\  & \leq  \frac{\alpha}{4\pi^{3} }\int_{0}^{\infty}dw\, w^{\frac{1}{2}}\frac{1}{(w+(k')^{2})^{2} }= \frac{\alpha}{4|k^{\prime}|\pi^{2} }  .   
\end{align}

Going back to~(\ref{Katrina}), I have the relations
\begin{align*}
\frac{d}{dt}\langle m \rangle_{t} & \geq \frac{1}{4}\int_{|v|\leq \frac{1}{2}|k|  } dv\,j(v) \Big|{ }_{\scriptscriptstyle{Q}}\big\langle k \big|e^{-\textup{i}v X} H^{\frac{1}{2}} e^{\textup{i}v X}-e^{\textup{i}vX} H^{\frac{1}{2}} e^{-\textup{i}vX}    \big| k  \big\rangle_{\scriptscriptstyle{Q}} \Big|^{2} \\ &= \frac{1}{4}\int_{|v|\leq \frac{1}{2}|k|  } dv\,j(v)\Big|  \big\langle k+v \big|\, H^{\frac{1}{2}}\, \big| k+v  \big\rangle -\big\langle k-v \big|\, H^{\frac{1}{2}}\, \big| k-v  \big\rangle     \Big|^{2}+\mathit{O}\Big(\frac{1}{1+|\beta(k)|}\Big)   \\ &=  \frac{1}{4}\int_{|v|\leq \frac{1}{2}|k|  } dv\,j(v)\Big(   |k+v|-|k-v|    \Big)^{2}     +\mathit{O}\Big(\frac{1}{1+|\beta(k)|}\Big)
\\ & =  \frac{1}{4}\int_{|v|\leq \frac{1}{2}|k|  } dv\,j(v) (2v)^{2}+\mathit{O}\Big(\frac{1}{1+|\beta(k)|}\Big)=\sigma+\mathit{O}\Big(\frac{1}{1+|\beta(k)|}\Big).
\end{align*}
The first equality follows from~(\ref{UsesCond}), and the second  follows by~(\ref{NextToLast}) for $k'=k+v$ and because $|k'|^{-1}=\mathit{O}\big(\frac{1}{1+|\beta(k)|}\big)$ by the restriction $|v|\leq \frac{1}{2}|k|$.  The third inequality is by the definition of $\sigma$ and Chebyshev's inequality through  $\int_{|v|\geq \frac{1}{2}|k|  } dv\,\frac{j(v)}{\mathcal{R}}\,v^{2}\leq 4|k|^{-2}\varsigma$, where $\varsigma$ is the fourth moment of $\frac{j(v)}{\mathcal{R}}$.

\end{proof}

\begin{lemma}
Let $K_{t^{-}}=k$ for $|k|\gg 1$ and $|dL_{t}|=|v|\leq |k|$.  There exists an $\mathbf{a}>0$ such that
$$\textup{Var}\big[\calE(K_{t})\big|\calF_{t^{-}},dL_{t}  \big]\leq \frac{\mathbf{a}}{1+|\beta(k)|}.$$

\end{lemma}

\begin{proof}
Note that
$$
\mathbb{E}\big[E(K_{t})-E(K_{t^{-}})\big|\calF_{t^{-}},|dL_{t}|  \big] = \frac{1}{2}{ }_{\scriptscriptstyle{Q}} \big\langle k\big| e^{-\textup{i}vP}He^{\textup{i}vP}+ e^{\textup{i}vP}He^{-\textup{i}vP}-2H  \big|k\big\rangle_{\scriptscriptstyle{Q}}\leq v^{2},
$$
where the  inequality follows from the same computation as in Part (1) of Prop.~\ref{StochLem}.  However, the quantity $\textup{Var}\big[\calE(K_{t})\big|\calF_{t^{-}},dL_{t}  \big]$ appears in the follow equation: 
\begin{align}\label{Tortilla}
\mathbb{E}\big[E(K_{t})-E(K_{t^{-}})\big|\calF_{t^{-}},|dL_{t}|  \big] =&\frac{1}{2}\sum_{\pm}\textup{Var}\big[\calE(K_{t})\big|\calF_{t^{-}},dL_{t} =\pm v \big]\nonumber  \\ & + \frac{1}{2}\Big(\mathbb{E}\big[\calE(K_{t})\big|\calF_{t^{-}},dL_{t}  =v \big]-\mathbb{E}\big[\calE(K_{t})\big|\calF_{t^{-}},dL_{t}=-v \big] \Big)^{2}\nonumber \\  &+\mathbb{E}\big[\calE(K_{t})\big|\calF_{t^{-}},|dL_{t}|=|v|  \big]^{2}-\mathcal{E}^{2}(k) .
\end{align}
The bottom line is positive.  The second line of~(\ref{Tortilla}) can be written as 
$$\Big(\mathbb{E}\big[\calE(K_{t})\big|\calF_{t^{-}},dL_{t}  =v \big]-\mathbb{E}\big[\calE(K_{t})\big|\calF_{t^{-}},dL_{t}=-v \big]\Big)^{2}= \Big|{ }_{\scriptscriptstyle{Q}}\big\langle K_{t^{-}} \big|e^{-\textup{i}v X} H^{\frac{1}{2}} e^{\textup{i}v X}-e^{\textup{i}vX} H^{\frac{1}{2}} e^{-\textup{i}vX}    \big| K_{t^{-}}  \big\rangle_{\scriptscriptstyle{Q}}\Big|^{2}$$
and the right side is  $v^{2}+\mathit{O}\big(\frac{1}{1+|\beta(k)| } \big)$ by the argument in the proof for Part (4) of Prop.~\ref{Basics}.  Since all three terms on the right side of~(\ref{Tortilla}) are positive, it follows that $\textup{Var}\big[\calE(K_{t})\big|\calF_{t^{-}},dL_{t}= v \big]$ is $\mathit{O}\big(\frac{1}{1+|\beta(k)| } \big)$ as claimed.

\end{proof}

 \section{The limiting classical process}\label{SecClassical}
  
 I now shift my focus entirely to the classical stochastic processes $K_{r}$ whose probability density evolves according to the equation~(\ref{Master}).     In Sect.~\ref{SecSubMart}, I prove Thm.~\ref{SubMartCLT}, which stated that $t^{-\frac{1}{2}}\mathcal{E}_{st}$, $s\in[0,1]$ converges in distribution as $t\rightarrow \infty$ to the absolute value of a Brownian motion.  Section~\ref{SubSecRefl} contains various lemmas related to the random variables $\frac{\tau_{n+1}-\tau_{n}}{|K_{\tau_{n}}|}$ being approximately exponentially distributed with expectation $\nu^{-1}$, where  $\tau_{n},\tau_{n+1}$ are successive reflection times.  Finally, Sect.~\ref{SubSecMain} contains proofs of the lemmas in Sect.~\ref{SecProofOutline} and completes the proof of  Thm.~\ref{Main}.

\subsection{A submartingale central limit theorem}\label{SecSubMart}

The difference between the quantities $|k|$ and $\mathcal{E}(k)=E^{\frac{1}{2}}(k)$ is bounded by a constant, and the difference is even $\mathit{O}(|k|^{-1})$ for $|k|\gg 1$ as a consequence of the Kr\"onig-Penney relation~(\ref{Energies}).   Working with the process $\mathcal{E}_{r}$ is advantageous, since it is a submartingale by Part (2) of Prop.~\ref{Basics}, and  the increasing part of the Doob-Meyer decomposition for the square $\mathcal{E}_{r}^{2}=E_{r}$ increases linearly.  Thus, the processes $t^{-\frac{1}{2}}|K_{st}|$ and $t^{-\frac{1}{2}}\mathcal{E}_{st}$, $s\in[0,1]$ are close, and I will  work with the latter.    
As before, $M_{r}$ and $A_{r}$ will denote the martingale and increasing parts the Doob-Meyer decomposition for $\mathcal{E}_{r}$.  The main result of this section is proof of Thm.~\ref{SubMartCLT}.  One of the key inputs for the proof is Lem.~\ref{EnergyLemma}, which yields that $\mathcal{E}_{r}$ typically spends most of the interval $r\in [0,t]$ with values $\mathcal{E}_{r}\gg 1$.  This is important, since some estimates, such as in Part (4) of Prop.~\ref{Basics},  only become effective when $\mathcal{E}_{r}$ is large.  Lemma~\ref{LittleEm} offers a lower  bound for the predictable quadratic variation of the martingale $m_{t}$ from Prop.~\ref{Basics}. This lower bound shows that $\langle m,m\rangle_{r}$ essentially grows linearly with rate $\sigma$, and when combined with Part (2) and (3) of Prop.~\ref{Basics}, this implies that $t^{-\frac{1}{2}}M_{st}$ and  $t^{-\frac{1}{2}}m_{st}$ are nearly equal. 

The result of the lemma below holds for a general class of positive processes whose squares are submartingales, whose variances increase linearly, and that satisfy some Lindberg condition (to avoid the situation where the jumps are  very large but very infrequent).  The lemma is only stated for $\mathcal{E}_{r}$ here.  The lemma implies that $\mathcal{E}_{r}$ typically spends most of a time interval $r\in [0,t]$ with values on the order of $t^{\frac{1}{2}-\rho}$, $0<\rho<\frac{1}{2}$ for large $t$, where ``most" is  all except for a total duration on the order  $\mathit{O}(t^{\rho})$.  Lemma 3.4 in~\cite{Newton} is similar, although the proof here employs different techniques.  In the more general situation in which the second momentum of $\mathcal{E}_{r}$ grows at varying rates, similar statements can be made for a stochastic time change of $\mathcal{E}_{r}$.

\begin{lemma}\label{EnergyLemma}
Let $\varrho_{1},\varrho_{2},\varrho_{3}\geq 0$, $\varrho_{1}+\varrho_{2} +\varrho_{3}=\frac{1}{2}$, and $\varrho_{1}$ be strictly positive.  For $\epsilon>0$,  define $T_{\epsilon,t}=\int_{0}^{1}ds\,\chi(t^{-\varrho_{1} }\mathcal{E}_{st}\leq \epsilon  )$.  For all $\epsilon, \delta>0$  ,
$$\limsup_{t\rightarrow \infty}\, t^{\varrho_{2}}  \mathbb{P}\big[T_{\epsilon,t}\geq t^{-\varrho_{3} }\delta    \big]< 64\sigma^{-\frac{1}{2}} \frac{\epsilon}{\delta}.       $$
\end{lemma}

\begin{proof}
I begin with an application of Chebyshev's inequality to obtain
$$ t^{\varrho_{2}}\mathbb{P}\big[T_{\epsilon,t}\geq t^{-\varrho_{3}}\delta    \big]\leq \frac{ t^{\frac{1}{2}-\varrho_{1} } }{\delta  } \mathbb{E}\big[T_{\epsilon,t} \big].$$
  
Set $\varsigma_{0}=\varpi_{1}=0$, and define the stopping times $\varsigma_{n},\varpi_{n}\leq t$ such that 
\begin{align*}
\varpi_{n}= \min \{r\in (\varsigma_{n-1},\infty)\,\big| \,\mathcal{E}_{r}\leq \epsilon t^{\varrho_{1} }         \},\quad \quad
\varsigma_{n}= \min \{ r\in (\varpi_{n},\infty)\,\big| \, \mathcal{E}_{r}\geq 2\epsilon t^{\varrho_{1} }            \}.
\end{align*}
Let $\Xi_{t}$ be the number $\varpi_{n}$'s less than $t$.  I clearly have  $$T_{\epsilon,t}\leq t^{-1}\sum_{n=1}^{\Xi_{t}}\varsigma_{n}-\varpi_{n}.    $$
Moreover, notice that
\begin{align}\label{IncursionTime}
\mathbb{E}\big[T_{\epsilon,t}  \big]\leq \mathbb{E}\big[  t^{-1}\sum_{n=1}^{\Xi_{t}}\varsigma_{n}-\varpi_{n}\big]\leq t^{-1}\mathbb{E}\big[ \Xi_{t}\big]\sup_{n\in \mathbb{N}}\mathbb{E}\big[\varsigma_{n}-\varpi_{n}\,\big|\, n\leq \Xi_{t}   \big]  . 
\end{align}
I hence have an upper bound in terms of the expectation for the number of upcrossings $\Xi_{t}$ and the expectation for the duration of a single upcrossing $ \varsigma_{n}-\varpi_{n}$ conditioned on the event $n\leq \Xi_{t}$.  By the submartingale uncrossing inequality~\cite{Chung}, I have the first inequality below 
\begin{align}\label{Numero}
\mathbb{E}\big[ \Xi_{t}\big]\leq  \frac{ \mathbb{E}\big[\mathcal{E}_{t}  \big]}{2\epsilon t^{\varrho_{1}}-\epsilon t^{\varrho_{1} }    }\leq   \frac{ \mathbb{E}\big[\mathcal{E}_{t}^{2}  \big]^{\frac{1}{2}} }{\epsilon t^{\varrho_{1} }  }\leq      \frac{(\mathbb{E}\big[\mathcal{E}_{0}^{2}    \big]+t\sigma )^{\frac{1}{2}} }{\epsilon t^{\varrho_{1} } }<2\sigma^{\frac{1}{2}} t^{\frac{1}{2}-\varrho_{1}} \epsilon^{-1}.
\end{align}
The last inequality is for $t$ large enough.  

Next, I focus on the expectation of the incursions $\varsigma_{n}-\varpi_{n}$.  Whether or not the event $n\leq \Xi_{t}$ occurred will be known at time $\varpi_{n}$, so    
$$\sup_{n\in \mathbb{N}}\mathbb{E}\big[\varsigma_{n}-\varpi_{n}\,\big|\, n\leq \Xi_{t}   \big]\leq \sup_{n\in \mathbb{N},\,\omega\in \mathcal{F}_{\varpi_{n}} }\mathbb{E}\big[\varsigma_{n}-\varpi_{n}\,\big|\,\mathcal{F}_{\varpi_{n}} \big]. $$
In examining the expression on the right side above, I will set $\varpi_{n}=0$ and $\varsigma=\varsigma_{n}-\varpi_{n}$. 
Since $\sigma t$ is the increasing part of the Doob-Meyer decomposition for $\mathcal{E}^{2}_{t}-\mathcal{E}^{2}_{0}$, the optional sampling theorem gives 
$$\mathbb{E}_{\omega}\big[\varsigma\wedge \mathcal{T}\big]= \sigma^{-1}\mathbb{E}_{\omega}\big[\mathcal{E}^{2}_{\varsigma\wedge \mathcal{T} }-\mathcal{E}^{2}_{0}  \big]. $$
If the process $\mathcal{E}_{t}$ were constrained to make jumps bounded by $J>0$, then I could immediately reason that      
$$\mathbb{E}_{\omega}[\varsigma]=\lim_{\mathcal{T}\rightarrow \infty}   \sigma^{-1}\mathbb{E}_{\omega}\big[\mathcal{E}^{2}_{\varsigma \wedge \mathcal{T} }-\mathcal{E}^{2}_{0}\big]\leq \sigma^{-1} (2\epsilon t^{\varrho_{1}}+J)^{2},   $$
since $\mathcal{E}_{s}$ is either still in the interval $[0,2\epsilon t^{\varrho_{1} }]$ at time $s=\varsigma \wedge \mathcal{T}$ or has jumped out by a value of at most $J$.  Plugging~(\ref{Numero}) and the bound $\sigma^{-1} (2\epsilon t^{\varrho_{1}}+J)^{2}$ for $\sup_{n\in \mathbb{N}}\mathbb{E}\big[\varsigma_{n}-\varpi_{n}\,\big|\, n\leq \Xi_{t}   \big]$ into~(\ref{IncursionTime}) leads to the desired result.  

For the case in which there is not a cap on the jumps, it is necessary to be more  careful.  Consider a modified process in which jumps $\mathcal{E}_{s}-\mathcal{E}_{s-} $ greater than $2\epsilon t^{\varrho_{1}}$ are removed.  The modified process can be written as 
$$
\mathcal{E}_{r}^{\prime}=\sum_{n=1}^{\mathcal{N}_{r} }(\mathcal{E}_{t_{n}}-\mathcal{E}_{t_{n}^{-}})\chi( \mathcal{E}_{t_{n}}-\mathcal{E}_{t_{n}^{-}}\leq 2\epsilon t^{\varrho_{1} } ),    $$
for Poisson times $t_{n}\leq r$.  The modified process $\mathcal{E}_{r}^{\prime}$  is still positive, and for large enough $t$, $(\mathcal{E}_{r}^{\prime})^{2}$ will still be a submartingale such that the increasing part of its Doob-Meyer decomposition grows nearly at the linear rate $\sigma$ (although growing at a lesser rate, such as $\frac{\sigma}{2}$, is sufficient).

Let $\varsigma^{\prime} $ be  analogously defined to $\varsigma$ as the hitting time that $\mathcal{E}_{r}^{\prime}$ jumps out of $[0,\,2\epsilon t^{2\varrho_{1} }]$.  By the definitions of these processes, I always have $\varsigma \leq \varsigma^{\prime} $,
\begin{align*}
\mathbb{E}_{\omega}\big[ \varsigma \big] &=\lim_{\mathcal{T}\rightarrow \infty}\mathbb{E}_{\omega}\big[ \varsigma \wedge \mathcal{T}  \big]\leq \lim_{\mathcal{T}\rightarrow \infty}   \mathbb{E}_{\omega}\big[ \varsigma^{\prime}\wedge \mathcal{T}  \big]\\ &\leq \lim_{\mathcal{T}\rightarrow \infty}   2\sigma^{-1} \mathbb{E}_{\omega}\big[(\mathcal{E}_{\varsigma^{\prime}\wedge \mathcal{T}  }^{\prime})^{2}-(\mathcal{E}_{0}^{\prime})^{2}  \big]\leq  32\sigma^{-1}\epsilon^{2}t^{2\varrho_{1} },   
\end{align*}
 since $\mathcal{E}_{\varsigma^{\prime}\wedge \mathcal{T}  }^{\prime}\leq 4\epsilon t^{\varrho_{1} }$. Combining this inequality with~(\ref{Numero}) in~(\ref{IncursionTime}) finishes the proof.

\end{proof}

I will make a few comments about the Markov chain $\mathbf{\theta}_{n}=\Theta(K_{t_{n}})$, which is the sequence of momenta $K_{t_{n}}$ contracted to the torus $\mathbb{T}=[-\frac{1}{4},\frac{1}{4})$ at the Poisson times $t_{n}$.  Since the lattice components of the momentum jumps live on $\Z$, they do not influence the jump rates for the contracted process $\mathbf{\theta}_{n}$.  Define the map $\langle\cdot\rangle:L^{1}(\R)\longrightarrow L^{1}(\mathbb{T})$ by
\begin{align}\label{DefTorusContraction}
\langle q\rangle(\theta)= \sum_{i\in \Z}q(\theta+\frac{i}{2}), \quad \quad  \quad \quad  \theta \in \big[-\frac{1}{4},  \frac{1}{4}\big).   
\end{align}
If $q_{t_{n-1}}$, $q_{t_{n}}$ are the distributions for the momentum at times $t_{n-1}$,$t_{n}$, then $\langle q_{t_{n}}\rangle =\mathbf{T}\langle q_{t_{n-1}}\rangle$, where the operator $\mathbf{T}$ has integral kernel
$$   \mathbf{T}(\theta_{2},\theta_{1}) =\mathcal{R}^{-1}\sum_{i\in \Z} j\big(\theta_{2}-\theta_{1}+\frac{i}{2}\big),\quad \quad \theta_{1},\theta_{2}\in \big[-\frac{1}{4},\frac{1}{4}\big).  $$ 
As a consequence of (2) of List~\ref{Assumptions}, $\|\langle j\rangle \|_{\infty}=\sup_{-\frac{1}{4}\leq \theta\leq \frac{1}{4}}\sum_{i\in \Z}j(\theta+\frac{i}{2})$ is finite, and thus $\mathbf{T}$ maps $L^{1}$ functions to $L^{\infty}$:  
\begin{align}\label{Crumpet}
\|\mathbf{T}h \|_{\infty} \leq \mathcal{R}^{-1}\|h\|_{1}\sum_{i\in \Z}j\big(\theta+\frac{i}{2}\big),  \quad \quad \quad h\in L^{1}\big(\mathbb{T}\big)\leq \frac{\mu}{\mathcal{R}} \|h\|_{1}. 
\end{align}

By~(\ref{Crumpet}), the process $K_{r}$  forgets its previous locations on the torus $\mathbb{T}$ exponentially fast. In particular, the process $\Theta(K_{r})$ will not not be disproportionately concentrated around the value zero.  This is important for the proof of the following lemma, which relies on Part (4) of Prop.~\ref{Basics}, since the bound in Part (4) of Prop.~\ref{Basics} is only useful for $|\beta(K_{r})|=\frac{1}{2}|\Theta(K_{r})||\mathbf{n}(K_{r})|\gg 1$, i.e., when $|\Theta(K_{r})|$ is not too close to zero.

\begin{lemma}\label{LittleEm}
The predictable quadratic variation for  $m_{r}$ satisfies
$$\mathbb{E}\Big[\sup_{s\in [0,1]}\big|\sigma s-t^{-1}\langle m,m\rangle_{st}\big|\Big]\leq \frac{g\log(1+t)}{t^{\frac{1}{6}}}  $$   
for some constant $g>0$ and all $t\in \R_{+}$.

\end{lemma}

\begin{proof}
By Parts (2) and (3) of Prop.~\ref{Basics}, I have the inequality $\sigma\geq \frac{d}{dr}\langle m,m\rangle_{r}$.  Thus, the expression $\sup_{s\in [0,1]}\big|\sigma s-t^{-1}\langle m,m\rangle_{st}\big|$ is equal to $\sigma-t^{-1}\langle m,m\rangle_{t}$.  For the mean-$\mathcal{R}^{-1}$ exponential waiting times $e_{n}=t_{n+1}-t_{n}$ between Poisson times,  I can write
\begin{align}\label{Bengali}
t^{-1}\langle m,m\rangle_{t}=t^{-1}\sum_{n=0}^{\mathcal{N}_{t}}\mathcal{V}(K_{t_{n} })e_{n}- t^{-1}\Big(\sum_{n=1}^{\mathcal{N}_{t}}e_{n}-t\Big)\mathcal{V}(K_{\mathcal{N}_{t}}), 
\end{align}
where the second term on the right subtracts over-counting in the first term and will be of  order $\mathit{O}(t^{-1})$, since $  \mathcal{V}$ is bounded. 

Approximating $\langle m,m\rangle_{t}$ by the process  $s_{t}:=\sum_{n=1}^{\mathcal{N}_{t}}\mathbb{E}\big[\mathcal{V}(K_{t_{n}})\,\big|\,\mathcal{F}_{t_{n-1}}\big]  e_{n}$, yields a difference bounded by 
$$t^{-2}\mathbb{E}\big[\big(\langle m,m\rangle_{t}-s_{t}   \big)^{2}\big]\leq \sigma\mathcal{R} t^{-1}+\mathit{O}(t^{-2}),  $$
since $\sup_{k}\mathcal{V}(k)\leq \sigma$ and the $e_{n}$'s are independent of each other and everything else.  The conditional expectation  $\mathbb{E}\big[\mathcal{V}(K_{t_{n}})\,\big|\,\mathcal{F}_{t_{n-1}}\big]$ is a function $Q:\R\rightarrow [0,\,\sigma]$ of $K_{t_{n-1}}$.  My expression of choice will be $\int_{0}^{t}dr Q(K_{r})$, which is nearly equal to $s_{t}$ except for two extra boundary terms of order $\mathit{O}(t^{-1})$.      

The key input for the proof will be an upper bound for  $\sigma-Q(k)$ when $|k|\geq 2t^{\frac{1}{6}}$.  The probability that a Poisson jump $v$ has absolute value $|v|>t^{\frac{1}{6}}$ will decay superpolynomially for large $t$ by (1) of List~\ref{Assumptions}, and thus  
\begin{align}\label{Puff}
\sigma -Q(k)=\int_{[-t^{\frac{1}{6}},\,t^{\frac{1}{6}}  ]}dv\,\sum_{i\in \Z}\frac{j(v)}{\mathcal{R}}\big|\kappa_{v}( k,i)\big|^{2}\big|\sigma-\mathcal{V}(k+i+v    )\big|+\mathit{O}(t^{-\frac{1}{6}}) .  
\end{align}
By Part (1) of Lem.~\ref{BadTerms}, the sum of $\big|\kappa_{v}( k,i)\big|^{2}$ for $i\notin I(k,v)$ will be $\leq ct^{-\frac{1}{3}}$ for some $c>0$.  Since $|v|\leq t^{\frac{1}{6}}$, the values $|k+i+v |$ must be $\geq t^{\frac{1}{6}}$ for $i\in I(k,v)$. Applying Part (4) of Prop.~\ref{Basics}, the integral in~(\ref{Puff}) is bounded by 
\begin{align*}
&\int_{[-t^{\frac{1}{6}},\,t^{\frac{1}{6}}  ]}dv\,\sum_{i\in I(k  ,\,v )  }\frac{j(v)}{\mathcal{R}}\big|\kappa_{v}( k,i)\big|^{2} \frac{\mathbf{a}}{1+|\beta( k+i+v  )|}+\mathit{O}(t^{-\frac{1}{6}}) \nonumber  \\ &\leq \mathcal{R}^{-1}\Big(\sup_{-\frac{1}{4}\leq \theta \leq \frac{1}{4}} \sum_{i\in\Z} j\big(\theta+\frac{i}{2}\big)\Big)\int_{-\frac{1}{4}}^{\frac{1}{4}}d\theta \frac{\mathbf{a} }{1+| \lfloor t^{\frac{1}{6}}\rfloor \theta| }+\mathit{O}(t^{-\frac{1}{6}}) \\  &\leq \frac{2\mathbf{a}\mu}{\mathcal{R}} \frac{\log(1+t^{\frac{1}{6}} )}{t^{\frac{1}{6}}}+\mathit{O}(t^{-\frac{1}{6}}), 
\end{align*}
where I have used that the values $k+i+v$ have absolute value $\geq t^{\frac{1}{6}}$ and the shifts $i\in \Z$ do not change the values $k+v$ modulo $\frac{1}{2}$.  The supremum in the second inequality is smaller than $\mu$ by (2) of List~\ref{Assumptions}. Thus, for $|k|>2t^{-\frac{1}{6}}$, the above gives
\begin{align}\label{Curtsey}
\sigma -Q(k)=  \frac{2\mathbf{a}\mu}{\mathcal{R}} \frac{\log(1+t^{\frac{1}{6}} )}{t^{\frac{1}{6}}}+\mathit{O}(t^{-\frac{1}{6}}).
\end{align}

Now, I am ready to bound the expectation of $\int_{0}^{t}dr\,\big(\sigma-Q(K_{r})\big)$.  Define $T_{t}=t^{-1}\int_{0}^{t}dr\chi(|K_{r}|\leq 2t^{\frac{1}{6}} ) $.  By considering the events $T_{t}> t^{\frac{1}{6}}$ and $T_{t}\leq  t^{\frac{1}{6}}$, I get the bound
\begin{align*}
\mathbb{E}\Big[t^{-1}\int_{0}^{t}dr\big(\sigma-Q(K_{r})\big)  \Big]\leq & \sigma\mathbb{P}\big[T_{t}\geq t^{-\frac{1}{6}} \big]+\sigma \mathbb{E}\big[T_{t}\chi(T_{t}\leq t^{\frac{1}{6}})\big] + \sup_{|k|>  2t^{-\frac{1}{6}} }\big(\sigma-Q(k)\big)  \\   \leq & \sigma\mathbb{P}\big[T_{t}\geq t^{-\frac{1}{6}} \big]+\sigma t^{-\frac{1}{6}}+ \frac{2\mathbf{a}\mu}{\mathcal{R}} t^{-\frac{1}{6}} \log(t). 
\end{align*}  
  For the event $T_{t}\leq t^{\frac{1}{6}}$, the  third term on the right corresponds to the part of the trajectory in which $|K_{r}|$ is   $\geq 2t^{\frac{1}{6}}$, and the second inequality uses~(\ref{Curtsey}).    By Lem.~\ref{EnergyLemma} with $\delta=1$, $\epsilon=2$, and $\varrho_{1} =\varrho_{2}=\varrho_{3}=\frac{1}{6}$,   the probability that $T_{t}\geq t^{-\frac{1}{6}}$ has order $t^{-\frac{1}{6}}$ (since  $|K_{r}|\approx\mathcal{E}_{r}+\mathit{O}(|K_{r}|^{-1})$).

\end{proof}

The proof that the processes $t^{-\frac{1}{2}}\mathcal{E}_{st}$, $s\in [0,1]$ converge as $t\rightarrow \infty$ to the absolute value of a Brownian motion  relies on the following features of $\mathcal{E}_{t}$: 
\begin{enumerate}
\item The process $\mathcal{E}_{r}$ is a positive submartingale with martingale component $M_{r}$.

\item  The increasing part of the Doob-Meyer decomposition for $\mathcal{E}_{r}^{2}$ is $\mathcal{E}_{0}^{2}+ \sigma r$.

\item  When $\mathcal{E}_{r}\gg 1$, then $\sigma-\frac{d}{dt}\langle M,M \rangle_{t}$ goes to zero (at least on average over long  periods of time).

\item The quadratic variation of $M_{r}$ satisfies  a Lindberg condition. 

\end{enumerate} 
Theorem~\ref{SubMartCLT} follows from a more general theorem assuming conditions of the type above. In principle, if (2) does not strictly hold but the expectation of $\mathcal{E}_{t}$ tends to grow with $t$, then one can obtain this assumption by considering a stochastic time change $\tau_{t}$:  $\mathcal{E}_{\tau_{t}}=\mathcal{E}_{t}^{\prime}$.   Assumption (3) is required to guarantee that the process $t^{-\frac{1}{2}}\mathcal{E}_{st}$  behaves diffusively away from zero.  It forbids, for instance, that the process $\mathcal{E}_{r}$ has the deterministic trajectory $\mathcal{E}_{r}=\sigma^{\frac{1}{2}}r^{\frac{1}{2}}$.  To make use of (3), I will need to verify that  $\mathcal{E}_{t}$ spends most of the time at ``high" values, and this verification is made through Lem.~\ref{EnergyLemma}.   Assumption (4) is to guarantee asymptotic negligibility as is required in the martingale central limit theorem.  Theorem~\ref{SubMartCLT} is similar to Thm. 4.1 in~\cite{Newton}. \vspace{.3cm} 


\begin{proof}[Proof of Thm.~\ref{SubMartCLT}]\text{  }\\
All ``convergence in law" in this proof will refer to the uniform metric. Let $M_{r}$ and $A_{r}$ be the components of the Doob-Meyer decomposition for $\mathcal{E}_{r}$, and define the submartingale $G_{r}=M_{r}+\sup_{0\leq s\leq r}-M_{s}$.  The convergence in law of $t^{-\frac{1}{2}}\mathcal{E}_{st}$ is implied by the following statements:
\begin{enumerate}[(i).]
\item The processes $(\sigma t)^{-\frac{1}{2}}G_{st}$, $s\in [0,1]$ converge in law as $t\rightarrow \infty$ to the absolute value of a standard Brownian motion.   
\item The random variables $\sup_{0\leq s\leq 1}  t^{-\frac{1}{2}}\big|A_{st}-\sup_{0\leq r\leq st}-M_{r}\big|$ converge in probability to zero as $t\rightarrow \infty$. 
\end{enumerate}

\noindent  (i).\hspace{.15cm}  If  $(\sigma t)^{-\frac{1}{2}} M_{st}$ converges in law  to a Brownian motion $\mathbf{B}_{s}$ as $t\rightarrow \infty$,  it follows that $t^{-\frac{1}{2}}G_{st}$ converges in law to $L(B_{s})$, since the map $F:L^{\infty}([0,1])$ defined by $F(f)(s)=f(s)+\sup_{r\in[0,s]}-f(r)$ is  bounded with respect to the supremum norm.  Moreover, by L\'evy's law, $F(B_{s})$ has the same distribution as the absolute value of a Brownian motion~\cite[Ch.6]{Karatzas}. 

 Part (3) of Prop.~\ref{Basics} states that $M_{r}$ is a sum of the uncorrelated martingales $m_{r}$ and $M_{r}-m_{r}$.  By Lem.~\ref{LittleEm}, the predictable quadratic variation for $t^{-\frac{1}{2}}m_{st}$, $0\leq s\leq 1 $ converges in distribution to $\sigma s$.  By~\cite[Thm.VIII.2.13]{Pollard}, it will follow that $t^{-\frac{1}{2}}m_{st}$ converges in law to a Brownian motion with diffusion constant $\sigma$ if the following Lindberg condition holds as $t\rightarrow \infty$:
$$ t^{-1}\mathbb{E}\Big[ \sup_{0\leq r\leq t}|m_{r}-m_{r-}|^{2}   \Big]\longrightarrow 0.  $$
The quadratic variation of $m_{r}$ is dominated by the quadratic variation of the L\'evy process by Part (3) of Prop.~\ref{Basics}, and hence I have the first inequality below.  
\begin{align*}
 t^{-1}\mathbb{E}\Big[ \sup_{0\leq r\leq t}|m_{r}-m_{r-}|^{2}   \Big]\leq & t^{-1}\mathbb{E}\Big[ \sup_{0\leq n\leq \mathcal{N}_{t} }|v_{n}|^{2}   \Big]\\ \leq & t^{-1}\mathbb{E}\Big[\sum_{n=1}^{\mathcal{N}_{t}}|v_{n}|^{4}   \Big]^{\frac{1}{2}} \\ =& t^{-\frac{1}{2}}\Big(\int_{\R}dv\,j(v)\,v^{4}   \Big)^{\frac{1}{2}} . 
 \end{align*}
The second inequality is Jensen's, and the equality holds because the jumps $v_{n}$ occur with L\'evy rate $j(v)$.   By the assumptions on the jump rates $j(v)$, the fourth moment is finite, and the bottom line tends to zero for large $t$.  That proves the convergence in law for $t^{-\frac{1}{2}}m_{st}$.

 Next, I argue that $t^{-\frac{1}{2}}\sup_{0\leq s\leq 1} |M_{st}-m_{st}|$ converges weakly to zero as $t\rightarrow \infty$.  It will then follow that  $t^{-\frac{1}{2}}M_{st}$ converges to a Brownian motion.   Since the martingales $m_{r}$ and $M_{r}-m_{r}$ are uncorrelated, the formula holds:
  \begin{align}\label{Lamma}
\sigma\geq \frac{d}{dr}\langle M,M\rangle_{r}=    \frac{d}{dr}\langle m,m\rangle_{r}+\frac{d}{dr}\langle M-m,M-m\rangle_{r}.
\end{align} 
With  Doob's inequality, the inequality~(\ref{Lamma}), and Lem.~\ref{LittleEm}, I have
\begin{align*}
t^{-1}\mathbb{E}\big[\sup_{0\leq r\leq t} | M_{r}-m_{r}|^{2}\big]\leq & 4t^{-1}\mathbb{E}\big[ | M_{r}-m_{r}|^{2}   \big]\\ = &4t^{-1}\mathbb{E}\big[ \langle M-m, M-m\rangle_{t}   \big]\\  \leq &4\mathbb{E}\big[\sigma-t^{-1}\langle m,m\rangle_{t}\big] \longrightarrow 0.
\end{align*}

\noindent (ii).\hspace{.15cm}  Note that $G_{r}$ is the smallest positive submartingale that has Doob-Meyer decomposition with $M_{r}$ as the martingale part, so $\mathcal{E}_{r}\geq G_{r}$.  I begin with a technical point extending the discussion of (i).  Let $\mathcal{A}_{s}^{(t)}$ be the increasing part in the Doob-Meyer decomposition for $t^{-1}G_{st}^{2}$.  I should expect that $\mathcal{A}_{s}^{(t)}$ converges to $\sigma s$, since $(\sigma t)^{-\frac{1}{2}}G_{st}$ converges to the absolute value of a standard Brownian motion, and the increasing part of the Doob-Meyer decomposition for the square of a standard Brownian motion $\mathbf{B}_{s}$ is $s$.  My first task will be to show there is convergence in distribution
\begin{align}\label{FirstisFirst}
\sup_{0\leq s \leq 1}\big(\sigma s -\mathcal{A}_{s}^{(t)}\big)\Longrightarrow 0   
\end{align}
for large $t$. 

The value $\mathcal{A}_{s}^{(t)}$ has the form $\mathcal{A}_{s}^{(t)}=t^{-1}\langle M,M\rangle_{st}-Z_{s}^{(t)}$, where    $Z_{s}^{(t)}$ is a term related to the times when $t^{-\frac{1}{2} }G_{st}$ is likely to jump to zero: $Z_{s}^{(t)}:=  t^{-1}\mathcal{R}\int_{0}^{st}\,dr F(\mathcal{E}_{r},G_{r})   $ for
$$F(G_{r},K_{r}): =\mathbb{E}\Big[\big|M_{r}-M_{r^{-}}+G_{r} \big|^{2}\chi(M_{r}-M_{r^{-}}\leq -G_{r})\,\Big| G_{r},\,K_{r}, \,\mathcal{N}_{r}-\mathcal{N}_{r^{-}}\neq 0    \Big].$$  
The dependence on $K_{r}$ enters through the statistics for $M_{r}-M_{r^{-}}$ conditioned on the event that $r\in \R_{+}$ is a Poisson time (i.e. $\mathcal{N}_{r}-\mathcal{N}_{r^{-}}\neq 0$).  It follows that $t^{-1}\langle M,M\rangle_{st} -\mathcal{A}_{s}^{(t)} $ is increasing in $s$ and so     
  $$\sup_{0\leq s\leq 1}\big(\sigma s  -\mathcal{A}_{s}^{(t)}\big)= \big(\sigma-t^{-1}\langle M,M\rangle_{t}\big)+Z_{1}^{(t)}.           $$
By Part (1), the first term on the right converges in distribution to zero as $t\rightarrow \infty$.  

Now to address $Z_{1}^{(t)}$.  I always have the inequality $F(G_{r},K_{r})\leq \frac{\sigma}{\mathcal{R}}$, and for  the regime $|G_{r}|\gg 1$, then
\begin{align*}
 F(G_{r},K_{r})& =\mathbb{E}\Big[\big|M_{r}-M_{r^{-}}+G_{r} \big|^{2}\chi(M_{r}-M_{r^{-}}\leq -G_{r})\,\Big| G_{r},\,K_{r}, \,\mathcal{N}_{r}-\mathcal{N}_{r^{-}}\neq 0    \Big]\nonumber  \\ & \leq \mathbb{E}\Big[\big|M_{r}-M_{r^{-}} \big|^{2}\chi(|M_{r}-M_{r^{-}}|\geq G_{r})\,\Big| G_{r},\,K_{r}, \,\mathcal{N}_{r}-\mathcal{N}_{r^{-}}\neq 0    \Big]\nonumber  \\ & =
 \mathit{O}\Big(\frac{1}{1+|\beta(K_{r})     | }     \Big).   
\end{align*}
To see the order equality, I first break $M_{r}-M_{r^{-}}$ into the two parts  $(M-m)_{r}-(M-m)_{r^{-}}$ and $m_{r}-m_{r^{-}}$ and use the simple inequality
\begin{multline*}
\big|M_{r}-M_{r^{-}}+G_{r} \big|^{2} \chi(|M_{r}-M_{r^{-}}|\geq G_{r})\\ \leq  4\big| m_{r}-m_{r^{-}} \big|^{2}\chi(|m_{r}-m_{r^{-}}|\geq \frac{1}{2}G_{r})+4 \big|(M-m)_{r}-(M-m)_{r^{-}} \big|^{2}  
\end{multline*}
to separate the terms in the expectation.   By the reasoning in the proof of Part (3) of Prop.~\ref{Basics}, $|m_{r}-m_{r^{-}}|^{2}$ is smaller than the square of the L\'evy jump $|L_{r}-L_{r^{-}}|^{2}$.  Since the L\'evy jumps have exponential tails, the expectation of  $|m_{r}-m_{r^{-}}|^{2}\chi(|m_{r}-m_{r-}|\geq \frac{1}{2}G_{r})$ must decay exponentially in $G_{r}$ (and therefore be negligible).   The expectation of $\big|(M-m)_{r}-(M-m)_{r^{-}}\big|^{2}$ is smaller than $\frac{\mathbf{a}}{1+|\beta(K_{r^{-}})| }$ by the proof of Part (4) of Prop.~\ref{Basics}.  Note that $\mathcal{E}_{r}\geq G_{r}$ and $\mathcal{E}_{r}\approx |K_{r}|$ to conclude that when $G_{r}$ is large, then $|K_{r}|$ is also large.  Doubling the bound to cover the smaller  term $|m_{r}-m_{r^{-}}|^{2}$ above, I can write       
\begin{align}\label{LikeEm}
\mathbb{E}\Big[ t^{-1}\mathcal{R}\int_{0}^{t}dr F(G_{r},K_{r})\Big]\leq  \mathbb{E}\Big[ t^{-1}\int_{0}^{t}dr\frac{8\mathcal{R}\mathbf{a}  }{1+|\beta(K_{r})|}  \Big].
\end{align}
Define $T_{\epsilon,t}^{\prime} =t^{-1}\int_{0}^{1}dr\,\chi( t^{-\frac{1}{2}} G_{st} \leq \epsilon)$.  For some $c> 0$, the following inequality holds:  $$\limsup_{t\rightarrow \infty} \mathbb{P}[ T_{\epsilon,t}^{\prime}<\delta]\leq c\frac{\epsilon}{\delta}.$$   This can be either proved using that $t^{-\frac{1}{2}}G_{st}$ converges to  the absolute value of a Brownian motion and a little estimate involving Gaussians or by a small modification of the proof of Lem.~\ref{EnergyLemma} with $\mathcal{E}_{r}$ replaced by $G_{r}$.  Hence, for any $\delta>0$ I can pick an $\epsilon$ small enough to make the probabilities $\mathbb{P}[ T_{\epsilon,t}^{\prime}<\delta]$ uniformly small for large $t$.  By a similar argument as in the proof of Prop.~\ref{LittleEm}, the right side of~(\ref{LikeEm}) converges to zero, which finishes the proof of~(\ref{FirstisFirst}).      
 
With~(\ref{FirstisFirst}) as a tool, I go to the core of the proof.  Since $t^{-1}\mathcal{E}_{st}^{2}$ and $t^{-1}G_{st}^{2}$ are both submartingales, taking Doob-Meyer decompositions allows me to write their difference as a martingale $\mathcal{M}_{s}^{(t)}$ plus the difference of the predictable increasing parts $\sigma s$ and $ \mathcal{A}_{s}^{(t)}$.  I can also write the difference between $t^{-1}\mathcal{E}_{st}^{2}$ and $t^{-1}G_{st}^{2}$  in an expression involving $A_{st} - \sup_{0\leq r\leq st}-M_{r}$ as follows 
\begin{align}\label{Corinthians}
 2t^{-1}\mathcal{E}_{st}(\mathcal{A}_{st}  -\sup_{0\leq r\leq st}-M_{r}  )+t^{-1}(A_{st} - \sup_{0\leq r\leq st}-M_{r}  )^{2} &=  t^{-1}\mathcal{E}_{st}^{2}-t^{-1}G_{st}^{2} \nonumber  \\  &=\mathcal{M}_{s}^{(t)}+\big(\sigma s -\mathcal{A}_{s}^{(t)}\big). 
 \end{align}
Since the left side is made up of positive terms, I can immediately conclude that  
 $$\sup_{0\leq s\leq 1} t^{-1}(A_{st} - \sup_{0\leq r\leq st}-M_{r})^{2}\leq  \sup_{0\leq s\leq 1}\mathcal{M}_{s}^{(t)}+ \sup_{0\leq s\leq 1}\big(\sigma s -\mathcal{A}_{s}^{(t)}\big).  $$  
Hence, showing that the terms on the right converge in distribution to zero for large $t$ will be sufficient.  I have already shown $\sup_{0\leq s\leq 1}\big(\sigma s -\mathcal{A}_{s}^{(t)}\big)$ converges to zero, and now I turn to the martingale component.         

Since the left side of~(\ref{Corinthians}) is positive, it follows that $-\mathcal{M}_{s}^{(t)}\leq \sigma s -\mathcal{A}_{s}^{(t)}$ for all $s\in [0,1]$, and thus
 \begin{align}\label{AACP}
  \mathbb{E}\big[|\mathcal{M}_{s}^{(t)}|\chi(\mathcal{M}_{s}^{(t)}<0    )\big]\leq \mathbb{E}\big[ \sigma s -\mathcal{A}_{s}^{(t)} \big].
 \end{align}
Let $\tau$ be the stopping time when $\mathcal{M}_{s}^{(t)}$ first jumps  $\geq \epsilon $ or $\tau=1$ when values above $\epsilon$ are not reached by time $s=1$.  By the optional sampling theorem, 
 I have the equality $\mathbb{E}\big[\mathcal{M}_{\tau}^{(t)}   \big]=0 $.  Since $\mathcal{M}_{\tau}^{(t)} $ has mean zero, the equality below holds: 
\begin{align}
 \mathbb{P}\Big[\sup_{0\leq s\leq 1}\mathcal{M}_{s}^{(t)}\geq \epsilon \Big]& \leq  \epsilon^{-1} \mathbb{E}\big[\mathcal{M}_{\tau}^{(t)}  \chi(\mathcal{M}_{\tau}^{(t)}>\epsilon )\big]\leq  \epsilon^{-1} \mathbb{E}\big[ |\mathcal{M}_{\tau}^{(t)} | \chi(\mathcal{M}_{\tau}^{(t)}>0 )\big] \nonumber \\  &= \epsilon^{-1}\mathbb{E}\big[|\mathcal{M}_{\tau}^{(t)}|\chi(\mathcal{M}_{\tau}^{(t)}<0)  \big]\leq \epsilon^{-1}\mathbb{E}\big[|\mathcal{M}_{1}^{(t)}|\chi(\mathcal{M}_{1}^{(t)}<0)  \big] \nonumber  \\ &\leq \epsilon^{-1}\mathbb{E}\big[ \sigma  -\mathcal{A}_{1}^{(t)}\big].          
\end{align}    
The first inequality is Chebyshev's and the last is~(\ref{AACP}).  Since the random variables $\sigma  -\mathcal{A}_{1}^{(t)}$ are bounded by $\sigma$ and converge to zero in distribution as $t\rightarrow \infty$, the expectation on the right tends to zero.

\end{proof}


\subsection{The waiting-times between momentum reflections }\label{SubSecRefl}

This section focuses on bounding and characterizing the length of the time intervals between ``Bragg reflections" for the momentum process.  The stochastic process $K_{r}$ exhibits a singular rate of sign-changing in comparison to an ordinary random walk due to the lattice components of the momentum jumps.  The lattice jumps yielding reflections are ``macroscopic", since they transport the momentum to a value that would  require many steps through the L\'evy component to reach.  This unusual sign-flipping characteristic was not apparent for the study of $\mathcal{E}_{r}=E^{\frac{1}{2}}(K_{r})$ in Sect.~\ref{SecSubMart}, since $\mathcal{E}_{r}$ depends only on the absolute value of the momentum.   The results below are directed toward showing that the waiting-time $\tau_{m+1}-\tau_{m}$ between successive reflections $\tau_{m},\,\tau_{m+1}$ is approximately an exponential distribution with mean $\nu^{-1}|K_{\tau_{m}}|$ and is independent of the subsequent fluctuations in momentum.  This idealized behavior emerges for high beginning momentum $|K_{\tau_{m}}|\gg 1$.

It will be useful to work with an idealization of the momentum process that makes the same L\'evy jumps with rate $j(v)$  except that the additional lattice jump $m\in \Z$  takes only  the values $0$ or $-\mathbf{n}(k+v)$ with the probabilities defined below.  Define the functions $\Pi_{-}(k)= \frac{\alpha^{2}}{8\pi^{2}}(k^{2}+\frac{\alpha^{2}}{16\pi^{2}})^{-1}$, $\mathbf{R}_{-}(k)=\Pi_{-}\big(\frac{1}{2}\Theta(k)\mathbf{n}(k)\big)$, and $\mathbf{R}_{+}=1-\mathbf{R}_{-}$.  If the current state of the momentum is $k$ and the L\'evy increment is $v$, then the conditional probabilities for the lattice component are defined to be
\begin{align}\label{AlterStat}
m=\left\{  \begin{array}{cc}   0  & \quad \mathbf{R}_{+}(k+v),  \\ \text{ } & \text{ }\\ -\mathbf{n}(k+v) & \quad  \mathbf{R}_{-}(k+v).   \end{array} \right. 
\end{align} 
This process describes a particle whose momentum is reflected at the next Poisson time based on a weighted coin flip whose bias is determined by the sum of the current position $k$ and next L\'evy increment $v$.  Of course, the simplified statistics is a useful approximation only for time intervals when the momentum is high.  Note that the idealized  process closely resembles the process proposed in Sect.~\ref{SecButchered}.  I will refer to probabilities and expectations in the law above by  $\widetilde{\mathbb{P}}$ and $\widetilde{\mathbb{E}}$, respectively.  For a technical reason explained below, I define the particle to make an additional jump $-\mathbf{n}(K_{0})$ with probability $\mathbf{r}_{-}(K_{0})$ at an infinitesimal time after zero .

Consider a particle beginning with a momentum $k$ with $|k|\gg 1$ and making jumps according to the simplified law above, and let $\tau$ be the first reflection time, i.e., the non-trivial lattice jump.  I will sketch why the random variable $\frac{\tau}{|k|}$ is approximately an exponential with mean $\nu^{-1}$.  It is equivalent to  consider a discrete-time random walk $X_{n}$, $n\in \mathbb{N}$ beginning from $k$ with jump increment density $\frac{j(v)}{\mathcal{R}}$ and a  hitting time $\hat{N}$ marking the first tails outcome for coin tosses with tails probabilities $\mathbf{R}_{-}(X_{n})$.   As long as $X_{n}$ does not stray too far from $k$, the probabilities  $\mathbf{R}_{-}(X_{n})$ are approximately
\begin{align}\label{Shimmy}
\mathbf{R}_{-}(X_{n})\approx  \Pi_{-}\big(|k|\Theta(X_{n})\big).
\end{align}
The problem is thus further reduced to the contracted  Markov chain $\widehat{\theta}_{n}:=\Theta(X_{n})$ living on the torus $\mathbb{T}=[-\frac{1}{4},\frac{1}{4})$ with the transition operator $\mathbf{T}:L^{1}(\mathbb{T})$ (defined above Lem.~\ref{LittleEm}).   The probabilities $\Pi_{-}\big(|k|\widehat{\theta}_{n}\big)$ decay rapidly for $|\widehat{\theta}_{n}|\gg |k|^{-1} $, so  $\widehat{\theta}_{n}$ will typically require many steps to  score a  reflection. Since  the chain $\widehat{\theta}_{n}$ is exponentially ergodic to the uniform distribution over $\mathbb{T}$, the probability of a reflection on a given time step is approximately $2\int_{\mathbb{T}}d\theta\Pi_{-}(|k|\theta)\approx \frac{\alpha}{|k|} $.  The random variable $\frac{\hat{N}}{|k|}$ should thus be close in law to a mean-$\alpha^{-1}$ exponential for $|k|\gg 1$.  This conclusion is consistent with the ansatz~(\ref{Shimmy}), since the fluctuations in the random walk $X_{n}$ over $[0,\hat{N}]$ will be on the order $\mathit{O}(|k|^{\frac{1}{2}})$, which is small compared to $|k|$.  Since the Poisson times have rate $\mathcal{R}$ and $\nu:=\mathcal{R}\alpha$, the distribution for  $\frac{\tau}{|k|}$ will be approximately a mean-$\nu^{-1}$ exponential.

  The connection between the momentum process and the simplified law~(\ref{AlterStat}) is not obvious, since the approximations for the probabilities $|\kappa_{v}(k,n)|^{2}$ in Prop.~\ref{BadTerms} reduce  for $|k|\gg 1$ to expressions $\mathbf{r}_{\epsilon_{1}}(k) \mathbf{r}_{\epsilon_{2}}(k+v) $, $\epsilon_{1},\epsilon_{2}\in \{\pm\}$ when $n\in I(k,v)$ and zero for $n\notin I(k,v)$.  For a hint of the link, note the identity
\begin{align}\label{Huffington}
\mathbf{R}_{-}(k)=2\mathbf{r}_{-}(k) \mathbf{r}_{+}(k).
\end{align} 
  An intermediary law between the original momentum process and~(\ref{AlterStat}) is given by the process that accompanies a L\'evy jump $v$ from the momentum $k$ with four possible lattice jumps with values and probabilities given by:
\begin{align}\label{AlterStat2}
m=\left\{  \begin{array}{cc}   -\mathbf{n}(k) &  \mathbf{r}_{-}(k) \mathbf{r}_{+}(k+v),    \\ -\mathbf{n}(k+v)  & \mathbf{r}_{+}(k)\mathbf{r}_{-}(k+v), \\ \mathbf{n}(k)-\mathbf{n}(k+v) & \mathbf{r}_{-}(k) \mathbf{r}_{-}(k+v), \\  0 & \mathbf{r}_{+}(k)\mathbf{r}_{+}(k+v).    \end{array} \right. 
\end{align}
Naturally, when $\mathbf{n}(k)=\mathbf{n}(k+v)$, then the probabilities are summed.  I refer to these statistics by $\widetilde{\mathbb{P}}^{ \prime},\widetilde{\mathbb{E}}^{ \prime }$.  There is an underlying equivalence between the laws~(\ref{AlterStat}) and~(\ref{AlterStat2}) such that they may be embedded on a single probability space in which the trajectories match for specific realizations except for time periods around reflection times where the values become staggered.  To demonstrate this, I will consider only skeletal chains, since the waiting-times between jumps are independent of the states. As before, let $v_{m}$, $m\in \mathbb{N}$ be a sequence of independent random variables with density $\frac{j(v)}{\mathcal{R}}$. Let the (non stationary) Markov chain $\tilde{X}_{n}$  have increments whose law depends on the parity of $n\geq 1$ as
 \begin{align*}
\tilde{X}_{n}-\tilde{X}_{n-1}=     \left\{  \begin{array}{ccc}  0  & \quad  \mathbf{r}_{+}\big(\tilde{X}_{n-1} \big) &\quad $n$ \text{ odd}, \\ -\mathbf{n}\big(  \tilde{X}_{n-1} \big)  & \quad \mathbf{r}_{-}\big(\tilde{X}_{n-1} \big)  & \quad   $n$ \text{ odd} ,    \\ 
v_{\frac{n}{2}}  & \quad \mathbf{r}_{+}\big( \tilde{X}_{n-1} +v_{\frac{n}{2}}\big)    & \quad  $n$ \text{ even}, \\ v_{\frac{n}{2}}- \mathbf{n}\big(  \tilde{X}_{n-1} +v_{\frac{n}{2}}\big)  & \quad \mathbf{r}_{-}\big( \tilde{X}_{n-1} +v_{\frac{n}{2}}\big)    & \quad  $n$ \text{ even},    \end{array} \right.
\end{align*}
where the second column on the right lists the probabilities for the lattice component.  Thus, the jumps $v_{m}$ occur at even time steps and coins are flipped at each time step to determine whether there is an extra lattice jump.  The chains $\tilde{X}_{2m+1}$ and $\tilde{X}_{2m}$ for $m\geq 0$ are each Markovian and have the statistics of~(\ref{AlterStat}) and (\ref{AlterStat2}), respectively.   This result depends on the symmetries $j(v)=j(-v)$ and $\mathbf{r}_{-}(k)=\mathbf{r}_{-}(k-\mathbf{n}(k) )$.  The two idealized laws are interchangeable for the purposes of this section.

Let $S:\R\rightarrow \{\pm 1\}$ be the sign function. I will use the term ``sign-flip" in a technical sense to refer to a Poisson time $t_{n}$ such that for some odd $m\in \mathbb{N}$ 
\begin{align}\label{RealDeal}
 S(K_{t_{n-m-1} })=S(K_{t_{n-m} })=-S(K_{t_{n}})=-S(K_{t_{n+1}})    
 \end{align}
and $ S(K_{t_{n-r} })=-S(K_{t_{n-r+1} }) $   for $ r\in [1,m)$.
In other words, there is an odd-numbered  sequence of  sign-changes ending at the time $t_{n}$.  Only sign-flips in which there is a single sign-change are likely to occur in practice.  The complication in the definition results from the fact that there is a comparatively high probability of a sign-change at the Poisson time after a sign-change, and I would like to avoid counting these ``fake" sign-flips.   For example, suppose the beginning momentum is $k$ with $|k|\gg 1$ and the next two L\'evy jumps are $v_{1}$ and $v_{2}$.  By the approximate statistics~(\ref{AlterStat2}), the probability assigned  for the lattice jump $-\mathbf{n}(k+v_{1})$ at the first Poisson time is $\mathbf{r}_{+}(k)\mathbf{r}_{-}(k+v_{1})$, and the resulting landing point is $k+v_{1}-\mathbf{n}(k+v_{1})\approx -k-v_{1}$.  The value $\mathbf{r}_{-}(k+v_{1})$ can not be too small if this jump is likely to occur.  The probability that the particle makes another lattice jump $\mathbf{n}(k+v_{1})$ at the second Poisson time to land at  $k+v_{1}-v_{2}$ is approximately
$$\mathbf{r}_{-}\big(k+v_{1}-\mathbf{n}(k+v_{1})\big)\mathbf{r}_{+}\big(k+v_{1}+v_{2}-\mathbf{n}(k+v_{1})\big)  \geq  \frac{1}{2}\mathbf{r}_{-}\big(k+v_{1}\big),  $$
where the inequality is by the lower bound $ \mathbf{r}_{+}\geq \frac{1}{2}$ and the symmetry $\mathbf{r}_{-}(k-\mathbf{n}(k))=\mathbf{r}_{-}(k)$.  Thus, the second sign-change has a nonnegligible probability $\geq \frac{1}{2}\mathbf{r}_{-}(k+v_{1})$ regardless of the value $v_{2}$.

  A sign-flip time is not a hitting time, since it requires the information from the following Poisson verifying that the sign of the momentum does not change again at next the momentum jump.  However, this will not give much information.  Consider $K_{0}=k$ with $|k|\gg1$, and let $q^{\prime},q\in L^{1}(\R)$ be the distribution for the first Poisson time with and without conditioning to not change sign, respectively.  Using the estimates (1) and (3) from Lem.~\ref{BadTerms}, it can be proven that the probability of a sign-change can be at most $\approx \frac{1}{2}$.  Hence, I will have the  bound
 \begin{align}\label{WithOrWithoutYou}
 q^{\prime}(p)\leq 2 q(p),  \hspace{2cm} p\in \R.
 \end{align}
In particular, there are no density peaks near lattice points $\frac{1}{2}\Z$, where approximations tend to be weaker. 

The following proposition is the main result of Sect.~\ref{SubSecRefl} and the proof is in Sect.~\ref{SecFlipidy}.  The inequalities in Prop.~\ref{TimeFlip} have the following purposes: Parts (1) and (2) are characterizations for how close the random variables $\frac{\tau_{m+1}-\tau_{m}   }{|K_{\tau_{m}}|}$ are to mean-$\nu^{-1}$ exponentials for $|K_{\tau_{m}}|\gg 1$.   Part (3) bounds the amount of time that the momentum process spends performing ``fake sign-flips" before the actual sign-flip occurs in the technical sense~(\ref{RealDeal}).  Part (4) states that the probability of a macroscopic fluctuation in the absolute value  of the momentum before a sign-flip is small.  Part (5) shows that the fluctuation in the absolute value of the momentum over the time interval between sign-flips is nearly uncorrelated with the length on the time interval.  Notice that in Prop.~\ref{TimeFlip} the momentum process is conditioned not make a sign-flip at the first Poisson time.  This is exactly the situation that I have for the process $K_{r}$, $r\in [\tau_{m},\infty)$ when  given the information $\widetilde{\mathcal{F}}_{\tau_{m}^{-}}$ and $\tau_{m}$ is a sign-flip.  This does not explicitly cover the case that $\tau_{m}$ is not a sign-flip  as when $\tau_{m}=\varsigma_{j} $ for some $j\in \mathbb{N}$ or $|K_{\tau_{m}}|\notin [\frac{1}{2}|K_{\tau_{m-1}}|, \frac{3}{2} |K_{\tau_{m-1}}|]$. These events do not occur frequently enough to generate nonnegligible contributions to the quantities that I study (in fact, Part (4) of Prop.~\ref{TimeFlip} bounds the probability of the event $|K_{\tau_{m}}|\notin [\frac{1}{2}|K_{\tau_{m-1}}|, \frac{3}{2} |K_{\tau_{m-1}}|]$).   Without conditioning the  momentum to remain the same sign after the first jump,  the estimates in Parts (1) and (2) may be distorted if $K_{0}$ is near a lattice point $\frac{1}{2}\Z$, since  $\tau$ will have a comparatively high probability of occurring at the first Poisson time. However,  Prop.~\ref{TimeFlip} still offers bounds for the moments of $\frac{\tau}{|k|}$ in that case.   There are other natural conditions on the beginning momentum that will yield the results of Prop.~\ref{TimeFlip}.  The assumptions  that $K_{r}$ begins with the value $K_{0}=k$ and is conditioned to avoid a sign change at the first Poisson time can be replaced by the assumptions that $K_{r}$ begins in a distribution given by a density $q\in L^{1}(\R)$  concentrated around the value $k$ and that the contracted density $\langle q \rangle \in L^{1}(\mathbb{T}) $ is bounded.       

\begin{proposition} \label{TimeFlip}
Let $\zeta>0$, $K_{0}=k$ for $|k|\gg 1$, and $K_{r}$ be conditioned not to make a sign change at the first Poisson time (i.e. $ S(K_{0})=S(K_{t_{1}})$).  Define $\tau$ to be  the first time that either $K_{r}$ has a sign-flip or  $|K_{r}|$ jumps out of the set $[\frac{1}{2}|k|,\,\frac{3}{2}|k|]$ depending on what occurs first.  For fixed $\zeta,\, m>0$, there exist $C$ and $\gamma_{0}$ such that the following inequalities hold for all $k$ and $0< \gamma\leq \gamma_{0}$:

\begin{enumerate}
\item $\Big|\mathbb{E}\big[ \big(\frac{\tau}{|k|}\big)^{m}  \big]-m!\nu^{-m}   \Big|\leq \frac{C}{ |k|^{1-\zeta}},  $

\item 
$\Big|\mathbb{E}\big[ e^{\gamma\frac{\tau}{|k|}} \big]-\frac{\nu}{\nu-\gamma}   \Big|\leq \frac{C}{ |k|^{1-\zeta}},  $

\item 
$ \mathbb{E}\Big[\int_{0}^{\tau}dr\chi\big(S(K_{r})\neq S(K_{0})   \big)   \Big] \leq C, $
\item 
$\mathbb{P}\Big[ |K_{\tau}|\notin \big[\frac{1}{2}|k|,\,\frac{3}{2}|k|\big]     \Big]\leq \frac{C\log(|k|)}{|k|^{2}},     $

\item $ \mathbb{E}\Big[ \big( |K_{\tau}|-|K_{0}|   \big)\tau  \Big] \leq  C|k|^{1+\zeta}. $

\end{enumerate}

\end{proposition}

Recall by the remarks preceding  Lem.~\ref{LittleEm} that the chain $\widehat{\theta}_{n}=\Theta(K_{t_{n}})$ on the torus $\mathbb{T}=[-\frac{1}{4},\frac{1}{4})$ is Markovian with bounded transition operator $\sup_{\theta_{1},\theta_{2}\in\mathbb{T}} \mathbf{T}(\theta_{2},\theta_{1})<\frac{\mu}{\mathcal{R}}$.  Since the lattice jumps do not appear in $\mathbb{T}$, the contracted process for the idealized momentum process is exactly the same.  The torus chain has  detailed balance, since the kernel is symmetric $\mathbf{T}(\theta_{2},\theta_{1})=\mathbf{T}(\theta_{1},\theta_{2})$ as a consequence of the symmetry $j(v)=j(-v)$.  The process is thus time-reversible, and its stationary state is the uniform distribution.  Due to the boundedness of the kernel, the process is exponentially ergodic, and even in the supremum norm I have the existence of an $\varepsilon>0$ such that    
\begin{align}\label{Cabbage}
\sup_{\theta \in [-\frac{1}{4},\frac{1}{4})}\big|\big(\mathbf{T}^{n}h\big)(\theta)-2  \big|\leq e^{-n\varepsilon }\sup_{\theta \in [-\frac{1}{4},\frac{1}{4})}\big|h(\theta)-2  \big|.  
\end{align}

The laws $\mathbb{P}$ and $\widetilde{\mathbb{P}}^{ \prime }$ determine two probability measures on the sequence of momenta $K_{t_{j}}$.   I denote the total variation distance between the measures on sequences of length $M$   by $\|\cdot\|_{\textup{Var},M}$.   Let $A_{M}$ be the event $ |K_{t_{j}}| \in [\frac{1}{2}|k|,\frac{3}{2}|k|]$ for $0\leq j\leq M $, and $\chi(A_{M})\mathbb{P}$ be the positive measure that agrees with $\mathbb{P}$ for events in $A_{M}$ and assigns events in $A^{c}_{M}$ zero weight. 
 The key inputs for the proof of Lem.~\ref{LemTotVar} are the estimates (1) and (3) of Lem.~\ref{BadTerms}.

\begin{lemma}\label{LemTotVar}
Let $K_{0}=k$. There is a $c>0$ such that for all $k$,
$$ \big\|\chi(A_{M})\mathbb{P}-\chi(A_{M})\widetilde{\mathbb{P}}^{ \prime}\big\|_{\textup{Var},M } \leq  c M\frac{ \log(|k|)  }{ |k|^{2} }.    $$

\end{lemma}

\begin{proof}
Let  $\widetilde{\mathbb{P}}^{(n)}$ refer to the process that obeys the law~(\ref{AlterStat2}) for the first $n$ jumps and the original law for the remaining jumps.   The variation norm for the difference $\chi(A_{M})\mathbb{P}-\chi(A_{M})\widetilde{\mathbb{P}}^{ \prime}$  is bounded by
\begin{align*}
 \big\|\chi(A_{M})\mathbb{P}-\chi(A_{M})\widetilde{\mathbb{P}}^{ \prime}\big\|_{\textup{Var},M} \leq & \sum_{n=1}^{M } \big\|\chi(A_{M})\widetilde{\mathbb{P}}^{(n-1)}-\chi(A_{M})\widetilde{\mathbb{P}}^{(n)}\big\|_{\textup{Var},M}\\ \leq & M \sup_{n\in \mathbb{N}} \big\|\chi(A_{M})\widetilde{\mathbb{P}}^{(n-1)}-\chi(A_{M})\widetilde{\mathbb{P}}^{(n)}\big\|_{\textup{Var},M } ,
\end{align*}
where I have used the triangle inequality with a telescoping sum determined by inserting the measures $\chi(A)\widetilde{\mathbb{P}}^{(n)}$.  
Let  $dL_{n}$ and $d\ell_{n}$  denote the L\'evy and lattice jumps at the $n$th Poisson time.  The laws $\widetilde{\mathbb{P}}^{(n)}$ and $\widetilde{\mathbb{P}}^{(n-1)}$ differ only at the $n$th Poisson time in the probabilities that they assign for the lattice jump $d\ell_{n}=m\in \Z$.  
 I can bound a single term from the sum above through the inequality
\begin{align}\label{Franklin}
 \big\|\chi(A_{M})\widetilde{\mathbb{P}}^{(n-1)}-\chi(A_{M})\widetilde{\mathbb{P}}^{(n)}\big\|_{\textup{Var},M } < & \int_{\R} dp\big|\hat{q}_{n}(p)-\hat{q}_{n}^{\prime}(p)\big|\nonumber \\ \leq & \int_{\R}dp\,\hat{q}_{n-1}^{\prime}(p)\int_{\R}dv\,\frac{j(v)}{\mathcal{R}} \sum_{m\in \Z}\nonumber \\  \Big|\mathbb{P}\big[d\ell_{n}=m   \,\big|\,dL_{n}=v, K_{\tau_{n-1}}=p\big] -& \widetilde{\mathbb{P}}^{ \prime } \big[d\ell_{n}=m    \,\big|\,dL_{n}=v, K_{\tau_{n-1}}=p \big]    \Big|    ,
\end{align}       
where $\hat{q}_{n}$ and $\hat{q}_{n}^{\prime}$ are the densities
$$\hat{q}_{n}(p)=\mathbb{E}^{(n-1)}\big[\delta(K_{t_{n}}-p)\chi\big(A_{n}\big)     \big] \quad \text{and} \quad  \hat{q}_{n}^{\prime}(p)=\widetilde{\mathbb{E}}^{ \prime} \big[\delta(K_{t_{n}}-p)\chi\big(A_{n} \big)     \big].$$
The first inequality~(\ref{Franklin}) follows since $A_{M}\subset A_{n}$ for $M\geq n$.  I will show that the right side of~(\ref{Franklin}) is $\mathit{O}(\frac{\log(|k|)}{|k|^{2}}    ) $ for $|k|\gg 1$.  
 
By (1) of List~\ref{Assumptions} and Chebyshev's inequality, I can replace the integration $\int_{\R}dv$ in~(\ref{Franklin})  by the restriction $|v|\leq \frac{1}{4}|k|$ with an error that  decays exponentially for $|k|\gg 1$  (since the integrand is bounded by one).   Under $\widetilde{\mathbb{P}}^{ \prime}$, only the lattice jumps $m\in I(p,v)$ have nonzero probability. For the law $\mathbb{P}$,  Part (1) of Lem.~\ref{BadTerms} yields that for $|p|\in [\frac{1}{2}|k|,\frac{3}{2}|k|]$ and $|v|\leq \frac{1}{4}|k|$,
$$\sum_{m\notin I(p,p+v) } \mathbb{P}\big[d\ell=m    \,\big|\,dL_{n}=v, K_{\tau_{n-1}}=p \big]=\sum_{m\notin I(p,p+v)}\big|\kappa_{v}(p,m)\big|^{2}=\mathit{O}(|k|^{-2}),  $$ 
when $|k|\gg 1$.  Since the density $\hat{q}_{n}'(p)$ is supported on the interval $|p|\in [\frac{1}{2}|k|,\frac{3}{2}|k|]$,  I can restrict the summation $\sum_{m}$ in~(\ref{Franklin}) to $m\in I(p,v)$ with an error $\mathit{O}(|k|^{-2})$.

Next, I treat the difference~(\ref{Franklin}) for $m\in I(p,v)$. The different cases are handled similarly, and I will take $m=-\mathbf{n}(p)\neq -\mathbf{n}(p+v)$.  I make the restriction $-\mathbf{n}(p)\neq -\mathbf{n}(p+v)$, because I will apply Part (3) of Lem.~\ref{BadTerms} to obtain the errors for the approximations of $\big|\kappa_{v}\big(p,-\mathbf{n}(p)\big)\big|^{2}$, and  the cases $\mathbf{n}(p)= \mathbf{n}(p+v)$  and  $\mathbf{n}(p)\neq \mathbf{n}(p+v)$ correspond to different approximations in the lemma.  By definition,  $$\widetilde{\mathbb{P}}'\big[d\ell_{n}=-\mathbf{n}(p)   \,\big|\,dL_{n}=v, K_{\tau_{n-1}}=p\big]=\mathbf{r}_{-}\big(p\big)\mathbf{r}_{+}\big(p+v\big),$$  so I am interested in the difference 
 \begin{align}\label{BlogPost}
 \int_{\R}dp\,\hat{q}_{n-1}^{\prime}(p)\int_{\substack{|v|\leq \frac{1}{4}|k|\\ \mathbf{n}(p)\neq \mathbf{n}(p+v)   } } dv\,\frac{j(v)}{\mathcal{R}}\Big| \big|\kappa_{v}\big(p,-\mathbf{n}(p)\big)\big|^{2}-\mathbf{r}_{-}\big(p\big)\mathbf{r}_{+}\big(p+v\big)\Big| .
 \end{align}
I will change summation variables through $p= m_{1}+\theta_{1}$ and $p+v=m_{2}+\theta_{2}$ for $m_{1},m_{2}\in \frac{1}{2}\Z$ and $\theta_{1},\theta_{2}\in [-\frac{1}{4},\frac{1}{4})$.  The domain $|v|\leq \frac{1}{4}|k|$ is roughly the same as $|m_{1}-m_{2}|\leq \frac{1}{4}|k|$, so the  expression~(\ref{BlogPost}) is approximately
\begin{multline}\label{Discount}
\sum_{\substack{m_{1},m_{2}\in \frac{1}{2}\Z \\ 0<|m_{1}-m_{2}|\leq \frac{1}{4}|k| } }\mathcal{R}^{-1}\int_{[-\frac{1}{4},\frac{1}{4})} d\theta_{1}d\theta_{2}\,\hat{q}^{\prime}_{n-1}(m_{1}+\theta_{1} )j(m_{2}-m_{1}+\theta_{2}-\theta_{1})\\ \times \Big| \big|\kappa_{m_{2}-m_{1}+\theta_{2}-\theta_{1} }\big(m_{1}+\theta_{1},-2m_{1}\big)\big|^{2}-\mathbf{r}_{-}(m_{1}+\theta_{1})\mathbf{r}_{+}(m_{2}+\theta_{2})    \Big|.
\end{multline}
 By Part (3) of Lem.~\ref{BadTerms}, there is a $c>0$ such that the sum of terms $m_{1}\neq m_{2}$ in~(\ref{Discount}) is smaller than  
\begin{align*}
 \sum_{\substack{ m_{1},m_{2}\in \frac{1}{2}\Z \\ 0<|m_{1}-m_{2}|\leq \frac{1}{4}|k|   }   }\mathcal{R}^{-1}\int_{[-\frac{1}{4},\frac{1}{4})}  d\theta_{1}& d\theta_{2}\,\hat{q}^{\prime}_{n-1}(m_{1}+\theta_{1} )j(m_{2}-m_{1}+\theta_{2}-\theta_{1})\frac{c}{|m_{1}+\theta_{1}|(1+|m_{1}\theta_{1}|)} \\ & \leq \frac{4c }{\mathcal{R}|k|} \|\langle \hat{q}^{\prime}_{n-1}\rangle \|_{\infty}\|\langle j\rangle \|_{\infty}\int_{[-\frac{1}{4},\frac{1}{4})} d\theta_{1}  \frac{1}{1+|\frac{1}{4}k\theta_{1}| }\\ &\leq  \frac{32c}{\mathcal{R}}\|\langle \hat{q}_{n-1}^{\prime} \rangle \|_{\infty}\|\langle j\rangle \|_{\infty} \frac{\log(1+\frac{|k|}{16})  }{|k|^{2}},
\end{align*}
where  the map $\langle \cdot \rangle: L^{1}(\R)\rightarrow L^{1}(\mathbb{T})$ defined in~(\ref{DefTorusContraction})  contracts densities on $\R$ to densities on the torus $\mathbb{T}=[-\frac{1}{4},\frac{1}{4})$.  In the first inequality above, I have bounded $|m_{1}+\theta_{1}|^{-1}(1+|m_{1}\theta_{1}|)^{-1}$ from below by $|\frac{k}{4}|^{-1}(1+|\frac{k\theta}{4}|)^{-1}$ and commuted the sums with the integrals.
By the discussion preceding Lem.~\ref{LittleEm},   
\begin{align}\label{Karate}
\|\langle \hat{q}_{n}^{\prime}\rangle\|_{\infty} \leq \|\langle \mathcal{T}q_{n-1}^{\prime}\rangle\|_{\infty}=\|\mathbf{T}\langle \hat{q}_{n-1}^{\prime}\rangle\|_{\infty} \leq \mathcal{R}^{-1}\|\langle j \rangle\|_{\infty},
\end{align}
and  where $\mathcal{T}$ is the transition kernel for the law~(\ref{AlterStat2}).  The norm $\|\langle j\rangle\|_{\infty}$ is finite by (3) of List~\ref{Assumptions}.     Note that $\hat{q}^{\prime}_{n}(p)\leq (\mathcal{T}\hat{q}^{\prime}_{n-1})(p)$, since $ (\mathcal{T}\hat{q}^{\prime}_{n-1})(p)$ includes a contribution for $|p|$ outside the interval $\big[\frac{1}{2}|k|,\frac{3}{2}|k|\big]$.



\end{proof}

\begin{lemma}\label{ProbFlip}
 Let  $t_{j}$ for  $j\in \mathbb{N}$ be the Poisson times and $K_{0}=k$.  For $|k|$ large, there are $c,C>0$ such that the following inequalities hold: 
\begin{enumerate}
\item $\mathbb{P}\big[S(K_{0})=S(K_{t_{1}})=-S(K_{t_{2}})=- S(K_{t_{3}})\big ]\geq \frac{c}{|k|},$
 
 \item $\mathbb{P}\big[S(K_{0})=S(K_{t_{1}})=-S(K_{t_{2}}) \big ]\leq \frac{C}{|k|}$.
 \end{enumerate}
The same statements hold with the statistics $\mathbb{P}$ replaced by $\widetilde{\mathbb{P}}$.

\end{lemma}

\begin{proof}\text{ }\\

\noindent Part (1):\hspace{.15cm} The laws $\mathbb{P}$ and $\widetilde{\mathbb{P}}^{\prime}$ determine measures on the sequences $K_{t_{j}}$.  Let $\|\cdot\|_{\textup{Var},4}$ and the set $A_{4}$ be defined as in Lem~\ref{LemTotVar}.  The probability of the event $A_{4}^{c}$ is order $\mathit{O}(|k|^{-2})$ for both  $\mathbb{P}$ and $\widetilde{\mathbb{P}}^{\prime}$, since the  random variables $|K_{t_{j+1}}|-|K_{t_{j}}|$ have uniformly bounded second moments and by Chebyshev's inequality.  Thus, 
$$\big\|\mathbb{P}-\widetilde{\mathbb{P}}^{\prime}\big\|_{\textup{Var},\,4}+\mathit{O}\big(|k|^{-2}\big)=\big\|\chi(A_{4})\mathbb{P}-\chi(A_{4})\widetilde{\mathbb{P}}^{\prime}\big\|_{\textup{Var},\,4}=\mathit{O}\Big(\frac{\log(|k|)}{|k|^{2}}\Big), $$
where the second equality above is by Lem~\ref{LemTotVar}.  However, the variational distance $\|\mathbb{P}-\widetilde{\mathbb{P}}^{\prime}\|_{\textup{Var},\,4}$ will bound the difference between the probabilities
\begin{align*}
\Big|\mathbb{P}\big[S(K_{0})=S(K_{t_{1}})=-S(K_{t_{2}})=-S(K_{t_{3}}) \big ]- \widetilde{\mathbb{P}}^{ \prime}\big[S(K_{0})=S(K_{t_{1}})=-S(K_{t_{2}})=-S(K_{t_{3}}) \big ] \Big|.  
\end{align*} 
Therefore, I can substitute $\mathbb{P}$ with $\widetilde{\mathbb{P}}^{ \prime}$ with an error $\mathit{o}(|k|^{-1})$. 

One way in which the event $S(K_{0})=S(K_{t_{1}})=-S(K_{t_{2}})=-S(K_{t_{3}})$ can occur is if the L\'evy jumps satisfy $|v_{i}|\leq |k|^{\frac{1}{2}}$, $i=1,\,2,\,3$ and the corresponding lattice jumps are respectively $0, -\mathbf{n}(k+v_{1}+v_{2}), 0$.  This possibility gives a lower bound   
\begin{align*}
\widetilde{\mathbb{P}}^{ \prime }\big[ S(K_{0})&=S(K_{t_{1}})=-S(K_{t_{2}})=-S(K_{t_{3}})\big] \\ \geq & \int_{|v_{1}|,|v_{2}|,|v_{3}|  \leq |k|^{\frac{1}{2}} }\frac{j(v_{1})}{\mathcal{R}}\frac{j(v_{2})}{\mathcal{R}}\frac{j(v_{3})}{\mathcal{R}}   \mathbf{r}_{+}(k)\mathbf{r}_{+}^{2}(k+v_{1}) \mathbf{r}_{-}(k+v_{1}+v_{2}) \\  &\times\mathbf{r}_{+}\big(k+v_{1}+v_{2}- \mathbf{n}(k+v_{1}+v_{2}) \big) \mathbf{r}_{+}\big(k+v_{1}+v_{2}-\mathbf{n}(k+v_{1}+v_{2})+v_{3}\big).   
\end{align*}
I will use the identities  $\mathbf{r}_{+}(k')=\mathbf{r}_{+}\big(k'-\mathbf{n}(k')\big)$ and  $\Pi_{-}\big(2^{-1}\Theta(k')\mathbf{n}(k')\big)=  2\mathbf{r}_{+}(k')\mathbf{r}_{-}(k')$ for  $k'= k+v_{1}+v_{2}    $.  
Since $\mathbf{r}_{+}$ is  $\geq \frac{1}{2}$ and $j(v_{2})\geq \frac{1}{\mu}$ for $|v_{2}|\leq 1$  by (3) of List~\ref{Assumptions}, I have a bound from below given by  
\begin{align*}
\frac{1}{2\mu }&\int_{|v_{1}| \leq |k|^{\frac{1}{2}} }\frac{j(v_{1})}{\mathcal{R}}\int_{[-1,1]}dv_{2}\frac{j(v_{2})}{\mathcal{R}}   \Pi_{-}\big(2^{-1}\Theta(k+v_{1}+v_{2})\mathbf{n}(k+v_{1}+v_{2})  \big) \\ & \geq \frac{1}{2\mu \mathcal{R}}\int_{|v_{1}| \leq |k|^{\frac{1}{2}} }\frac{j(v_{1})}{\mathcal{R}}\int_{[-1,1]}dv_{2}\,\Pi_{-}\big(2^{-1} \Theta(k)\mathbf{n}(k+v_{1}+v_{2})  \big)    \\ & =\frac{\alpha}{\mu \mathcal{R}|k|},
\end{align*}
  where I have used that $\int_{\R}dk\,\Pi_{-}(k)=\alpha $.

\vspace{.5cm}

\noindent Part (2):\hspace{.15cm} By the argument in Part (1), I can approximate $\mathbb{P}$ with $\widetilde{\mathbb{P}}^{\prime}$.   There are four possible pairs of lattice jumps that will give a sign change for the  second jump and not the first: $0$ for the first and $-\mathbf{n}(k+v_{1})$ or $-\mathbf{n}(k+v_{1}+v_{2})$ for the second, or $\mathbf{n}(k)-\mathbf{n}(k+v_{1})$ for the first and $-2\mathbf{n}(k)+\mathbf{n}(k+v_{1})$ or $\mathbf{n}(k)-\mathbf{n}(k+v_{1})-\mathbf{n}(k+v_{1}+v_{2})$ for the second.  The cases require the same analysis (with more or less messy notation), so I will analyze the case in which the jumps are  $0$ and $-\mathbf{n}(k+v_{1})$:  
\begin{align*}
 \int_{|v_{1}|,|v_{2}|  \leq |k|^{\frac{1}{2}} }& \frac{j(v_{1})}{\mathcal{R}}\frac{j(v_{2})}{\mathcal{R}}  \mathbf{r}_{+}(k)\mathbf{r}_{+}(k+v_{1})\mathbf{r}_{-}(k+v_{1})\mathbf{r}_{+}(k+v_{1}+v_{2}) \\ & \leq \frac{1}{2} \int_{|v_{1}| \leq |k|^{\frac{1}{2}} }\frac{j(v_{1})}{\mathcal{R}} \Pi_{-}\big(2^{-1}\Theta(k+v_{1})\mathbf{n}(k+v_{1})\big)
\\  & \approx \frac{1}{2} \sum_{ \substack{ m\in \frac{1}{2}\Z , \\ |m-k|\leq \sqrt{|k|}  } }\int_{\mathbb{T}}d\theta  \frac{j(m+\theta)}{\mathcal{R}} \Pi_{-}\big(\theta m)    \\ & \leq   \frac{\alpha \|\langle j\rangle\|_{\infty}}{\mathcal{R}|k|  } .
 \end{align*}
In the first inequality, I have used that $\mathbf{r}_{+}\leq 1$.   The norm $\|\langle j\rangle \|_{\infty}$ is smaller than $\mu$ by assumption (2) of List~\ref{Assumptions}.

\end{proof}

\subsubsection{ Results for the idealized law}

Lemma~\ref{SubFraiser} concerns the simplified momentum process~(\ref{AlterStat}), or alternatively~(\ref{AlterStat2}), and it is applied in the proof of Lem.~\ref{Fraiser} to obtain the analogous result for the original momentum process.  Part (1) of the lemma states that random variables $\frac{\mathcal{N}_{\tau_{n+1}}-\mathcal{N}_{\tau_{n}}  }{|K_{\tau_{n}}|}$ are approximately mean-$\alpha^{-1}$  exponentials for $|K_{\tau_{n}}|\gg 1 $.  The proof is based on the heuristics in the paragraph~(\ref{Shimmy}).

\begin{lemma}\label{SubFraiser}
Let $\beta,\zeta>0$ and $K_{0}=k$ with $|k|\gg 1$ be conditioned not to change signs at the first Poisson time.  Let $\tau$ be the first Poisson time that either the process makes a sign flip or  its absolute value jumps out of the interval $[\frac{1}{2}|k|,\,\frac{3}{2}|k|]$.  For fixed $\zeta$, there is a $C>0$ such that for all $k\in \R$, the following inequalities hold: 
\begin{enumerate}
\item $\sup_{a \in \R_{+}} \Big| \widetilde{\mathbb{P}}\Big[\frac{\mathcal{N}_{\tau}}{|k|}>a  \Big]-e^{-\alpha a   }   \Big|\leq C |k|^{-1+\zeta},  $

\item 
$\Big|\widetilde{\mathbb{E}}\Big[\big(|K_{t_{M } }|-|K_{0}|\big)\,\chi(\mathcal{N}_{\tau}> M  )   \Big]  \Big| \leq C |k|^{\beta-1+\zeta} $, 

\end{enumerate}
where $M:=\lfloor |k|^{\beta } \rfloor$.

\end{lemma}

\begin{proof}
\text{ }\\

\noindent Part (1): \hspace{.15cm} Since I am considering the probability of the event $ \widetilde{\mathbb{P}}\big[\mathcal{N}_{\tau}\geq a|k|\big]$ that is determined when the first reflection occurs (or, less likely, when $|K_{\mathcal{N}_{\tau}}|$ leaves the interval $[\frac{1}{2}|k|,\frac{3}{2}|k|]$), the problem can be formulated in terms of a random walk $X_{n}=k+v_{1}+\dots v_{n}$, where  the increments $v_{n}$ are independent and have density $\frac{j(v)}{\mathcal{R}}$; and the reflections are decided by flipping coins with heads weight $\Pi_{+}\big(2^{-1}\mathbf{n}(X_{n})\Theta(X_{n})  \big)$.  The problem is thus translated to characterizing the distribution for the first tails outcome. This scheme ignores the special rule for the first momentum jump for the law~(\ref{AlterStat}), but that will not generate a difference.  I am also ignoring the technical definition~(\ref{RealDeal}) for a sign-flip, which is only relevant when working with the laws $\mathbb{P}$ and $\widetilde{\mathbb{P}}'$ for which the double flips are likely.     Let $\mathbf{F}\subset \mathbb{N}$ be the set of times $n$ with either a tail outcome or such that $|X_{n}|$ lands outside $[\frac{1}{2}|k|,\,\frac{3}{2}|k|]$.  Also, let $\hat{N}$ be the smallest element in $\mathbf{F}$.     

First, I will show that $\widetilde{\mathbb{P}}\big[\hat{N}\geq a|k|\big]$ decays exponentially for $a\gg 1$.  By a simpler argument than in the proof for  Part (1) of  Lem.~\ref{ProbFlip}, there is a $B>0$ such that 
\begin{align}\label{BZZZ}
\widetilde{\mathbb{P}}[  n \notin \mathbf{F}   \,\big|\,  \,X_{n-1}   ]\leq 1-\frac{B}{|k|}
\end{align}
for all $|X_{n-1}|\in \frac{1}{2}|k|$.  With an iterated conditional expectation and an application of~(\ref{BZZZ}), then
\begin{align*}
\widetilde{\mathbb{P}}\big[ \hat{N}>  n \big] &= \widetilde{\mathbb{E}}\big[\chi(1 \notin \mathbf{F})\cdots \chi( n\notin \mathbf{F})     \big] \\ &=  \widetilde{\mathbb{E}}\big[\chi(1 \notin \mathbf{F} )\cdots \chi( n-1 \notin \mathbf{F} )\widetilde{\mathbb{P}}[ n \notin \mathbf{F}   \,\big|  \,X_{n-1}   ]  \big] \\ &\leq \Big(1-\frac{B}{|k|}\Big)\widetilde{\mathbb{E}}\big[\chi(1 \notin \mathbf{F})\cdots \chi( n-1 \notin \mathbf{F})\big],
\end{align*}
where $B$ is the constant in the lower bound from Part (1) of Lem.~\ref{ProbFlip}.    By induction, $\widetilde{\mathbb{P}}\big[ \hat{N}>n \big]\leq (1-\frac{B}{|k|})^{n} $.  It follows that $\frac{\hat{N}}{|k|}$ has finite exponential moments $\widetilde{\mathbb{E}}[e^{\gamma|k|^{-1}\hat{N} }  ] $ for $0\leq \gamma< B $, which are uniformly bounded for all large $|k|$.

By the above $\widetilde{\mathbb{P}}\big[\hat{N}\geq m\big]$ decays exponentially with rate at least $(1-\frac{B}{|k|})^{m}$.  Hence, the values $\widetilde{\mathbb{P}}\big[\hat{N}\geq a |k|\big]$ and  $e^{-\alpha a  }$ will be both small for $a> |k|^{\zeta}$, and I can focus my analysis on $a \leq |k|^{\zeta}$.
Let me set $\hat{n}=\lfloor a |k| \rfloor$ and $\hat{m}= \lfloor |k|^{\zeta} \rfloor$. I can rewrite $\widetilde{\mathbb{P}}\big[\hat{N}\geq \hat{n}\big]$ as   
\begin{align}\label{Abel}
\widetilde{\mathbb{P}}\big[\hat{N}> \hat{n}\big] & =\widetilde{\mathbb{E}}\big[\chi( 1\notin \mathbf{F})\cdots \chi( \hat{n}\notin \mathbf{F})    \big]\nonumber \\ &=\widetilde{\mathbb{E}}\Big[\chi( 1 \notin \mathbf{F})\cdots \chi( \hat{n}-1\notin \mathbf{F})\widetilde{\mathbb{P}}\big[ \hat{n}\notin \mathbf{F}\,\big|\,\mathbf{F}\cap (\hat{n}-\hat{m},\,\hat{n})=\emptyset,\,X_{\hat{n}-\hat{m}}  \big]    \Big]\nonumber \\ &=\widetilde{\mathbb{P}}\big[\hat{N}> \hat{n}-1   \big]\int_{\R}dp\,r_{\hat{n},\hat{m}}(p)\,\widetilde{\mathbb{P}}\big[\hat{m} \notin \mathbf{F}\,\big|\,\hat{N}\geq \hat{m},\,X_{0}=p  \big], 
\end{align}
where $r_{\hat{n},\hat{m}}(p)$ is the density for $X_{\hat{n}-\hat{m}}$ conditioned on the event $\hat{N}> \hat{n}-1 $.  The third equality above holds by the strong Markov property.  By definition, the density $r_{\hat{n},\hat{m}}(p)$ will have its support over $|p|\in [\frac{1}{2}|k|,\frac{3}{2}|k|]$.  I will argue that the integral in the last line of~(\ref{Abel}) has the $|k|\gg 1$ asymptotics
\begin{align}\label{Curses}
\int_{\R}dp\,r_{\hat{n},\hat{m}}(p)\,\widetilde{\mathbb{P}}\big[\hat{m} \notin \mathbf{F}\,\big|\,\hat{N}\geq \hat{m},\,X_{0}=p  \big]=1-\frac{\alpha}{|k|}+\mathit{O}\big(|k|^{2\zeta-2}\big).
\end{align}
The result~(\ref{Curses}) can be applied inductively in~(\ref{Abel}) for $\widetilde{\mathbb{P}}[\hat{N}> \hat{n}-1],\dots, \widetilde{\mathbb{P}}[\hat{N}> \hat{m}+1]$ to obtain
 \begin{align*}
 \widetilde{\mathbb{P}}\big[\hat{N}> \hat{n}\big] &= \Big(1-\frac{\alpha}{|k|}\Big)^{\hat{n}-\hat{m}}\widetilde{\mathbb{P}}[\hat{N}> \hat{m}  ]+ \mathit{O}(\gamma|k|^{2\zeta-2})\sum_{j=\hat{m}+1}^{\hat{n}-1} \Big(1-\frac{\alpha}{|k|}\Big)^{\hat{n}-j-1} \widetilde{\mathbb{P}}\big[\hat{N}> j\big] \\ & =
 e^{-\alpha a  }+\mathit{O}( |k|^{2\zeta-1})e^{-A a} ,
 \end{align*}  
where $A=B\wedge \alpha $.  The second equality follows by a simple argument yielding that $\big|\mathbb{P}\big[\hat{N}\geq \hat{m}\big]-1\big|=\mathit{O}( |k|^{\zeta-1})$.  Since $\zeta>0$ is arbitrary, the result would follow.  

I will separate the analysis showing~(\ref{Curses}) into proofs of the statements (i)-(iii) below.  For $|p-k|\leq |k|^{\frac{1}{2}+\zeta}$ and $|k|\gg 1$, the following statements hold:
\begin{enumerate}[(i).]

\item The difference between the left side of~(\ref{Curses}) and 
 \begin{align}\label{Restrict}
\int_{|p-k|\leq |k|^{\frac{1}{2}+\zeta} }dp\,r_{\hat{n},\hat{m}}(p)\,\widetilde{\mathbb{P}}\big[\hat{m} \notin \mathbf{F}\,\big|\,\hat{N}\geq \hat{m},\,X_{0}=p  \big]
\end{align}
decays superpolynomially in $|k|$.

\item $\widetilde{\mathbb{P}}\big[\hat{m}\notin \mathbf{F}\,\big|\,\hat{N}\geq \hat{m},\,X_{0}=p  \big]= 1-\frac{\alpha}{|p|}+\mathit{O}(|k|^{-2+\zeta})$

\item $\int_{|p-k|\leq |k|^{\frac{1}{2}+\zeta} }dp\,r_{\hat{n},\hat{m}}(p)\big(1-\frac{\alpha}{|p|}\big)=1-\frac{\alpha}{|k|}+\mathit{O}\big( |k|^{2\zeta-2}   \big)$

\end{enumerate}
\vspace{.3cm}

\noindent (i).\hspace{.15cm}  The value  $\int_{|p-k| > |k|^{\frac{1}{2}+\zeta} }dp\,r_{\hat{n},\hat{m}}(p)$ is smaller than the probability of the event $\sup_{0\leq n\leq \hat{n}}|X_{n}-k|>
|k|^{\frac{1}{2}+\zeta}$.  The probability $\widetilde{\mathbb{P}}\big[\sup_{0\leq n\leq \hat{n}}X_{n}-k>
|k|^{\frac{1}{2}+\zeta}\big]$ will decay superpolynomially fast for $|k|\gg 1$, since
\begin{align}\label{NotMuch}
e^{|k|^{\frac{\zeta}{2}}} \widetilde{\mathbb{P}}\Big[\sup_{0\leq n\leq \hat{n}}X_{n}-k>|k|^{\frac{1}{2}+\zeta }\Big] & \leq \widetilde{\mathbb{E}}\Big[\sup_{0\leq n\leq \hat{n}}e^{  (X_{n}-k)|k|^{-\frac{1}{2}-\frac{\zeta}{2}  }}
  \Big]\nonumber \\ & \leq 4 \widetilde{\mathbb{E}}\Big[ e^{(X_{\hat{n} }-k)|k|^{-\frac{1}{2}-\frac{\zeta}{2} }  } \Big]\nonumber \\ &\leq 4\big(1+2\sigma^{\prime} |k|^{-1-\zeta }\big)^{\hat{n}}\nonumber \\ & \leq 4e^{2\sigma^{\prime} },  
\end{align}
where the first inequality is Chebyshev's, the second is Doob's maximal inequality, and $\sigma^{\prime}$ is the second moment of $ \frac{j(v)}{\mathcal{R}}$. The last inequality uses the cap $\hat{n}\leq |k|^{1+\zeta} $.  The same argument can be applied to bound $\widetilde{\mathbb{P}}\big[\sup_{0\leq n\leq \hat{n}}k-X_{n}>
|k|^{\frac{1}{2}+\zeta}\big]$.

\vspace{.4cm}

\noindent (ii).\hspace{.15cm}  I first show that, when beginning from $X_{0}=p$ for $|p|\in [\frac{3}{4}|k|,\frac{5}{4}|k|]$, the probability of the event $\hat{m}\notin \mathbf{F}$ is nearly the same with or without conditioning on the information  $\hat{N}\geq \hat{m}$.  I will bound the probability  $\widetilde{\mathbb{P}}\big[\hat{m} \notin \mathbf{F}\,\big|\,X_{0}=p  \big]$ in (ii$'$) below.  First, I need to show  that $ \widetilde{\mathbb{P}}\big[\hat{N}\geq \hat{m}\,\big|\, X_{0 }=p \big]$ is nearly $1$.  This requires another inductive argument concluding that   
 \begin{align*}
 \widetilde{\mathbb{P}}\big[\hat{N}\geq \hat{m}\,\big|\, X_{0} =p \big]  &=\widetilde{\mathbb{E}}\big[\chi\big(\hat{N}\geq \hat{m}-1\big)\widetilde{\mathbb{P}}\big[ \hat{m} \notin \mathbf{F}\,\big|\,  X_{\hat{m}-1}\big]  \,\big|\, X_{0}=p \big]  \\ &\geq \widetilde{\mathbb{E}}\big[\chi\big(\hat{N}\geq \hat{m}-1 \big) \,\big|\, X_{0}=p \big]\Big(1-\frac{2\alpha \|\langle j\rangle  \|_{\infty}}{\mathcal{R}|k|}\Big)+\mathit{O}(|k|^{-2})\\  & \geq 1-\frac{2\alpha \| \langle j \rangle \|_{\infty}}{\mathcal{R}}|k|^{\zeta-1} +\mathit{O}(|k|^{\zeta-2}).
 \end{align*}
The first inequality follows from the analysis below.   I can assume  that  $X_{\hat{m}}=X_{\hat{m}-1}+v_{\hat{m}}$ has absolute value $|X_{\hat{m}-1}+v_{\hat{m}}|\geq \frac{1}{2}|k|$ by~(\ref{NotMuch}), so 
\begin{align*}
1-\widetilde{\mathbb{P}}\big[ \hat{m} \notin \mathbf{F}\,\big|\,  X_{\hat{m}-1}\big] &\approx  \int_{|X_{\hat{m}-1}+v|\geq \frac{1}{2}|k| } dv\,\frac{j(v)}{\mathcal{R}}\Pi_{-}\big( 2^{-1}\mathbf{n}(X_{\hat{m}-1}+v)\Theta(X_{\hat{m}-1}+v    )\big)\\ &\leq \frac{1}{\mathcal{R}}\int_{[-\frac{1}{4},\frac{1}{4}]  } d\theta \,\langle j \rangle \big(\theta-\Theta(X_{\hat{m}-1})\big) \Pi_{-}\big(2^{-1}|k|\theta \big)\\ & = \frac{\|\langle  j \rangle \|_{\infty}}{\mathcal{R}  }\int_{[-\frac{1}{4},\frac{1}{4}]  } d\theta  \,\Pi_{-}\big(2^{-1}|k|\theta \big)\\  &\leq \frac{2\alpha \| \langle j \rangle \|_{\infty}}{\mathcal{R}|k|}. 
\end{align*}
 
 Now, I use that  $ \widetilde{\mathbb{P}}\big[\hat{N}\geq \hat{m}\,\big|\, X_{0} =p \big]$ is close to $1$.  Let $h, h^{\prime}$ be the distributions for $X_{\hat{m}-1}$ when beginning from $X_{0}=p$ and conditioned or not conditioned on the event $\hat{N}\geq \hat{m}$,  respectively.  By similar reasoning as above  
\begin{align}
\big|\widetilde{\mathbb{P}}\big[\hat{m} \notin \mathbf{F}\, \big|\,& \hat{N}\geq \hat{m},\,X_{0}=p  \big]-\widetilde{\mathbb{P}}\big[\hat{m} \notin \mathbf{F}\,\big|\,X_{0}=p  \big]   \big|\nonumber  \\ & = 
 \Big| \int_{\R}dw\,\big( h(w)-h^{\prime}(w)\big) \int_{\R} dv\, \frac{j(v)}{\mathcal{R}} \Pi_{-}\big(2^{-1}\mathbf{n}(w+v)\Theta(w+v)\big) \Big| \nonumber \\ & \leq \|h-h^{\prime}\|_{1} \Big(\frac{2\alpha \|\langle  j \rangle\|_{\infty}}{\mathcal{R}|k|}+\mathit{O}\big(|k|^{-2}\big) \Big)\nonumber \\ & \leq \frac{8\alpha^{2} \| \langle j \rangle \|_{\infty}^{2} }{\mathcal{R}^{2}}|k|^{\zeta-2} +\mathit{O}\Big(|k|^{\zeta-3}\Big),
\end{align}
where the last inequality follows from
 $$ \|h-h^{\prime}\|_{1} \leq 2\widetilde{\mathbb{P}}\big[\hat{N}\geq \hat{m}\,\big|\, X_{0}=p \big]\leq \frac{4\alpha \|\langle j \rangle \|_{\infty}}{\mathcal{R}}|k|^{\zeta-1} +\mathit{O}\big(|k|^{\zeta-2}\big).   $$
 \vspace{.2cm}
\noindent (ii$'$).\hspace{.15cm} Approximating the unconditioned expression $\widetilde{\mathbb{P}}\big[\hat{m} \notin \mathbf{F}\,\big|\,X_{0}=p  \big]$ is more straightforward.  I will denote the density $\widetilde{\mathbb{P}}\big[X_{\hat{m}}=w \,\big|\,\,X_{0}=p  \big]= T^{\hat{m}}\big(\delta(\cdot -p)\big)$ by $q_{\hat{m},p}$.  I have the following approximations
\begin{align*}
\widetilde{\mathbb{P}}\big[\hat{m} \notin \mathbf{F}\,\big|\,\,X_{0}=p\big]  \approx & \int_{\R}dw\,q_{\hat{m},p}(w)\, \Pi_{+}\big(2^{-1}\mathbf{n}(w)\,\Theta(w)\big) \\  \approx  & \int_{|w-p|\leq \hat{m}  }dw\,q_{\hat{m},p}(w)\,  \Pi_{+}\big(2^{-1}\mathbf{n}(w)\,\Theta(w)\big) \\  = &\int_{|w-p|\leq \hat{m}}\,dp\, q_{\hat{m},p}(w) \Pi_{+}\big(2^{-1}\mathbf{n}(p)\,\Theta(w)\big)\\ &+ \mathit{O}\Big(|p|^{\zeta-1}\int_{\R }\,dp\, q_{\hat{m},p}(w) \Pi_{-}\big(2^{-1} \mathbf{n}(p)\,\Theta(w)\big)\Big)
.  
\end{align*}
  By the same argument as~(\ref{NotMuch}), the probability that  $|K_{r}|$ leaves the interval $[p- \hat{m},\,p+\hat{m}]$ after $\hat{m}$ steps after beginning from $p$ will be superpolynomially small.  Hence, the approximation in the second line above with the restriction of the integration to $[p- \hat{m},\,p+\hat{m}]$ will have a superpolynomially small error.  The error in the last equality comes from the inequality            
\begin{align}\label{Aphrodite}
\big|\Pi_{+}(x)-   \Pi_{+}(y)\big|\leq  \frac{40\pi^{2}}{\alpha^{2}} \Pi_{-}(x)\Big|\frac{|x|-|y|}{x}\Big|,\quad \text{when} \quad  \big||x|-|y|\big|\leq \frac{1}{2}|x|
\end{align}
for $x=2^{-1}\mathbf{n}(p)\,\Theta(w)$ and $y=2^{-1}\mathbf{n}(w)\,\Theta(w)$.   Finally, $\Pi_{+}\big(2^{-1}\mathbf{n}(p)\,\Theta(w)\big)$ only depends on $w$ modulo $\frac{1}{2}$, so I have that   
\begin{align*}
\int_{|w-p|\leq \hat{m}  }\,dp\, q_{m,p}(w) \Pi_{+}\big(2^{-1}\mathbf{n}(p)\,\Theta(w)\big)=& \int_{[-\frac{1}{4},\frac{1}{4})}\,d\theta\, \langle q_{\hat{m},p}'\rangle (\theta) \Pi_{+}\big(2^{-1}\mathbf{n}(p)\,\theta\big)\\ \leq & \int_{[-\frac{1}{4},\frac{1}{4})}\,d\theta\, \langle q_{\hat{m},p}\rangle (\theta) \Pi_{+}\big(2^{-1}\mathbf{n}(p)\,\theta\big) \\ 
\approx & 2\int_{[-\frac{1}{4},\frac{1}{4})}\,d\theta\, \Pi_{+}\big(2^{-1}\mathbf{n}(p)\,\theta\big)\\ = & 1-\frac{\alpha}{|p|}  +\mathit{O}(|k|^{-2}),
\end{align*}
where $q_{m,p}'(w):=q_{m,p}(w)\chi(|w-p|\leq \hat{m} ) $.   The second approximation above follows by~(\ref{Cabbage}), since  $\langle q_{\hat{m},p}\rangle=\mathbf{T}^{\hat{m}}\big(\langle \delta(\cdot-p)\rangle \big)$ must be superpolynomially close for $|k|\gg 1$ to the uniform distribution $2$.  
The order equality  on the bottom line uses that $|p|\geq \frac{1}{2}|k|$.  \vspace{.4cm}

\noindent (iii). \hspace{.15cm} Plugging in the result (i) to~(\ref{Restrict}), and expanding $|p|^{-1}$ around $k$ results in 
\begin{multline}
\int_{|p-k|\leq |k|^{\frac{1}{2}+\zeta} }dp\,r_{\hat{n},\hat{m}}(p)\,\Big(1-\frac{\alpha }{|p|} \Big) \\ =\Big( 1-\frac{\alpha}{|k|}\Big)+ \frac{\alpha}{|k|^{2}} \int_{|y|\leq |k|^{\frac{1}{2}+\zeta} }dy\,y\,r_{\hat{n},\hat{m}}(y+k)+\mathit{O}\big( |k|^{2\zeta-2} \big).
\end{multline}
It remains to show that $ \int_{|y|\leq |k|^{\frac{1}{2}+\zeta} }dy\,y\,r_{\hat{n},\hat{m}}(y+k) $ is $\mathit{O}(|k|^{2\zeta})$.
This will require an analysis of the density $r_{\hat{n},\hat{m}}$ and related densities $r^{\prime}_{\hat{n},\hat{m}},r^{\prime \prime}_{\hat{n},\hat{m}}$,   which can be written
\begin{align}\label{RFirst}
r_{\hat{n},\hat{m} }(p)&= \frac{ \widetilde{\mathbb{E}}\big[\Pi_{+}\big(2^{-1}\mathbf{n}(X_{1})\Theta(X_{1})\big)\cdots \Pi_{+}\big(2^{-1}\mathbf{n}(X_{\hat{n}-1})\Theta(X_{\hat{n}-1})\big)\delta\big(X_{\hat{n}-\hat{m}}-p \big)    \big]}{ \mathbb{E}\big[\Pi_{+}\big(2^{-1}\mathbf{n}(X_{1})\Theta(X_{1})\big)\cdots \Pi_{+}\big(2^{-1}\mathbf{n}(X_{\hat{n}-1})\Theta(X_{\hat{n}-1})\big)    \big]  }, \\ \label{RSecond} r^{\prime}_{\hat{n},\hat{m} }(p)&:=\frac{\widetilde{\mathbb{E}}\big[\Pi_{+}\big(2^{-1}\mathbf{n}(k)\Theta(X_{1})\big)\cdots \Pi_{+}\big(2^{-1}\mathbf{n}(k)\Theta(X_{\hat{n}-1})\big)\delta\big(X_{\hat{n}-\hat{m}}-p \big)    \big] }{\widetilde{\mathbb{E}}\big[\Pi_{+}\big(2^{-1}\mathbf{n}(k)\Theta(X_{1})\big)\cdots \Pi_{+}\big(2^{-1}\mathbf{n}(k)\Theta(X_{\hat{n}-1})\big)  \big] },  \\ \label{RThird}   
r_{\hat{n},\hat{m}}^{\prime \prime}(p)&:=\frac{\widetilde{\mathbb{E}}\big[\Pi_{+}\big(2^{-1}\mathbf{n}(k)\Theta(X_{\hat{m}+1})\big)\cdots \Pi_{+}\big(2^{-1}\mathbf{n}(k)\Theta(X_{\hat{n}-1})\big)\delta\big(X_{\hat{n}-\hat{m}}-p \big)    \big]    }{\mathbb{E}\big[\Pi_{+}\big(\mathbf{n}(k)\Theta(X_{\hat{m}+1})\big)\cdots \Pi_{+}\big(2^{-1}\mathbf{n}(k)\Theta(X_{\hat{n}-1})\big)   \big]     }.
\end{align}
I claim  that the following order equalities hold:
\begin{align}\label{Dist}
\sup_{|y|\leq |k|^{\frac{1}{2}+\zeta}}  \Big|\frac{r_{\hat{n},\hat{m} }(y+k)}{r_{\hat{n},\hat{m} }^{\prime}(y+k)}-1\Big|=\mathit{O}\big(|k|^{\zeta-\frac{1}{2}}\big)\hspace{.1cm} \text{ and }  \sup_{|y|\leq |k|^{\frac{1}{2}+\zeta}}  \Big|\frac{r_{\hat{n},\hat{m} }^{ \prime } (y+k)}{r_{\hat{n},\hat{m} }^{\prime \prime}(y+k)}-1\Big|=\mathit{O}\big(|k|^{\zeta-1 }\big).
\end{align}
I will return to a discussion of~(\ref{Dist}) below in (iii$'$).  It  follows from~(\ref{Dist}) that 
$$\Big|\int_{|y|\leq |k|^{\frac{1}{2}+\zeta} }dy\,y\,\big(r_{\hat{n},\hat{m} }(y+k)-r^{\prime \prime }_{\hat{n},\hat{m} }(y+k)\big)\Big|\leq  \mathit{O}\big(|k|^{\zeta-\frac{1}{2}}\big)\int_{|y|\leq |k|^{\frac{1}{2}+\zeta} }dy |y|r^{\prime \prime}_{\hat{n},\hat{m} }(y+k)    =\mathit{O}\big(|k|^{2\zeta}\big).$$
In that case,  it is sufficient to bound the expression $\int_{|y|\leq |k|^{\frac{1}{2}+\zeta} }dy\,y\,r^{\prime \prime }_{\hat{n},\hat{m} }(y+k)$.  The density $r^{\prime \prime }_{\hat{n},\hat{m} }$ can be written as 
\begin{align*}
&r_{\hat{n},\hat{m}}^{\prime \prime}(p) \\ &=\frac{ \int_{\R}dw\,q_{\hat{m} }(w) \widetilde{\mathbb{E}}\big[\Pi_{+}\big(2^{-1}\mathbf{n}(k)\Theta(X_{\hat{m}+1})\big)\cdots \Pi_{+}\big(2^{-1}\mathbf{n}(k)\Theta(X_{\hat{n}-1})\big)\delta\big(X_{\hat{n}-\hat{m}}=p \big) \,\big|\,X_{\hat{m}}=w   \big]    }{ \widetilde{\mathbb{E}}\big[\Pi_{+}\big(2^{-1}\mathbf{n}(k)\Theta(X_{\hat{m}+1})\big)\cdots \Pi_{+}\big(2^{-1}\mathbf{n}(k)\Theta(X_{\hat{n}-1})\big)   \big]    },
\end{align*}
where $q_{\hat{m}}$ is density for the random variable $X_{\hat{m}}$.  The density $r_{\hat{n},\hat{m}}^{\prime \prime}$ corresponds to the distribution for $X_{\hat{n}-\hat{m}}$ conditioned on the event $\mathbf{F} \cap [\hat{m}+1,\hat{n}-1]=\emptyset$.  Moreover, the integral of $q_{\hat{m}}(w)\chi(|w-k|> |k|^{\zeta})$ will be superpolynomially small, since the jumps in the random walk $X_{i}$, $i\in[1,\hat{m}]$ have mean zero and exponential tails.  Hence, I can split $X_{\hat{n}-\hat{m}}-k$ in to a sum of components $X_{\hat{m}}-k$ and $X_{\hat{n}-\hat{m}}-X_{\hat{m}}$, where  $|X_{\hat{m}}-k|$ is typically small and $X_{\hat{n}-\hat{m}}-X_{\hat{m}}$ has density $U_{\hat{m},\hat{n},w}\in L^{1}(\R)$ given by
$$ U_{\hat{m},\hat{n},w}(p):= \frac{ \widetilde{\mathbb{E}}\big[\Pi_{+}\big(2^{-1}\mathbf{n}(k)\Theta(X_{\hat{m}+1})\big)\cdots \Pi_{+}\big(2^{-1}\mathbf{n}(k)\Theta(X_{\hat{n}-1})\big)\delta\big(X_{\hat{n}-\hat{m}}-X_{\hat{m}}=p \big) \,\big|\,X_{\hat{m}}=w \big] }{ \widetilde{\mathbb{E}}\big[\Pi_{+}\big(2^{-1}\mathbf{n}(k)\Theta(X_{\hat{m}+1})\big)\cdots \Pi_{+}\big(2^{-1}\mathbf{n}(k)\Theta(X_{\hat{n}-1})\big)   \big]  } 
$$
 when $X_{\hat{m}}=w$ and conditioned on the event $\mathbf{F} \cap [\hat{m}+1,\hat{n}-1]=\emptyset $.
Since the arguments of the expectations defining $U_{\hat{m},\hat{n},w}$ depend only on the contracted random walk  $\Theta(X_{m})$, the density satisfies the shift invariance $U_{\hat{m},\hat{n},w}(p)=U_{\hat{m},\hat{n},w+\frac{1}{2}m}(p)$ for $m\in \Z$.  This shift symmetry gives the    inequality in the second line below:
\begin{align}\nonumber 
\Big|\int_{|y|\leq |k|^{\frac{1}{2}+\zeta}}dy \,y\, r_{\hat{n},\hat{m}}^{\prime \prime}(y+k)\Big|=  &  \mathit{O}\big(|k|^{\zeta}\big)+\Big|\int_{|w-k|\leq |k|^{\frac{1}{2}+\zeta}}dw q_{\hat{m}}(w)   \int_{|y|\leq |k|^{\frac{1}{2}+\zeta}}dy\,y\, U_{\hat{m},\hat{n},w}(y)\Big|
\\  \leq   & \mathit{O}\big(|k|^{\zeta}\big)+\Big|\int_{\mathbb{T}}d\theta\langle q_{\hat{m}}\rangle(\theta)   \int_{|y|\leq |k|^{\frac{1}{2}+\zeta}}dy\,y\, U_{\hat{m},\hat{n},\theta}(y)\Big|. \label{FireFight} 
\end{align}
The equality above follows by approximating $X_{\hat{n}-\hat{m}}-k$ using $X_{\hat{n}-\hat{m}}-X_{\hat{m}}$ and restricting the integration over $w\in \R$ to  $ |w-k|\leq |k|^{\frac{1}{2}+\zeta}$. 

 The density $\langle q_{\hat{m}}\rangle \in L^{1}(\mathbb{T})$ will be superpolynomially close to the uniform  distribution over $\mathbb{T}=[-\frac{1}{4},\frac{1}{4})$ by~(\ref{Cabbage}). 
However, if $\langle q_{\hat{m}}\rangle$ is replaced by $2$, then the resulting expression is zero, since  the symmetry $U_{\hat{m},\hat{n},\theta}(y)=U_{\hat{m},\hat{n},-\theta}(-y)$ holds as a consequence of the symmetry $j(v)=j(-v)$ for the  jump rates.  Thus, the last term on the right side of~(\ref{FireFight}) is also $\mathit{O}(|k|^{\zeta})$.  
\vspace{.4cm}

\noindent (iii$'$).\hspace{.15cm}  I will focus on bounding the difference   $\big|\frac{r_{\hat{n},\hat{m} }(y+k)}{r_{\hat{n},\hat{m} }^{\prime}(y+k)}-1\big|$ over the domain $|y|\leq |k|^{\frac{1}{2}+\zeta}$, since the order equality on the right side of~(\ref{Dist}) follows by simpler analysis.  Let $G$ be the event $\sup_{0\leq j\leq \hat{n}}\big|X_{j}-k|\leq |k|^{\frac{1}{2}+\zeta}$.  By~(\ref{NotMuch}), $1-\mathbb{P}[G]$ decreases superpolynomially with large $|k|$.  The difference between the expressions in the denominators of~(\ref{RFirst}) and~(\ref{RSecond}) is smaller than
\begin{align*}
 \widetilde{\mathbb{E}}\Big[& \big|\Pi_{+}\big(2^{-1}\mathbf{n}(k)\Theta(X_{1})\big)\cdots\Pi_{+}\big(2^{-1}\mathbf{n}(k)\Theta(X_{\hat{n}-1})\big)  \\  &-\Pi_{+}\big(2^{-1}\mathbf{n}(X_{1})\Theta(X_{1})\big)\cdots\Pi_{+}\big(2^{-1}\mathbf{n}(X_{\hat{n}-1})\Theta(X_{\hat{n}-1})\big)    \big|\,\Big|\,G  \Big]  \\ \leq & \frac{40\pi^{2}}{\alpha^2}\sum_{j=1}^{\hat{n}-1}\widetilde{\mathbb{E}}\Big[\Pi_{+}\big(2^{-1}\mathbf{n}(k)\Theta(X_{1})\big)\cdots \Pi_{+}\big(2^{-1}\mathbf{n}(k)\Theta(X_{j-1})\big)\\ &\times \Big|\frac{X_{j}-k}{k}\Big|\Pi_{-}\big(2^{-1}\mathbf{n}(k)\Theta(X_{j})\big)\Pi_{+}\big(2^{-1}\mathbf{n}(k)\Theta(X_{j+1})\big)\cdots \Pi_{+}\big(2^{-1}\mathbf{n}(X_{\hat{n}-1})\Theta(X_{\hat{n}-1})\big)\,\Big| \,G   \Big].  
\end{align*}
  The above inequality follows from inserting a telescoping sum and using the inequality~(\ref{Aphrodite}).

  Since I have conditioned on the event $G$, then $\big|\frac{X_{j}-k}{k}\big|\leq |k|^{\zeta-\frac{1}{2}}$.  Removing the conditioning on $G$ will make another superpolynomially small error, and a single term from the sum above is bounded by
\begin{align*}
 |k|^{\zeta-\frac{1}{2}}
\widetilde{\mathbb{E}}\Big[&\Pi_{+}\big(2^{-1}\mathbf{n}(k)\Theta(X_{1})\big)\cdots \Pi_{+}\big(2^{-1}\mathbf{n}(k)\Theta(X_{j-1})\big)\\ &\times \Pi_{-}\big(2^{-1}\mathbf{n}(k)\Theta(X_{j})\big)\Pi_{+}\big(2^{-1}\mathbf{n}(k)\Theta(X_{j+1})\big)\cdots \Pi_{+}\big(2^{-1}\mathbf{n}(X_{\hat{n}-1})\Theta(X_{\hat{n}-1})\big)  \Big]\\ \leq & 2 \mu^{4}|k|^{\zeta-\frac{3}{2}}
\widetilde{\mathbb{E}}\Big[\Pi_{+}\big(2^{-1}\mathbf{n}(k)\Theta(X_{1})\big)\cdots  \Pi_{+}\big(2^{-1}\mathbf{n}(X_{\hat{n}-1})\Theta(X_{\hat{n}-1})\big)  \Big].
\end{align*}
The conditional density for $\Theta(X_{j})$ given $\Theta(X_{j-1})$ is bounded by $\mu$ as a result of assumption (2) from List~\ref{Assumptions} (and the same for~$\Theta(X_{j+1})$ given $\Theta(X_{j})$).  I can make use of this to obtain bounds in which $\Theta(X_{j})$ is decoupled from the other variables, and where $\Theta(X_{j})$ and $\Theta(X_{j+1})$ are integrated with respect to the constant density $\mu$ over $\mathbb{T}$.  Integrating out $X_{j}$ yields a factor $2\int_{\mathbb{T}}d\theta \Pi_{+}\big(2^{-1}\mathbf{n}(k)\theta\big)\approx \frac{\alpha}{|k|}< \frac{2\alpha}{|k|}$. Moreover,  I can put the variable back in the expression by inserting $2\Pi_{+}\big(2^{-1}\mathbf{n}(k)\Theta(X_{j}))$, since $2\Pi_{+}\geq 1$.    The conditional density for $\Theta(X_{j})$ given $\Theta(X_{j-1})$ (and $\Theta(X_{j+1})$ given $\Theta(X_{j})$) is bounded from below by $\mu^{-1}$ by (2) of List~\ref{Assumptions}.  I can use this to recouple $X_{j}$ with the other variables, so that the expression bounding the difference is a constant multiple of the denominator in the definition for $r^{\prime}_{\hat{n},\hat{m}}$.   

The same argument as above applies in order to bound the difference between the numerators of~(\ref{RFirst}) and~(\ref{RSecond}).

  \vspace{.4cm}

\noindent Part (2):  \hspace{.15cm} I will sketch the proof.  Let $\zeta<\beta$, $M=\lfloor |k|^{\beta}\rfloor$, and  $\hat{m}=\lfloor |k|^{\zeta}\rfloor$.  As in  Part (1), I can phrase the problem in terms of the random walk $X_{n}$.    Since there is no actual sign-flipping for $X_{n}$ (only coin tossing),  I can remove the absolute values from $X_{M}, X_{0}$ in the expression:     
\begin{align*}
\widetilde{\mathbb{E}}\Big[(X_{M}-X_{0})\,\chi(\hat{N}> M  )   \Big]  \approx \Big| \widetilde{\mathbb{E}}\Big[\Pi_{+}\big(2^{-1}\mathbf{n}(X_{1})\Theta(X_{1})\big)\cdots \Pi_{+}\big(2^{-1}\mathbf{n}(X_{M})\Theta(X_{ M})\big) (X_{M}-k)      \Big]\Big|,
\end{align*}
where the approximation uses that there is a superpolynomially small probability that the event  $\hat{N}\leq M$   occurs because   $|X_{n}-k|\geq\frac{|k|}{2}$ for some $n\leq M$.  By the same reasoning as in  Part (1), I can replace  the factors $\mathbf{n}(X_{r})$ by $\mathbf{n}(k)$ in the arguments of the functions $\Pi_{+}$ with an error of order $\mathit{O}(|k|^{\zeta-1})$, and I can also remove the factors $\Pi_{+}\big(2^{-1}\mathbf{n}(X_{j})\Theta(X_{j})\big)$ from the expectations for $j\leq \hat{m}$.  By the triangle inequality, I can bound the right side above by 
\begin{multline}\label{GrassFast}
  \Big|\widetilde{\mathbb{E}}\Big[\Pi_{+}\big(2^{-1}\mathbf{n}(k)\Theta(X_{\hat{m}+1})\big)\cdots \Pi_{+}\big(2^{-1}\mathbf{n}(k)\Theta(X_{ M})\big) (X_{M}-X_{\hat{m}} )    \Big] \Big| \\ +\sum_{n=0}^{\hat{m}-1}\Big|\widetilde{\mathbb{E}}\Big[\Pi_{+}\big(2^{-1}\mathbf{n}(k)\Theta(X_{\hat{m}+1})\big)\cdots \Pi_{+}\big(2^{-1}\mathbf{n}(k)\Theta(X_{ M})\big) (X_{n+1}-X_{n})    \Big] \Big|.
 \end{multline}
 The first term will be superpolynomially close to zero for large $|k|$, since the distribution $\langle q_{\hat{m}}\rangle$ for $\Theta(X_{\hat{m}})$ will be superpolynomially close to the uniform distribution and by the symmetry argument used in the proof of  Part (1).  The probability of the event $\hat{N}\in [n+2,M]$ when $X_{n+1}=p$ for $n<\hat{N} $  and $|p-k|\leq \frac{1}{4}|k|$    will be $\mathit{O}\big(|k|^{\beta-1}\big)$, 
 so the $n$th term on the lower line of~(\ref{GrassFast}) is equal to
  $$\Big|\widetilde{\mathbb{E}}\Big[\Pi_{+}\big(2^{-1}\mathbf{n}(k)\Theta(X_{\hat{m}+1})\big)\cdots \Pi_{+}\big(2^{-1}\mathbf{n}(k)\Theta(X_{n+1})\big) (X_{n+1}-X_{n})    \Big] \Big| +\mathit{O}\big(|k|^{\beta-1}\big).   $$
Finally, I have the equality   
$$\widetilde{\mathbb{E}}\Big[ \Pi_{+}\big(2^{-1}\mathbf{n}(k)\Theta(X_{n+1})\big) (X_{n+1}-X_{n})\,\Big|\,X_{n}  \Big]=-\int_{\R}dv\, v\, \frac{j(v)}{\mathcal{R}}\,\Pi_{-}\big(2^{-1}\mathbf{n}(k)(X_{n}+v)\big)   .      $$  
I can use the exponential decay of $j(v)$ to cap the jumps $v$ by $|k|^{\zeta}$ with a small error, so the expression above is approximately    
$$\int_{|v|\leq |k|^{\zeta} }dv\, v\, \frac{j(v)}{\mathcal{R}}\,\Pi_{-}\big(2^{-1}\mathbf{n}(k)(X_{n}+v)\big)\leq  |k|^{\zeta }\frac{\|\langle j\rangle \|_{\infty} }{\mathcal{R}}\,\int_{\mathbb{T}}d\theta \Pi_{-}\,\big(2^{-1}\mathbf{n}(k)\theta\big)=\mathit{O}\big(|k|^{\zeta-1}\big),    $$
where I have followed the usual method for bounding expressions with integrals including  $j(v)$ and $\Pi_{-} $.

\end{proof}

\subsubsection{ Proof of Proposition~\ref{TimeFlip} } \label{SecFlipidy}

Part (1) of the lemma below makes another step toward showing that the duration between between successive sign-flip times $\tau_{n},\tau_{n+1}$ is approximately an exponential distribution with mean $\nu^{-1} |K_{\tau_{n}}|$. 

\begin{lemma}\label{Fraiser}
Let $\tau,\zeta, \beta, k$ be as in the statement of Lemma~\ref{Fraiser}.  There exists a $C>0$ such that all $k$
\begin{enumerate}
\item $\sup_{a \in \R_{+}} \Big| \mathbb{P}\Big[\frac{\mathcal{N}_{\tau}}{|k|}>a  \Big]-e^{-\alpha a   }   \Big|\leq C |k|^{-1+\zeta},  $

\item 
$\Big|\mathbb{E}\Big[\big(|K_{t_{M } }|-|K_{0}|\big)\,\chi(\mathcal{N}_{\tau}> M  )   \Big]  \Big| \leq C |k|^{\beta-1+\zeta} $, 

\end{enumerate}
where $M:=\lfloor |k|^{\beta } \rfloor$.

\end{lemma}

\begin{proof}\text{ }\\

\noindent Part (1):\hspace{.15cm} By Part (1) of Lemma~\ref{SubFraiser} and the equivalence between $\widetilde{\mathbb{P}}$ and $\widetilde{\mathbb{P}}^{\prime }$, it is sufficient to show  that
\begin{align} \label{Nordstrums} 
\sup_{a \in \R_{+} }
\Big| \mathbb{P}\Big[\frac{\mathcal{N}_{\tau}}{|k|}> a  \Big]- \widetilde{\mathbb{P}}^{\prime }\Big[\frac{\mathcal{N}_{\tau}}{|k|}>a  \Big] \Big|=\mathit{O}( |k|^{\zeta-1})
\end{align} 
for all $\zeta>0$. 

  First, I will show that the probability $\mathbb{P}[\mathcal{N}_{\tau}> n ]$ decays exponentially for large $n$.    Let $\mathbf{F}\subset \mathbb{N}$ be the set of times $n$ such that $t_{n}$ is a sign-flip  or $|K_{t_{n}}|$ jumps out of $[\frac{1}{2}|k|,\frac{3}{2}|k|]$.  I have the following relations
  \begin{align*}
\mathbb{P}\Big[\mathcal{N}_{\tau}> 3n  \Big] &=\mathbb{E}\Big[  \prod_{j=1}^{3n}\chi\big(j\notin \mathbf{F}     \big) \Big] \leq  \mathbb{E}\Big[  \prod_{j=1}^{n}\chi\big(3j-1\notin \mathbf{F}     \big) \Big]  
\\ & = \mathbb{E}\Big[  \Big(\prod_{j=1}^{n-1}\chi\big(3j-1\notin \mathbf{F}     \big)\Big) \mathbb{P}\big[  3n-1\notin \mathbf{F}     \,\big|\, K_{t_{3n-3}}  \big] \Big]\\ &\leq \Big(1-\frac{B'}{|k|}\Big) \mathbb{E}\Big[  \prod_{j=1}^{n-1}\chi\big(3j-1\notin \mathbf{F}     \big) \Big],
\end{align*}
where the second inequality uses the Markov property and Part (1) of Lemma~\ref{ProbFlip} to bound $ \mathbb{P}\big[  3n-1\notin \mathbf{F}     \,\big|\, K_{t_{3n-3}}  \big]$.  My purpose of removing the terms $j=0,1\,\textup{mod}\, 3$ in the product above was to accommodate the technical definition for a sign-flip, which requires information from the next Poisson time. 
 This argument can be applied inductively to obtain a bound for $\mathbb{P}[\mathcal{N}_{\tau}> n ]$ that decays exponentially as $(1-\frac{B}{|k|})^{n}$ for $B=\frac{B'}{3}$.  I remark that  the expectations $\mathbb{E}[e^{\gamma \frac{ \mathcal{N}_{\tau} }{|k|} }] $ must be uniformly bounded for  $0\leq \gamma<B$ and $k$.  The same arguments hold for the law $\widetilde{\mathbb{P}}^{\prime }$.  Since the probabilities $\mathbb{P}\big[ \frac{\mathcal{N}_{\tau}}{|k|}>a \big]$ and $\widetilde{\mathbb{P}}^{\prime }\big[\frac{\mathcal{N}_{\tau}}{|k|}>a \big]$ decay exponentially in $a\in \R_{+}$ with a rate that is uniform  in $k$,  I can take a cut-off $|a|\leq |k|^{\frac{\zeta}{2}} $ as in the proof of Part (1) of Lemma~\ref{SubFraiser}.


Let the set $A_{M}$ and the variational norm $\|\cdot\|_{\textup{Var},M }$ be defined as in Lem~\ref{LemTotVar}.  Notice that the event $\frac{\mathcal{N}_{\tau}}{|k|}> a$ is contained in the set $A_{M}$ for $M=\lfloor |k|^{1+\frac{\zeta}{2}}\rfloor$.    The variational norm for the difference $\chi(A_{M})\mathbb{P}-\chi(A_{M})\widetilde{\mathbb{P}}^{\prime}$ bounds the difference in probabilities for any event contained in $A_{M}$, and thus  
the first inequality below holds:
\begin{align*}
\sup_{a\leq |k|^{\frac{\zeta}{2} } }  \Big| \mathbb{P}\Big[\frac{\mathcal{N}_{\tau}}{|k|  }> a \Big]-\widetilde{\mathbb{P}}^{\prime}\Big[\frac{\mathcal{N}_{\tau}}{|k|}>a  \Big]    \Big| & \leq \big\|\chi(A_{M})\mathbb{P}-\chi(A_{M})\widetilde{\mathbb{P}}^{\prime}\big\|_{\textup{Var},\,M } \\ & \leq cM\frac{\log(|k|)}{|k|^{2}}=\mathit{O}(|k|^{\zeta-1}).
\end{align*}
The second inequality is by Lem~\ref{LemTotVar}.

\vspace{.5cm}

\noindent Part (2):\hspace{.15cm} If $|K_{t_{M} }|-|K_{0}|$ were a bounded random variable, then this result would follow immediately from the proof of Lem.~\ref{LemTotVar}.  By Part (2) of Lemma~\ref{SubFraiser}, it is sufficient to prove that
\begin{align}
 \Big|\mathbb{E}\Big[\Big(|K_{t_{M } }|-|K_{0}|\Big)\chi\big(\mathcal{N}_{\tau}> M  \big)   \Big]-\widetilde{\mathbb{E}}^{ \prime}\Big[\Big(|K_{t_{M } }|-|K_{0}|\Big)\chi\big(\mathcal{N}_{\tau}> M  \big)   \Big]  \Big| = \mathit{O}\Big(\frac{\log(|k|)}{|k|^{1-\beta } }\Big).
\end{align}

  As in the proof of Lem~\ref{LemTotVar}, I will break down the difference into a sum of differences between the expectations $\mathbb{E}^{(n-1)},\,\mathbb{E}^{(n)}$ of $|K_{t_{M}}|-|K_{0}|$ conditioned on $\mathbf{N}_{\tau}>M$:
\begin{align}\label{Attention}
  \sum_{n=1}^{M}  \Big|\widetilde{\mathbb{E}}^{(n-1)}\Big[\Big(|K_{t_{M } }|-|K_{0}|\Big)\chi\big(\mathcal{N}_{\tau}> M  \big)   \Big]-\widetilde{\mathbb{E}}^{(n)}\Big[\Big(|K_{t_{M } }|-|K_{0}|\Big)\chi\big(\mathcal{N}_{\tau}> M  \big)   \Big]  \Big|.  
\end{align}
   The difference  $K_{t_{n}}-K_{t_{n-1}}$ is equal to $ v_{n}+\ell_{n}$, where $v_{n},\, \ell_{n}$ are the $n$th L\'evy and lattice jumps respectively.  For $m\in \Z$, let $\Gamma(m;\,p,v)$ be the closest element to $m$ in $I(p,v)$ (where ties can be assigned arbitrarily).  It is convenient to split $\ell_{n}$ into two parts $\Gamma(\ell_{n};\, K_{t_{n-1}},v_{n})$ and $\ell_{n}-\Gamma(\ell_{n}; K_{t_{n-1}},v_{n})$ that are treated separately.  Let $K_{t_{M},n}^{\prime}$ be the momentum at time $t_{M}$ if the $n$th lattice jump is replaced by $\Gamma(\ell_{n}; K_{t_{n-1}},v_{n})$:
    $$K_{t_{M},n}^{\prime}:=K_{t_{M}}+K_{t_{n-1}}-K_{t_{n} }+ \Gamma(\ell_{n};\,K_{t_{n}},v_{n})+v_{n}.$$
I have the following bound for the $n$th term of the sum~(\ref{Attention}) when $K_{t_{M}}$ is replaced by $K_{t_{M},n}^{\prime}$ in the expectation $\widetilde{\mathbb{E}}^{(n-1)}$:    
\begin{align*}
\Big|\widetilde{\mathbb{E}}^{(n-1)}\Big[&\Big(|K^{\prime}_{t_{M},n}|-|K_{0}|\Big)\chi\big(\mathbf{N}_{\tau}>M\big)   \Big]  -     \widetilde{\mathbb{E}}^{(n)}\Big[\Big(|K_{t_{M}}|-|K_{0}|\Big) \chi\big(\mathbf{N}_{\tau}>M\big)     \Big] \Big|\nonumber  \\  \leq &\int_{\R}dp\,\hat{q}_{n-1}'(p)\int_{\R}dv\frac{j(v)}{\mathcal{R}} \nonumber \\ & \times \sum_{i\in I(k,v)  } |v+i|\,\left|\sum_{ \substack{ m\in \Z \text{ with}\\ i=\Gamma(m;\, k_{n-1},v)} }\big|\kappa_{v}\big(p,\,m\big)\big|^{2} -\widetilde{\mathbb{P}}^{ \prime}\big[d\ell=i   \,\big|\,dL_{n}=v, K_{t_{n-1}}=p\big]     \right| \nonumber ,     
\end{align*}
where the density $\hat{q}_{n}'$ is defined as in the proof of Lem~\ref{LemTotVar}. Recall that the $\hat{q}_{n}'(p)$ is supported in the set $[\frac{1}{2}|p|,\frac{3}{2}|p|]$.  By Part (1) of Lem.~\ref{BadTerms}, the sum over $|\kappa_{v}\big(p,\,m\big)|^{2}$ for $m\notin I(p,v)$ is $\mathit{O}(|k|^{-2})$.  Thus, it is sufficient bound 
\begin{align}\label{Nukki}
\int_{\R}dp\, \hat{q}_{n-1}'(p)& \int_{|v|\leq |k|^{\frac{1}{2}}  }dv \frac{j(v)}{\mathcal{R}}\,\sum_{i\in I(k,v)  }  |v+i|\nonumber \\   & \times \Big| \big|\kappa_{v}\big(p,\,i\big)\big|^{2} -\widetilde{\mathbb{P}}^{(0)}\big[d\ell=i   \,\big|\,dL_{n}=v, K_{t_{n-1}}=p\big]     \Big|,
\end{align}
where the restriction of the integration in $v$ to the set $|v|\leq |k|^{\frac{1}{2}} $ will have a superpolynomially small error by the decay of $j(v)$.  I can apply the same estimates as in the proof of Lem~\ref{LemTotVar} to show that~(\ref{Nukki}) is $\mathit{O}\Big(\frac{\log(|k|)}{|k|}\Big)$.



To complete the proof, I must bound the error of having substituted $|K_{t_{M}}|$ by $|K^{\prime}_{t_{M},n}|$ in the expectation $\widetilde{\mathbb{E}}^{(n-1)}$ above:  
\begin{multline*}
\widetilde{\mathbb{E}}^{(n-1)}\Big[\Big| |K_{t_{M}}|- |K^{\prime}_{t_{M}}| \Big| \chi\big(\mathcal{N}_{\tau}>M\big)   \Big]\leq\int_{\R}dp\,\hat{q}_{n-1}'(p)\int_{\R}dv\,\frac{j(v)}{\mathcal{R}}\sum_{m\in \Z}\big|m-\Gamma(m;\,p,v)\big|\, \big|\kappa_{v}\big(p,\,m\big)\big|^{2}. 
 \end{multline*}
Again employing a cutoff $|v|\leq |k|^{\frac{1}{2}}$ and using that the support of $\hat{q}_{n}'(p)$ is over $|p|\in [\frac{1}{2}|k|,\,\frac{3}{2}|k|]$,  I can apply Part (2) of Lem.~\ref{BadTerms}, since $\big|m-\Gamma(m;\,p,v)\big|=\textup{dist}\big(m,\,I(k,v)    \big)$.  This yields a $c>0$ such that for all $k$,  
$$\widetilde{\mathbb{E}}^{(n)}\Big[\Big| |K_{t_{M}}|- |K^{\prime}_{t_{M}}| \Big| \chi\big(\mathcal{N}_{\tau}>M\big)   \Big]\leq \frac{4c}{|k|}.  $$

Therefore, a single term from the  sum~(\ref{Attention}) is $\mathit{O}\big(\frac{\log(|k|)}{|k|})$.  Since there are $M=\lfloor |k|^{\beta} \rfloor $ terms,  the result follows.

 \end{proof}

\begin{proof}[Proof of Prop.~\ref{TimeFlip}] \text{ } \\

\noindent Part (1):\hspace{.15cm} Similar to Part (2).

\vspace{.5cm}

\noindent Part (2):\hspace{.15cm} The Poisson times $t_{n}$ can be thought of as the sum of independent mean-$\mathcal{R}^{-1}$ exponentials $e_{m}$, $m\in \mathbb{N}$: $t_{n}=\sum_{m=1}^{n}e_{m}$.
It is  sufficient to prove the statement of (2) with $\tau$ replaced by $\mathcal{N}_{\tau}$ and $\nu$ replaced by $\frac{\nu}{\mathcal{R}} =\alpha $.  This can be seen with the short calculation:
$$\mathbb{E}\left[ e^{\gamma\frac{\tau}{|k|} }\right]= \mathbb{E}\Big[ \Big(\frac{\mathcal{R} }{\mathcal{R}-\frac{\gamma}{|k|}}\Big)^{\mathcal{N}_{\tau}} \Big]= \mathbb{E}\Big[ e^{-\mathcal{N}_{\tau}\log\big(1-\frac{\gamma  }{\mathcal{R} |k| }    \big) }  \Big]= \mathbb{E}\Big[ e^{\gamma \mathcal{R}^{-1} \mathcal{N}_{\tau}   }  \Big]+\mathit{O}\big(|k|^{-1}\big).   $$

  By the remark in  the proof of Part (1) of Lem.~\ref{Fraiser},
the expectation of $e^{\gamma\frac{\mathcal{N}_{\tau}}{|k|}}$ is uniformly bounded  for all $k\in \R $ and  $\gamma$ in $0\leq \gamma_{0} \leq \gamma_{0} $ for small enough $\gamma_{0}>0$.  Using the following expectation formula for a random variable $X$:
\begin{align*}
 \mathbb{E}\big[e^{\gamma X }\big]  =\int_{0}^{\infty}da\,\gamma\,e^{a\gamma }  \mathbb{P}\big[ X\geq a\big] , 
\end{align*}  
  I can bound the difference in the case $\gamma< \gamma_{0}$:
  \begin{align*} 
& \Big| \mathbb{E}\left[e^{\gamma \frac{\mathcal{N}_{\tau}}{|k|}}\right]   
 - \frac{\alpha }{\alpha-\gamma} \Big|  \\ & \leq  \int_{0}^{ \gamma_{0}^{-1}  \log(|k|) }da\, \gamma\,e^{a\gamma }\Big|\mathbb{P}\Big[\frac{\mathcal{N}_{\tau}}{|k|}\geq a\Big]-e^{-a\alpha }\Big|   
  + \int_{B^{-1} \log(|k|) }^{\infty}da\,\gamma\,e^{a\gamma }  \Big(\mathbb{P}\Big[\frac{\mathcal{N}_{\tau}}{|k|}\geq a\Big] +e^{-a\alpha }\Big) \ \\ & \leq   2C\gamma \,|k|^{2\zeta-1}+ 2\int_{\gamma_{0}^{-1} \log(|k|) }^{\infty}da\,\gamma\,e^{a(\gamma-\gamma_{0})  }  = \mathit{O}(|k|^{2\zeta-1}).
\end{align*}
The second inequality uses Part (1) of Lem.~\ref{Fraiser} for the first term.  For the second term of the second line, the probability $\mathbb{P}\big[\frac{\mathcal{N}_{\tau}}{|k|}\geq a\big]$ is smaller than  $e^{-a\gamma_{0}}$ by the proof of Part (1) of Lem.~\ref{Fraiser}.

\vspace{.5cm}

\noindent Part (3):\hspace{.15cm} The quantity $\int_{0}^{\tau}dr\chi\big(S(K_{r})\neq S(K_{0})\big)$ sums the amount of time  that $K_{r}$ spends up to time $\tau$ with the opposite sign of what it began with.  By the definition of $\tau$, $K_{r}$ cannot spend two consecutive Poisson times $t_{n-1},t_{n}\leq \tau$ with $S(K_{t_{n-1}})= S(K_{t_{n}})\neq S(K_{0})$, and thus $S(K_{t_{n}})\neq S(K_{0})$ implies $S(K_{t_{n-1}})\neq S(K_{t_{n}})$.    

I can rewrite $\int_{0}^{\tau}dr\chi\big(S(K_{r})\neq S(K_{0})\big)$  in terms of the exponential waiting times $e_{n}$  as follows:
\begin{align*}
\mathbb{E}\Big[\int_{0}^{\tau}dr\chi\big(S(K_{r})\neq S(K_{0})\big)\Big] = &\mathbb{E}\Big[\sum_{n=2}^{\mathcal{N}_{\tau}-1}\chi\big(S(K_{t_{n}})\neq S(K_{0})\big)e_{n}\Big] \\ = & \mathcal{R}^{-1}\mathbb{E}\Big[\sum_{n=2}^{\mathcal{N}_{\tau}-1}\chi\big(S(K_{t_{n}})\neq S(K_{0})\big)\Big]\\  \leq  & \mathcal{R}^{-1}\mathbb{E}\Big[\sum_{n=2}^{\mathcal{N}_{\tau}-1}\mathbb{P}\big[S(K_{t_{n}})\neq S(K_{t_{n-1}}) \big|\,K_{t_{n-2}}\big] \Big]\\ \leq & \frac{c}{\mathcal{R}|k| }\mathbb{E}\big[\mathcal{N}_{\tau}\big]= c \mathbb{E}\Big[\frac{\tau}{|k|}\Big]<  \frac{2c }{\nu },    
\end{align*}
where I do not count $n=1$, since  the sign of the momentum has been conditioned not to change at the first Poisson time.  The probabilities in the summand on the third line are smaller than a constant multiple $c>0$ of $|k|^{-1}$   
by  the same argument as in the Part (2) of Lem.~\ref{ProbFlip}. To apply the argument from Lem.~\ref{ProbFlip},  I  must use that $|K_{r}|\geq \frac{1}{2}|k|$ for $r\leq  \tau $.     The last inequality is for  $|k|$ large enough, since  $\mathbb{E}\big[\frac{\tau}{|k|}\big] $ approaches $\nu^{-1}$ at $k\rightarrow \infty$ by Part (1).

\vspace{.5cm}

\noindent Part (4):\hspace{.15cm} In the notation from Part (1), $e_{\mathcal{N}_{\tau}+1}$ is the waiting time from $\tau$ to the next Poisson time. The random variable $\tau$ is not a hitting time, since determining $\tau$ requires information from the time $\tau^{\prime}:=\tau+e_{\mathcal{N}_{\tau}+1}$,  but $\tau^{\prime}$ is a hitting time.  The results for Parts (1) and (2) will also hold for $\tau^{\prime}$, since it includes  only an extra contribution from an additional independent exponential random variable with mean $\mathcal{R}^{-1}$.  
    
I will apply a submartingale argument using the fact that $\mathcal{E}_{s}\approx |K_{s}|+\mathit{O}(|K_{s}|^{-1})$ when $|K_{s}|\gg 1$. As before, let  $M_{r},A_{r}$ be the Doob-Meyer components  of the submartingale $\mathcal{E}_{r}-\mathcal{E}_{0}$.  Recall that  $M_{r}$ is a sum of two uncorrelated submartingales $m_{r}$ and $M_{r}- m_{r}$ by Part (3) of Prop.~\ref{Basics}.  By an argument similar to Lem.~\ref{LittleEm}, I can  show $\langle m,m\rangle_{\tau}$ is close to $\sigma \tau$ and thus $\langle M-m,M-m\rangle_{\tau}$ is small, since $\frac{d}{dr}\langle m,m\rangle_{r}+\frac{d}{dr}\langle M-m,M-m\rangle_{r}\leq \sigma$.  To see that $\sigma \tau-\langle m,m\rangle_{\tau}$ is small, notice that
\begin{multline*}
\mathbb{E}\big[\sigma\tau-\langle m,m\rangle_{\tau} \big]=\mathbb{E}\Big[\sum_{n=0}^{\infty}\chi\big( \mathcal{N}_{\tau}> n\big)\big(\sigma-\mathcal{V}(K_{t_{n}})   \big) e_{n}   \Big] \\ \leq \mathbb{E}\Big[\sum_{n=0}^{\infty}\chi\big(\mathcal{N}_{\tau}> n-1\big)e_{n}\mathbb{E}\Big[\big(\sigma-\mathcal{V}(K_{t_{n}})\big)\chi \big(|K_{t_{n}}|\in \big[\frac{1}{2}|k|,\,\frac{3}{2}|k|\big]     \big)\,\Big|\,\mathcal{F}_{t_{n-1}} \Big]    \Big],
\end{multline*}
where the inequality follows because the event $\mathcal{N}_{\tau}\geq n$ is contained by the intersection of $\mathcal{N}_{\tau}\geq n-1$ and $|K_{t_{n}}|\in\big[\frac{1}{2}|k|,\,\frac{3}{2}|k|\big]$.   The exponential random variables $e_{n}$ are independent of everything else, so they can be left outside the conditional expectation.  By the same argument as in the proof of Lem.~\ref{LittleEm},
$$\mathbb{E}\Big[\big(\sigma-\mathcal{V}(K_{t_{n}})\big)\chi \Big(|K_{t_{n}}|\in \big[\frac{1}{2}|k|,\,\frac{3}{2}|k|\big]\Big) \,\Big|\,\mathcal{F}_{t_{n-1}} \Big]\leq  \frac{4\mathbf{a}}{\mathcal{R}} \| \langle j \rangle  \|_{\infty} \frac{\log(|k|) }{|k|  }.    $$   
Using the above and the result of Part (1), 
$$\mathbb{E}\big[\langle M-m,M-m\rangle_{\tau}\big]\leq \frac{4\mathbf{a}}{\mathcal{R}} \|\langle  j  \rangle \|_{\infty} \frac{\log(|k|) }{|k|  }\mathbb{E}\big[\tau \big]\leq \frac{ 4\nu \mathbf{a}}{\mathcal{R} } \| \langle j \rangle  \|_{\infty} \log(|k|).  $$

Now applying standard martingale arguments, I can bound $M_{r}-m_{r}$ and $A_{r}$.  
\begin{align*}
\mathbb{E}\Big[\sup_{0\leq r\leq \tau }\big| M_{r}- m_{r}+A_{r}   \big|^{2}\Big]  &  \leq 2\mathbb{E}\Big[\sup_{0\leq r\leq \tau^{\prime} }\big| M_{r}-m_{r}  \big|^{2}\Big] + 2\mathbb{E}\big[A_{\tau }^{2}\big] \\ &\leq 8\mathbb{E}\big[(M_{\tau^{\prime} }-m_{\tau^{\prime} })^{2}\big]+\frac{2\sigma^{2}}{|k|^{2}}\mathbb{E}\big[ \tau^{2}  \big]  \\  &= 8\sigma\mathbb{E}\big[\langle M-m,M-m\rangle_{\tau^{\prime}}  \big]+\frac{2\sigma^{2}}{|k|^{2}}\mathbb{E}\big[ \tau^{2}  \big]\\ & \leq \frac{4\nu \mathbf{a}}{\mathcal{R}} \| \langle j  \rangle \|_{\infty} \log(|k|) + 8\nu^{2} \sigma^{2}.             
\end{align*}
The second inequality is Doob's for the first term, and for the second term I use that $A_{\tau}\leq\frac{\sigma \tau}{|k|}$, which is a consequence of the identity $\sigma=\frac{d}{dr}\langle M,M\rangle_{r}+2\mathcal{E}_{r}\frac{d}{dr}A_{r}$ along with the fact that $\mathcal{E}_{r}\approx|K_{r}|\geq \frac{1}{2}|k|$ for $r\in [0,\tau)$.  By Chebyshev's inequality, it follows that the probability  $\sup_{0\leq r\leq \tau }\big| M_{r}- m_{r}+A_{r}   \big|^{2}\geq \frac{|k|^{2} }{4}$ will be $\mathit{O}\big(\frac{\log(|k|)}{|k|^{2}} \big)$. 

The final part of the story is $m_{\tau}$.  The $m_{r}$ martingale is well-behaved, since the absolute values of its jumps are less than those of the L\'evy process by Part (3) of Prop.~\ref{Basics} and the L\'evy jumps have exponential tails.  On the other hand, the moments of $\tau$ are finite by Part (1).  The probability  $|m_{\tau}|>t^{\frac{1}{2}+\delta}$, $\delta>0$ will decay superpolynomially,  since $\tau$ has order $\mathit{O}(|k|)$ for $|k|\gg 1$.

\vspace{.4cm}

\noindent Part (5):\hspace{.15cm} This follows by an extension of the analysis in Part (2) of Lemma~\ref{Fraiser}, which I will not include.

 \end{proof}

\subsection{Proof of Lemmas~\ref{LemMartApprox}-\ref{LemLindberg} }\label{SubSecMain}

 For $r,t\in \R_{+}$, $s\in [0,1]$, and $n\in \mathbb{N}$, let the processes $M_{r}$, $A_{r}$, $Y_{s}^{(t)}$, $\mathbf{m}_{s}^{(t)}$, and $\mathbf{N}_{r}$; the filtration $\widetilde{\mathcal{F}}_{r}$; and the times $\tau_{n}$  be defined as in Sect.~\ref{SecProofOutline}.  In the proofs of this section, I will treat all the times $\tau_{m}$ as if they occur through sign-flips rather than elaborating on the exceptional occurrences in which $\tau_{m}=\varsigma_{j}$ for some $j\in \mathbb{N}$ or $|K_{\tau_{m}}|\notin \big[\frac{1}{2}|K_{\tau_{m-1}}|,\,\frac{3}{2}|K_{\tau_{m-1}}|\big]$.  Although it is slightly incorrect to neglect those cases, the omission avoids some messy and unenlightening case considerations, and the estimates in the proofs below imply that those cases have negligible contribution.  Whereas the number of sign-flips will be on the order $\mathit{O}(t^{\frac{1}{2}})$, the expected number of $\varsigma_{j}$ over the time interval $[0,t]$ has the bound $ \sigma^{\frac{1}{2}}t^{\frac{1}{8}+\iota } $ for $0<\iota\ll 1$ by~(\ref{Numero}), and the event that $|K_{\tau_{m}}|\notin \big[\frac{1}{2}|K_{\tau_{m-1}}|,\,\frac{3}{2}|K_{\tau_{m-1}}|\big]$ for some $\tau_{m}\leq t$ is unlikely to occur  (see~(\ref{Itakeskus})).

 \begin{proof}[Proof of Lem.~\ref{LemLindberg}]  The proof of the Lindberg condition for $M_{r}$ follows from the analysis contained in (i) of the proof of Thm.~\ref{SubMartCLT}.  Also, the random variables $t^{-\frac{1}{2}}M_{st}$ for $t\in \R_{+}$ and $s\in[0,1]$ are uniformly integrable, since  
 $$  \sup_{s\in[0,1],\,t\in \R_{+} } \mathbb{E}\big[\big(t^{-\frac{1}{2}}M_{st}\big)^{2}    \big] = t^{-1}\mathbb{E}\big[\langle M_{t},M_{t}\rangle \big]\leq \sigma,     $$
 where the second inequality is by Part (2) of Prop.~\ref{Basics}.  
 
To see the Lindberg condition for $\mathbf{m}_{s}^{(t)}$, notice that   
\begin{align}\label{Willful}
\sup_{0\leq s\leq 1}\big| \mathbf{m}_{s}^{(t)}-\mathbf{m}_{s^{-}}^{(t)}     \big| =  &  t^{-\frac{5}{4}} \sup_{0\leq n\leq \mathbf{N}_{t}}  \Delta\tau_{n}| K_{\tau_{n}}| \nonumber \\ \leq &  t^{-\frac{1}{4}} \Big(\sup_{0\leq r\leq t}t^{-\frac{1}{2}}|K_{r}|\Big)^{2}\Big( \sup_{1\leq n\leq \mathbf{N}_{t}}\frac{\Delta\tau_{n}}{|K_{\tau_{n}}|} \Big).   
\end{align} 
The random variables $t^{-\frac{1}{8}} \big(\sup_{0\leq r\leq t}t^{-\frac{1}{2}}|K_{r}|\big)^{2}$ converge to zero as $t\rightarrow \infty$, since $t^{-\frac{1}{2}}|K_{st}|$ converges to a Brownian motion with respect to  the uniform metric by Thm.~\ref{SubMartCLT}.  Moreover, by Part (2) of Prop.~\ref{TimeFlip}, there is a $C>0$ such that for small enough $\gamma>0$, the expectations
$\mathbb{E}\big[e^{\gamma\frac{\Delta\tau_{n}}{|K_{\tau_{n}}|}}\,\big|\,\widetilde{\mathcal{F}}_{\tau_{n}^{-}}\big]$ are smaller than $C$ for all $n$ and $t\gg 1$.   With Chebyshev's inequality, $\mathbb{P}\big[\frac{\Delta\tau_{n}}{|K_{\tau_{n}}|}>t^{\frac{1}{8}}\,\big|\,\widetilde{\mathcal{F}}_{\tau_{n}^{-}}\big]\leq Ce^{-\gamma t^{\frac{1}{8}} }$.  Applying this with an inductive argument using conditional expectations, then
\begin{align*}
\mathbb{P}\Big[ \sup_{1\leq n\leq \mathbf{N}_{t} } \frac{\Delta\tau_{n}}{|K_{\tau_{n}}|}\leq t^{\frac{1}{8}}  \Big]\geq \mathbb{E}\left[\Big(1-Ce^{-\gamma t^{\frac{1}{8}}}\Big)^{\mathbf{N}_{t}}  \right] \geq \mathbb{E}\left[\Big(1-Ce^{-\gamma t^{\frac{1}{8}}}\Big)^{\mathcal{N}_{t}}\right]\longrightarrow 1.     
\end{align*}
The second inequality follows because the number of sign-flips will be less than the number of Poisson times, and the number of Poisson times is typically $\propto t$.    
 
 The random variables $\mathbf{m}_{s}^{(t)}$ for  $s\in [0,1]$, $t\in \R_{+}$ are uniformly integrable, since the variances are uniformly bounded: 
 \begin{align}\label{Sis}  \sup_{s\in[0,1],\,t\in \R_{+} } \mathbb{E}\left[\big(\mathbf{m}_{s}^{(t)}\big)^{2}    \right]= &\mathbb{E}\Big[ t^{-\frac{5}{2}}\sum_{n=1}^{\mathbf{N}_{t} }|K_{\tau_{n}}|^{2}|\Delta\tau_{n}|^{2}    \Big]\nonumber  \\ \leq &\frac{8}{\nu} \mathbb{E}\Big[ t^{-\frac{5}{2}}\sum_{n=1}^{\mathbf{N}_{t}-1 }|K_{\tau_{n}}|^{3} \Delta\tau_{n} \Big]+\frac{4}{\nu} \mathbb{E}\Big[ t^{-\frac{5}{2}}\sup_{0\leq r\leq t}|K_{r}|^{4} \Big] \nonumber  \\  \leq & \frac{8}{\nu} \mathbb{E}\Big[ t^{-\frac{3}{2}}\sup_{0\leq r\leq t}|K_{r}|^{3}  \Big]+\mathit{O}\big(t^{-\frac{1}{2}}\big)=\mathit{O}(1) .
 \end{align}
 The second inequality uses that $ \sum_{n=1}^{\mathbf{N}_{t}-1 }\Delta\tau_{n}\leq t$.  For the first inequality, I use nested conditional expectations twice to replace the factor $|\Delta\tau_{n}|^{2}$ by $\frac{2}{\nu}\Delta\tau_{n}|K_{\tau_{n}}|$ in the sum.  In doing so, I invoke Part (1) of Prop.~\ref{TimeFlip} to get the approximations
 $$ \mathbb{E}\big[|\Delta\tau_{n}|^{2}\,\big|\,\widetilde{\mathcal{F}}_{\tau_{n}^{-}}\big]    \approx  2\nu^{-2}|K_{\tau_{n}}|^{2}\hspace{1cm}\text{and}\hspace{1cm}   |K_{\tau_{n}}|^{2}\approx \nu|K_{\tau_{n}}|\mathbb{E}\big[|\Delta\tau_{n}|\,\big|\,\widetilde{\mathcal{F}}_{\tau_{n}^{-}}\big].   $$  
For each approximation, I multiply the upper bound by a factor of $2$ to cover the error of the approximation.   Uniform bounds of $\mathbb{E}\big[ \sup_{0\leq r\leq t} t^{-\frac{m}{2}}|K_{r}|^{m} \big] $, $m=3,4$ for large $t$ are obtained by using that $|K_{r}|\leq \mathcal{E}_{r}$ and applying Doob's maximal inequality to the submartingale $\mathcal{E}_{r}$.  Bounds on the fourth moments of $\mathcal{E}_{r}$ are contained in Part (1) of Prop.~\ref{Basics}.

 \end{proof}

\begin{lemma} \label{Crave}
Fix some $0<\iota<\frac{3}{16}$ in the definition for the times $\tau_{n}$, $n\in \mathbb{N}$.  As $t\rightarrow \infty$, there is convergence in probability 
  $$\sup_{0\leq s\leq 1}\Big| Y_{s}^{(t)}- t^{-\frac{5}{4}}\sum_{n=1 }^{\mathbf{N}_{st}-1} \int_{\tau_{n}}^{\tau_{n}+\Delta \tau_{n}  }dr K_{r}\Big| \Longrightarrow 0 .    $$

 \end{lemma}
 
\begin{proof} 
 
 Since my estimates will depend on $|K_{r}|\gg 1$, I first show that the contribution to $Y_{s}^{(t)}$ that  accumulates when $|K_{r}|\ll |k|^{\frac{3}{8}}$ will be negligible.  Let $\delta>0$, then
\begin{align}\label{LowEnergy}
\mathbb{P}\Big[ \sup_{0\leq s\leq 1} t^{-\frac{5}{4}}\int_{0}^{st}dr |K_{r}| \chi\left( |K_{r}|\leq 2  t^{\frac{3}{8}-\iota} \right)  >\delta        \Big] &  \leq\mathbb{P}\Big[ \int_{0}^{1}ds\chi\left(t^{-\frac{3}{8}+\iota }|K_{st}|\leq 2 \right)> \delta t^{-\frac{1}{8}+\iota}    \Big]\nonumber   \\    &\leq 128\, \sigma^{-\frac{1}{2}}\delta^{-1}t^{-2\iota}, 
\end{align}
where the last inequality is for $t$ large enough by Lem.~\ref{EnergyLemma} with $\varsigma_{1}=\frac{3}{8}-\iota$, $\varsigma_{2}=2\iota$,  $\varsigma_{3}=\frac{1}{8}-\iota$, and $\epsilon=2$.  Hence, for any $\delta$, I can pick  $t$ large enough to make~(\ref{LowEnergy})  arbitrarily small.  It follows that  
 $$\sup_{0\leq s\leq 1}\Big| Y_{s}^{(t)}- t^{-\frac{5}{4}}\sum_{n=1 }^{\mathbf{N}_{st}} \int_{\tau_{n}}^{(\tau_{n}+\Delta \tau_{n})\wedge st  }dr K_{r}\Big| \Longrightarrow 0 $$
as  $t\rightarrow \infty$, since the contribution during incursions is negligible.    The upper bound $(\tau_{n}+\Delta \tau_{n})\wedge st$ for the integrals can be replaced by $ \tau_{n}+\Delta \tau_{n}$, since 
 $$  \sup_{0\leq n\leq \mathbf{N}_{t}} t^{-\frac{5}{4}}\int_{\tau_{n}}^{\tau_{n}+\Delta\tau_{n}  }dr |K_{r}|\leq   t^{-\frac{1}{4}} \Big(\sup_{0\leq r\leq t}t^{-\frac{1}{2}}|K_{r}|\Big)^{2}\Big( \sup_{1\leq n\leq \mathbf{N}_{t}}\frac{\Delta\tau_{n}}{|K_{\tau_{n}}|} \Big). $$
  The right side goes to zero for large $t$ by the discussion following~(\ref{Willful}).   For the same reason, I can replace the upper bound of the sum by $\mathbf{N}_{st}-1$.

 \end{proof}

For the remainder of the section, I will set $\iota=0$ in the definition for the excursion intervals $[\varsigma_{n-1},\varpi_{n})$, and thus indirectly for the definition of the times $\tau_{n}$.  The rate of decay computed for various expressions in the proofs would be a little slower if I kept $\iota>0$, but $\iota$ can be chosen arbitrarily small anyway. The following proof will rely heavily on applications of Lem.~\ref{Fraiser} and Prop.~\ref{TimeFlip}.

\begin{proof}[Proof of Lem.~\ref{LemMartApprox}] The proof of Part (1) follows from similar analysis as in (i) of the proof for Lem.~\ref{SubMartCLT}.   For Part (2), I can approximate $Y_{s}^{(t)}$ by the expression $t^{-\frac{5}{4}}\sum_{n=1 }^{\mathbf{N}_{st}-1} \int_{\tau_{n}}^{\tau_{n}+\Delta \tau_{n}  }dr K_{r}$ for $t\gg 1$ by Lem.~\ref{Crave}.  The result follows by showing the convergences in probability (i)-(iii) below.  The differences in (i) and (ii) involve coarse-graining approximations in which the random time intervals $\Delta\tau_{m}$ are parsed into shorter intervals $\Delta_{m,n}$ with duration on the order $\mathit{O}\big(|K_{\tau_{m}}|^{\beta}\big)$ for some $\beta$ chosen from  the interval $(0,\frac{1}{3})$.  I define the following notations:
\begin{eqnarray*}
\omega(m)&:=& \left\lfloor |K_{\tau_{m}}|^{\beta }\right\rfloor, \\
L_{m}&:=&   \left\lfloor \frac{\mathcal{N}_{\tau_{m}+\Delta\tau_{m} }, -\mathcal{N}_{\tau_{m}} }{\omega(m) }\right\rfloor , \\
\Gamma_{m,n}&:=&t_{\mathcal{N}_{\tau_{m}}+nL_{m}},\\
\Delta_{m,n}&:=&\Gamma_{m,n+1} -\Gamma_{m,n}.
\end{eqnarray*}
Moreover, I define   $\underline{\Gamma_{m,n}}$  
as the Poisson time preceding  $\Gamma_{m,n}$ when $n\geq 1$ and $\underline{\Gamma_{m,0}}=\tau_{m}$.  Also, I will abuse notation by identifying  $\mathcal{F}_{\underline{\Gamma_{m,0}}}$ with $\widetilde{\mathcal{F}}_{\tau_{m}^{-}}$.

 I will show the following convergences to zero in probability:
\begin{enumerate}[(i).]
\item $ \sup_{0\leq s\leq 1}\Big| t^{-\frac{5}{4}}\sum_{m=1}^{\mathbf{N}_{st}-1} \Big(\int_{\tau_{m}}^{\tau_{m}+\Delta\tau_{m} }drK_{r}- S(K_{\tau_{m}})\sum_{n=0}^{L_{m}-1}\Delta_{m,n}|K_{\Gamma_{m,n} } |\Big)\Big|\Longrightarrow 0, $

\item $ \sup_{0\leq s\leq 1}\Big|t^{-\frac{5}{4}}\sum_{m=1   }^{\mathbf{N}_{st}-1}\Big( S(K_{\tau_{m}})\sum_{n=0}^{L_{m}-1}\Delta_{m,n}|K_{\Gamma_{m,n} }|-  K_{\tau_{m}}\Delta\tau_{m}  \Big)   \Big| \Longrightarrow 0, $

\item $  \sup_{0\leq s\leq 1}\Big| t^{-\frac{5}{4}}\sum_{m=1 }^{\mathbf{N}_{st}-1} K_{\tau_{m}}\,\mathbb{E}\big[ \Delta\tau_{m} \,|\,\widetilde{\mathcal{F}}_{\tau_{m}^{-}}  \big] \Big| \Longrightarrow 0.  $
\end{enumerate}
The expression $t^{-\frac{5}{4}}\sum_{m=1   }^{\mathbf{N}_{st}-1}K_{\tau_{m}}(\Delta\tau_{m} -\mathbb{E}\big[ \Delta\tau_{m} \,|\,\widetilde{\mathcal{F}}_{\tau_{m}^{-}}  \big] ) $ is obtained by the  right term in  (ii) minus the expression in (iii) and differs from the expression for $\mathbf{m}_{s}^{(t)}$   by the substitution in the upper summand of $\mathbf{N}_{st}$  with $\mathbf{N}_{st}-1$.  By the Lindberg condition in Lemma~\ref{LemLindberg}, the difference is negligible for $t\gg 1$.

\vspace{.4cm}

\noindent (i).\hspace{.15cm}  Over an interval $r\in [\tau_{m},\tau_{m}+\Delta \tau_{m})$, the process $K_{r}$ has the same sign except for isolated Poisson times at which it jumps to the opposite sign and back again at the next Poisson time.  My first step will be to bound the net effect of these rogue sign changes. The first inequality below follows because $|K_{r}|\in \big[\frac{1}{2}|K_{\tau_{n}}|,\,\frac{3}{2}|K_{\tau_{n}}|\big]$: 
\begin{align}\label{KindaBlue}
\mathbb{E}\Big[\sup_{0\leq s\leq 1} t^{-\frac{5}{4}}&\sum_{m=1}^{\mathbf{N}_{st}-1}\int_{\tau_{m}}^{\tau_{m}+\Delta \tau_{m}}dr\chi\big(K_{r}\neq K_{\tau_{m}}\big)|K_{r}|\Big]\nonumber  \\ &\leq \mathbb{E}\Big[ t^{-\frac{5}{4}}\sum_{m=1}^{\mathbf{N}_{t}-1}\frac{3}{2}|K_{\tau_{m}}|\mathbb{E}\Big[\int_{\tau_{m}}^{\tau_{m}+\Delta \tau_{m}}dr\chi\big(K_{r}\neq K_{\tau_{m}}\big)\,\Big|\,\widetilde{\mathcal{F}}_{\tau_{m}^{-}}\Big]\Big]\nonumber   \\ &\leq C\mathbb{E}\Big[ t^{-\frac{5}{4}}\sum_{m=1}^{\mathbf{N}_{t}-1}\frac{3}{2}|K_{\tau_{m}}| \Big]\nonumber \\ &\leq  3C\nu t^{-\frac{5}{4}}\mathbb{E}\Big[\sum_{m=1}^{\mathbf{N}_{t}-1}\mathbb{E}\big[\Delta\tau_{m}   \,\big|\,\widetilde{\mathcal{F}}_{\tau_{m}^{-}}\big]  \Big]\nonumber  \\ & \leq 3C\nu t^{-\frac{1}{4}}\longrightarrow 0.
\end{align}
The second inequality follows by Part (3) of Prop.~\ref{TimeFlip}.  For the third inequality, I have used that $\mathbb{E}\big[\Delta\tau_{m}   \,\big|\,\widetilde{\mathcal{F}}_{\tau_{m}^{-}}\big]\approx \nu^{-1}|K_{\tau_{m}} |$ for large  $t$ by Part (1) of Prop.~\ref{TimeFlip}, and I doubled the bound to cover the error.  The last inequality holds since the  $\Delta\tau_{m}$'s sum up to less than $t$. Hence, I can take the sign of $K_{r}$ to be constant over the time intervals $[\tau_{m},\tau_{m}+\Delta\tau_{m})$.

It is useful to approximate the expressions in (i) by replacing $K_{r}$ with $S(K_{r})\mathcal{E}_{r}$, since $\mathcal{E}_{r}$ is a submartingale with convenient analytic properties.  By~(\ref{KindaBlue}) and because $|K_{r}|=\mathcal{E}_{r}+\mathit{O}(|K_{r}|^{-1})$, I have that    
$$\sup_{0\leq s\leq 1}\Big| t^{-\frac{5}{4}}\sum_{m=1}^{\mathbf{N}_{st}-1}\Big(\int_{\tau_{m}}^{\tau_{m}+\Delta\tau_{m}}drK_{r}- S(K_{\tau_{m}})\int_{\tau_{m}}^{\tau_{m}+\Delta\tau_{m}}dr\,\mathcal{E}_{r} \Big)   \Big|\longrightarrow 0.  $$
I can also replace $|K_{r}|$ with $\mathcal{E}_{r}$ in the expression $t^{-\frac{5}{4}}\sum_{m=1   }^{\mathbf{N}_{st}-1}S(K_{\tau_{m}}) \sum_{n=0}^{L_{m}-1}\Delta_{m,n}\big|K_{\Gamma_{m,n} }\big| $, and the difference (i) reduces to
\begin{multline}\label{Mink}
\Big| t^{-\frac{5}{4}}\sum_{m=1}^{\mathbf{N}_{st}-1}S(K_{\tau_{m}})\int_{\tau_{m}}^{\tau_{m}+\Delta\tau_{m}}dr\,\mathcal{E}_{r} - t^{-\frac{5}{4}}\sum_{m=1   }^{\mathbf{N}_{st}-1}S(K_{\tau_{m}}) \sum_{n=0}^{L_{m}-1}\Delta_{m,n}\mathcal{E}_{\Gamma_{m,n}} \Big| \\ \leq  t^{-\frac{5}{4}}  \Big|\sum_{ m=1 }^{\mathbf{N}_{st}-1}\sum_{n=0}^{L_{m}-1}
\int_{\Gamma_{m,n} }^{\Gamma_{m,n+1} }dr\big(\mathcal{E}_{r}-\mathcal{E}_{\Gamma_{m,n} }\big)        \Big|  +t^{-\frac{5}{4}}\Big(\sup_{0\leq r\leq t}\mathcal{E}_{r}\Big)  \sum_{m=1}^{\mathbf{N}_{t}-1}\Delta_{m,L_{m}} ,
\end{multline}
where the second term bounds the under-counting by $\Delta\tau_{m}-(\Gamma_{m,L_{m}}-\tau_{m})\leq \Delta_{m,L_{m}}$ for the length of the interval $[\tau_{m},\,\tau_{m}+\Delta\tau_{m})$. 

Next, I will show that the second term on the second line of~(\ref{Mink}) goes to zero.  The intervals between successive Poisson times have mean $\mathcal{R}^{-1}$ and are independent of everything else, which gives the  equality below:   
\begin{align}\label{Mink2}
\mathbb{E}\Big[\sum_{m=1 }^{\mathbf{N}_{t}-1}\Delta_{m,L_{m}}\Big] &= \frac{1}{\mathcal{R}}\mathbb{E}\Big[\sum_{m=1 }^{\mathbf{N}_{t}-1}\omega(m)\Big]\leq \frac{t^{-\frac{3}{8}(1-\beta)} }{\mathcal{R}} \mathbb{E}\Big[\sum_{m=1}^{\mathbf{N}_{t}-1}|K_{\tau_{m}}|\Big]\nonumber \\ &\leq   \frac{2\nu t^{-\frac{3}{8}(1-\beta)} }{\mathcal{R}} \mathbb{E}\Big[\sum_{m=1 }^{\mathbf{N}_{t}-1}\mathbb{E}\big[\Delta\tau_{m}\,\big|\,\widetilde{\mathcal{F}}_{\tau_{m}^{-}}  \big] \Big]\leq \frac{2\nu t^{\frac{5}{8}+\frac{3}{8}\beta} }{\mathcal{R}} .  
\end{align}
The first inequality follows because $\omega(m)\leq |K_{\tau_{m}}|^{\beta}$, and $|K_{\tau_{m}}|\geq t^{\frac{3}{8}}$.  By  Part (1) of Prop.~\ref{TimeFlip}, I have that $\mathbb{E}\big[\Delta\tau_{m}\,|\,K_{\tau_{m}}\big]\geq  \frac{1}{2\nu} |K_{\tau_{m}}| $ for large enough $t$.  The  last inequality follows by removing the nested conditional expectations and $\sum_{m=1 }^{\mathbf{N}_{t}-1}\Delta\tau_{m}<t$.  Recall that $\beta\in(0,\frac{1}{3})$, so I have $ \frac{5}{8}+\frac{3}{8}\beta <\frac{3}{4}$.  For a  $\delta \in (0,\frac{1}{8}-\frac{3}{8}\beta)$, then~(\ref{Mink2}) implies that $t^{-\frac{3}{4}+\delta}\sum_{m=1 }^{\mathbf{N}_{t}-1}\Delta_{m}$ converges to zero in probability for large $t$, and Thm.~\ref{SubMartCLT} implies that $t^{-\frac{1}{2}-\delta}\sup_{0\leq r\leq t}\mathcal{E}_{r}$ converges to zero in probability.  Hence, the product goes to zero.

  To bound the first term on the second line of~(\ref{Mink}), I use the Doob-Meyer decomposition $\mathcal{E}_{s}=M_{s}+A_{s}$ and the triangle inequality to get
\begin{align}\label{Instanbul}
t^{-\frac{5}{4}}\mathbb{E}\Big[  \sup_{0\leq s\leq 1}  \Big|& \sum_{ m=1 }^{\mathbf{N}_{st}-1}\sum_{n=0}^{L_{m}-1}
\int_{\Gamma_{m,n}  }^{\Gamma_{m,n+1} }dr\big(\mathcal{E}_{r}-\mathcal{E}_{\Gamma_{m,n}}\big) \Big| \Big]\nonumber  \\  \leq  & t^{-\frac{5}{4}}\mathbb{E}\Big[  \sum_{ m=1}^{\mathbf{N}_{t}-1}\sum_{n=0}^{L_{m}-1}
\mathbb{E}\Big[\Big|\int_{\Gamma_{m,n} }^{\Gamma_{m,n+1} }dr\big(M_{r}-M_{\Gamma_{m,n} }\big)\Big| \,  \Big|\,\mathcal{F}_{\Gamma_{m,n}   }  \Big]\Big] \nonumber \\ &+t^{-\frac{5}{4}}\mathbb{E}\Big[ \sum_{ m=1  }^{\mathbf{N}_{t}-1}\sum_{n=0}^{L_{m}-1}
\int_{\Gamma_{m,n} }^{\Gamma_{m,n+1} }dr\big(A_{r}-A_{\Gamma_{m,n} }\big) \Big],
\end{align}
where  I have inserted nested conditional expectations for the martingale term.     
For the $A_{r}$ term, recall that $\mathcal{E}_{r}^{2}$ is a submartingale with increasing part $\mathcal{E}_{0}^{2}+\sigma t$ by Part (1) of Prop.~\ref{Basics}.   It follows that $\frac{d}{dr}A_{r}\leq \frac{1}{2\mathcal{E}_{r}}(\sigma-\frac{d}{dr}\langle M,M\rangle_{r})\leq \sigma t^{-\frac{3}{8}}$ for the ``high energy" part of the trajectory  $\mathcal{E}_{r}\approx |K_{r}|\geq t^{\frac{3}{8}}$.  The last line in~(\ref{Instanbul}) is therefore less than $\sigma t^{-\frac{5}{8}}$.

 For a single pair $m,n$ from the sum on the second line of~(\ref{Instanbul}), the following inequalities hold:
\begin{align*}
\mathbb{E}\Big[\Big|\int_{\Gamma_{m,n} }^{\Gamma_{m,n+1} }dr\big(M_{r}-M_{\Gamma_{m,n} }\big)\Big| \,  \Big|\,\mathcal{F}_{\Gamma_{m,n}   }  \Big] &\leq \mathbb{E}\Big[\Big| \int_{\Gamma_{m,n} }^{\Gamma_{m,n+1} }dr \big(M_{r}-M_{\Gamma_{m,n} }\big)   \Big|^{2}\,  \Big|\,\mathcal{F}_{\Gamma_{m,n}  } \Big]^{\frac{1}{2}} \\ &= \mathbb{E}\Big[\Big|\int_{\Gamma_{m,n} }^{\Gamma_{m,n+1} }dM_{r}\big(\Gamma_{m,n+1} -r\big)   \Big|^{2}\,  \Big|\,\mathcal{F}_{\Gamma_{m,n}}    \Big]^{\frac{1}{2}}\\ &=  \mathbb{E}\Big[\int_{\Gamma_{m,n} }^{\Gamma_{m,n+1} }dr\frac{d\langle M,M\rangle_{r}}{dr} \big(\Gamma_{m,n+1} -r\big)^{2} \,  \Big|\,\mathcal{F}_{\Gamma_{m,n}  }  \Big]^{\frac{1}{2}} \\ & \leq ( \frac{\sigma}{3})^{\frac{1}{2}} \Delta_{m,n}^{\frac{3}{2}}.
\end{align*}
  The first inequality above is Jensen's inequality.   The last inequality holds since the predictable quadratic variation $\langle M,M\rangle_{r}$ grows at a rate $\leq \sigma$. 

  The second line of~(\ref{Instanbul}) is  bounded by $2( \frac{\sigma}{3})^{\frac{1}{2}}$ multiplied by
\begin{align}\label{Mathews}
  t^{-\frac{5}{4}}\mathbb{E}\Big[  \sum_{ m=1}^{\mathbf{N}_{t}-1}\sum_{n=0}^{L_{m}-1}
\Delta_{m,n}^{\frac{3}{2}}\Big]\leq & \mathcal{R}^{-\frac{3}{2}}\Big(\sup_{s\geq 1}\frac{\Gamma(s+\frac{3}{2})}{s^{\frac{3}{2}}\Gamma(s)   } \Big) t^{-\frac{5}{4}}\mathbb{E}\Big[  \sum_{ m=1}^{\mathbf{N}_{t}-1}\sum_{n=0}^{L_{m}-1}
\omega^{\frac{3}{2}}(m)\Big]\nonumber \\ = & \mathcal{R}^{-\frac{1}{2}}\Big(\sup_{s\geq 1}\frac{\Gamma(s+\frac{3}{2})}{s^{\frac{3}{2}}\Gamma(s)   } \Big)  t^{-\frac{5}{4}}\mathbb{E}\Big[  \sum_{ m=1}^{\mathbf{N}_{t}-1}\omega^{\frac{1}{2}}(m)\sum_{n=0}^{L_{m}-1}
\Delta_{m,n}\Big]\nonumber  \\ \leq &  \mathcal{R}^{-\frac{1}{2}}\Big(\sup_{s\geq 1}\frac{\Gamma(s+\frac{3}{2})}{s^{\frac{3}{2}}\Gamma(s)   } \Big) t^{-\frac{1}{4}}\mathbb{E}\Big[\sup_{0\leq r\leq t}\mathcal{E}_{r}^{\frac{\beta}{2}} \Big],
\end{align}
where $\Gamma:\R_{+}\rightarrow \R_{+}$ is the gamma function, and it should not to be confused with the times $\Gamma_{n,m}$.  The  first inequality and equality  above follow because  $\Delta_{m,n}$ is a sum of $\omega(m)$ independent  mean-$\mathcal{R}^{-1}$ exponential random variables.   For the second inequality in~(\ref{Mathews}), I have used that  $ \sum_{ m=1}^{\mathbf{N}_{t}-1}\sum_{m=0}^{L_{m}-1}\Delta_{m,n} \leq t$ and  $\omega_{m}\leq \sup_{0\leq r\leq t}|K_{r}|^{\beta}$. Finally, by Jensen's inequality 
\begin{align*}
t^{-\frac{1}{4}}\mathbb{E}\Big[\sup_{0\leq r\leq t}\mathcal{E}_{r}^{\frac{\beta}{2}} \Big]  & \leq t^{-\frac{1}{4}} \mathbb{E}\Big[\sup_{0\leq r\leq t}\mathcal{E}_{r}^{2} \Big]^{\frac{\beta}{4}}\leq  2^{\frac{\beta}{2}} t^{-\frac{1}{4}} \mathbb{E}\Big[\mathcal{E}_{t}^{2} \Big]^{\frac{\beta}{4}}\\ &=  t^{-\frac{1}{4}}\big( \mathbb{E}\big[\mathcal{E}_{0}^{2} \big]+t\sigma\big)^{\frac{\beta}{4}}\propto t^{\frac{\beta}{4}-\frac{1}{4} }\rightarrow 0,   \end{align*}
where the first inequality is Jensen's and the second is Doob's.  The equality is by Part (1) of Prop.~\ref{Basics}.

\vspace{.5cm}

\noindent (ii). \hspace{.15cm}  The difference is bounded by
\begin{align}\label{HorseCrap}
\sup_{0\leq s\leq 1}\Big|t^{-\frac{5}{4}}\sum_{m=1   }^{\mathbf{N}_{st}-1} S(K_{\tau_{m}})\sum_{n=0}^{L_{m}-1} \Delta_{m,n}\big( |K_{\Gamma_{m,n} }|-  |K_{\tau_{m}}| \big) \Big| +2t^{-\frac{5}{4}}\Big(\sup_{0\leq s\leq 1}|K_{s}| \Big) \sum_{m=1   }^{\mathbf{N}_{t}-1}\Delta_{m}.  
\end{align}
The second term decays to zero by the argument in (i).  By  partial summation, the left sum in~(\ref{HorseCrap}) is equal to  
\begin{align}\label{Demogogue}
t^{-\frac{5}{4}}\sum_{m=1   }^{\mathbf{N}_{st}-1} S(K_{\tau_{m}})\sum_{n=1}^{L_{m}-1} \big( |K_{\Gamma_{m,n} }|-  |K_{\Gamma_{m-1,n}}| \big)\big(\Gamma_{m,L_{m}}-\Gamma_{m,n}   \big)&\approx \nonumber   \\
t^{-\frac{5}{4}}\sum_{m=1   }^{\mathbf{N}_{st}-1} S(K_{\tau_{m}})\sum_{n=1}^{L_{m}-1} \big( |K_{\underline{\Gamma_{m,n}} }|-  |K_{\Gamma_{m-1,n}}| \big)\big(\Gamma_{m,L_{m}}-\Gamma_{m,n}   \big)  & \approx  t^{-\frac{5}{4}}\sum_{m=1   }^{\mathbf{N}_{st}-1} H_{m}, 
\end{align}
where  $H_{m}$ is defined as 
\begin{align*} 
 H_{m}:=S(K_{\tau_{m}})\sum_{n=1}^{\infty} \big( |K_{\underline{\Gamma_{m,n}} }| -|K_{\Gamma_{m,n-1}}| \big)\big(\tau_{m}+\Delta\tau_{m}-\Gamma_{m,n}   \big)\chi\big(\tau_{m}+\Delta\tau_{m}>\Gamma_{m,n}  \big).
\end{align*}
  For technical reasons involving conditioning, it will be convenient to replace $K_{\Gamma_{m,n}}$ by $K_{\underline{\Gamma_{m,n}} }$ as in the first approximation of~(\ref{Demogogue}).  The sum $ t^{-\frac{5}{4}}\sum_{m=1   }^{\mathbf{N}_{t}-1} \sum_{n=0}^{L_{m}-1} \Delta_{m,n}\big| |K_{\Gamma_{m,n} }|-  |K_{\underline{\Gamma_{m,n}} } |  \big|  $
is easy to bound using  $|K_{r}|\approx \mathcal{E}_{r}$ and techniques used in (i).  The second approximation in~(\ref{Demogogue}) replaces $\Gamma_{m,L_{m}}$ by  $\tau_{m}+\Delta\tau_{m} $  in the expression, and the resulting error decays by the argument in~(\ref{Mink2}) again. 

Since the probability that $\sup_{0\leq r\leq t}|K_{r}|\geq \epsilon^{-1} t^{\frac{1}{2}}$ is small for $1\gg \epsilon>0$, I have the equality $t^{-\frac{5}{4}}\sum_{m=1   }^{\mathbf{N}_{st}-1} H_{m}=t^{-\frac{5}{4}}\sum_{m=1   }^{\mathbf{N}_{st}-1} H_{m}\chi(|K_{\tau_{m}}|\leq \epsilon^{-1}t^{\frac{1}{2}})$ with probability close to one.  Introducing cutoff's will be useful to avoid problems with higher moments of $K_{r}$.   So far I have made only minor adjustments to the expression.  The strategy to show  that $t^{-\frac{5}{4}}\sup_{0\leq s\leq 1}\Big|\sum_{m=1   }^{\mathbf{N}_{st}-1} H_{m}\Big|$ converges in probability to zero will be to prove: \vspace{.2cm} 

\noindent \hspace{.1cm} (\textup{ii}$'$).\hspace{.15cm} $t^{-\frac{5}{4}}\mathbb{E}\Big[\sum_{m=1}^{\mathbf{N}_{t}-1} \big| \mathbb{E}\big[H_{m}\,\big|\,\widetilde{\mathcal{F}}_{\tau_{m}^{-}} \big]\big|\chi\big(|K_{\tau_{m}}|\leq \epsilon^{-1}t^{\frac{1}{2}}\big) \Big]\longrightarrow 0$,  \vspace{.2cm}

\noindent \hspace{.1cm} (\textup{ii}$''$).\hspace{.15cm} $t^{-\frac{5}{2}}\mathbb{E}\Big[\sum_{m=1}^{\mathbf{N}_{t}-1} \big(\mathbb{E}\big[H_{\tau_{m} }^{2}\,\big|\,\widetilde{\mathcal{F}}_{\tau_{m}^{-}}   \big]  -\mathbb{E}\big[H_{m}\,\big|\,\widetilde{\mathcal{F}}_{\tau_{m}^{-}}   \big]^{2}\big)\chi\big(|K_{\tau_{m}}|\leq \epsilon^{-1}t^{\frac{1}{2}}\big)   \Big]\longrightarrow 0$.  \vspace{.3cm}\\
  By an analogous Lindberg condition as for $\mathbf{m}_{s}^{(t)}$, the upper summands above can be taken to be either $\mathbf{N}_{t}-1$ or $\mathbf{N}_{t}$ depending on convenience.  Since the sum of $(H_{m}- \mathbb{E}\big[ H_{m}\,\big|\,\widetilde{\mathcal{F}}_{\tau_{m}^{-}}  \big])\chi(|K_{\tau_{m}}|\leq \epsilon^{-1}t^{\frac{1}{2}}) $ for $m\in [1,\mathbf{N}_{st}]$ is a  $\mathcal{F}_{s}^{(t)}$-martingale, the convergence~(\textup{ii}$''$) along with Doob's maximal inequality shows that $t^{-\frac{5}{4}}\sup_{0\leq s\leq 1}\big|\sum_{m=1   }^{\mathbf{N}_{st}-1} H_{m}- \mathbb{E}\big[H_{m}\,\big|\,\widetilde{\mathcal{F}}_{\tau_{m}^{-}}  \big] \big|$ converges to zero.

 \vspace{.4cm}
\noindent (ii$'$).\hspace{.15cm}   A single term $\big|\mathbb{E}\big[ H_{m}\,\big|\,\widetilde{\mathcal{F}}_{\tau_{m}^{-}}   \big] \big|$ is bounded by: 
\begin{align}\label{FineMess}
\Big| \mathbb{E}\Big[&\sum_{n=1}^{\infty} \chi\Big(\Delta\tau_{m}+\tau_{m}>\Gamma_{m,n}\Big) \big(|K_{\underline{\Gamma_{m,n}} }|-  |K_{\Gamma_{m,n-1} }|\big)\,  \mathbb{E}\big[ \Delta\tau_{m}+\tau_{m}-\Gamma_{m,n} \big|\,\mathcal{F}_{\underline{\Gamma_{m,n}} }    \big] \,\Big|\,\widetilde{\mathcal{F}}_{\tau_{m}^{-}}\Big]\Big| \nonumber \\ \leq & \nu^{-1} |K_{\tau_{m}}|\, \Big|  \mathbb{E}\Big[\sum_{n=1}^{\infty} \chi\Big(\Delta\tau_{m}+\tau_{m}>\Gamma_{m,n}\Big) \big(|K_{\underline{\Gamma_{m,n}} }|-  |K_{\Gamma_{m,n-1} }|\big) \,\Big|\,\widetilde{\mathcal{F}}_{\tau_{m}^{-}}\Big]\Big| \nonumber  \\  & +
 \, \nu^{-1} \Big|  \mathbb{E}\Big[\sum_{n=1}^{\infty} \chi\Big(\Delta\tau_{m}+\tau_{m}>\Gamma_{m,n}\Big) \big(|K_{\underline{\Gamma_{m,n}} }|-  |K_{\Gamma_{m,n-1} }|\big)\,\big( |K_{\underline{\Gamma_{m,n}}}|-|K_{\tau_{m}}|\big)  \,\Big|\,\widetilde{\mathcal{F}}_{\tau_{m}^{-}}\Big]\Big| 
\nonumber  \\ &  +
 \,C  \big(\frac{3}{2}\big)^{\zeta} |K_{\tau_{m}}|^{\gamma}  \mathbb{E}\Big[\sum_{n=1}^{\infty} \chi\Big(\Delta\tau_{m}+\tau_{m}>\Gamma_{m,n}\Big) \Big||K_{\underline{\Gamma_{m,n}} }|-  |K_{\Gamma_{m,n-1} }|\Big|\, \,\Big|\,\widetilde{\mathcal{F}}_{\tau_{m}^{-}}\Big] ,
\end{align}
where I have applied the triangle inequality with 
\begin{align} \label{Belief}
\mathbb{E}\big[ \Delta\tau_{m}+\tau_{m}-\Gamma_{m,n} \big|\,\mathcal{F}_{\underline{\Gamma_{m,n}} }    \big]=&\nu^{-1} |K_{\tau_{m}}|+\big(\nu^{-1} |K_{\underline{\Gamma_{m,n}}}|-\nu^{-1} |K_{\tau_{m}}|\big)\nonumber \\ &+ \big( \mathbb{E}\big[ \Delta\tau_{m}+\tau_{m}-\Gamma_{m,n} \big|\,\mathcal{F}_{\underline{\Gamma_{m,n}} }    \big]-\nu^{-1} |K_{\underline{\Gamma_{m,n}}}|\big).      
\end{align}
For the third term in~(\ref{FineMess}), I have applied an analog of Part (1) of Prop.~\ref{TimeFlip} to get a $C>0$ for a given $\zeta>0$ such that 
\begin{align}\label{Faith}
\Big|\mathbb{E}\big[ \Delta\tau_{m}+\tau_{m}-\Gamma_{m,n}\big|\,\mathcal{F}_{\underline{\Gamma_{m,n}}}    \big]-\nu^{-1} |K_{\underline{\Gamma_{m,n}}}| \Big|\leq C|K_{\underline{\Gamma_{m,n}}}|^{\zeta}\leq C \big(\frac{3}{2}\big)^{\zeta} \big|K_{\tau_{m}}\big|^{\zeta},
\end{align}
 where the second inequality uses that $|K_{\underline{\Gamma_{m,n}}}|\leq \frac{3}{2}|K_{\tau_{m}}|$ for  $\Gamma_{m,n}<\Delta\tau_{m}+\tau_{m}$.  I have said ``analog" above because the situation is not identical to Prop.~\ref{TimeFlip}, since I am  conditioning with respect to $\mathcal{F}_{\underline{\Gamma_{m,n}}}$ rather than conditioning that the sign of the momentum does not change on the Poisson time following $\Gamma_{m,n}$.  The choice made earlier in the proof to deal with $\mathcal{F}_{\underline{\Gamma_{m,n}}}$ rather than $\mathcal{F}_{\Gamma_{m,n}}$ is useful now, since the conditional density for $K_{\Gamma_{m,n}}$ and $\Theta(K_{\Gamma_{m,n}})$ given $\mathcal{F}_{\underline{\Gamma_{m,n}}}$ will be bounded.  It follows by the discussion preceding the statement of Prop.~\ref{TimeFlip} that I can apply the same proof as in Part (1) of Prop.~\ref{TimeFlip} to this case.

The expectation in the second line  of~(\ref{FineMess}) can be rewritten as
\begin{align}\label{Flout}
\mathbb{E}\Big[\sum_{n=1}^{\infty} &\chi\Big(\Delta\tau_{m}+\tau_{m}>\Gamma_{m,n-1}\Big) \Big|\mathbb{E}\Big[\big(|K_{\underline{\Gamma_{m,n}} }|-  |K_{\Gamma_{m,n-1} }|\big)\chi\Big(\Delta\tau_{m}+\tau_{m}>\Gamma_{m,n}\Big) \Big|\, \mathcal{F}_{\underline{\Gamma_{m,n-1}}  }\Big] \Big|\, \Big|\,\widetilde{\mathcal{F}}_{\tau_{m}^{-}}\Big]\nonumber  \\ &\leq  c \mathbb{E}\Big[\sum_{n=1}^{\infty} \chi\Big(\Delta\tau_{m}+\tau_{m}>\Gamma_{m,n-1}\Big) |K_{\Gamma_{m,n-1}}|^{\beta-1+\zeta  }   \,\Big|\,\widetilde{\mathcal{F}}_{\tau_{m}^{-}}\Big]\nonumber \\ &\leq  2c  |K_{\tau_{m}}|^{\beta-1+\zeta }\mathbb{E}\Big[\sum_{n=1}^{\infty} \chi\Big(\frac{\mathcal{N}_{\Delta\tau_{m}+\tau_{m}}-\mathcal{N}_{\tau_{m}}}{\omega(m)}>n-1\Big)    \,\Big|\,\widetilde{\mathcal{F}}_{\tau_{m}^{-}}\Big] \nonumber   \\ &\leq  2c  |K_{\tau_{m}}|^{\beta-1+\zeta }\frac{1}{\omega(m)}\mathbb{E}\big[\mathcal{N}_{\Delta\tau_{m}+\tau_{m}}-\mathcal{N}_{\Delta\tau_{m}}   +1\,\big|\,\widetilde{\mathcal{F}}_{\tau_{m}^{-}}\big] \nonumber   \\
&\leq \frac{4c }{\alpha }  |K_{\tau_{m}}|^{\zeta } , 
\end{align}
where replacing $\chi\big(\Delta\tau_{m}+\tau_{m}>\Gamma_{m,n}\big)$ by $\chi\big(\Delta\tau_{m}+\tau_{m}>\Gamma_{m,n-1}\big)$ on the first line makes the expression larger. The first inequality is for some $c>0$ by Part (2) of Lem.~\ref{Fraiser}.  For the second inequality, I have used that  $|K_{r}|\geq \frac{1}{2}|K_{\tau_{m}}|$ for $r\in \big[\tau_{m},\,\tau_{m}+\Delta\tau_{m}\big]$, and I rewrote the argument of $\chi$ in terms of the corresponding counts for the Poisson times.  The third equality is from the expectation formula $\sum_{n=1}^{\infty}\mathbb{P}[X> n]\leq   \mathbb{E}[X]$ for a positive random variable $X$. The fourth inequality in~(\ref{Flout}) follows from the relations
$$ \mathbb{E}\big[\mathcal{N}_{\Delta\tau_{m}+\tau_{m}}-\mathcal{N}_{\tau_{m}}   \,\big|\,\widetilde{\mathcal{F}}_{\tau_{m}^{-}}\big]= \mathcal{R} \mathbb{E}\big[\Delta\tau_{m}  \,\big|\,\widetilde{\mathcal{F}}_{\tau_{m}^{-}}\big]\leq \frac{2}{\alpha} |K_{\tau_{m}}|^{-1 },$$
 where the inequality is for large enough times by Part (1) of Lem.~\ref{Fraiser}.

For the expression on the third line~(\ref{FineMess}), I split the factor 
$ |K_{\underline{\Gamma_{m,n}} }|-  |K_{\tau_{m}}|$ into a sum of the terms $ |K_{\Gamma_{m,n-1}}|-  |K_{\tau_{m}}|$ and $  |K_{\underline{\Gamma_{m,n}} }|- |K_{\Gamma_{m,n-1} }|$.  For $ |K_{\underline{\Gamma_{m,n}} }|-  |K_{\tau_{m}}| $, I introduce a nested conditional expectation with respect to $\mathcal{F}_{\underline{\Gamma_{m,n-1}}}$ as follows: 
\begin{align*}
\nu^{-1} \mathbb{E}&\Big[ \sum_{n=1}^{\infty} \chi\Big(\Delta\tau_{m}+\tau_{m}>\Gamma_{m,n-1}\Big)\big|\mathbb{E}\big[ |K_{\underline{\Gamma_{m,n}} }|-  |K_{\Gamma_{m,n-1} }|\,\big| \,\mathcal{F}_{\underline{\Gamma_{m,n-1}}}\big]\big|\,  \big| |K_{\Gamma_{m,n-1} }|-  |K_{\tau_{m}}|\big| \,\Big|\,\widetilde{\mathcal{F}}_{\tau_{m}^{-}}\Big] \\  & \leq c\nu^{-1} |K_{\tau_{m}}|^{\beta+\zeta-1} \mathbb{E}\Big[\sum_{n=1}^{\infty} \chi\Big(\frac{\mathcal{N}_{\Delta\tau_{m}+\tau_{m}}-\mathcal{N}_{\tau_{m}}}{\omega(m)}>n-1\Big) \big| |K_{\Gamma_{m,n-1} }|-  |K_{\tau_{m}}|\big| \,\Big|\,\widetilde{\mathcal{F}}_{\tau_{m}^{-}}\Big] \\  & \leq \frac{c}{\nu} |K_{\tau_{m}}|^{\beta+\zeta-1}\,\frac{\mathcal{R}}{\omega(m)}\mathbb{E}\big[ (\Delta\tau_{m})^{2}
 \,\big|\,\widetilde{\mathcal{F}}_{\tau_{m}^{-}}\big]^{\frac{1}{2}}  \mathbb{E}\Big[ \sup_{0\leq r\leq \Delta \tau_{m}}\big| |K_{\tau_{m}+r }|-  |K_{\tau_{m}}|\big|^{2}   \,\Big|\,\widetilde{\mathcal{F}}_{\tau_{m}^{-}}\Big]^{\frac{1}{2}}\\ & \propto |K_{\tau_{m}}|^{\zeta} ,
\end{align*}
The first inequality is by Part (2) of Lem.~\ref{Fraiser} and  $|K_{\Gamma_{m,n-1} }|\geq \frac{1}{2}|K_{\tau_{m}}|$ for $r\in \big[\tau_{m},\,\tau_{m}+\Delta\tau_{m}\big]$.  The first expectation on the third line is $\approx \frac{2}{\nu^{2}}|K_{\tau_{m}}|$ by Part (1) of Prop.~\ref{TimeFlip}.  The second expectation on the third line is also $\mathit{O}(|K_{\tau_{m}}|)$, which can be shown by approximating $\mathcal{E}_{r}\approx |K_{r}|$,  writing  $\mathcal{E}_{\underline{\Gamma_{m,n}} }-  \mathcal{E}_{\Gamma_{m,n-1}}$ in terms of the Doob-Meyer components $M_{r},\,A_{r}$, and applying the relation $\sigma=\frac{d}{dr}\langle M,M\rangle_{r}+2\mathcal{E}_{r}\frac{d}{dr}A_{r}$ in the standard inequalities as  before.   The term corresponding to  $  |K_{\underline{\Gamma_{m,n}} }|- |K_{\Gamma_{m,n-1} }|$ has a simpler analysis yielding a bound $\mathit{O}\big( |K_{\tau_{n}}|^{1+\zeta}\big)$.

The fourth line of~(\ref{FineMess}) is $\mathit{O}\big( |K_{\tau_{n}}|^{\zeta-\frac{\beta}{2} }  \big)$.  By the above analysis, the term $\big|\mathbb{E}\big[ H_{m}\,\big|\,\widetilde{\mathcal{F}}_{\tau_{m}^{-}}   \big] \big|$  is bounded by a constant multiple of $|K_{\tau_{m}}|^{1+\zeta}$ (since $\zeta<\frac{1}{4}$).  Thus,  (ii$'$) is bounded by a constant multiple of     
\begin{align}\label{Juice}
t^{-\frac{5}{4}}\mathbb{E}\Big[\sum_{m=1}^{\mathbf{N}_{t}-1}|K_{\tau_{m}}|^{1+\zeta}\chi\big(|K_{\tau_{m}}|\leq \epsilon^{-1}t^{\frac{1}{2}} \big)  \Big]& \leq \epsilon^{-\zeta } t^{-\frac{5}{4}+\frac{\zeta}{2}  }\mathbb{E}\Big[\sum_{m=1}^{\mathbf{N}_{t}-1}|K_{\tau_{m}}|  \Big]\nonumber \\ & \leq 2\nu \epsilon^{-\zeta } t^{-\frac{5}{4}+\frac{\zeta}{2}  }\mathbb{E}\Big[ \sum_{m=1}^{\mathbf{N}_{t}-1}\mathbb{E}\big[\Delta\tau_{m}\,\big|\,\widetilde{\mathcal{F}}_{\tau_{m}^{-}}\big]   \Big]\nonumber  \\ &  = 2\nu \epsilon^{-\zeta } t^{-\frac{5}{4}+\frac{\zeta}{2}  }\mathbb{E}\Big[ \sum_{m=1}^{\mathbf{N}_{t}-1}\Delta\tau_{m}  \Big]\nonumber   \\ & \leq  2\nu \epsilon^{-\zeta } t^{-\frac{1}{4}+\frac{\zeta}{2}  },
\end{align}
where the inequalities follow by the standard method using  Part (1) of Prop.~\ref{TimeFlip} and  $\sum_{m=1}^{\mathbf{N}_{t}-1}\Delta\tau_{m}<t$. 
Hence, the expression in (ii$'$) goes to zero for large $t$.

\vspace{.4cm}

\noindent (ii$''$). \hspace{.15cm} A similar argument as in (ii$'$) yields that 
$$ t^{-\frac{5}{2}}\mathbb{E}\Big[\sum_{m=1}^{\mathbf{N}_{t}-1}\mathbb{E}\big[H_{m}\,\big|\,\widetilde{\mathcal{F}}_{\tau_{m}^{-}} \big]^{2}\chi\big(|K_{\tau_{n} } |>\epsilon^{-1}t^{\frac{1}{2}} \big) \Big]$$
 is bounded by a constant multiple of $\epsilon^{-2\zeta}t^{-1+\zeta}$.  Thus, this term also vanishes as $t\rightarrow \infty$.     The expression (ii$''$) requires an analysis of $\mathbb{E}\big[ H_{\tau_{m} }^{2}\,\big|\,\widetilde{\mathcal{F}}_{\tau_{m}^{-}}\big]$.  I begin with the inequality 
 $$\big|\mathbb{E}\big[ H_{\tau_{m} }^{2}\,\big|\,\widetilde{\mathcal{F}}_{\tau_{m}^{-}}\big]\big|\leq 2\Big|\mathbb{E}\Big[\sum_{1\leq n'< n< L_{m} }G_{n',n}\,\Big|\,\widetilde{\mathcal{F}}_{\tau_{m}^{-}}\Big]\Big|+\Big|\mathbb{E}\Big[\sum_{1\leq n< L_{m} }G_{n,n}\,\Big|\,\widetilde{\mathcal{F}}_{\tau_{m}^{-}}\Big]\Big|,   $$
where $G_{n',n}$ has the form
\begin{align}\label{Budapest}
G_{n',n}:=& \chi\Big(\Delta \tau_{m}+\tau_{m}>\Gamma_{m,n}\Big) \big(|K_{\underline{\Gamma_{m,n^{\prime}} } }|-  |K_{\Gamma_{m,n^{\prime}-1} }|\big)\, \big(|K_{\underline{\Gamma_{m,n} } }|-  |K_{\Gamma_{m,n-1} }|\big) \nonumber \\ &\times\mathbb{E}\Big[ \big(\Delta \tau_{m}+\tau_{m}-\Gamma_{m,n}\big)^{2}+(\Gamma_{m,n}-\Gamma_{m,n'})\big(\Delta \tau_{m}+\tau_{m}-\Gamma_{m,n}\big)\, \Big|\,\mathcal{F}_{\underline{\Gamma_{m,n} } }   \Big]. \, 
\end{align}

The conditional expectation in~(\ref{Budapest}) can be  written as
\begin{align}\label{Gospel}
\mathbb{E}\Big[  \big(\Delta \tau_{m}+\tau_{m}-&\Gamma_{m,n}\big)^{2}+(\Gamma_{m,n}-\Gamma_{m,n'})\big(\Delta \tau_{m}+\tau_{m}-\Gamma_{m,n}\big) \Big|\,\mathcal{F}_{\underline{\Gamma_{m,n} } } \Big]\nonumber  \\  = & \Big(2\nu^{-2} K^{2}_{\tau_{m} }+\nu^{-1} (\Gamma_{m,n}-\Gamma_{m,n'})\, |K_{\tau_{m} }|\Big)  \nonumber \\ &+ \Big( 2\nu^{-2}\big( K^{2}_{\Gamma_{m,n} }- K^{2}_{\tau_{m} } \big)  +\nu^{-1} \big(\Gamma_{m,n}-\Gamma_{m,n'}\big)\, \big( |K_{\Gamma_{m,n}}|- |K_{\tau_{m} }| \big)\Big)\nonumber \\ &+\frak{Er}, 
\end{align}
where the error term is bounded by 
$$|\frak{Er}|\leq  C\big(\frac{3}{2}\big)^{1+\zeta} |K_{\tau_{m} }|^{1+\zeta}+ C\big(\frac{3}{2}\big)^{\zeta}\big|\Gamma_{m,n}-\Gamma_{m,n'}\big|\,|K_{\tau_{m} }|^{\zeta},  $$
for $C,\zeta>0$ defined as before.  This is analogous to~(\ref{Belief}) and~(\ref{Faith}) in  (i$'$), and follows by  Part (1) of Lem.~\ref{FineMess} and  because $|K_{\underline{\Gamma_{m,n}}}|\leq \frac{3}{2}|K_{\tau_{m}}|$ for  $\Gamma_{m,n}<\Delta\tau_{m}+\tau_{m}$.  I will not go through the details for the bounds of the three terms on the right side of~(\ref{Budapest}), since they admit the same procedure to break them down as applied before (except, for instance, the inequality $\sum_{n=1}^{\infty}\mathbb{P}[X>n]\leq \mathbb{E}[X]$ is replaced by $\sum_{n=1}^{\infty}n\mathbb{P}[X>n]\leq \mathbb{E}[X^{2}]$). The end result is that $\mathbb{E}\big[ H_{\tau_{m} }^{2}\,\big|\,\widetilde{\mathcal{F}}_{\tau_{m}^{-}}\big]$ is smaller than a constant multiple of $|K_{\tau_{m}}|^{3}$, where the leading term comes from the diagonal sum $n=n'$.  By the same argument in~(\ref{Juice}), the sum of terms $|K_{\tau_{m}}|^{3}$ is bounded by
 $$ t^{-\frac{5}{2}}\mathbb{E}\Big[\sum_{m=1}^{\mathbf{N}_{t}-1}|K_{\tau_{m}}|^{3}\chi\big(\big|K_{\tau_{m}}|\leq \epsilon^{-1}t^{\frac{1}{2}} \big)  \Big]< 2\nu \epsilon^{-2}t^{-\frac{1}{2}},
$$ 
 which goes to zero for large $t$. 

\vspace{.5cm}

\noindent (iii).\hspace{.15cm}  I will first show that there is a vanishing error in replacing the terms $\mathbb{E}\big[ \Delta\tau_{m}\,\big|\,\widetilde{\mathcal{F}}_{\tau_{m}^{-}}\big]$ in the expression by $\nu^{-1} |K_{\tau_{m}}|$.     To bound the difference, I apply  Part (1) of Prop.~\ref{TimeFlip} to get the first and second inequalities below for $C>0$ depending on my choice of $0<\zeta<\frac{1}{2}$:  
\begin{align*}
\mathbb{E}\Big[ \sup_{0\leq s\leq 1} t^{-\frac{5}{4}}&\sum_{m=1}^{\mathbf{N}_{st}-1}  |K_{\tau_{m}}|\, \big|\mathbb{E}\big[ \Delta\tau_{m}\,\big|\,\widetilde{\mathcal{F}}_{\tau_{m}^{-}}\big]-\nu^{-1} |K_{\tau_{m}}   |\big| \Big]\\ &\leq C\mathbb{E}\Big[t^{-\frac{5}{4}}\sum_{m=1}^{\mathbf{N}_{t}-1} |K_{\tau_{m}}|^{1+\zeta}  \Big] \leq  2C\nu \mathbb{E}\Big[t^{-\frac{5}{4}}\sum_{m=1}^{\mathbf{N}_{t}-1}|K_{\tau_{m}}|^{\zeta}   \mathbb{E}\big[ \Delta\tau_{m}\,\big|\,\widetilde{\mathcal{F}}_{\tau_{m}^{-}}\big]  \Big] \\ & \leq \frac{2C\nu }{ t^{\frac{1}{4}}}\mathbb{E}\Big[\sup_{0\leq r\leq t} |K_{r}|^{\zeta}  \Big]\leq \frac{2C\nu}{ t^{\frac{1}{4}}}\mathbb{E}\Big[\sup_{0\leq r\leq t} |K_{r}|^{2}  \Big]^{\frac{\zeta }{2}} \\ & < \frac{4C\nu}{ t^{\frac{1}{4}}}\mathbb{E}\big[ \mathcal{E}_{t}^{2}  \big]^{\frac{\zeta }{2}}=  \frac{4C\nu (\mathbb{E}[\mathcal{E}_{0}^{2}]+ \sigma t )^{\frac{\zeta}{2} }  }{ t^{\frac{1}{4}}} \longrightarrow 0.
\end{align*}
The third inequality follows by removing the nested conditional expectations and  $\sum_{m=1}^{\mathbf{N}_{t}-1}\Delta\tau_{m}<t$. The fourth inequality is Jensen's, and  the fifth  employs $|K_{r}|\leq \mathcal{E}_{r}$ and Doob's maximal inequality to the positive submartingale $\mathcal{E}_{r}$.  The equality is Part (1) of Prop.~\ref{Basics}.    

I am left to bound the sum $\sum_{m=1}^{\mathbf{N}_{st}-1} K_{\tau_{m}} |K_{\tau_{m}}|$.  Let $\mathbf{G}\subset \R_{+}$ be the set of all times $\tau_{m}$ such that  $\tau_{m}\in[\varsigma_{j},\,\varpi_{j+1})$ and $\tau_{m}$ is an even-numbered  flip time following $\varsigma_{j}$ (and this includes $\tau_{m}=\varsigma_{j}$).   Denote the number of sign-flips in the interval  $(\varsigma_{j},\varpi_{j+1})$  by $\mathbf{n}_{j}$.   The sum of terms $ K_{\tau_{m}} |K_{\tau_{m}}|$ can be written as
\begin{align}\label{HakiSak}
\sum_{j=1}^{\Upsilon_{st} }&\sum_{\substack{\tau_{m}\in [\varsigma_{j},\,\varpi_{j+1})\\  \tau_{m+1}\leq st }  } \, K_{\tau_{m}} |K_{\tau_{m}}|\nonumber  \\ & \approx \sum_{\substack{m=1 \\ \tau_{m}\in \mathbf{G} }   }^{\mathbf{N}_{st}-1}\,\Big( K_{\tau_{m}  } |K_{\tau_{m}}|+K_{\tau_{m}+\Delta \tau_{m}  } |K_{\tau_{m}+\Delta \tau_{m} }|\Big)-   \sum_{\substack{j=1, \\ \varsigma_{j}=\tau_{m}  } }^{\Upsilon_{st}-1}\chi(\mathbf{n}_{j}\text{ even}   )K_{\tau_{m}+\Delta \tau_{m}  }|K_{\tau_{m}+\Delta \tau_{m}}|,
\end{align}
 where I have neglected a single extra boundary term $K_{\tau_{m}}| K_{\tau_{m}}|$ that may occur in the last incomplete excursion.

 I can immediately treat the boundary sum on the right side of~(\ref{HakiSak}).  Since $\tau_{m}+\Delta \tau_{m}$ occurs during an incursion, I have that  $|K_{\tau_{m}+\Delta \tau_{m}}|\leq t^{\frac{3}{8}}$, and thus the  first inequality below holds:
$$t^{-\frac{5}{4}}\mathbb{E}\Big[ \sum_{\substack{j=1, \\ \varsigma_{j}=\tau_{m}  }  }^{\Upsilon_{t}-1}|K_{\tau_{m}+\Delta \tau_{m}}|^{2}  \Big]\leq t^{-\frac{1}{2}}\mathbb{E}\big[ \Upsilon_{t} \big]\leq 2\sigma^{\frac{1}{2}}t^{-\frac{3}{8}}.  $$
The second inequality is from~(\ref{Numero}) in the proof of Lem.~\ref{EnergyLemma} and relies on the submartingale upcrossing inequality.  

For the first sum on the right side of~(\ref{HakiSak}), it is convenient to write $K_{\tau_{m} } |K_{\tau_{m}}|+K_{\tau_{m+1}} |K_{\tau_{m+1}}|$ as a sum of   the two terms
\begin{align}\label{KrisKros}
 \left(S(K_{\tau_{m} })+S(K_{\tau_{m}+\Delta\tau_{m}  })   \right) \big|K_{\tau_{m}+\Delta\tau_{m} } \big|^{2}\quad \text{and}\quad  -S(K_{\tau_{m}  })\left( |K_{\tau_{m}+\Delta\tau_{m}}|^{2}-|K_{\tau_{m}}|^{2}     \right). 
 \end{align}  
  I can bound the sum of  terms on the left of~(\ref{KrisKros}) through the inequality 
\begin{align}
\sup_{0\leq s\leq 1}t^{-\frac{5}{4}}\Big|\sum_{\substack{m=1 \\ \tau_{m}\in \mathbf{G} }   }^{\mathbf{N}_{st}-1} \, &\big(S(K_{\tau_{m}  })+S(K_{\tau_{m}+\Delta\tau_{m}  })   \big) \big|K_{\tau_{m}+\Delta\tau_{m}    } \big|^{2}\Big| \nonumber   \\ &\leq  2t^{-\frac{5}{4}}\Big(\sup_{0\leq r\leq t}|K_{r}|^{2}\Big) \sum_{m=1    }^{\mathbf{N}_{t}-1} \chi\big(S(K_{\tau_{m}  })=S(K_{\tau_{m}+\Delta\tau_{m}  })    \big). 
\end{align}
For any $\delta>0$, the factor $t^{-1-\delta}\sup_{0\leq r\leq t}|K_{r}|^{2}$ will go to zero for large $t$.  By definition, $S(K_{\tau_{m}  })=S(K_{\tau_{m}+\Delta\tau_{m}  }) $ occurs  when $|K_{\tau_{m}+\Delta\tau_{m}}|$ jumps out of the interval  $\big[\frac{1}{2}|K_{\tau_{m}  }|,\, \frac{3}{2}|K_{\tau_{m}  }| \big]$. By inserting nested conditional expectations, I have the equality below
\begin{align}\label{Itakeskus}
t^{-\frac{1}{4}+\delta}\mathbb{E}\Big[\sum_{m=1  }^{\mathbf{N}_{t}-1} \chi\big(S(K_{\tau_{m}  })=S(K_{\tau_{m}+\Delta\tau_{m}  })    \big)\Big]& = t^{-\frac{1}{4}+\delta}\mathbb{E}\Big[\sum_{m=1  }^{\mathbf{N}_{t}-1} \mathbb{P}\Big[|K_{\tau_{m}+\Delta\tau_{m} }|\notin \big[\frac{1}{2}|K_{\tau_{m}  }|,\, \frac{3}{2}|K_{\tau_{m}  }|\big]\,\Big|\, \widetilde{\mathcal{F}}_{\tau_{m}  ^{-}} \Big]   \Big]\nonumber \\ & \leq \frac{3C}{8 } \frac{\log(t)}{t^{1-\delta}}   \mathbb{E}\big[\mathbf{N}_{t}   \big],
\end{align}
where the inequality follows since $|K_{\tau_{m}}|\geq t^{\frac{3}{8}} $ for all times $\tau_{m}$ and by Part (4) of Prop.~\ref{TimeFlip}. The expectation $\mathbb{E}[\mathbf{N}_{t}]$ has a rough bound given by routine arguments: 
\begin{align*}
\mathbb{E}\big[\mathbf{N}_{t}-1\big]& \leq t^{-\frac{3}{8}}\mathbb{E}\Big[\sum_{m=1  }^{\mathbf{N}_{t}-1}|K_{\tau_{m}}| \Big] \leq 2\nu t^{-\frac{3}{8}}\mathbb{E}\Big[\sum_{m=1  }^{\mathbf{N}_{t}-1}\mathbb{E}\big[\Delta\tau_{m}\,\big|\,\widetilde{\mathcal{F}}_{\tau_{m}^{-}} \big] \Big]\\ &\leq t^{-\frac{3}{8}}\mathbb{E}\Big[\sum_{m=1  }^{\mathbf{N}_{t}-1}\Delta\tau_{m}\Big]\leq t^{\frac{5}{8}}. 
\end{align*}
The first inequality is from $|K_{\tau_{m}}|>t^{\frac{3}{8}}$.  For the second and third inequalities, I have applied Part (1) of Prop.~\ref{TimeFlip} and $\sum_{m=1  }^{\mathbf{N}_{t}-1}\Delta \tau_{m}\leq t$, respectively.  Hence, the right side of~(\ref{Itakeskus}) will go to zero for any choice of $\delta<\frac{3}{8}$.

The sum of terms 
$$\sum_{\substack{m=1 \\ \tau_{m}\in \mathbf{G} }   }^{\mathbf{N}_{st}-1} \, -S(K_{\tau_{m}  })\big( |K_{\tau_{m}+\Delta\tau_{m}}|^{2}-|K_{\tau_{m}}|^{2} \big)  $$
on the right side of~(\ref{KrisKros}) is easier, since I can approximate $|K_{r}|$ by the submartingale  $\mathcal{E}_{r}=M_{r}+A_{r}$ as  before.  It is convenient to write $|K_{\tau_{m}+\Delta\tau_{m}}|^{2}-|K_{\tau_{m}}|^{2}$ as a sum of $\big(|K_{\tau_{m}+\Delta\tau_{m}}|-|K_{\tau_{m}}|\big)^{2}$ and $2\big(|K_{\tau_{m}+\Delta\tau_{m}}|-|K_{\tau_{m}}|\big)|K_{\tau_{m}}| $ and to treat the two terms separately.

\end{proof}

\begin{proof}[Proof of Lem.~\ref{LemMartApproxQuad}]
I will show Part (1).  Part (2) follows by similar analysis relying on Part (5) of Prop.~\ref{TimeFlip}.  Part (3) is a consequence of the convergence in law of $t^{-\frac{1}{2}}M_{st}$ over $s\in [0,1]$ as $t\rightarrow \infty$ by Thm.~\ref{SubMartCLT}, the Lindberg condition for $t^{-\frac{1}{2}}M_{st}$ in Lem.~\ref{LemLindberg},  and~\cite[Cor.VI.6.7]{Jacod}. 

The quadratic variation of $\mathbf{m}^{(t)}$ has the following form:
 $$
  [\mathbf{m}^{(t)}]_{s}= t^{-\frac{5}{2}}\sum_{m=1}^{\mathbf{N}_{st}}K_{\tau_{m}}^{2}\big(\Delta\tau_{m}- \mathbb{E}\big[\Delta\tau_{m}\,\big|\, \widetilde{\mathcal{F}}_{\tau_{m}^{-}}\big]   \big)^{2} .$$  
For technical convenience, I will work with the expression on the right with the upper summand $\mathbf{N}_{st}$ replaced by $\mathbf{N}_{st}-1$.  The difference is negligible by Lem.~\ref{LemLindberg}.  I will show the following convergences in probability:
\begin{enumerate}[(i).]  
\item  $\sup_{0\leq s\leq 1}\Big|  t^{-\frac{5}{2}}\sum_{m=1}^{\mathbf{N}_{st}-1}  \Big(K_{\tau_{m}}^{2}\big(\Delta\tau_{m}- \mathbb{E}\big[\Delta\tau_{m}\,\big|\, \widetilde{\mathcal{F}}_{\tau_{m}^{-}}\big]   \big)^{2}- \nu^{-1} |K_{\tau_{m} }|^{3}\Delta\tau_{m}  \Big)\Big| \Longrightarrow 0,  $

\item  $\sup_{0\leq s\leq 1}\Big|   \nu^{-1} t^{-\frac{5}{2}}\sum_{m=1}^{\mathbf{N}_{st}-1}  |K_{\tau_{m} }|^{3}\Delta\tau_{m}  - \nu^{-1}\int_{0}^{s}dr\big|t^{-\frac{1}{2} }K_{rt}  \big|^{3} \Big| \Longrightarrow 0. $
\end{enumerate}

\vspace{.4cm}

\noindent (i).  \hspace{.15cm} For small $\epsilon$, $\mathbb{P}\big[\sup_{0\leq r\leq t}|t^{-\frac{1}{2}}K_{r}|>\epsilon^{-1}  \big]$ is small.   Hence, for $\hat{K}_{r}=K_{r}\chi(|K_{r}|\leq \epsilon^{-1})$, 
\begin{align}\label{Littlewood}
\sum_{m=1}^{\mathbf{N}_{st}-1}  K_{\tau_{m}}^{2}\big(\Delta\tau_{m}- \mathbb{E}\big[\Delta\tau_{m}\,\big|\, \widetilde{\mathcal{F}}_{\tau_{m}^{-}}\big]   \big)^{2}
\end{align}
 is typically equal for all $s\in [0,1]$ to the same expression with $K_{r}$ replaced by $\hat{K}_{r}$.  The advantage of working with $\hat{K}_{r}$ is that it has arbitrarily many moments.  The following is a martingale with respect to the filtration $\widetilde{\mathcal{F}}_{s}^{(t)}$:
 \begin{align}\label{Hardy}
 \mathbf{W}_{s}^{(t),\epsilon} :=t^{-\frac{5}{2}}\sum_{m=1}^{\mathbf{N}_{st}}  \hat{K}_{\tau_{m}}^{2}\Big(\big(\Delta\tau_{m}- \mathbb{E}\big[\Delta\tau_{m}\,\big|\, \widetilde{\mathcal{F}}_{\tau_{m}^{-}}\big]   \big)^{2}- \mathbb{E}\Big[\big(\Delta\tau_{m}- \mathbb{E}\big[\Delta\tau_{m}\,\big|\, \widetilde{\mathcal{F}}_{\tau_{m}^{-}}\big]   \big)^{2}\,\Big|\, \widetilde{\mathcal{F}}_{\tau_{m}^{-}}\Big] \Big). 
\end{align}
  Using Part (1) of Prop.~\ref{TimeFlip} for the first two inequalities below,  the second moment of the above martingale is bounded through the inequalities 
\begin{align*}
\mathbb{E}\Big[t^{-5} \sum_{m=1}^{\mathbf{N}_{st}-1} & \hat{K}_{\tau_{m}}^{4}\mathbb{E}\Big[\Big(\big(\Delta\tau_{m}- \mathbb{E}\big[\Delta\tau_{m}\,\big|\, \widetilde{\mathcal{F}}_{\tau_{m}^{-}}\big]   \big)^{2}- \mathbb{E}\big[\big(\Delta\tau_{m}- \mathbb{E}\big[\Delta\tau_{m}\,\big|\, \widetilde{\mathcal{F}}_{\tau_{m}^{-}}\big]   \big)^{2}\,\big|\, \widetilde{\mathcal{F}}_{\tau_{m}^{-}}\big] \Big)^{2}\,\Big|\,\widetilde{\mathcal{F}}_{\tau_{m}^{-}}\Big]   \Big] \\ & \leq \frac{32}{ \nu^{4}} \mathbb{E}\Big[t^{-5} \sum_{m=1}^{\mathbf{N}_{st}-1}  \hat{K}_{\tau_{m}}^{8}   \Big]\\ &\leq  \frac{64}{ \nu^{3}}\mathbb{E}\Big[t^{-5} \sum_{m=1}^{\mathbf{N}_{st}-1}  |\hat{K}_{\tau_{m}}|^{7}\mathbb{E}\big[\Delta\tau_{m}  \,\big|\,\widetilde{\mathcal{F}}_{\tau_{m}^{-}} \big]   \Big]\\ & \leq \frac{64}{\nu^{3}}\epsilon^{-7}t^{-\frac{1}{2}}\longrightarrow 0.
\end{align*}
 I have multiplied the bounds in the first two inequalities by $2$ to cover the error terms for the approximations of the moments  $\mathbb{E}\big[(\Delta\tau_{m})^{m}\,\big|\, \widetilde{\mathcal{F}}_{\tau_{m}^{-}}\big]$ from Prop. 6.3.  The last inequality follows by $|\hat{K}_{\tau_{m}}|\leq \epsilon^{-1}t^{\frac{1}{2}}$ and removing the nested conditional expectations and  $\sum_{m=1}^{\mathbf{N}_{t}-1}\Delta\tau_{m}<t$.  By Doob's maximal inequality,   
$ \mathbb{E}\big[\sup_{0\leq s\leq 1} \big|\mathbf{W}_{s}^{(t),\epsilon}\big|^{2} \big]$  tends to zero for any fixed $\epsilon$.  By similar applications of Part (1) of Prop.~\ref{TimeFlip} as above, the difference 
$$\mathbb{E}\Big[ \sup_{0\leq s\leq 1} \Big| t^{-\frac{5}{2}}\sum_{m=1}^{\mathbf{N}_{st}-1}  \hat{K}_{\tau_{m}}^{2}\Big( \mathbb{E}\Big[\big(\Delta\tau_{m}- \mathbb{E}\big[\Delta\tau_{m}\,\big|\, \widetilde{\mathcal{F}}_{\tau_{m}^{-}}\big]   \big)^{2}\,\Big|\, \widetilde{\mathcal{F}}_{\tau_{m}^{-}}\Big]-\nu^{-1}|\hat{K}_{\tau_{m}}|  \mathbb{E}\big[\Delta\tau_{m}\,\big|\, \widetilde{\mathcal{F}}_{\tau_{m}^{-}}\big]   \Big) \Big| \Big]=\mathit{O}(t^{\zeta-\frac{1}{2}})$$
tends to zero.  The difference between 
$$\nu^{-1} t^{-\frac{5}{2}}\sum_{m=1}^{\mathbf{N}_{st}}  |\hat{K}_{\tau_{m}}|^{3}\mathbb{E}\big[\Delta\tau_{m}\,\big|\, \widetilde{\mathcal{F}}_{\tau_{m}^{-}}\big] \quad\text{and}\quad \nu^{-1} t^{-\frac{5}{2}}\sum_{m=1}^{\mathbf{N}_{st}}  |\hat{K}_{\tau_{m}}|^{3}\Delta\tau_{m}$$ is a $\widetilde{\mathcal{F}}_{s}^{(t)}$-martingale and tends to zero by a similar (but simpler) argument as for~(\ref{Hardy}).  Finally, the process $\nu^{-1} t^{-\frac{5}{2}}\sum_{m=1}^{\mathbf{N}_{st}-1}  |\hat{K}_{\tau_{m}}|^{3}\Delta\tau_{m} $ is equal to the same expression with $|\hat{K}_{\tau_{m}}|$ replaced by $|K_{\tau_{m}}|$ with probability arbitrarily close to one as $\epsilon\rightarrow 0$.

\vspace{.5cm} 
 
 \noindent (ii).\hspace{.15cm} First, I show  the sum over $|K_{\tau_{m}}|^{3}\Delta\tau_{m}$ can be replaced by a sum over $\int_{\tau_{m}}^{\tau_{m}+\Delta \tau_{m}}dr |K_{r}|^{3}$.  This approximation reduces to some martingale analysis using that $|K_{r}|\approx \mathcal{E}_{r}$:  
\begin{align}\label{Trip}
\Big| \nu^{-1} t^{-\frac{5}{2}}\sum_{m=1}^{\mathbf{N}_{st}-1}& \int_{\tau_{m}}^{\tau_{m}+\Delta\tau_{m}}dr |K_{r}|^{3}-\nu^{-1} t^{-\frac{5}{2}}\sum_{m=1}^{\mathbf{N}_{st}-1}  |K_{\tau_{m}}|^{3}\Delta\tau_{m}    \Big|\nonumber \\ &\leq \nu^{-1}  \sup_{0\leq r\leq t}|K_{r}|^{2} \Big( t^{-\frac{5}{2}}\sum_{m=1}^{\mathbf{N}_{st}-1} \int_{\tau_{m}}^{\tau_{m}+\Delta\tau_{m}}dr\big| |K_{r}|-|K_{\tau_{m}}| \big|\Big). 
\end{align} 
For any $\delta>0$, the random variables $ t^{-1-\delta}\sup_{0\leq r\leq t}|\hat{K}_{r}|^{2}$  tends to zero, since $t^{-\frac{1}{2}}|K_{st}|$ converges to the absolute value of a Brownian motion in the uniform metric.   The quantity $|K_{r}| $ in the sum on the right side of~(\ref{Trip}) can be replaced by $\mathcal{E}_{r}$ with an error $\mathit{O}(|K_{r}|^{-1})$. Moreover, I have the inequality
\begin{align}\label{Delic}
 t^{-\frac{3}{2}+\delta}\mathbb{E}\Big[\sup_{0\leq s\leq 1} \sum_{m=1}^{\mathbf{N}_{st}-1} \int_{\tau_{m}}^{\tau_{m}+\Delta\tau_{m}}dr\big| \mathcal{E}_{r}-\mathcal{E}_{\tau_{m}} \big| \Big] \leq &  t^{-\frac{3}{2}+\delta}\mathbb{E}\Big[ \sum_{m=1}^{\mathbf{N}_{t}-1} \mathbb{E}\Big[\int_{\tau_{m}'}^{\tau_{m}''}dr \big|M_{r}-M_{\tau_{m}}\big|  \,\Big|\, \mathcal{F}_{\tau_{m}'}\Big] \Big]
\nonumber \\ &+ t^{-\frac{3}{2}+\delta}\mathbb{E}\Big[ \sum_{m=1}^{\mathbf{N}_{t}-1} \int_{\tau_{m}}^{\tau_{m}+\Delta\tau_{m} }dr \big|A_{r}-A_{\tau_{m}}\big| \Big],
\end{align}
where $\tau_{n}',\tau_{n}''$ are the Poisson times following $\tau_{n}$ and  $\tau_{m}+\Delta\tau_{m}$, respectively.  Both $\tau_{n}',\tau_{n}''$ are hitting times with respect to the  filtrations $\mathcal{F}_{r},\widetilde{\mathcal{F}}_{r}$.  For the inequality above, I have used the triangle inequality with $\mathcal{E}_{r}=M_{r}+A_{r}$ and introduced nested conditional expectations for the martingale term.  Changing the lower bounds of the integration from $\tau_{m}$ to $\tau_{m}'$ for the martingale term doesn't change the value, since $M_{r}=M_{\tau_{m}}$ for $r\in [\tau_{m},\tau_{m}')$.  

For the drift term, I can use that $\sigma= \frac{d}{ds}\langle M,M\rangle_{r}+2\mathcal{E}_{s}\frac{d}{dr}A_{r}$ to get the bound $A_{r}-A_{\tau_{m}}\leq \int_{\tau_{m}}^{r}dv\frac{2\sigma}{\mathcal{E}_{v}}$. However, since $\mathcal{E}_{v}\geq t^{-\frac{3}{8}}$ over the excursion periods, the first inequality below holds:
$$
\mathbb{E}\Big[ \sum_{m=1}^{\mathbf{N}_{t}-1} \int_{\tau_{m}}^{\tau_{m}+\Delta\tau_{m} }dr \big|A_{r}-A_{\tau_{m}}\big| \Big]\leq  \sigma t^{-\frac{3}{8}}\mathbb{E}\Big[ \sum_{m=1}^{\mathbf{N}_{t}-1} (\Delta\tau_{m} )^{2}\Big]\leq \frac{2}{\nu }  t^{\frac{5}{8}}\mathbb{E}\Big[ \sup_{0\leq r\leq t}|K_{r}|\Big]=\mathit{O}(t^{\frac{9}{8}}).  
$$
The second inequality uses that $\sum_{m=1}^{\mathbf{N}_{t}-1}\Delta\tau_{m}<t$ and Part (1) of Prop.~\ref{TimeFlip} through the 
 standard tricks:
\begin{align*}
\mathbb{E}\Big[ \sum_{m=1}^{\mathbf{N}_{t}-1} (\Delta\tau_{m} )^{2}\Big]&=\mathbb{E}\Big[ \sum_{m=1}^{\mathbf{N}_{t}-1}  \mathbb{E}\big[ \big(\Delta\tau_{m}\big)^{2}  \,\big|\,\widetilde{\mathcal{F}}_{\tau_{m}^{-}} \big] \Big]\leq \frac{4}{\nu^{2}}\mathbb{E}\Big[ \sum_{m=1}^{\mathbf{N}_{t}-1}  K_{\tau_{m}}^{2} \Big]\leq  \frac{8}{\nu}\mathbb{E}\Big[ \sum_{m=1}^{\mathbf{N}_{t}-1} K_{\tau_{m}} \mathbb{E}\big[ \Delta\tau_{m}  \,\big|\,\widetilde{\mathcal{F}}_{\tau_{m}^{-}}\big]  \Big]\\ &= \frac{8}{\nu}\mathbb{E}\Big[ \sum_{m=1}^{\mathbf{N}_{t}-1} \Delta\tau_{m} K_{\tau_{m} } \Big]\leq \frac{8 t}{\nu} \mathbb{E}\Big[ \sup_{0\leq r\leq t}|K_{r}|\Big], 
\end{align*}
where the first and second inequalities hold for large enough $t$.  Thus, the second line of~(\ref{Delic}) vanishes  as $t\rightarrow \infty$.  The martingale term is handled by similar arguments.  


By the above, I can work with the sum of the integrals $\int_{\tau_{m}}^{\tau_{m}+\Delta \tau_{m}}dr |K_{r}|^{3}$.  By the definition of the low energy incursions, the momentum has the upper bound $|K_{r}|\leq 2t^{\frac{3}{8}}$ and thus
$$\sup_{0\leq s\leq 1} \Big| \nu^{-1} t^{-\frac{5}{2}}\sum_{m=1}^{\mathbf{N}_{st}-1} \int_{\tau_{m}}^{\tau_{m}+\Delta\tau_{m}}dr |K_{r}|^{3}- \nu^{-1} t^{-\frac{5}{2}}\int_{0}^{\tau_{\mathbf{N}_{st} } }dr |K_{r}|^{3} \Big|\leq \frac{8}{\nu } t^{-\frac{3}{8}}.  $$ 
 The final remainder $\sup_{0\leq s\leq 1}t^{-\frac{5}{2}}\int_{\tau_{\mathbf{N}_{st} } }^{st}dr\, |K_{r}|^{3}$  for the difference between $t^{-\frac{5}{2}}\int_{0}^{\tau_{\mathbf{N}_{st} } }dr |K_{r}|^{3}$ and  the final expression $t^{-\frac{5}{2}}\int_{0 }^{st}dr |K_{r}|^{3}$ is smaller than  $(\sup_{\tau_{m}\leq t}\Delta\tau_{m}) \sup_{0\leq r\leq t}|K_{r}|^{3} $, which vanishes  by the same argument as for the proof of the Lindberg condition for $\mathbf{m}_{s}^{(t)}$ in Lem.~\ref{LemLindberg}.

\end{proof}

\section*{Acknowledgments}
I thank Mark Fannes and Christian Maes for useful discussions.   This work is supported by the Belgian Interuniversity Attraction Pole P6/02, the Marie Curie funded Research training network project MRTN-CT-2006-035651, Acronym CODY, of the European Commission, and European Research Council grant No. 227772.  I also benefited from NSF FRG grant DMS-0757581 during a visit to the Department of Mathematics at U.C. Davis.

\begin{appendix}

\section{Existence and uniqueness of the quantum dynamical semigroup}\label{AppendixSemigroup}

Lindblad equations with unbounded generators can pose technical difficulties regarding the existence, construction, and uniqueness of  their corresponding quantum dynamical semigroups.  This situation is analogous to that for  Kolmogorov equations in classical Markovian dynamics.  A quantum problem, however, would be considered ``solved" if it was somehow reduced to a classical problem, such as in the case of models that are fully translation invariant~\cite{Holevo}.

As before let $\textup{D}(H)\subset L^{2}(\R)$ denote the domain of the Hamiltonian $H$.  The mathematical definition of the quantum dynamics will require the form generator $ \mathcal{L}:\textup{D}(H)\times \mathcal{B}_{1}\big(L^{2}(\R)\big)\times \textup{D}(H)\rightarrow \C$ given by
 \begin{align*} 
 \mathcal{L}\big(\psi_{1};\, \rho;\, \psi_{2}\big):=\big\langle \Big(  \frac{\textup{i}}{\lambda}H-2^{-1}\Psi^{*}(I)\Big) \psi_{1} \big| \rho\,\psi_{2}\big\rangle
+ \big\langle  \psi_{1} \big|  \rho \Big(  \frac{\textup{i}}{\lambda}H-2^{-1}\Psi^{*}(I)\Big)\psi_{2}\big\rangle  + \big\langle \psi_{1};\,\Psi(\rho) \psi_{2}\big\rangle   .  
\end{align*}
The form generator is designed  to draw the operation of the unbounded terms away from $\rho$ and on to the vectors $\psi_{1},\psi_{2}$, which have a restricted domain. In my case, the map $\Psi$ has finite operator norm, since by the complete positivity of $\Psi$, its operator norm is equal to the norm of $\Psi^{*}(I)\in \mathcal{B}\big(L^{2}(\R)\big)$.  Conveniently, $\Psi^{*}(I)=\mathcal{R}I$ has norm $\mathcal{R}$.  If $\Psi$ were an unbounded map, I would have to look for a convenient Kraus decomposition $\Psi(\rho)=\sum_{j}A_{j}^{*}\rho A_{j} $, and the last term  above would be $\sum_{j}  \big\langle  A_{j}\psi_{1};\,\rho\, A_{j} \psi_{2}\big\rangle $.

 A semigroup of maps $\Phi_{t}$ on $\mathcal{B}_{1}\big(L^{2}(\R)\big)$ is said to be \textit{conservative} if $\Tr[\Phi_{t}(\rho)]=\Tr[\rho]$ for all $t\in \R_{+}$ and positive $\rho \in \mathcal{B}_{1}\big(L^{2}(\R)\big) $.  Alternatively, this can be expressed by the adjoint semigroup as $\Phi_{\lambda,t}^{*}(I)=I$.  The representation of the semigroup 
 $\Phi_{\lambda,t}$  in (2) of Lemma~\ref{LindbladTech}  is a pseudo-Poisson property in analogy with classical Markovian semigroups~\cite[X.1]{Feller}.  The map $\Phi_{\lambda,t}$ is an expectation over the Hamiltonian flow  interrupted at discrete random times by the operation of a transition map $\mathbf{T}:=\mathcal{R}^{-1}\Psi$, where the times occur according to a Poisson clock with rate $\mathcal{R}$.

\begin{lemma}\label{LindbladTech}  There is a unique strongly continuous, conservative semigroup of maps $\Phi_{\lambda,t}$ on $\mathcal{B}_{1}\big(L^{2}(\R)\big)$ satisfying
\begin{align}\label{Blahh}
\langle \psi_{1}|\Phi_{\lambda,t}(\rho)\,\psi_{1}\rangle= \langle \psi_{1}|\rho\,\psi_{1}\rangle +\int_{0}^{t}dr\,  \mathcal{L}\big(\psi_{1}; \Phi_{\lambda,r}(\rho); \psi_{2} \big)  
\end{align}
for all $\psi_{1},\psi_{2}\in \textup{D}(H)$ and $\rho \in \mathcal{B}_{1}\big(L^{2}(\R)\big)$.  
\begin{enumerate}
\item The semigroup $\Phi_{\lambda,t}$ can be written in the form $\Phi_{\lambda,t}(\rho)=  \mathbb{E}\big[U_{\lambda,t}(\xi)\rho U_{\lambda,t}^{*}(\xi) \big]$,
where $U_{\lambda,t}(\xi)= e^{-\frac{\textup{i}(t-t_{n} ) }{\lambda}H}e^{\textup{i}v_{n}X}\cdots    e^{-\frac{i(t_{2}-t_{1}) }{\lambda}H} e^{\textup{i}v_{1}X}$, the expectation is with respect to a L\'evy process with rate density $j(v)$, and $\xi=\big(t_{1},v_{1}; \dots   ;t_{n},v_{n}\big) $ is the realization of the process over the time interval $[0,t]$.  
\item
Alternatively, $\Phi_{\lambda,t}(\rho)=  \mathbb{E}\big[\Phi_{t,\xi}^{(\lambda)}(\rho)\big]$, where the expectation is with respect to a Poisson clock with rate $\mathcal{R}$, the Poisson times over the interval $[0,t]$ are $\xi=(t_{1},\dots,t_{n})$, and $\Phi_{t,\xi}^{(\lambda)}:\mathcal{B}_{1}\big(L^{2}(\R)\big)$ is defined by
$$\Phi_{t,\xi}^{(\lambda)}(\rho):=\mathcal{R}^{-n}    e^{-\frac{\textup{i}(t-t_{n})}{\lambda}H}  \Psi(   \cdots  e^{-\frac{\textup{i}(t_{2}-t_{1})}{\lambda}H}\Psi(e^{-\frac{\textup{i}t_{1}}{\lambda}H}\rho e^{\frac{\textup{i}t_{1}}{\lambda}H})e^{\frac{\textup{i}(t_{2}-t_{1})}{\lambda}H}\cdots)e^{\frac{\textup{i}(t-t_{n })}{\lambda}H}.  $$
 .   

\end{enumerate}

\end{lemma}

\begin{proof}
By~\cite{Davies}, solving the integral equation~(\ref{Blahh}) is  equivalent to solving
\begin{align*}
\langle \psi_{1}|\Phi_{\lambda,t}(\rho)\,\psi_{1}\rangle= e^{-\mathcal{R}t}\langle  e^{\frac{\textup{i}t}{\lambda} H}\psi_{1}|\rho\,e^{\frac{\textup{i}t}{\lambda} H}\psi_{1}\rangle +\int_{0}^{t}dr\,  e^{-\mathcal{R}(t-r)}\Big\langle e^{\frac{\textup{i}(t-r)}{\lambda} H}\psi_{1};  \Psi\big(\Phi_{\lambda,r}(\rho)\big) e^{\frac{\textup{i}(t-r)}{\lambda} H}\psi_{2} \Big\rangle . 
\end{align*}
In most unbounded cases, there is a technical issue in checking that $A_{\lambda}:-\frac{\textup{i}}{\lambda}H-2^{-1}\Psi^{*}(I)$ defines an $m$-accretive operator over an appropriate domain, and therefore generates a contractive semigroup on $L^{2}(\R)$.  However, since $ \Psi^{*}(I)=\mathcal{R}I$, this is trivial given that $H$ is self-adjoint, and I can factor $e^{tA_{\lambda}}=e^{-t\frac{\mathcal{R}}{2} }e^{ -\frac{\textup{i}t}{\lambda}H }  $.  Since $\Psi$ and $ e^{\frac{\textup{i}t}{\lambda} H}$ are bounded,  a  solution to the above integral equation can be constructed by the Dyson series
\begin{align}\label{Pygmy}
 \Phi_{\lambda,t}=e^{-t\mathcal{R}}\sum_{n=0}^{\infty}\mathcal{R}^{n}\int_{0\leq t_{1} \dots t_{n}\leq t} \Phi_{t,\xi}^{(\lambda)},  
 \end{align}
 where $\Phi_{t,\xi}^{(\lambda)}$ and $\xi$ are defined as in statement (2) of the lemma.  The semigroup  $\Phi_{\lambda,t}$ is strongly continuous, since $\Psi$ is bounded and the group $e^{ -\frac{\textup{i}t}{\lambda}H } $ is strongly continuous.  To check conservativity, it can be computed that $\Phi_{\lambda,t}^{*}(I)=I$ using  $\Psi^{*}(I)=\mathcal{R}I$ and that $e^{ -\frac{\textup{i}t}{\lambda}H }$ is unitary.

  The summation on the right side of~(\ref{Pygmy}) is equal to $ \mathbb{E}\big[\Phi_{t,\xi}^{(\lambda)}(\rho)\big]$, where the expectation is with respect to the a Poisson process with rate $\mathcal{R}$.  The other stochastic representation for $\Phi_{\lambda,t}$ is obtained by expanding  $\Psi$ in its integral form.  The resulting integrals can be commuted, because the integrands are completely positive maps.

\end{proof}

\section{The fiber decomposition for one-dimensional periodic Schr\"odinger equations  }\label{AppendixFiber}

The reader is directed to~\cite{Reed} for a more detailed discussion of the structure of periodic Schr\"odinger equations.   A Schr\"odinger Hamiltonian $P^{2}+V(X)$ with a period-$2\pi$ potential satisfies 
$$  e^{\textup{i}2\pi P} (P^{2}+V(X))e^{-\textup{i}2\pi P}=P^{2}+V(X+2\pi)=P^{2}+V(X).   $$
Commuting with $e^{\textup{i}2\pi P}$ implies that  $P^{2}+V(X)$ must have  invariant spaces corresponding to the spectral values for $e^{\textup{i}2 \pi P}$ (i.e. the unit circle in $\C$).   The Hilbert space $\mathcal{H}=L^{2}(\R)$ admits a fiber decomposition
\begin{align}\label{FiberInt}
\mathcal{H}= \int_{ [-\frac{1}{2},\frac{1}{2}) }^{\oplus}d\phi\,\calH_{\phi},  \quad  \quad \quad  \mathcal{H}_{\phi}\cong L^{2}\big([-\pi,\pi) \big), 
\end{align}
with fiber-maps sending $f\in L^{2}(\R)$ to $[f]_{\phi}\in L^{2}\big([-\pi,\pi) \big)$ and formally defined according to the partial Fourier transform 
$$[f]_{\phi}(x)=\frac{1}{(2\pi)^{\frac{1}{2}}}\sum_{n\in \Z}e^{-\textup{i}2\pi n \phi}f(x+2\pi n),\quad \quad x\in [-\pi,\pi).       $$
From the Fourier transform formula above,  the eigen space property for  $e^{\textup{i}2\pi P}$ is clear: $ [e^{\textup{i}2\pi P}f]_{\phi}(x)=e^{ \textup{i}2\pi \phi } [f]_{\phi}   $.  The parameter $\phi$ is called the \textit{quasimomentum} or the \textit{crystal momentum}, and its domain $[-\frac{1}{2},\frac{1}{2})$ is the \textit{first Brillouin zone}.  The operation of the Hamiltonian $H$ on the $\phi$-fiber is   
$  [Hf]_{\phi}= H_{\phi}[f]_{\phi}$ for   
$$H_{\phi}= -\big(\frac{d^{2}}{dx^{2}}\big)_{\phi}+V(x),$$
where  $V(x)$ is interpreted as a multiplication operator on $L^{2}\big([-\pi,\pi) \big)$, and
 $(\frac{d^{2}}{dx^{2}})_{\phi}$ is the Laplacian with boundary conditions
$$g(-\pi )=e^{\textup{i}2\pi\phi}g(\pi)\quad \quad\text{and} \quad \quad \frac{d g}{dx}(-\pi)=e^{\textup{i}2\pi\phi}\frac{dg}{dx}(\pi). $$

 For each $\phi\in [-\frac{1}{2},\frac{1}{2}) $, the self-adjoint operator $H_{\phi}$ has compact resolvent.  The eigenvalues are  non-degenerate for $\phi \neq -\frac{1}{2},0 $  and are labeled progressively as $E_{n,\phi}$ by a parameter $n\in \mathbb{N}$ called the \textit{band index}.   When $\phi \neq -\frac{1}{2},0 $, the pair $ (n,\phi)$ is related in the extended-zone scheme to a parameter $k\in \R-\frac{1}{2}\Z$ through the relations 
\begin{align}\label{Extended}
k=\phi\,\textup{mod}\, 1, \hspace{1cm}\text{and}\hspace{1cm}   n= \left\{  \begin{array}{cc}  2|k-\phi|    &  S(k)=S(\phi) ,  \\  \quad & \quad \\   2|k-\phi|-1   &  S(k)=-S(\phi),    \end{array} \right.  
\end{align}
where $S:\R\rightarrow \{\pm 1\}$ is the sign function. The assignment of $ (n,\phi)$ for  $\phi \in \{-\frac{1}{2},0\}$ is a matter of convention that I am not concerned with, since the set $ \{-\frac{1}{2},0\}$ has measure zero in the direct integral~(\ref{FiberInt}).  The dispersion relation $E:\R\rightarrow \R_{+}$ is defined as $E(k)=E_{n,\phi}$ for $\phi \neq -\frac{1}{2},0$ and $k$ related to $(n,\phi)$ as above, and I hold the convention that $E(k)$ is symmetric and left-continuous for $k\geq 0$.

For $\phi \neq -\frac{1}{2},0$,  let $\psi_{n,\phi}$ be a normalized eigenvector for $H_{\phi}$ with eigenvalue $E_{n,\phi}$.    I can pick the eigenvectors    $\psi_{n,\phi}$ to  vary continuously (and, in fact, smoothly~\cite[Thm.XIII.90]{Reed}) as elements in $L^{2}\big([-\pi,\pi)\big)$ for $\phi \in (-\pi,0)$ and $\phi\in (0,\pi)$.   Given $f\in \mathcal{H}$, I can assign an extend-zone scheme representation $\widehat{f}\in L^{2}(\R)$ through
$$   \widehat{f}(k)= \langle \psi_{n,\phi}| [f]_{\phi}\rangle       , \hspace{2cm}   k \in \R-\frac{1}{2}\Z,   $$
where $n$,$\phi$ are determined by $k$ as above. Again, the assignment of $\widehat{f}(k)$ for $k\in \frac{1}{2}\Z$ is arbitrary.  In analogy with the position and momentum operators, I can define a self-adjoint operator  $P_{\scriptscriptstyle{Q}}$ that acts on the domain 
$ \big\{ f\in \mathcal{H}\,|\, \int_{\R}dk|\widehat{f}(k)|^{2}<\infty\big\} $ as multiplication in the extended-zone scheme representation
$$   \widehat{(P_{\scriptscriptstyle{Q}}f)}(k)=k\widehat{f}(k).       $$
By standard operator calculus, I can define functions of  $P_{\scriptscriptstyle{Q}}$, and the Hamiltonian is given by $H=E(P_{\scriptscriptstyle{Q}})$.  For $k\in \R-\frac{1}{2}\Z$, I define a tempered distribution $| k\rangle_{\scriptscriptstyle{Q}} $  such that for an element $f\in \mathcal{S}(\R)$ in Schwartz space, 
\begin{align*}
\langle  f| k\rangle_{\scriptscriptstyle{Q}} := & \sum_{n}e^{\textup{i}2\pi n\phi}\int_{[-\pi,\pi)}dx   \overline{f}(x+2\pi n)  \psi_{n,\phi}( x  )  \\ =&\int_{[-\pi,\pi)}dx\overline{[f]}_{\phi}(x)\psi_{n,\phi}( x  )  :=\overline{\widehat{f}(k)},  
\end{align*}
where $(n,\phi)$ is determined  by $k$.  In the usual senses, I have the formal relations ${  }_{\scriptscriptstyle{Q}}\langle k'| k\rangle_{\scriptscriptstyle{Q}} =\delta(k'-k)$ and $P_{\scriptscriptstyle{Q}}| k\rangle_{\scriptscriptstyle{Q}}=k | k\rangle_{\scriptscriptstyle{Q}}$.

\subsection{Eigenket conventions for the Dirac comb }

There remains much choice in the phase for the eigenvectors $\psi_{n,\phi}\in L^{2}\big([-\pi,\pi)\big)$ determining the extend-zone scheme of the last section.  For the Dirac comb, I will fix the definition with a specific closed expression for the eigenvectors.  With the correspondence of $(n,\phi)\in \mathbb{N}\times (-\frac{1}{2},\frac{1}{2})$ with $k\in \R-\frac{1}{2}\Z$ in~(\ref{Extended}), I will replace the subscript by $k$: $\widetilde{\psi}_{k}:=\psi_{n,\phi}$.  The eigenfunctions $\widetilde{\psi}_{k}\in L^{2}\big([-\pi,\pi)\big)$ for $k\in \R-\frac{1}{2}\Z$ are given by
\begin{align*}
\widetilde{\psi}_{k}(x)= N_{k}^{-\frac{1}{2}}\left\{  \begin{array}{cc} \frac{e^{\textup{i}2\pi (\mathbf{q}(k)-k)  }-1    }{ e^{\textup{i}2\pi(k +\mathbf{q}(k))}  -1 } e^{-\textup{i} x \mathbf{q}(k) }+e^{\textup{i}2\pi( \mathbf{q}(k)-k) }  e^{\textup{i} x \mathbf{q}(k)  } &  -\pi\leq x\leq 0  ,  \\  \quad & \quad \\ \frac{e^{\textup{i}2\pi(\mathbf{q}(k)-k)  }-1     }{1-e^{-\textup{i}2 \pi(k +\mathbf{q}(k))} }e^{-\textup{i} x \mathbf{q}(k)  }+   e^{\textup{i}x  \mathbf{q}(k)  }     & 0\leq x< \pi ,   \end{array} \right.  
\end{align*}
where $N_{k}>0$ is a normalization, and $  \mathbf{q}:\R\rightarrow \R $ is defined as in~(\ref{Energies}).  For $|k|\gg 1$, I approximately have that $\mathbf{q}(k)\approx k$, and the Bloch function $\widetilde{\psi}_{k}$ is approximately the plane wave  $(2\pi)^{-\frac{1}{2}}e^{\textup{i}xk}$, except when $k+q(k)\approx 2k$ is near an integer.

\subsection{The diagonal in the extended-zone scheme representation}\label{SecDiaTech} 

A density matrix $\rho\in \mathcal{B}_{1}\big(\mathcal{H}\big)$ determines a probability density $[\rho]_{\scriptscriptstyle{D}}\in L^{1}(\R)$ corresponding to the distribution in the extended-zone scheme variable.  Intuitively, this is given by  the diagonal of the integral kernel $\rho(k_{1},k_{2}):={ }_{\scriptscriptstyle{Q}}\langle k_{1}|\rho |k_{2}\rangle_{\scriptscriptstyle{Q}}$, although kernels are only defined a.e. $\R\times \R$ for Hilbert-Schmidt operators, so this does not offer a rigorous  definition without some additional condition on the kernel such as continuity.  For a rigorous definition,  notice that there is a unique probability measure $\mu_{\rho}$ on $\R$  such that  for all $g\in L^{\infty}(\R)$
$$ 
 \Tr\big[  g(P_{\scriptscriptstyle{Q}}) \rho \big]  =   \int_{\R}d\mu_{\rho}(k)\,g(k).      $$
This follows by the Riesz representation theorem, since the left side is positive for $g\geq 0$, bounded in absolute value by $\|g\|_{\infty}$, and equal to one for $g=1_{\R}$.  By the continuity of the spectrum of $P_{\scriptscriptstyle{Q}}$, the measure must be continuous with respect to Lebesgue measure, and I denote the Radon-Nikodym derivative of $\mu_{\rho}$ by $[\rho]_{\scriptscriptstyle{D}}\in L^{1}(\R)$.

\subsection{Invariant fibers for the Lindblad dynamical semigroup  }

The rate of kicks from the gas is invariant of the spatial location of the particle.  The total dynamics is thus invariant under spatial shifts by $2\pi$.  For the dynamical maps $\Phi_{\lambda, t}:\mathcal{B}_{1}\big(L^{2}(\R)\big)$, this feature is expressed as the covariance 
\begin{align}\label{2PIShift}
\Phi_{\lambda,t}\big(e^{\textup{i}2\pi P }\rho e^{-\textup{i} 2\pi P }    \big)= e^{\textup{i}2\pi P }\Phi_{\lambda, t}(\rho)e^{-\textup{i} 2\pi P }. 
\end{align}
Not surprisingly, this implies that the Banach space $\mathcal{B}_{1}\big(L^{2}(\R)\big)$ decomposes into invariant fibers indexed by the Brillouin zone.  This can be understood formally by the statement that the kernel values $  { }_{\scriptscriptstyle{Q}}\langle k|\rho |k+n+\phi  \rangle_{\scriptscriptstyle{Q}}$ for $k\in \R$, $n\in \Z$ do not interact dynamically for different  $\phi\in [-\frac{1}{2},\frac{1}{2})$.  This holds also with  the kets $|k \rangle_{\scriptscriptstyle{Q}}$ replaced by the standard momentum kets $|k\rangle$.  
  Mathematically, it is easiest to discuss the invariant spaces for the adjoint semigroup  $\Phi_{\lambda,t}^{*}$. Let $\mathcal{A}_{\mathbb{T}}\subset \mathcal{B}\big(L^{2}(\R)\big)$ be the algebra of all bounded operators commuting with $e^{\textup{i} 2\pi P }$.  
By the analogous covariance property~(\ref{2PIShift}) for $\Phi_{\lambda,t}^{*}$,  it follows that $\Phi_{t}^{*}(\mathcal{A}_{\mathbb{T}})\subset \mathcal{A}_{\mathbb{T}}$.  More generally, the Banach spaces $e^{\textup{i}\phi X }\mathcal{A}_{\mathbb{T}}$ for $\phi\in [-\frac{1}{2},\frac{1}{2})$ will also be invariant by the Weyl commutation relation $$e^{\textup{i}2\pi  P }e^{\textup{i}\phi X } e^{-\textup{i} 2\pi P }= e^{\textup{i}2\pi \phi } e^{\textup{i}\phi X }.$$

\section{Dispersion without the Dirac comb }\label{AppendixNoComb}

The elementary lemma below implies that the spatial dispersion for the particle scales as $t^{\frac{3}{2}}$ for times $t\gg 1$ when the Dirac comb is not present.  This scaling agrees with the classical case.

\begin{lemma}
Let $\rho_{\lambda,t}\in \mathcal{B}_{1}\big(L^{2}(\R)\big) $ satisfy the Lindblad equation~(\ref{TheModel}) with $\alpha=0$.  Also let  $Q_{\lambda,t}(x):=\langle x| \rho_{\lambda,t}|x\rangle  $ be the position distribution for the particle.   As $t\rightarrow \infty$, the renormalized density $  t^{\frac{3}{2}} Q_{\lambda,t}(t^{\frac{3}{2}}x)$ converges in law to a Gaussian with variance $\frac{4\sigma}{3\lambda^{2}}$.

\end{lemma}

\begin{proof}
I will show pointwise convergence as $t\rightarrow \infty$ for the characteristic functions $\varphi_{\lambda,t}$ of  $  t^{\frac{3}{2}} Q_{\lambda,t}(t^{\frac{3}{2}}x)$.  The function $\varphi_{\lambda,t}$ can be written in terms of a trace formula involving $ \rho_{\lambda,t} $ by the following:
$$\varphi_{\lambda,t}(v):= \int_{\R}dx\, t^{\frac{3}{2}} Q_{\lambda,t}(t^{\frac{3}{2}}x)e^{\textup{i}vx }=  \int_{\R}dx\,  Q_{\lambda,t}(x)e^{\textup{i}t^{-\frac{3}{2}}v x }=\Tr\big[e^{\textup{i}t^{-\frac{3}{2}}v X}  \rho_{\lambda,t}  \big],  $$
where $X$ is the position operator.  The right side is a special case of the quantum characteristic function for the matrix $\rho_{\lambda, t}$.  The quantum characteristic function has a closed factored form  given by
$$ \Tr\big[e^{\textup{i}v X+\textup{i}q P}  \rho_{\lambda,t}  \big]= \Tr\big[ \Phi_{\lambda,t}^{*}\big(e^{\textup{i}v X+\textup{i}q P} \big) \rho  \big] =e^{\int_{0}^{t}dr\big( \phi (q+ \frac{2}{\lambda}(t-r)v  ) -\phi(0) \big)  } \Tr\big[e^{\textup{i}v X+\textup{i}q P}  \rho \big],$$
where $v,q\in \R$ and $\phi(q):=\int_{\R}dvj(v)e^{\textup{i} vq}   $.  The second equality above can be seen through the expressions for $\Phi_{\lambda,t}$ in Lem.~\ref{LindbladTech},  Weyl's intertwining relations, and the relation  $e^{\frac{\textup{i}t}{\lambda}P^{2}}f(X)e^{-\frac{\textup{i}t}{\lambda}P^{2} }  = f(X+\frac{2t}{\lambda}P)$ for bounded functions $f:\R\rightarrow \C$.  
 With the above formula, 
\begin{align*}
\varphi_{\lambda,t}(v)=e^{\int_{0}^{t}dr\big( \phi (t^{-\frac{3}{2}} \frac{2}{\lambda}(t-r)v  ) -\phi(0) \big)  } \Tr\big[e^{\textup{i}t^{-\frac{3}{2}}v X}  \rho \big]\longrightarrow e^{-\frac{2\sigma}{3\lambda^{2}}v^{2}},
\end{align*}
since $\phi'(0)=0$ and $\phi''(0)=\sigma $.  Hence, $  t^{\frac{3}{2}} Q_{\lambda,t}(t^{\frac{3}{2}}x)$ converges in distribution to a Gaussian with variance $\frac{4\sigma}{3\lambda^{2}}$.

\end{proof}

\end{appendix}

\end{document}